%% file: main.tex
\documentclass[prodmode,acmtods]{acmsmall}
\usepackage[ruled]{algorithm2e}
\usepackage{graphicx}
\usepackage{balance}  
\usepackage{epsfig}
\usepackage{subfigure}
\usepackage{amsmath}
\usepackage{float}
\usepackage{amssymb}
\usepackage{latexsym}
\usepackage{epic}
\usepackage{url}
\usepackage[english]{babel}
\usepackage{epstopdf}

\usepackage[noend]{algorithmic}
\usepackage{tabularx}
\usepackage{times}
\usepackage{dsfont}
\usepackage{amssymb}
\usepackage{multicol}
\usepackage{blindtext}
\usepackage{afterpage}

\usepackage{etex}
\reserveinserts{18}
\usepackage{morefloats}

\renewcommand{\algorithmcfname}{ALGORITHM}
\SetAlFnt{\small}
\SetAlCapFnt{\small}
\SetAlCapNameFnt{\small}
\SetAlCapHSkip{0pt}
\IncMargin{-\parindent}
\newcommand{\eat}[1]{}

\newcommand{\noi}{\noindent}
\newcommand{\Vol}{\mathrm{Vol}}

\long\def\comment#1{}
\begin{document}

\markboth{G. Yuan et al.}{Optimizing Batch Linear Queries under Exact and Approximate Differential Privacy}

\title{Optimizing Batch Linear Queries under Exact and Approximate Differential Privacy}
\author{
Ganzhao Yuan \affil{South China University of Technology, China}
Zhenjie Zhang \affil{Advanced Digital Sciences Center, Singapore}
Marianne Winslett \affil{Advanced Digital Sciences Center, Singapore; University of Illinois at Urbana-Champaign, USA}
Xiaokui Xiao \affil{Nanyang Technological University, Singapore}
Yin Yang \affil{Hamad Bin Khalifa University, Qatar}
Zhifeng Hao \affil{South China University of Technology, China; Guangdong University of Technology, China}
}

\begin{abstract}
Differential privacy is a promising privacy-preserving paradigm for
statistical query processing over sensitive data. It works by
injecting random noise into each query result, such that it is
provably hard for the adversary to infer the presence or absence of
any individual record from the published noisy results. The main
objective in differentially private query processing is to maximize
the accuracy of the query results, while satisfying the privacy
guarantees. Previous work, notably \cite{LHR+10}, has suggested that with an appropriate strategy, processing a batch of correlated
queries as a whole achieves considerably higher accuracy
than answering them individually. However, to our knowledge there is currently no practical solution to find such a strategy for an arbitrary query batch; existing methods either return strategies of poor quality (often worse than naive methods) or require prohibitively expensive computations for even moderately large domains. Motivated by this, we propose low-rank mechanism (LRM), the first practical differentially private technique for answering batch linear queries with high accuracy. LRM works for both exact (i.e., $\epsilon$-) and approximate (i.e., ($\epsilon$, $\delta$)-) differential privacy definitions. We derive the utility guarantees of LRM, and provide guidance on how to set the privacy parameters given the user's utility expectation. Extensive experiments using real data demonstrate that our proposed method consistently outperforms state-of-the-art query processing solutions under differential
privacy, by large margins.

\end{abstract}




\acmformat{G. Yuan, Z. Zhang, M. Winslett, X. Xiao, Y. Yang, Z. Hao. Optimizing Batch Linear Queries under Exact and Approximate Differential Privacy}

\begin{bottomstuff}
Yuan is supported by NSF-61402182. Zhang, Winslett, Xiao are supported by SERC 102-158-0074 from Singapore's A*STAR. Xiao is also supported by SUG Grant M58020016
from Nanyang Technological University and AcRF Tier 2 Grant ARC19/14 from Ministry of Education, Singapore. Hao is supported by NSF-(61100148, 61202269, 61472089), and Key Technology Research and Development Programs of Guangdong Province (2010B050400011,2013B051000076). We thank Li and Miklau for sharing us code for the Adaptive Mechanism, and the anonymous referees for their valuable comments.

Author's addresses:
Yuan (corresponding author), Department of Mathematics, South China University of Technology, Guangzhou, China, ecgzhyuan@scut.edu.cn.
Zhang, Advanced Digital Sciences Center, Singapore, zhenjie@adsc.com.sg.
Winslett, Advanced Digital Sciences Center, Singapore; Department of Computer Science, University of Illinois at Urbana-Champaign, IL, USA, winslett@illinois.edu.
Xiao, School of Computer Engineering, Nanyang Technological University, Singapore, xkxiao@ntu.edu.sg.
Yang, Division of Information and Communication Technologies, College of Science, Engineering and Technology, Hamad Bin Khalifa University, Qatar, yyang@qf.org.qa.
Hao, Faculty of Computer Science, Guangdong University of Technology, China; School of Computer Science and Engineering, South China University of Technology, China, mazfhao@scut.edu.cn.

\end{bottomstuff}

\maketitle
\newpage
\input{ intro }
\input{ related }
\input{ prelim }

\input{ decomp }
\input{ opt }
\input{ ext }
\input{ exp }
\input{ conc }
%
\bibliographystyle{acmsmall}
\bibliography{LowRankDP}

\appendix
\clearpage
\input{ app }

\input{ app2 }

\end{document}

%% file: intro.tex
\section{Introduction}\label{sec:intro} 
Differential privacy \cite{DMNS06} is an emerging
paradigm for publishing statistical information over sensitive data,
with strong and rigorous guarantees on individuals' privacy. Since
its proposal, differential privacy has attracted extensive research
efforts, such as in cryptography \cite{DMNS06}, algorithms
\cite{DRV10,HT10,MT07}, database management \cite{DWHL11,HRMS10,LHR+10,RN10,XBHG11,XWG10,PYZWY13}, data mining
\cite{BLS+10,FS10}, social network analysis \cite{rastogi2009relationship,hay2009accurate,sala2011sharing} and machine learning \cite{BLR08,CMS11,RBHT09}.
The main idea of differential privacy is to inject random noise into
aggregate query results, such that the adversary cannot infer, with
high confidence, the presence or absence of any given record $r$ in
the dataset, even if the adversary knows all other records in the
dataset besides $r$. The adversary's maximum confidence in inferring private
information is controlled by a user-specified parameter $\epsilon$, called the \emph{privacy budget}. Given $\epsilon$, the main goal of query processing under differential privacy is to
maximize the utility/accuracy of the (noisy) query answers, while
satisfying the above privacy requirements.

This work focuses on a common class of queries called \emph{linear
counting queries}, which is the basic operation in many statistical
analyses. Similar ideas apply to other types of linear queries,
e.g., linear sums. Figure \ref{fig:intro-example}(a) illustrates an
example electronic medical record database, where each record
corresponds to an individual. Figure \ref{fig:intro-example}(b)
shows the exact number of HIV+ patients in each state, which we
refer to as \emph{unit counts}. A linear counting query in this
example can be any linear combination of the unit counts. For
instance, let $x_{NY}$, $x_{NJ}$, $x_{CA}$, $x_{WA}$ be the patient
counts in states NY, NJ, CA, and WA respectively; one possible
linear counting query is $x_{NY}+x_{NJ}+x_{CA}+x_{WA}$, which
computes the total number of HIV+ patients in the four states listed
in our example. Another example linear counting query is
$x_{NY}/19+x_{NJ}/8+x_{CA}/37$, which calculates the weighted
average of patient counts in states NY, NJ and CA, with weights set
according to their respective population sizes. In general, we are
given a database with $n$ unit counts, and a batch $QS$ of $m$
linear counting queries. The goal is to answer all queries in $QS$
under differential privacy, and maximize the expected
overall accuracy of the queries.

\begin{figure} [htb]
\centering
\begin{tabular}{cc}
{\includegraphics[height=0.2\textwidth]{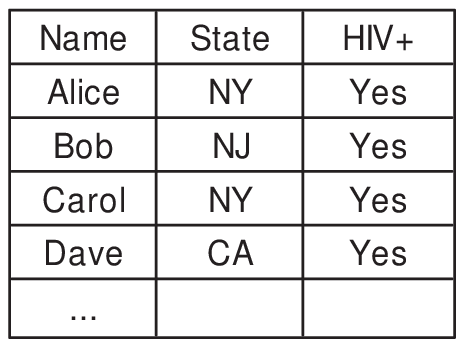}} &
{\includegraphics[height=0.2\textwidth]{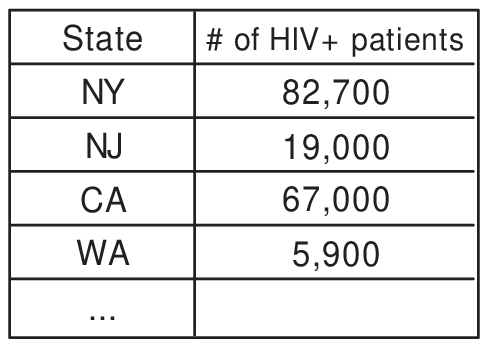}} \\
(a) Patient records & (b) Statistics on HIV+ patients
\end{tabular}
\caption{Example medical record database}
\label{fig:intro-example}
\end{figure}

Straightforward approaches to answering a batch of linear counting queries usually lead to sub-optimal result accuracy. Consider processing the query set $Q=\{q_1, q_2, q_3\}$ under the $\epsilon$-differential privacy definition, detailed in Section 3. One naive solution, referred to as \emph{noise on result} (\emph{NOR}), is to process each query independently, e.g., using the Laplace mechanism \cite{DMNS06}. This method fails to exploit the \emph{correlations} between different queries. Consider a batch of three different queries $q_1=x_{NY}+x_{NJ}+x_{CA}+x_{WA}$, $q_2=x_{NY}+x_{NJ}$, $q_3=x_{CA}+x_{WA}$. Clearly, the three queries are correlated since $q_1=q_2+q_3$. Thus, an alternative strategy for answering these queries is to process only $q_2$ and $q_3$, and use their sum to answer $q_1$. As will be explained in Section 3, the amount of noise added to query results depends upon the \emph{sensitivity} of the query set, which is defined as the maximum possible total change in query results caused by adding or removing a single record in the original database. Under $\epsilon$-differential privacy, the sensitivity of the query set $\{q_2, q_3\}$ is 1, because adding/removing a patient record in Figure \ref{fig:intro-example}(a) affects at most one of $q_2$ and $q_3$ (i.e., $q_2$ if the record is associated with state NY or NJ, and $q_3$ if the state is CA or WA), by exactly 1. On the other hand, the query set $\{q_1, q_2, q_3\}$ has a sensitivity of $2$ (under the $\epsilon$-differential privacy definition), since a record in the above 4 states affects both $q_1$ and one of $q_2$ and $q_3$. According to the Laplace mechanism, the variance of the added noise to each query is $2\Delta^2/\epsilon^2$, where $\Delta$ is the sensitivity of the query set, and $\epsilon$ is the user-specified privacy budget. Therefore, processing $\{q_1, q_2, q_3\}$ directly incurs a noise variance of $(2\times 2^2)/\epsilon^2$ for each query; on the other hand, executing $\{q_2, q_3\}$ leads to a noise variance of $(2\times 1^2)/\epsilon^2$ for each of $q_2$ and $q_3$, and their sum $q_1=q_2+q_3$ has a noise variance of $(2\times2)/\epsilon^2=4/\epsilon^2$. Clearly, the latter method obtains higher accuracy for all queries.

Another simple solution, referred to as \emph{noise on data} (\emph{NOD}), is to process each unit count under differential privacy, and combine them to answer the given linear counting queries. Continuing the example, this method computes the noisy counts for $x_{NY}$, $x_{NJ}$, $x_{CA}$ and $x_{WA}$, and uses their linear combinations to answer $q_1$, $q_2$, and $q_3$. This approach overlooks the correlations between different unit counts. In our example, $x_{NY}$ and $x_{NJ}$ (and similarly, $x_{CA}$ and $x_{WA}$) are either both present or both absent in every query, and, thus, can be seen as a single entity. Processing them as independent queries incurs unnecessary accuracy costs when re-combining them. In the example, NOD adds noise with variance $2/\epsilon^2$ to each unit count, and their combinations to answer $q_1$, $q_2$, and $q_3$ have noise variance $8/\epsilon^2$, $4/\epsilon^2$ and $4/\epsilon^2$, respectively. NOD's result utility is also worse than the above-mentioned strategy of processing $q_2$ and $q_3$, and adding their results to answer $q_1$.

In general, the query set $Q$ may exhibit complex correlations among different queries and among different unit counts. As a consequence, it is non-trivial to obtain the best strategy to answer $Q$ under differential privacy. For instance, consider the following query set:
\begin{eqnarray*}
q_1&=&2x_{NJ}+x_{CA}+x_{WA}\\
q_2&=&x_{NJ}+2x_{WA}\\
q_3&=&x_{NY}+2x_{CA}+2x_{WA}
\end{eqnarray*}

NOR is clearly a poor choice, since it incurs a sensitivity of 5 under the $\epsilon$-differential privacy definition
(e.g., a record of state WA affects $q_1$ by 1, and $q_2$ and $q_3$
by 2 each). The sensitivity of NOD remains 1, and it answers $q_1$,
$q_2$, and $q_3$ with noise variance $2\times(2^2+1^2+1^2)/\epsilon^2$,
$2\times(1^2+2^2)/\epsilon^2$ and $2\times(1^2+2^2+2^2)/\epsilon^2$ respectively, leading to a
sum-square error (SSE) of $40/\epsilon^2$. The optimal strategy in
terms of SSE in this case computes the noisy results of $q'_1 = x_{NY}/8 + x_{WA}$, $q'_2= -{3}x_{NY}/8 - x_{CA}$ and $q'_3 = x_{NY}/4 - x_{NJ}$.
Then, it obtains the results for $q_1$, $q_2$, and $q_3$ as follows.
\begin{eqnarray*}
q_1&=&q'_1 - q'_2 - 2q'_3\\
q_2&=&2q'_1 - q'_3\\
q_3&=&2q'_1 - 2q'_2
\end{eqnarray*}


The sensitivity of the above method is also 1, because (i) adding/removing a record of state NJ, CA and WA can only affect queries $q'_3$, $q'_2$ and $q'_1$, respectively, by at most 1; (ii) adding/removing a record of state NY causes the results of $q'_1$, $q'_2$ and $q'_3$ to change by at most 1/8, 3/8, and 1/4, respectively, leading to a maximum total change of 1/8+3/8+1/4=1. We introduce the formal definition of sensitivity later in Section \ref{sec:pre}. Hence, independent random noise of variance $2 \times 1^2 / \epsilon^2  = 2 / \epsilon^2$ is injected to the results of each of $q'_1$, $q'_2$ and $q'_3$. Their combination $q_1=q'_1 - q'_2 - 2q'_3$ thus has a noise variance of $2\times(1^2+(-1)^2+(-2)^2)/\epsilon^2 = 12 /\epsilon^2$. Similarly, combining $q'_1-q'_3$ to answer $q_2$ and $q_3$ as above incur a noise variance of $2\times(2^2+(-1)^2)/\epsilon^2 = 10/\epsilon^2$ and $2\times(2^2+(-2)^2)/\epsilon^2 = 16/\epsilon^2$ respectively. The SSE for queries $q_1-q_3$ is thus $12 /\epsilon^2 + 10 /\epsilon^2 + 16 /\epsilon^2 = 38 /\epsilon^2$.

Observe that the there is no simple pattern in
the query set or the optimal strategy. Since there is an infinite
space of possible strategies, searching for the best one is a
challenging problem.

Li et al. \cite{LHR+10} first formalize the above observations (i.e., answering a correlated query set with an effective strategy) into the \emph{matrix mechanism}. However, as we explain in Section \ref{sec:related:matrix}, the original matrix mechanism lacks a practical implementation, because the solutions in \cite{LHR+10} for finding a good strategy are either inefficient (which incur prohibitively high computational costs for even moderately large domains), or ineffective (which rarely obtain strategies that outperform naive methods NOD/NOR). Later, Li and Miklau \cite{li2012adaptive} propose the adaptive mechanism, which can be seen as an implementation of the matrix mechanism. This method, however, still incurs some drawbacks as discussed in Section \ref{sec:related:matrix}, which limit its accuracy. Motivated by this, we
propose the first practical realization of the matrix mechanism, called the \emph{low-rank mechanism} (\emph{LRM}), based on the theory of low-rank matrix approximation. LRM applies to both $\epsilon$-differential privacy and ($\epsilon$, $\delta$)-differential privacy, two most commonly used differential privacy definitions today. We analyze the utility of LRM under ($\xi$, $\eta$)-usefulness \cite{BLR08}, a popular utility measure. Extensive experiments demonstrate that LRM significantly outperforms existing solutions in terms of result accuracy, sometimes by orders of magnitude.


The rest of the paper is organized as follows. Section
\ref{sec:related} reviews previous studies on differential privacy.
Section \ref{sec:pre} provides formal definitions for our problem.
Section \ref{sec:formulation} presents the mechanism formulation of
LRM under $\epsilon$-differential privacy. Section \ref{sec:opt} discusses
how to solve the optimization problem in LRM. Section \ref{sec:approx} extends LRM to answer queries under $(\epsilon,\delta)$-differential privacy. Section \ref{sec:exp}
verifies the superiority of our proposal through an extensive
experimental study. Finally, Section \ref{sec:concl} concludes the
paper.


%% file: related.tex
\section{Related Work}\label{sec:related}

Section \ref{sec:related:general} surveys general-purpose mechanisms for enforcing differential privacy. Section \ref{sec:related:matrix} presents two methods that are closely related to the proposed solution, namely the matrix mechanism and the adaptive mechanism.

\subsection{Differential Privacy Mechanisms}\label{sec:related:general}

Differential privacy was first formally presented in \cite{DMNS06},
though some previous studies have informally used similar models,
e.g., \cite{DN03}. The Laplace mechanism \cite{DMNS06} is the first generic mechanism for enforcing differential privacy, which works
when the output domain is a multi-dimensional Euclidean space. McSherry and Talwar \cite{MT07} propose the exponential mechanism, which applies to
any problem with a measurable output space. The generality of the exponential mechanism makes it an important tool in the design of many other
differentially private algorithms, e.g., \cite{CPS+12,XZXYY12,XZXYY13,MT07}.


The original definition of differential privacy is $\epsilon$-differential privacy, which focuses on providing a strong and rigorous definition of privacy. Besides this, another popular definition is ($\epsilon$, $\delta$)-differential privacy, which can be seen as an approximate version of $\epsilon$-differential privacy. In many applications, ($\epsilon$, $\delta$)-differential privacy provides a similarly strong privacy definition, while enabling simpler and/or more accurate algorithms. One basic mechanism for enforcing ($\epsilon$, $\delta$)-differential privacy is the Gaussian mechanism, which injects Gaussian noise to the query results calibrated to the $\mathcal{L}_2$ sensitivity of the queries \cite{DworkKMMN06}. \cite{hardt2012beating} employ $k$ Gaussian measurements strategy to compute the low rank approximations of large matrices. However, ($\epsilon$, $\delta$)-differential privacy might be unsatisfactory in certain situations. For example, \cite{de2012lower} demonstrate that ($\epsilon$, $\delta$)-differential privacy is weaker than $\epsilon$-differential privacy in terms of mutual information even when $\delta$ is negligible. The proposed solution applies to both definitions of differential privacy. We present details of these two privacy definitions in Section \ref{sec:pre}.

Linear query processing is of particular interest in both the theory and database communities, due to its wide range of applications. To minimize the error of linear queries under differential privacy requirements, several methods try to build a synopsis of the original database, such as Fourier transformations \cite{RN10}, wavelets \cite{XWG10} and hierarchical trees \cite{HRMS10}. The compressive mechanism \cite{LZMY11} reduces the amount of noise necessary to satisfy differential privacy, for datasets with a sparse representation. By publishing a noisy synopsis under $\epsilon$-differential privacy, these methods are capable of answering an arbitrary number of linear queries. However, most of these methods obtain good accuracy only when the query selection criterion is a continuous range; meanwhile, since these methods are not workload-aware, their performance for a specific workload tends to be sub-optimal.

Workload-aware algorithms address this problem, which optimize the overall accuracy of a set of given linear queries. This work falls into this category. Notable workload-aware methods include (i) Multiplicative Weights / Exponential Mechanism (MWEM) \cite{hardt2012simple},(ii) the Matrix Mechanism \cite{LHR+10} and (iii) the Adaptive Mechanism \cite{li2012adaptive}. MWEM publishes a synthetic dataset optimized towards the given linear query set. In particular, it provides a beautiful theoretical bound on the maximum error of the given queries, which grows sublinearly to the number of records in the dataset, and logarithmically with the number of queries. In practice, however, this bound tends to be loose as it is derived from worst-case scenarios. {Meanwhile, the target problem of MWEM is different from ours, as we focus on answering a given set of linear queries rather than publishing synthetic data. Nevertheless, MWEM can be applied to our problem, and we compare it against the proposed solution in the experiments.} The Matrix Mechanism and the Adaptive Mechanism share some common features as the proposed solution, and we explain them in detail in Section \ref{sec:related:matrix}. 2.2. It is worth mentioning that as our experiments shows, the proposed solution outperforms all previous methods in terms of overall error, on a variety of datasets and workload types.

Recently, \cite{nikolov2013geometry} proposes a workload decomposition method that injects \emph{correlated} Gaussian noise to the query results to satisfy ($\epsilon$, $\delta$)-differential privacy. They prove that their solution provides an $\mathcal{O}((\log m)^2)$ approximation to the optimal mechanism, where $m$ is the number of queries. However, this method is infeasible in practice, since it involves computing minimum enclosing ellipsoids (MEE), for which the current best algorithm takes $m^{\mathcal{O}(m)} n$ time, where $n$ is the number of unit counts. \cite{nikolov2013geometry} suggests using approximation method for computing MEE, e.g. Khachiyan's algorithm \cite{Todd20071731}. This approximation algorithm still takes high order polynomial time to converge, which makes it prohibitively expensive for practical applications.

Several theoretical studies have derived lower bounds for the noise level for processing linear queries under differential privacy \cite{DN03,HT10}. Notably, Dinur and Nissim \cite{DN03} prove that any perturbation mechanism with maximal noise of scale $\mathcal{O}(n)$ cannot possibly preserve personal privacy, if the adversary is allowed to ask all possible linear queries, and has exponential computation capacity. By reducing the computation capacity of the adversary to polynomial-bounded Turing machines, they show that an error scale $\Omega(\sqrt{n})$ is necessary to protect any individual' privacy. More recently, Hardt and Talwar \cite{HT10} have significantly
tightened the error lower bound for answering a batch of linear
queries under differential privacy. Given a batch of $m$ linear
queries, they prove that any $\epsilon$-differential privacy
mechanism leads to squared error of at least
$\Omega(\epsilon^{-2}m^3\Vol(W))$, where $\Vol(W)$ is the volume of
the convex body obtained by transforming the $\mathcal{L}_1$-unit
ball into $m$-dimensional space using the linear transformations in
the workload $W$. This paper extends their
analysis to low-rank workload matrices.

Another related line of research concerns answering queries \emph{interactively} under differential privacy. In this setting, the system process queries one at a time, without knowing any future query. Clearly, this problem is more difficult that the non-interactive setting described so far, where the system knows all queries in the workload in advance. Most notably, Hardt et al. propose the Private Multiplicative Weights Mechanism (PMWM) \cite{hardt2010multiplicative}, whose error is asymptotically optimal with respect to the number of queries answered. The MWEM method described above \cite{hardt2012simple} applies similar ideas to the non-interactive setting. Besides PMWM, Hardt et al. \cite{HT10} propose the $K$-norm Mechanism whose error level almost reaches the lower bound derived in the same paper. Roth et al. introduce the Median Mechanism \cite{RothR10} for answering arbitrary queries interactively. However, both the $K$-norm Mechanism and the Median Mechanism rely on uniform sampling in a high-dimensional convex body \cite{DyerFK91}, which theoretically takes polynomial time, but is usually too expensive to be applied in practice.

Besides linear queries, differential privacy is also applicable to more complex queries in various research areas, due to its strong privacy guarantee. In the field of data mining, Friedman and Schuster \cite{FS10} propose the first algorithm for building a decision tree under differential privacy. Mohammed et al. \cite{MCF+11} study the same problem, and propose an improved solution based on a generalization strategy coupled with the exponential mechanism. Ding et al. \cite{DWHL11} investigate the problem of differentially private data cube publication. They present a randomized materialized view selection algorithm, which reduces the overall error, and preserves data consistency.

In the database literature, a plethora of methods have been proposed to optimize the accuracy of differentially private query processing. A tutorial on database-related differential privacy technologies can be found in \cite{YZMWX12}. Cormode et al. \cite{CPS+12} investigate the problem of multi-dimensional indexing under differential privacy, with the novel idea of assigning different amounts of privacy budget to different levels of the index. Peng et al.
\cite{PYZWY12} propose the DP-tree, which obtains improved accurate for higher dimensional data. Xu et al. \cite{XZXYY12,XZXYY13} optimize the procedure of building a differentially private histogram, whose method combines dynamic programming for optimal histogram computation and the exponential mechanism. \cite{LiQSC12} study the problem of how to perform frequent itemset mining on transaction databases while satisfying differential privacy, with the novel approach of constructing a basis set and then using it to find the most frequent patterns.

In addition, differential privacy for modeling security in social networks has also received much attention in recent literature. \cite{rastogi2009relationship} considers answering subgraph counting queries in a social network. Their solution assumes a Bayesian adversary whose prior is drawn from a distribution. They compute a high probability upper bound on the local sensitivity of the data and then answer by adding noise proportional to that bound. \cite{hay2009accurate} shows how to privately approximate the degree distribution in the edge adjacency model of a graph. Also, \cite{sala2011sharing} develop a differentially private graph model based on \emph{dk-series} reconstruction. Their approach mainly extracts a graph's detailed structure into degree correlation statistics and inject noise into the resulting dataset and generates a synthetic graph.

Lastly, differential privacy is also becoming a hot topic in the machine
learning community, especially for learning tasks involving
sensitive information, e.g., medical records. In \cite{CMS11},
Chaudhuri et al. propose a generic differentially private learning
algorithm, which requires strong convexity of the objective
function. Rubinstein et al. \cite{RBHT09} study the problem of SVM
learning on sensitive data, and propose an algorithm to perturb the
kernel matrix with performance guarantees, when the gradient of the loss function
satisfies the Lipschitz continuity property. Zhang et al. propose functional mechanism and for a large class of
optimization-based analyses \cite{zhang2012functional}. Later, they propose the PrivGene framework, which combines genetic algorithms and an enhanced version of exponential mechanism for differentially private model fitting \cite{zhang2013privgene}. General differential privacy techniques have also been applied to real
systems, such as network trace analysis \cite{MM10} and private recommender systems \cite{MM09}.

\subsection{Matrix Mechanism and Adaptive Mechanism}\label{sec:related:matrix}

In the seminal work of \cite{LHR+10}, Li et al. propose the matrix mechanism (MM), which formalizes the intuition that a batch of correlated linear queries can be answered more accurately under $\epsilon$-differential privacy, by processing a different set of queries (called the \emph{strategy}) and combining their results. Specifically, given a workload of linear counting queries, MM first constructs a \emph{workload matrix} $W$ of size $m$$\times$$n$, where $m$ is the number of queries, and $n$ is the number of unit counts. The construction of the workload matrix is elaborated further in Section 3. After that, MM searches for a \emph{strategy matrix} $A$ of size $r$$\times$$n$, where $r$ is a positive integer. Intuitively, $A$ corresponds to another set of linear queries, such that every query in $W$ can be expressed as a linear combination of the queries in $A$. The matrix mechanism then answers the queries in $A$ under $\epsilon$-differential privacy, and subsequently uses their noisy results to answer queries in $W$.

The main challenge for applying the matrix mechanism to practical workloads is to identify an appropriate strategy matrix $A$. Ref. \cite{LHR+10} provides two algorithms for this purpose. The first, based on iteratively solving a pair of related semidefinite programs, incurs $\mathcal{O}(m^3n^3)$ computational overhead, which is prohibitively expensive even for moderately large values of $m$ and $n$. The second solution (called \emph{approximate matrix mechanism} (\emph{AMM})) computes an $\mathcal{L}_2$ approximation of the optimal strategy matrix $A$. This method, though faster than the first one, still requires high CPU costs and memory consumption, and scales poorly with the domain size and query set cardinality. In order to test the approximate matrix mechanism with large data and query sets in our experiments, we have devised an improved solution, which we call the \emph{exponential smoothing mechanism} (\emph{ESM}), based on the problem formulation of approximate matrix mechanism in \cite{LHR+10}. ESM is at least as accurate as the method in \cite{LHR+10}, and yet much more efficient. Hence, in our experiments we use ESM in place of AMM. Appendix \ref{sec:appedinex:esm} provides details of ESM.

There are, however, two main drawbacks of ESM (and also vanilla AMM). First, the $\mathcal{L}_2$ approximation of the optimal strategy matrix often has poor quality. In fact, due to this problem, in our experiments we found that under $\epsilon$-differential privacy, the accuracy of ESM is often no better than the naive solution NOD that injects noise directly into the unit counts. A second and more subtle problem is that the formulation of the optimization program in AMM involves matrix inverse operators, which can cause numerical instability when the final solution (i.e., the strategy matrix) is of low rank, as explained in Appendix \ref{sec:appedinex:esm}. The proposed low-rank mechanism avoids both problems, and achieves significantly higher result accuracy as shown in our experiments.

The idea of matrix mechanism naturally extends to ($\epsilon$, $\delta$)-differential privacy, using the Gaussian mechanism instead of the Laplace mechanism as the fundamental building block. In this case, the optimization program is defined using $\mathcal{L}_2$ form, and the AMM formulation is equivalent to that of MM, meaning that AMM and ESM now solve the exact optimization program. Hence, in theory, AMM can obtain optimal results. However, in practice, both ESM and the AMM implementation in \cite{LHR+10} often fail to converge to the optimal strategy matrix, due to numerical instability incurred by the matrix inverse operator in the AMM formulation.

Recently, \cite{li2012adaptive} Li et al. propose another implementation of AMM, called the adaptive mechanism (AM). For any given workload $W$, AM attempts to find the best strategy matrix by computing the optimal nonnegative weights for the eigenvectors of the workload matrix $W$. Since the strategy matrix may have one or more columns whose $\mathcal{L}_2$-norm are less than the sensitivity, they refine the strategy matrix by appending some completing columns to the candidate strategy matrix without raising the sensitivity. Therefore, this post-processing step can reduce the expected error. AM incurs two serious drawbacks. First, it involves solving a complicated semidefinite program, and it is not known whether their solution to the program converges to the optimal solution. Second and more importantly, such multistep strategy in AM does not offer any guarantee on optimality. The proposed method LRM is free from these problems, and obtains significantly better performance as we show in the experiments. Appendix \ref{sec:appedinex:am} provides details of AM.

%% file: prelim.tex
\section{Preliminaries}\label{sec:pre}

We focus on answering a batch of linear counting queries $Q=\{q_1,q_2,\ldots,q_m\}$ over a sensitive database $D$. Each query $q_i \in Q$ is a linear combination of \emph{unit counts} in the data domain, denoted as $x_1,x_2,\ldots,x_n$. In the example of Figure \ref{fig:intro-example}, the sensitive database $D$ contains records corresponding to individual HIV+ patients; each unit count is the number of such patients in a state of the US; each query in the example is a linear combination of these state-level patient counts. Our goal is to answer $Q$ with minimum overall error, while satisfying differential privacy. In particular, we consider two definitions of differential privacy, namely $\epsilon$-differential privacy (i.e., the original definition of differential privacy) and ($\epsilon$,$\delta$)-differential privacy (a popular formulation of approximate differential privacy). Our solutions use the Laplace mechanism (resp., the Gaussian mechanism) as a fundamental building block to enforce $\epsilon$- (resp., ($\epsilon$, $\delta$)-) differential privacy. In the following, Section 3.1 presents the definition of $\epsilon$-differential privacy and the Laplace mechanism. Section 3.2 covers ($\epsilon$, $\delta$)-differential privacy and the Gaussian mechanism. Section 3.3 describes naive approaches to answering a batch of linear counting queries. Section 3.4 explains important properties of low-rank matrices that are used in our solutions. Table \ref{tab:notation} summarizes frequently used notations throughout the paper.
\begin{table}[!t]
\tbl{Summary of frequent notations\label{tab:notation}}{%
\begin{tabular}{|c|c|}
\hline
Symbol & Meaning\\
\hline
\hline $D$ & input database \\
\hline $n$ & number of unit counts \\
\hline $Q$ & input query set\\
\hline $m$ & number of queries in $Q$ \\
\hline $W$ & workload matrix, i.e., the matrix representation of $Q$ \\
\hline $B,L$ & a decomposition of $W$ satisfying $W\approx B\cdot L$\\
\hline $s$ & rank of workload matrix $W$\\
\hline $r$ & number of columns in $B$ (also number of rows in $L$) \\
\hline $Q(D)$ & exact answer of $Q$ on database $D$ \\
\hline $\Delta(Q)$ & $\mathcal{L}_1$ sensitivity of $Q$\\
\hline $\Theta(Q)$ & $\mathcal{L}_2$ sensitivity of $Q$\\
\hline $\epsilon,\delta$ & privacy parameters \\
\hline $\xi,\eta$ & utility parameters \\
\hline $\kappa(W)$ & generalized condition number of matrix $W$\\
\hline $\rho(W)$ & $\rho$-coherence of matrix $W$\\
\hline $|||X|||_1$ & maximum absolute column sum of matrix $X$ \\
\hline $|||X|||_2$ & spectral norm, maximum singular value of matrix $X$ \\
\hline $|||X|||_{\infty}$ & maximum absolute row sum of matrix $X$ \\
\hline $\|X\|_{*}$ & nuclear norm, sum of the singular values of matrix $X$ \\
\hline $\|X\|_{F}$ & Frobenius norm, square root of the sum of squared elements of matrix $X$\\
\hline
\end{tabular}}
\end{table}

\subsection{$\epsilon$-Differential Privacy and the Laplace Mechanism}

The basic idea behind the privacy guarantee of differential privacy is the indistinguishability between \emph{neighbor databases}. Two databases $D$ and $D'$ are called neighbor databases, iff. $D'$ can be obtained by adding or removing exactly one record from $D$. In the example of Figure \ref{fig:intro-example}, a neighbor database can be obtained by removing an individual from the original data, or by adding another one. For linear counting queries, the essential difference between two neighbor databases $D$ and $D'$ is that they differ on exactly one unit count, by exactly one. Formally, let $\{x_1,x_2,\ldots,x_n\}$ be the set of unit counts corresponding to $D$ and $\{x'_1,x'_2,\ldots,x'_n\}$ be the unit counts for $D'$. Then, there exists an $i$, $1 \leq i \leq n$, such that $x_j = x'_j$ for all $j \neq i$, and $|x_i-x'_i|=1$.

Given a set of queries $Q$, a randomized mechanism $M$ for answering $Q$ satisfies $\epsilon$-differential privacy, iff. for every possible pair of
neighbor databases $D$ and $D'$, the following inequality holds:

\begin{equation}
\forall R:~ \Pr(M(Q,D)= R)\leq e^{\epsilon} \Pr(M(Q,D')= R)
\end{equation}
where $R$ is any possible output of $M$, and $M(Q, D)$ (resp. $M(Q, D')$) is the output of $M$ given query set $Q$ and input database $D$ (resp., $D'$). This inequality indicates that given an output $R$ of $M$, the adversary can only have limited confidence for inferring whether the input database is $D$ or $D'$, regardless of his/her background knowledge. Since $D$ and $D'$ can be any two neighbor databases that differ in any record, the above inequality also limits the adversary's confidence for inferring the presence or absence of a record in the input database; hence, it provides plausible deniablity to any individual involved in the sensitive data.

The Laplace mechanism \cite{DMNS06} is a fundamental solution for enforcing $\epsilon$-differential privacy, based on the concept of \emph{$\mathcal{L}_1$ sensitivity}. Given a query set $Q$, its $\mathcal{L}_1$ sensitivity $\Delta(Q)$ is the maximum $\mathcal{L}_1$ distance between the exact results of $Q$ on any pair of neighbor databases $D$ and $D'$, Formally, we have:

\begin{equation}
\Delta(Q)=\max_{D,D'}\|Q(D),Q(D')\|_1
\end{equation}

Note that in the above equation, $D$ and $D'$ can be any pair of neighbor databases. Hence, $\Delta(Q)$ is a property of the query set $Q$ and the data domain, and it does not depend upon the actual sensitive data $D$. In the example of Figure \ref{fig:intro-example}, the $\mathcal{L}_1$ sensitivity of a single query $q_1=x_{NY}+x_{NJ}+x_{CA}+x_{WA}$ is 1, because any two neighbor databases $D$ and $D'$ differ on only one unit count (which can be one of $x_{NY}$, $x_{NJ}$, $x_{CA}$ or $x_{WA}$) by exactly 1. If we include $q_2=x_{NY}+x_{NJ}$ and $q_3=x_{CA}+x_{WA}$ to the query set $Q$, the $\mathcal{L}_1$ sensitivity of $Q=\{q_1, q_2, q_3\}$ is 2, because a change of 1 on any of $x_{NY}$, $x_{NJ}$, $x_{CA}$ or $x_{WA}$ affects the result of $q_1$ by 1, and either one (but not both) of $q_2$ and $q_3$ by 1, leading to a $\mathcal{L}_1$ distance of 2.

Given a database $D$ and a query set $Q$, the Laplace mechanism (denoted as $M_{Lap}$) outputs a randomized result set $R$ that follow the Laplace distribution with mean $Q(D)$ and scale $\frac{\Delta(Q)}{\epsilon}$, i.e.,

\begin{equation}
\Pr(M_{Lap}(Q,D)=R)\propto \exp\left(\frac{\epsilon}{\Delta(Q)}\|R-Q(D)\|_1\right)
\end{equation}

This is equivalent to adding independent Laplace noise to the exact result of each query in $Q$, i.e., $M(Q, D) = Q(D)+ Lap\left(\frac{\Delta(Q)}{\epsilon}\right)^m$, where $m$ is the number of queries in $Q$, and $Lap\left(\frac{\Delta(Q)}{\epsilon}\right)$ is a random variable
following zero-mean Laplace distribution with scale
$\lambda=\frac{\Delta(Q)}{\epsilon}$. The probability density function
of zero-mean Laplace distribution is:
\begin{equation}\label{pdf:laplace}
f(x) = \frac{1}{2\lambda} \exp\left(-\frac{\|x\|_1}{\lambda}\right)
\end{equation}

According to properties of the Laplace distribution, the variance of $Lap(\lambda)$ is $2\lambda^2 = \frac{2\Delta(Q)^2}{\epsilon^2}$. Since the Laplace noise injected to each of the $m$ query results is independent, the overall expected squared error of the query answers obtained by the Laplace mechanism is $\frac{2m\Delta(Q)^2}{\epsilon^2}$. In our running example in Figure \ref{fig:intro-example}, to answer the query set $Q=\{q_1=x_{NY}+x_{NJ}+x_{CA}+x_{WA}, q_2=x_{NY}+x_{NJ}, q_3=x_{CA}+x_{WA}\}$ under $\epsilon$-differential privacy, a direct application of the Laplace mechanism injects independent, zero-mean Laplace noise of scale $\frac{2}{\epsilon}$ to the exact result of each of $q_1$, $q_2$ and $q_3$, since the $\mathcal{L}_1$ sensitivity for this set of queries is 2, as discussed in Section \ref{sec:intro}. The overall squared error for $Q$ is thus $\frac{2 \times 3 \times 2^2}{\epsilon^2} = \frac{24}{\epsilon^2}$.

\subsection{($\epsilon$, $\delta$)-Differential Privacy and the Gaussian Mechanism}

$\epsilon$-differential privacy can be difficult to enforce, especially for queries with high $\mathcal{L}_1$ sensitivity, or those whose $\mathcal{L}_1$ sensitivity is difficult to analyze. Hence, relaxed versions of $\epsilon$-differential privacy have been studied in the past, among which a popular definition is the ($\epsilon$, $\delta$)-differential privacy, also called approximate differential privacy. This definition involves an additional parameter $\delta$, which is a non-negative real number controlling how closely this definition approximates $\epsilon$-differential privacy. Formally, let $Range(M)$ be the set of all possible outputs of a mechanism $M$. A randomized mechanism $M$ satisfies ($\epsilon$, $\delta$)-differential privacy, iff. for any two neighbor databases $D$ and $D'$, the following holds:

\begin{equation}
\forall \textbf{R} \subseteq Range(M):~ \Pr(M(Q,D) \in \textbf{R})\leq e^{\epsilon} \Pr(M(Q,D') \in \textbf{R}) + \delta
\end{equation}
where $\textbf{R}$ is any set of possible results of $M$. It can be derived that when $\delta=0$, ($\epsilon$, $\delta$)-differential privacy is equivalent to $\epsilon$-differential privacy. Accordingly, since $\delta$ is non-negative, any mechanism that satisfies $\epsilon$-differential privacy also satisfies ($\epsilon$, $\delta$)-differential privacy for any value of $\delta$. When $\delta>0$, ($\epsilon$, $\delta$)-differential privacy relaxes $\epsilon$-differential privacy by ignoring outputs of $M$ with very small probability (controlled by parameter $\delta$). In other words, an ($\epsilon$, $\delta$)-differentially private mechanism satisfies $\epsilon$-differential privacy with a probability controlled by $\delta$.

A basic mechanism for enforcing ($\epsilon$, $\delta$)-differential privacy is the \emph{Gaussian mechanism} \cite{dwork2006our}, which involves the concept of \emph{$\mathcal{L}_2$ sensitivity}. For any two neighbor databases $D$ and $D'$, the $\mathcal{L}_2$ sensitivity $\Theta(Q)$ of a query set $Q$ is defined as:

\begin{equation}\label{eq:L2sensitivity}
\Theta(Q) = \max_{D,D'} \|Q(D),Q(D')\|_2
\end{equation}

In the running example shown in Figure \ref{fig:intro-example}, the $\mathcal{L}_2$ sensitivity for the query set $Q=\{q_1=x_{NY}+x_{NJ}+x_{CA}+x_{WA}, q_2=x_{NY}+x_{NJ}, q_3=x_{CA}+x_{WA}\}$ is $\sqrt{2}$, since the exact results of $q_1$ (as well as one of $q_2$ and $q_3$) differ by at most 1 for any two neighbor databases, leading to an $\mathcal{L}_2$ sensitivity of $\sqrt{1^2+1^2}=\sqrt{2}$. Similar to $\mathcal{L}_1$ sensitivity, the $\mathcal{L}_2$ sensitivity $\Theta(Q)$ depends on the data domain $\mathds{D}$ and the query set $Q$, not the actual data. Given a database $D$ and a query set $Q$, the Gaussian mechanism (denoted by $M_{Gau}$) outputs a random result that follows the Gaussian distribution with mean $Q(D)$ and magnitude $\sigma = \frac{\Theta(Q)}{h(\epsilon,\delta)}$, where $h(\epsilon,\delta)=\frac{\epsilon}{\sqrt{8 \ln (2/ \delta )}}$. This is equivalent to adding $m$-dimensional independent Gaussian noise $Gau \left(\frac{\Theta(Q)}{h(\epsilon,\delta)}\right)^m$, in
which $Gau \left(\frac{\Theta(Q)}{h(\epsilon,\delta)}\right)$ is a random variable
following a zero-mean Gaussian distribution with scale
$\sigma=\frac{\Theta(Q)}{h(\epsilon,\delta)}$. The probability density function
of zero-mean Gaussian distribution is:

\begin{equation}\label{pdf:gaussian}
g(x) = \sqrt{ \frac{1}{2\pi \sigma^2}}
\exp \left(-\frac{\|x\|_2^2}{2\sigma^2}\right)
\end{equation}

According to properties of the Gaussian distribution, the variance of $Gau(\sigma)$ is $\sigma^2 = \frac{\Theta(Q)^2}{h(\epsilon,\delta)^2}$. Since independent Gaussian noise is injected to each of the $m$ query results, the total expected squared error for the query set is $\frac{m\Theta(Q)^2}{h(\epsilon,\delta)^2}$. In our running example in Figure \ref{fig:intro-example}, to answer the query set $Q=\{q_1=x_{NY}+x_{NJ}+x_{CA}+x_{WA}, q_2=x_{NY}+x_{NJ}, q_3=x_{CA}+x_{WA}\}$ under ($\epsilon$, $\delta$)-differential privacy, a direct application of the Gaussian mechanism injects independent, zero-mean Laplace noise of scale $\frac{\sqrt{2}}{h(\epsilon, \delta)}$ to the exact result of each of $q_1$, $q_2$ and $q_3$, since the $\mathcal{L}_2$ sensitivity for this set of queries is $\sqrt{2}$, according to Equation (\ref{eq:L2sensitivity}). The overall squared error for $Q$ is thus $\frac{3 \times (\sqrt{2})^2}{(h(\epsilon, \delta))^2} = \frac{48\ln(2/\delta)}{\epsilon^2}$.

\subsection{Naive Solutions for Answering a Batch of Linear Counting Queries}\label{sec:pre:query}

This paper focuses on answering a batch of linear counting queries, each of which is a linear combination of the unit counts of the input database $D$. Formally, given a weight vector $(w_1,w_2,\ldots,w_n)^T\in \mathbb{R}^{n}$, a linear counting query can be expressed as:

$$q(D)=w_1x_1+w_2x_2+\ldots+w_nx_n$$

We aim to answer a batch of $m$ linear queries, $Q=\{q_1,q_2,\ldots,q_m\}$. The query set $Q$
thus can be represented by a \emph{workload matrix} $W$ with $m$ rows
and $n$ columns. Each entry $W_{ij}$ in $W$ is the weight in query $q_i$ on the $j$-th unit count $x_j$. Since we do not use any other information of the input database $D$ besides the unit counts, in the following we abuse the notation by using $D$ to represent the vector of unit counts, i.e., $D=(x_1,x_2,\ldots,x_n)^T\in \mathbb{R}^{n}$. Hence, the query batch $Q$
can be answered by:
$$Q(D)=W D=\left(\sum_j W_{1j}x_j,\ldots,\sum_j W_{mj}x_j\right)^T\in\mathbb{R}^{m\times1}$$

Two naive solutions for enforcing differential privacy on a query batch are as follows.

\noindent\textbf{Noise on data (NOD). } The main idea of NOD is to add noise to each unit count. Then, the set of noisy unit counts are published, which can be used to answer any linear counting query. Because two neighbor databases differ on exactly one unit count, by exactly 1, both the $\mathcal{L}_1$ and the $\mathcal{L}_2$ sensitivity for the set of unit counts is 1, according to their respective definitions. NOD employs the Laplace mechanism to enforce $\epsilon$-differential privacy (or the Gaussian mechanism to enforce ($\epsilon$, $\delta$)-differential privacy) on the published unit counts, and then combines the noisy unit counts to answer the query batch $Q$. Let $M_{NOD, \epsilon}$ and $M_{NOD, (\epsilon, \delta)}$ denote the NOD mechanism for enforcing $\epsilon$-differential privacy and $(\epsilon, \delta)$-differential privacy, respectively. We have:

\begin{equation}
M_{NOD, \epsilon}(Q,D) = W\left(D+Lap\left(\frac{1}{\epsilon}\right)^n\right) \nonumber
\end{equation}

\begin{equation}\label{eqn:m_d:app}
M_{NOD, (\epsilon, \delta) (Q, D)}= W \left(D+ Gau \left( \frac{1}{ h(\epsilon,\delta)} \right)^n\right) \nonumber
\end{equation}
where $h(\epsilon,\delta)=\frac{\epsilon}{\sqrt{8 \ln (2/ \delta )}}$ as in the Gaussian mechanism.

Based on the analysis of the Laplace and Gaussian mechanisms, the expected squared error for $M_{NOD, \epsilon}$ and $M_{NOD, (\epsilon, \delta)}$ is
$\frac{2}{\epsilon^2}\sum_{i,j}W_{ij}^2$ and $\frac{1}{(h(\epsilon,\delta))^2}\sum_{i,j}W_{ij}^2$, respectively. For both privacy definitions, the error of NOD is
proportional to the squared sum of the entries in $W$.

\noindent\textbf{Noise on results (NOR). } NOR simply applies the Laplace mechanism (for $\epsilon$-differential privacy) or the Gaussian mechanism (for ($\epsilon$, $\delta$)-differential privacy) directly on the query set $Q$. Recall that each query $q_i \in Q$ is a linear combination of the unit counts, i.e., $q_i = \sum_j W_{ij}x_j$. Meanwhile, two neighbor databases differ on exactly one unit count, by exactly 1. Therefore, the sensitivity (both $\mathcal{L}_1$ and $\mathcal{L}_2$) of $q_i$ is $\max_j W_{ij}$, i.e., the maximum unit count weight in $q_i$. Regarding $Q$, its $\mathcal{L}_1$ sensitivity is $\Delta(Q)=\max_j \sum_i |W_{ij}|$, i.e., the highest column absolute sum \cite{LHR+10}. Similarly, its $\mathcal{L}_2$ sensitivity is $\Theta(Q) = \max_j \sqrt{\sum_i{W_{ij}^2}}$, i.e., the highest column $\mathcal{L}_2$ norm value\cite{LHR+10}. Thus, $M_{NOR, \epsilon}$ and $M_{NOR, (\epsilon, \delta)}$ output the following results.

\begin{equation}\label{eqn:m_r}
M_{NOR, \epsilon}(Q,D)= W D+Lap\left(\frac{\Delta(Q)}{\epsilon}\right)^m \nonumber
\end{equation}
\begin{equation}\label{eqn:m_r:app}
M_{NOR, (\epsilon, \delta) (Q, D)}= W  D+ Gau \left( \frac{\Theta(Q)}{h(\epsilon,\delta)} \right)^m \nonumber
\end{equation}
where $\Delta(Q)=\max_j \sum_i |W_{ij}|$, $\Theta(Q) = \max_j \sqrt{\sum_i{W_{ij}^2}}$, and $h(\epsilon,\delta)=\frac{\epsilon}{\sqrt{8 \ln (2/ \delta )}}$.

Similar to the analysis of the Laplace and the Gaussian mechanisms, the expected squared error of the $M_{NOR, \epsilon}$ on query $Q$ is $\frac{2m \Delta(Q)^2}{\epsilon^2} = \frac{2m \max_j \sum_i W_{ij}^2}{\epsilon^2}$, and the expected squared error of $M_{NOR, (\epsilon, \delta)}$ is $\frac{m\Theta(Q)^2}{h(\epsilon,\delta)^2} = \frac{m\max_j \sum_i
W^2_{ij}}{h(\epsilon,\delta)^2}$. An interesting observation is that under ($\epsilon$, $\delta$)-differential privacy, NOR obtains lower expected squared error than NOD, iff. $m\max_j\sum_i W^2_{ij}<\sum_{j}\sum_i W^2_{ij}$. Note that when $m\geq n$, this inequality can never hold, implying that NOR is more effective for when the number of queries $m$ is smaller than the number of unit counts $n$.

\subsection{Low-Rank Matrices and Matrix Norms}

The rank of a real-value matrix $W$ is the number of non-zero \emph{singular values} obtained by performing \emph{singular value decomposition} (SVD) of $W$. SVD decomposes $W$ of size $m \times n$ into the product of three matrices: $W=U \Sigma V$. $U$ and $V$ are row-wise and column-wise orthogonal matrices respectively, and $\Sigma$ is a diagonal matrix with positive real diagonal values, which are the singular values of $W$. Let $s$ be the number of such singular values, i.e., the rank of $W$. Then, Matrices $U$, $\Sigma$, and $V$ are of sizes $m\times s$, $s\times s$, and $s\times n$ respectively. SVD guarantees that $s \leq \min\{m,n\}$.

A matrix $W$ of size $m \times n$ whose rank is less than $\min\{m,n\}$ is called a \emph{low-rank matrix}. This happens when the rows and columns of $W$ are correlated. In the running example of Figure \ref{fig:intro-example}, the workload matrix corresponding to the query set $Q=\{q_1=x_{NY}+x_{NJ}+x_{CA}+x_{WA}, q_2=x_{NY}+x_{NJ}, q_3=x_{CA}+x_{WA}\}$ is a low-rank matrix, since the queries in $Q$ are correlated (i.e., $q_1=q_2+q_3$, and the unit counts are also correlated (e.g., $x_{NY}$ and $x_{NJ}$). The main idea of the proposed low-rank mechanism is to exploit the low-rank property of the workload matrix to reduce the necessary amount of noise required to satisfy differential privacy.

An important concept used in the proposed solution is the matrix norm, which is an extension of the notion of vector norms to matrices. Two common definitions of the matrix norm are: (i) Entrywise norm, which treats a matrix $W$ of size $m\times n$ simply as a vector of size $m\times n$ consisting of all entries of $W$, and applies one of the vector norm definitions. For example, applying the $\mathcal{L}_2$-norm to all entries in $W$ obtains $\|W\|_2= {(\sum_{i=1}^{m} \sum_{j=1}^{n} |W_{ij}|^2)}^{1/2}$, which is also called the \emph{Frobenius norm}, written as $\|W\|_F$. (ii) Induced norm (or Operator norm), defined by $|||W|||_p= \max_{x\neq 0}  \|Wx\|_p / \|x\|_p$, where $x$ is a vector of size $n$, and $\|x\|_p$ is the $\mathcal{L}_p$ norm of $x$. Notably, $|||W|||_1$ is simply the maximum absolute column sum of $W$, and $|||W|||_{\infty}$ is simply the maximum absolute row sum of the matrix $W$.

%% file: decomp.tex
\section{Workload Decomposition}\label{sec:formulation}

Recall that the example in Figure \ref{fig:intro-example} shows that sometimes it is best to answer a batch of linear counting queries $Q$ \emph{indirectly}, by first answering a set of intermediate linear counting queries under differential privacy, and combine their results to answer $Q$. The proposed low-rank mechanism (LRM) follows this idea. Specifically, given a workload matrix $W$ corresponding to the query set $Q$, LRM decomposes $W$ into the product of two matrices $W=BL$. $B$ is of size $m\times r$ and $L$ is of size $r\times n$. Here, $r$ is a parameter to be determined which specifies the number of intermediate queries; $L$ corresponds to the set of intermediate linear counting queries to answer under differential privacy; $B$ indicates how the results of these intermediate queries are combined to answer $Q$. The main challenge lies in how to choose the best decomposition that minimizes the overall error of $Q$, as there is a vast search space for possible decompositions. In this section, we model the search for the optimal matrix decomposition as a constrained optimization program, which is solved in the next section. For the ease of presentation, we focus on $\epsilon$-differential privacy in this and the next section, and defer the discussion of ($\epsilon$, $\delta$)-differential privacy until Section \ref{sec:approx}. In addition, we provide asymptotic error bounds for LRM in Appendix \ref{sec:appedinex:asymptotic:error:bound}.

In the following, Section \ref{sec:4.1} formalize LRM and the optimization program of workload decomposition. Section \ref{sec:utility} analyzes the result utility of LRM with the optimal workload decomposition, and discusses the selection of the privacy parameter $\epsilon$. Finally, Section \ref{sec:relaxation} presents a relaxed optimization program for workload decomposition which can further improve the accuracy of LRM for certain workloads.

\subsection{Optimization Program Formulation}\label{sec:4.1}

We first formalize LRM under $\epsilon$-differential privacy. Given $W$ and its decomposition $W=BL$, LRM first applies the Laplace mechanism to the intermediate queries specified by $L$. Let $\Delta(L)$ denote the $\mathcal{L}_1$ sensitivity of these intermediate queries. Similar to the case of NOR discussed in Section \ref{sec:pre:query}, $\Delta(L)$ is the maximum sum of absolute values of a column in $L$, which is:

$$\Delta(L)=\max_j \sum_i |L_{ij}|$$

Applying the Laplace mechanism, we obtain the noisy results of the intermediate queries:

$$LD+Lap\left(\frac{\Delta(L)}{\epsilon}\right)^r$$

where $D$ denotes the vector of unit counts. Next, LRM multiplies matrix $B$ with the noisy intermediate results, which essentially recombines the intermediate results to answer $Q$. Let $M_{LRM, \epsilon}(Q,D)$ denote LRM under $\epsilon$-differential privacy, we have:

\begin{equation}\label{eqn:part_mech}
M_{LRM, \epsilon}(Q,D)=B\left(LD+Lap\left(\frac{\Delta(L)}{\epsilon}\right)^r\right)
\end{equation}

Since $W=BL$, we have $Q(D)=WD=BLD$. Hence, the output $M_{LRM, \epsilon}(Q,D)$ can be seen as the sum of two components: $BLD$ and $B \cdot Lap\left(\frac{\Delta(L)}{\epsilon}\right)$. The former is the exact result of $Q$, and the latter is the noise added in order to satisfy differential privacy. Next we analyze the error of LRM. First we define the \emph{scale} of a decomposition, as follows.

\begin{definition}\label{def:scale}\textbf{Scale of a workload decomposition.} Given a workload decomposition $W=BL$, its scale $\Phi(B)$ is the squared sum of the entries in $B$, i.e., $\Phi(B)=\sum_{i,j}B_{ij}^2$.
\end{definition}

Meanwhile, we call $\Delta(L)$ the $\mathcal{L}_1$ \emph{sensitivity} of the decomposition $W=BL$. The following lemma shows that the expected squared error of LRM is linear to the scale of the decomposition, and quadratic to the $\mathcal{L}_1$ sensitivity of the decomposition.

\begin{lemma}\label{lem:decomp_error}
The expected squared error of $M_{LRM, \epsilon}(Q,D)$ using decomposition $W=BL$ is $\frac{2\Phi(B)\Delta(L)^2}{\epsilon^2}$.
\end{lemma}

\begin{proof}
According to Equation (\ref{eqn:part_mech}), $M_{LRM, \epsilon}(Q,D)-Q(D) = B\cdot
Lap\left(\frac{\Delta(L)}{\epsilon}\right)^r$. The expected squared error of the mechanism is thus $\left(\sum_{ij}B^2_{ij}\right)\frac{2(\Delta(L))^2}{\epsilon^2}$. Since
$\Phi(B)=\sum_{ij}B^2_{ij}$, the error can be rewritten as $\frac{2\Phi(B)(\Delta(L))^2}{\epsilon^2}$.
\end{proof}

Therefore, to find the best workload decomposition, it suffice to solve the optimal $B$ and $L$ that minimize $\Phi(B)\left(\Delta(L)\right)^2$, while satisfying $W=BL$. However, this optimization program is difficult to solve, because (i) the objective function involves the product of $\Phi(B)$ and the square of $\Delta(L)$, and (ii) $\Delta(L)$ may not be differentiable. To address this problem, we first prove an important property of workload decomposition, which implies that the exact value of $\Delta(L)$ is not important.

\begin{lemma}\label{lem:rescale}
Given a workload decomposition $W=BL$, we can always construct another decomposition $W=B'L'$ satisfying (i) $\Delta(L')=1$ and (ii) ($B'$, $L'$) lead to the same expected squared error of $M_{LRM, \epsilon}$ as ($B$, $L$), i.e.,
$$\Phi(B)\Delta(L)^2=\Phi(B')\left(\Delta(L')\right)^2=\Phi(B')$$
\end{lemma}

\begin{proof}
We obtain $B'$ and $L'$ by $B'=\Delta(L) B$, $L' = \frac{1}{\Delta(L)}L$. Based on the definition of $\mathcal{L}_1$ sensitivity, we have
$$\Delta(L')=\max_j\sum_i |L'_{ij}|=\max_j\sum_i\left|\frac{L_{ij}}{\Delta(L)}\right|=\frac{1}{\Delta(L)}\Delta(L)=1$$
Meanwhile, according to Definition \ref{def:scale}, we have:
$$\Phi(B')=\sum_{ij}(B'_{ij})^2=\sum_{ij}\Delta(L)^2(B_{ij})^2=\Phi(B)\Delta(L)^2$$
This leads to the conclusion of the lemma.
\end{proof}

It follows from the above lemma is that there must be an optimal decomposition with $\mathcal{L}_1$ sensitivity equal to 1, because we can always apply Lemma \ref{lem:rescale} to transform an optimal decomposition whose $\mathcal{L}_1$ sensitivity is not 1 to another optimal decomposition whose $\mathcal{L}_1$ sensitivity is 1. Therefore, it suffices to fix $\Delta(L)$ to 1 in the optimization program. Meanwhile, according to properties of the matrix trace, we have $\Phi(B)=\mbox{tr}(B^TB)$. Thus, we arrive at the following theorem.

\begin{theorem}\label{the:main}
Given the workload $W$, a workload decomposition $W=BL$ minimizes the expected squared error of the queries, if $(B,L)$ is the optimal solution to the following program:

\begin{equation}\label{eqn:opt-problem}
\begin{split}
\min_{B,L} \frac{1}{2}\mbox{tr}(B^TB)\\
\mbox{s.t.~~}  W=BL\\
\forall j \sum_i^r |L_{ij}|\leq 1
\end{split}
\end{equation}
\end{theorem}

The constant factor $1/2$ in the objective function above simplifies the notations in the following sections; it does not affect the optimal solution of the program. We omit the proof since it is already clear from the discussions above. Solving the above optimization program is rather difficult, since it involves a non-linear objective function and complex constraints. We present a relaxation of the problem in Section \ref{sec:relaxation}, and our solution in Section \ref{sec:opt}.

\subsection{Utility Analysis and Budget Selection}\label{sec:utility}

In practice, users are often unsure about how to set the privacy parameter $\epsilon$ involved in $\epsilon$-differential privacy. Instead, setting the desired \emph{utility} level of the query results is much more intuitive. Given the user-specified utility, this subsection derives the smallest $\epsilon$ value for LRM that satisfies the utility requirement. Note that smaller values of $\epsilon$ corresponds to stronger privacy protection. We use a common definition of query result utility called ($\xi$, $\eta$)-usefulness \cite{BLR08}, as follows.

\begin{definition}\label{def:usefulness}
Given a mechanism $M$, query set $Q$, sensitive data $D$, and parameters $\xi>0$ and $0<\eta<1$, we say that $M$ is $(\xi$, $\eta)$-useful with respect
to $Q$ and $D$ under the $\|\cdot\|_*$-norm if the following inequality holds:
$$ \Pr\left(\|M(Q,D) - Q(D)\|_* \geq \xi \right) \leq \eta$$
\end{definition}
where $\|\cdot\|_*$-norm can be any vector norm definition. In our analysis, we consider the $\|\cdot\|_1$-norm and the $\|\cdot\|_{\infty}$-norm.

Given user specified values of $\xi$ and $\eta$, we now derive the minimum value for $\epsilon$ with which LRM achieves ($\xi$, $\eta$)-usefulness. The derivation uses Markov's inequality and the Chernoff bound, as follows.

\begin{lemma}\label{lem:chernoff}
Markov's Inequality and the Chernoff Bound \cite{billingsley2012probability}. Given a non-negative random variable $X$ and $t>0$, the following inequality holds:
$$\Pr(X \geq t) \leq \frac{\mathds{E}[X]}{t}$$
Moreover, for any $s\geq 0$, we have:
$$\Pr(X\geq t) = \Pr(e^{sX}\geq e^{st}) \leq  \frac{\mathds{E}[e^{sX}]}{e^{st}}$$
\end{lemma}

The minimum $\epsilon$ value is given in the following theorem.

\begin{theorem}
\textbf{Utility of LRM under $\epsilon$-differential privacy.} Given query set $Q$, database $D$, and user-specified parameters $\xi>0$ and $0<\eta<1$, (i) $M_{LRM, \epsilon}$ with the optimal decomposition $W=BL$ solved from Program (\ref{eqn:opt-problem}) returns ($\xi$, $\eta$)-useful results of $Q$ on $D$ under the $\|\cdot\|_1$-norm, when the privacy parameter $\epsilon$ satisfies  $\epsilon \geq \left(2|||B|||_1 (s\cdot \ln2 - \ln \eta) \right) / \xi$. (ii) Meanwhile, $M_{LRM, \epsilon}$ with the optimal decomposition achieves ($\xi$, $\eta$)-usefulness under the $\|\cdot\|_{\infty}$-norm, when $\epsilon \geq \left(2|||B|||_{\infty} (  \sum_{i=1}^s \ln(\frac{i}{i-0.5}) - \ln \eta) \right) / \xi$.
\end{theorem}
\begin{proof}(i) We first prove the utility of LRM under the $\|\cdot\|_1$-norm. Let $X$ be the Laplace noise vector injected to the results of intermediate queries corresponding to $L$. We have:
\begin{eqnarray}
\|M_{P}(Q,D)-Q(D) \|_1 &=& \|B(LD+X)-WD\|_1 \nonumber\\
&=& \|B \cdot X\|_1 = |||B \cdot X|||_1 \leq |||B|||_1 \cdot |||X|||_1 = |||B|||_1 \cdot \|X\|_1\nonumber
\end{eqnarray}

According to the Laplace mechanism, $X_1,X_2,\cdots,X_r$ are i.i.d. random variables following the zero-mean Laplace distribution with scale $\Delta(L)/\epsilon$. Since $L$ is obtained by solving Program (\ref{eqn:opt-problem}), we have $\Delta(L)=1$. Therefore, the scale of each of the Laplace variable $X_i, 1\leq i \leq r$ is $1/\epsilon$. According to properties of the Laplace distribution, $|X_i|$ follows the exponential distribution with rate parameter equal to $\epsilon$. Let $Y=\|X\|_1=|X_1|+|X_2|+\cdots+|X_r|$. Then, according to properties of the exponential distribution, $Y$ follows the Erlang distribution. Specifically, the probability distribution function of $Y$ is:

\begin{equation}
\Pr\left(Y=x\right)= \frac{\epsilon^r x^{r-1}  e^{- \epsilon x}}{(r-1)!}dx \nonumber
\end{equation}

For any positive number $t$ such that $\mathds{E}[e^{tY}]$ exists, we have:
\begin{eqnarray}
\mathds{E}[e^{tY}] = \int_{0}^{\infty}{ e^{tx} \cdot \frac{
\epsilon^{r} x^{r-1} e^{-\epsilon x } }{(r-1)!} } dx =
(1- \frac{t}{\epsilon})^{-r}, t < \epsilon \nonumber
\end{eqnarray}
Moreover, for any real number $c$, according to Lemma \ref{lem:chernoff}, we have:

$$\Pr (Y > c) = \Pr(e^{tY}>e^{tc}) \leq \frac{ \mathds{E} [e^{tY}] }{ e^{ct}} = \frac{(1-\frac{t}{\epsilon})^{-r}} {e^{ct}} $$

Setting $t = \frac{\epsilon}{2}$ and $c=\frac{ \xi}{|||B|||_1}$, we obtain:
$$\Pr (Y > \frac{ \xi}{|||B|||_1}) \leq \frac{(\frac{1}{2})^{-r}} {e^{\frac{\xi \epsilon}{2|||B|||_1}}} $$

Therefore, we have:

\begin{eqnarray}
& \|M_{P}(Q,D)-Q(D)\|_1 \leq |||B|||_1\cdot Y \nonumber\\
& \Rightarrow \forall \xi, \Pr( \|M_{P}(Q,D)-Q(D)\|_1
\geq \xi) \leq \Pr( Y \geq \frac{ \xi }{|||B|||_1}) \leq \frac{(\frac{1}{2})^{-r}} {e^{\frac{\xi \epsilon}{2|||B|||_1}}} \nonumber\\
\end{eqnarray}
When $\epsilon \geq  \left(2|||B|||_1 \left(r\cdot
\ln2 - \ln \eta\right) \right)/\xi$, the above probability is thus bound by $\eta$. This finishes the proof for claim (i) in the theorem.

(ii) Next we focus on the $\|\cdot\|_{\infty}$-norm. Let $X$ denote the same meaning as in the proof of part (i). Then, we have:

\begin{eqnarray}
& \|M_{P}(Q,D)-Q(D)\|_{\infty}  = \|B \cdot X\|_{\infty} \leq |||B|||_{\infty} \cdot \|X\|_{\infty}\nonumber
\end{eqnarray}

The inequality above holds due to the fact that: $\|Rx\|_{\infty} \leq |||R|||_{\infty}\cdot \|x\|_{\infty}$ for any matrix $R$ and vector $x$. Let $Y=\|X\|_{\infty}=\max{(|X_1|,|X_2|,\cdots,|X_r|)}$. Similar to part (i) of the proof, each $|X_i|, 1 \leq i \leq r$ follows the exponential distribution with rate $\epsilon$. According to the memoryless property of the exponential distribution, we create a chain of variables, as follows:
\begin{equation}
Y = \max{(|X_1|,|X_2|,\cdots,|X_r|)} = X_{\lambda=r \epsilon} + X_{\lambda=(r-1) \epsilon} + \cdots + X_{\lambda= \epsilon}
\end{equation}
where each $X_{\lambda = x}$ denotes an independent exponential random variable with
rate $x$. Intuitively, $X_{\lambda=r \epsilon}$ models the distribution of the smallest value among $|X_1|$, $|X_2|$, $\cdots$, $|X_r|$; $X_{\lambda=(r-i+1) \epsilon}, 1 < i \leq r$ models the difference between the $i$-th smallest value and the $(i-1)$-th smallest value among $|X_1|$, $|X_2|$, $\cdots$, $|X_r|$. The sum thus yields the maximum value among $|X_1|$, $|X_2|$, $\cdots$, $|X_r|$.

Similar to the part (i) of the proof, we further derive:
\begin{eqnarray}
\mathds{E}[e^{tY}] =  \mathds{E}[e^{t\left(X_{\lambda=r \epsilon} +
X_{\lambda=(r-1) \epsilon } + \cdots + X_{\lambda= \epsilon}\right)}] =
\mathds{E}[e^{t\left(X_{\lambda=r \epsilon}\right)}]\cdot
\mathds{E}[e^{t\left(X_{\lambda=(r-1)
\epsilon}\right)}]\cdot\cdots\cdot\mathds{E}[e^{t\left(X_{\lambda=\epsilon}\right)}]
\nonumber
\end{eqnarray}
Because $\mathds{E}[e^{tX_{\lambda=a}}]=\int_{0}^{\infty}{ e^{tx} \cdot a e^{-ax} } dx = \frac{a}{a-t}$ for any $t<a$, we reach:

$$\forall t<\epsilon, \mathds{E}[e^{tY}] = \prod_{i=1}^{r} \frac{i \epsilon}{i \epsilon-t}$$.

Finally, according to Lemma \ref{lem:chernoff}, we have the following inequality:
\begin{eqnarray}
\Pr(Y > c) &&= \Pr(e^{t \cdot Y }>e^{tc}) \nonumber \\
&&\leq \frac{\mathds{E}[e^{t(Y)}]}{ e^{ct}} \nonumber \\
&&= \prod_{i=1}^r \frac{i \epsilon}{i \epsilon -t}   / e^{ct} \nonumber
\end{eqnarray}

With the choice of $t=\frac{\epsilon}{2}$ and
$c=\frac{ \xi }{|||B|||_{\infty}}$, we obtain:
\begin{eqnarray}
& \|M_{P}(Q,D)-Q(D)\|_{\infty} \leq  |||B|||_{\infty}\cdot \|X\|_{\infty} \nonumber\\
& \Rightarrow \forall \xi, \Pr( \|M_{P}(Q,D)-Q(D)\|_{\infty}
>\xi) \leq \Pr( \|X\|_{\infty} > \frac{ \xi }{|||B|||_{\infty}}) \nonumber\\
& \Rightarrow \forall \xi,  \Pr(\|M_{P}(Q,D)-Q(D)\|_{\infty} >\xi
) \leq  \left( \prod_{i=1}^{r} \frac{i \epsilon}{i \epsilon-t} \right)
/ e^{ct} = \left( \prod_{i=1}^{r} \frac{i \epsilon }{i
\epsilon - \epsilon /2}\right)    / e^{\frac{ \xi  \epsilon
}{2|||B|||_{\infty}}} \nonumber
\end{eqnarray}
When $\epsilon \geq \left(2|||B|||_{\infty} \left(  \sum_{i=1}^r \ln(\frac{i}{i-0.5})   - \ln \eta\right) \right)/\xi$, the above probability is bounded by $\eta$.
\end{proof}

%

\subsection{Relaxed Workload Decomposition}\label{sec:relaxation}

Program \ref{eqn:opt-problem} is rather difficult to solve, since it contains a non-linear objective and complex constraints. To devise a stable numerical solution, we relax the formulation so that $BL$ does not necessarily match $W$ exactly, but within a small error tolerance. To do this, we introduce a new parameter $\gamma$ to bound the difference between $W$ and $BL$ in terms of the Frobenius norm. This leads to the following optimization program:
\begin{equation}\label{eqn:relaxed-problem}
\begin{split}
\min_{B,L}~\frac{1}{2} \mbox{tr}(B^TB)\\
\mbox{s.t.~~}  \|W-BL\|_F\leq\gamma\\
\forall j \sum_i^r |L_{ij}|\leq 1
\end{split}
\end{equation}

The following theorem analyzes the error of LRM with the optimal decomposition obtained by solving Program (\ref{eqn:relaxed-problem}).

\begin{theorem}\label{the:new_error}
The expected squared error of $M_{LRM, \epsilon}(Q,D)$ using the optimal decomposition $(B,L)$ solved from Program (\ref{eqn:relaxed-problem}) is at most
$$2\mbox{tr}(B^TB)/\epsilon^2+\gamma\sum_{i}x^2_{i}$$
\end{theorem}
\begin{proof}
When $W\neq BL$, there are two sources of error. The first is the
added Laplace noise. According to Lemma
\ref{lem:decomp_error}, the error incurred by the Laplace noise is at most
$\frac{2}{\epsilon^2}\Phi(B)(\Delta(L))^2\leq
\frac{2}{\epsilon^2}\mbox{tr}(B^TB)$.

The second source of the error is due to the difference between $W$ and $BL$. The incurred expected squared error is bounded by:
\begin{eqnarray}
& & ((W-BL)D)^T(W-BL)D\nonumber\\
&\leq& \|W-BL\|^2_FD^TD\nonumber =
\|W-BL\|^2_F\sum^n_{i=1}x^2_i\nonumber
\end{eqnarray}

The inequality above is due to the Cauchy-Schwartz inequality. By
linearity of expectation, the expected squared errors can be simply
summed up. This leads to the conclusion of the theorem.

\end{proof}

While Theorem \ref{the:new_error} implies the possibility of
estimating the optimal $\gamma$, it is not practical to implement it
directly, because this estimation depends on the data, i.e.,
$\sum_{i}x^2_{i}$. In our experiments, we test different values of
$\gamma$, report their relative performance, and describe guidelines for setting the appropriate $\gamma$ independently of the underlying data.

%% file: opt.tex
\section{Solving for the Optimal Workload Decomposition}\label{sec:opt}

This section solves the relaxed workload decomposition problem defined in Program (\ref{eqn:relaxed-problem}). This program is rather difficult to solve, because it is neither convex nor differntiable. In the following, Section \ref{sec:opt:1} describes an effective and efficient solution, based on the inexact Augmented Lagrangian method \cite{Conn1997,Lin2010}. Section \ref{sect:convergence} proves that the proposed solution always converges, and analyzes its convergence rate.

\subsection{Solution Based on Augmented Lagrangian Method}\label{sec:opt:1}
Observe that Program (\ref{eqn:relaxed-problem}) is a constrained optimization problem with a large number of unknowns, a non-linear objective and rather complex constraints. Since there is no known analytic solution to such a problem, we focus on numerical solutions. Furthermore, Program (\ref{eqn:relaxed-problem}) is difficult to tackle even with numerical methods, due to three main challenges. First and foremost, there are a a set of \emph{non-differentiable} constraints $\forall j \sum_i^r |L_{ij}|\leq 1$, which rules out many generic techniques for solving constrained optimization problems, such as the Lagrange multiplier method, which are limited to problems with differentiable constraints. Second, the non-differentiable constraints involve the unknown matrix $L$, whereas the objective function involves another unknown matrix $B$, whose relationship to $L$ is rather complex (i.e., in constraint $\|W-BL\|_F\leq\gamma$); consequently, it is non-trivial to apply specialized methods for handling the non-differentiable constraints. Finally, Program (\ref{eqn:relaxed-problem}) is not convex with respect to the unknowns $B$ and $L$.

The main idea of the proposed solution is to break down Program (\ref{eqn:relaxed-problem}) into simpler, solvable subproblems. Since the most difficult part of Program (\ref{eqn:relaxed-problem}) is the existence of the non-differentiable constraints $\forall j \sum_i^r |L_{ij}|\leq 1$, we aim to break down the whole problem into subproblems with only these constraints, and an objective function that only involves the unknown $L$, not $B$. Then, we use a specialized technique to solve each of these subproblems. Specifically, we first eliminate the constraint $\|W-BL\|_F\leq\gamma\rightarrow 0$ using the augmented Lagrangian method, which runs in multiple iterations, each of which solves a subproblem with only the constraints $\forall j \sum_i^r |L_{ij}|\leq 1$. Then, inside each iteration, we remove $B$ from the objective function of the subproblem, by alternatively optimizing for $B$ and $L$. The result are subproblems with only the constraints $\forall j \sum_i^r |L_{ij}|\leq 1$ as well as an objective function that has only $L$ as unknowns. Each of these subproblems are then solved by applying a special solver called Nesterov's first order optimal gradient method \cite{Nesterov03}. An important optimization is that we apply the \emph{inexact} augmented Lagrangian method \cite{Conn1997,Lin2010}, which does not solve the subproblem exactly in each iteration exactly, leading both increased efficiency and stability.


Algorithm \ref{algorithm:Lagrange} shows the proposed solution for Program (\ref{eqn:relaxed-problem}). First, we apply the inexact augmented Lagrangian method to eliminate the linear constraint $\|W-BL\|_F\leq\gamma\rightarrow 0$, as follows: we add to the objective function (i) a positive penalty item $\beta \in \mathbb{R}$ and (ii) the Lagrange multiplier $\pi \in \mathbb{R}^{m \times n}$. $\beta$ and $\pi$ are iteratively updated, following the strategy in \cite{Conn1997,Lin2010}. In each iteration, the values of $\beta$ and $\pi$ are fixed, and the algorithm aims to find values for $B$ and $L$ that minimize the following subproblem:


\begin{eqnarray}\label{LowRankDP_Subproblem}
\min_{B,L} \mathcal{J}(B,L,\beta,\pi)  =  \frac{1}{2}\mbox{tr}(B^TB)+
\langle\pi,W-BL \rangle    + \frac{\beta}{2} \|W-BL\|_F^2  \\
~~s.t.~~ \forall j \sum_i |L_{ij}|\leq 1\nonumber
\end{eqnarray}

\begin{algorithm}[t]
\caption{\label{algorithm:Lagrange} {\bf Workload Matrix
Decomposition}}
\begin{algorithmic}[1]
\STATE  Initialize $\pi^{(0)}=\mathbf{0}\in \mathbb{R}^{m\times n},
\beta^{(0)}=1, k=1$
  \label{OutIter}\WHILE{not converged}
  \label{InIter}\WHILE {not converged}\label{main_alg:lag:begin}
  \STATE \label{step1} $B^{(k)} \leftarrow$ update $B$ using Equation (\ref{UpdateB})
  \STATE \label{step2} $L^{(k)} \leftarrow$ run Algorithm \ref{algorithm:Nesterov} to update $L$ according to Program (\ref{eqn:Lagrangian-Seq-L})
  \ENDWHILE\label{main_alg:lag:end}
  \STATE Compute $\tau = \|W-B^{(k)}L^{(k)}\|_F$
  \IF{$\tau$ is sufficiently small or $\beta$ is sufficiently large}
  \STATE return $B^{(k)}$ and $L^{(k)}$
  \ENDIF
  \IF{$k$ is a multiple of 10}
  \STATE $\beta^{(k+1)} = 2\beta^{(k)}$
  \ENDIF
  \STATE  \label{step3} $\pi^{(k+1)} = \pi^{(k)} + \beta^{(k+1)} \left(W-B^{(k)}L^{(k)}\right)$
  \STATE $k=k+1$
  \ENDWHILE
\end{algorithmic}
\end{algorithm}

Next we eliminate unknowns $B$ from the objective function of the above subproblem, Program \label{LowRankDP_Subproblem}. Observe that this is a bi-convex optimization problem with respect to $B$ and $L$, meaning that it is convex with respect to $B$ (resp. $L$), once we fix $L$ (resp. $B$) to a constant. Hence, we solve it by alternately optimizing $B$ and $L$ (lines \ref{main_alg:lag:begin}-\ref{main_alg:lag:end} of Algorithm 1). Note that following the inexact Augmented Lagrangian Multiplier methodology, it is not necessary to obtain the exact optimal values of $B$ and $L$, instead, a small number of iterations of the while-loop in lines \ref{step1}-\ref{step2} suffices. We first focus on optimizing $B$, treating $L$ as constant. Observe that $\mathcal{J}(\cdot)$ is convex with respect to $B$. Hence, the optimal $B$ can be obtained by solving $\frac{\partial \mathcal{J}}{\partial B}=0$. In particular, the gradient with respect to $B$ is:

$$\frac{\partial \mathcal{J}}{\partial B} = B - \pi L^T + \beta BLL^T - \beta WL^T$$

Solving $B$ from $\frac{\partial \mathcal{J}}{\partial B}=0$, we obtain:

\begin{equation}\label{UpdateB}
B=\left(\beta W L^T+\pi L^T\right) \left(\beta LL^T + I\right)^{-1}
\end{equation}

Next we show how to optimize $L$ with a fixed $B$. This is equivalent to the following quadratic program:

\begin{equation}\label{eqn:Lagrangian-Seq-L}
\begin{split}
\mathcal{G}(L)=\frac{\beta}{2} \mbox{tr}\left(L^TB^TBL\right) -  \mbox{tr}\left(\left(\beta W+\pi\right)^TBL\right)\\
    \mbox{s.t.~~}  \forall j \sum_i |L_{ij}|\leq 1
\end{split}
\end{equation}

\noi The gradient of the objective $\mathcal{G}(L)$ respect to $L$ in (\ref{eqn:Lagrangian-Seq-L}) can be computed as:
\begin{equation}\label{eqn:Gradient_L}
\frac{\partial \mathcal{G}}{\partial L} = \beta B^TBL-\beta B^TW - B^T \pi
\end{equation}
\noi For all $L', L''~\text{with} ~\forall j \sum_i |L'_{ij}|\leq 1, \forall j \sum_i |L''_{ij}|\leq 1$, we have the following inequalities:
\begin{eqnarray}\label{eqn:lipschitz}
\frac{\|\mathcal{G}(L') - \mathcal{G}(L'')\|_F}{\|L'-L''\|_F} &=& \frac{\|\beta B^TBL' - \beta B^TBL''\|_F}{\|L'-L''\|_F} \nonumber \\
&\leq& \frac{ |||\beta B^TB|||_2 \cdot \|L' - L''\|_F}{\|L'-L''\|_F}  = \beta \cdot |||B^TB|||_2 \nonumber
\end{eqnarray}
Therefore, the gradient of $\mathcal{G}(L)$ is Lipschitz continuous with parameter $\omega = \beta \cdot |||B^TB|||_2$.

We employ Nesterov's first order optimal gradient method \cite{Nesterov03} to solve the program in (\ref{eqn:Lagrangian-Seq-L}). Nesterov's method has a much faster
convergence rate than traditional methods such as the subgradient method or naive projected gradient descent. The updating rule in the projected gradient method is expressed as follows:
$$ L^{(t+1)} = \mathcal{P} ( L^{(t)} - \eta^{(t)} \frac{\partial \mathcal{G}}{\partial L^{(t)}}) $$
where $t$ denotes the iteration counter, $\mathcal{P}(L)$ denotes the $\mathcal{L}_1$ projection operator on any $L\in \mathbb{R}^{r \times n}$, $\eta > 0$ denotes the appropriate step size. One typical choice for $\eta$ is the inverse of the gradient lipschitz constant ${1}/{\omega}$, however, this can be sub-optimal when the gradient lipschitz constant is large. One can incooperate Beck et al.'s backtracking line search strategy to further accelerate the convergence of the projected gradient algorithm \cite{beck2009fast}. We adopt this line search strategy in our algorithm.

$L$ is updated by gradient descent while ensuring that the
$\mathcal{L}_1$ regularized constraint on $L$ is satisfied. This is done by the $\mathcal{L}_1$ projection operator, formulated as the following optimization problem:

\begin{equation}\label{eqn:L1Projection}
 \mathcal{P} (L) =  \arg\min_{\bar{L}\in \mathbb{R}^{r \times n}} \|\bar{L}-L\|_F^2 , s.t. ~~\forall j \sum_i |\bar{L}_{ij}|\leq 1,
\end{equation}

\noindent We observe that Equation (\ref{eqn:L1Projection}) can be decoupled
into $n$ independent $\mathcal{L}_1$ regularized sub-problems:
$$  \arg\min_{\bar{l}\in \mathbb{R}^{r \times 1}} \|\bar{l}-l\|_2^2 , s.t. ~~ \sum_i |\bar{l}_{i}|\leq 1 $$
where $l=L^{(t)}_{j}, j=1,2,\cdots,n$, $L^{(t)}_{j}$ is the $j^{th}$ column of $L^{(t)}$. Such a projection operator can be solved efficiently by $\mathcal{L}_1$ projection methods in $\mathcal{O}(r \log r)$ time \cite{Duchi2008}, as described in Algorithm \ref{algorithm:l1Projection}. The complete algorithm for solving Program (\ref{eqn:Lagrangian-Seq-L}) is summarized in Algorithm \ref{algorithm:Nesterov}.

%
%

\begin{algorithm}
\caption{\label{algorithm:l1Projection} {\bf Algorithm for
$\mathcal{L}_1$ Ball} Projection}
\begin{algorithmic}[1]
  \STATE input: A vector $l \in R^{r \times 1}$
  \STATE sort $l$ into $v$ such that $v_1 \geq v_2 \geq \cdots \geq v_r$
  \STATE find $\rho = \max \{i \in [r]: v_i -
  \frac{1}{i}\left(\sum_{k=1}^{i}v_k-1\right)>0\}$
  \STATE compute $\theta = \frac{1}{\rho}\left(
  {\sum_{i=1}^{\rho}v_i-1}\right)$
  \STATE output ${\bar{l}\in \mathbb{R}^{r \times 1}}$, s.t. $\bar{l}_i= \max(l_i - \theta, 0), i \in [r]$
\end{algorithmic}
\end{algorithm}

\begin{algorithm}
\caption{\label{algorithm:Nesterov} {\bf Nesterov's Projected Gradient Method}}
\begin{algorithmic}[1]
\STATE  input: $\mathcal{G}(L), \frac{\partial \mathcal{G}}{\partial L}, L^{(0)}$ \\
\STATE $\chi=r\cdot n \cdot 10^{-12}$, Lipschitz parameter: $\omega^{(0)}=1$\\
 \STATE Initializations: $L^{(1)}=L^{(0)}, \tau^{(-1)}=0, \tau^{(0)}=1, t=1$
  \WHILE{not converged}
      \STATE $\alpha=\frac{\tau^{(t-2)}-1}{\tau^{(t-1)}}$, $S=L^{(t)}+\alpha(L^{(t)}-L^{(t-1)})$
      \FOR{$j=0$ to $\cdots$}
      \STATE $\omega=2^j \omega^{(t-1)}$, $U = S - \frac{1}{\omega} \nabla_{S}$
      \STATE Project $U$ to the feasible set to obtain $L^{(t)}$ (i.e., solve Equation (\ref{eqn:L1Projection}))
      \IF{ \label{stopping1} $\|S-L^{(t)}\|_F<\chi$}
        \STATE  return;
      \ENDIF
      \STATE  Define function: $\mathcal{J}_{\omega,S}(U)=\mathcal{G}(S)+\langle \frac{\partial G}{\partial U}, U-S\rangle+\frac{\omega}{2}\|U-S\|_F^2$
      \IF{$ \mathcal{G}(L^{(t)}) \leq \mathcal{J}_{\omega,S}(U)$}
        \STATE $\omega^{(t)}=\omega; L^{(t+1)}=L^{(t)}$; break;
      \ENDIF
      \ENDFOR
      \STATE Set $\tau^{(t)}= \frac {1+ \sqrt{{1+4(\tau^{(t-1)})^2}}}{2}$
      \STATE $t=t+1$
      \ENDWHILE
\STATE return $L^{(t)}$
\end{algorithmic}
\end{algorithm}
\subsection{Convergence Analysis}\label{sect:convergence}
This subsection analyzes the convergence properties of the proposed workload decomposition algorithm. In each iteration, Algorithm \ref{algorithm:Lagrange} solves a sequence of Lagrangian subproblems by optimizing $B$ (step \ref{step1}) and $L$ (step \ref{step2}) alternatingly.
The algorithm stops when a sufficiently small $\gamma$ is obtained
or the penalty parameter $\beta$ is sufficiently large. It suffices
to guarantee that $L$ converges to a locally optimal solution \cite{Lin2010,wen2012alternating,wen2012solving}.

In general, penalty methods have the property that when the global (or local) minimizers of the
subproblem are found, every limit point is a global (or local) minimizer of the original problem \cite{fiacco1968nonlinear}. This property is preserved by the Augmented Lagrangian Multiplier counterparts. Therefore, the proposed solution for the workload decomposition problem converges, whenever the bi-convex optimization subproblem in Program (\ref{LowRankDP_Subproblem}) converges. Regarding the convergence properties of the bi-convex optimization subproblem, past study \cite{bertsekas1999} on bi-convex optimization has shown that block coordinate descent is guaranteed to converge to the stationary point for \emph{strictly convex} problems. However, the subproblem in Program (\ref{LowRankDP_Subproblem})
is not strictly convex (though it is convex); meanwhile, the subproblem may have multiple optimal solutions, which may cause problems to its convergence. Fortunately, for bi-convex optimizations which only involves two blocks, \cite{Grippo2000} shows that the strict convexity of the subproblem is not
required; every limit point of $\{B^{(k)},L^{(k)}\}$ is a stationary point. Accordingly, the bi-convex optimization subproblem exhibits nice convergence properties. In the following, we formalize and prove the convergence results of the proposed algorithm.


We first present the first order KKT conditions of the optimization problem in Program (\ref{eqn:relaxed-problem}). Introducing Lagrange multipliers $\mu \in \mathbb{R}^{n\times 1}$ and $\pi \in \mathbb{R}^{m\times n}$ for the inequality constraints $\forall j \sum_i^r |L_{ij}|\leq 1$ and linear constraints $W=BL$ respectively, we derive the
following KKT conditions of the optimization problem:

\begin{equation}\label{kkt:optimal}
\begin{split}
\mu \geq 0 ~~\text{(Non-Negativity)}   \\
W=BL,~\forall j \sum_i^r |L_{ij}|\leq 1 ~~\text{(Feasibility)}  \\
B  = \pi L^T,~0\in\sum_j^n{\mu_j \frac{\partial \sum_i^r |L_{ij}|
}{\partial L}} - B^T \pi ~~\text{(Optimality)}  \\
\forall j ~ \mu_j (\sum_i^r |L_{ij}|- 1) = 0 ~~\text{(Complementary
Slackness)}
\end{split}
\end{equation}

The following theorem establishes the convergence properties of the proposed algorithm, under the assumption that the iterates generated by Algorithm \ref{algorithm:Lagrange} exhibit no jumping behavior. Remark that the similar condition was used in \cite{wen2012alternating,wen2012solving}.

\begin {theorem} \label{LowRankDPConvergence}
\textbf{Convergence of Algorithm \ref{algorithm:Lagrange}.} Let
$X\triangleq(B,L,\pi)$ and $\{X^{(k)}\}_{k=1} ^{\infty}$ be the intermediate results of Algorithm \ref{algorithm:Lagrange} after the $k$-th iteration. Assume that $\{X^{(k)}\}_{k=1} ^{\infty}$
is bounded and $\lim_{k\rightarrow \infty} (X^{(k+1)}-X^{(k)}) = 0$. Then any accumulation point of $\{X^{(k)}\}_{k=1} ^{\infty}$
satisfies the KKT conditions presented in Equation (\ref{kkt:optimal}). In other words,
whenever $\{X^{(k)}\}_{k=1} ^{\infty}$ converges, it converges to a
first-order KKT optimal point.
\end {theorem}

\begin{proof}
Since $L^{(k+1)}$ is the global optimal solution of Program (\ref{eqn:Lagrangian-Seq-L}), by the
KKT optimal condition, there exist $\mu \geq 0, \mu \in\mathbb{R}^{n\times 1}$ and $L^{(k+1)}$ such that
the following equation holds:
\begin{equation} \label{KKT_Optimal_L}
0\in\frac{\partial \mathcal{G}}{\partial L^{(k+1)}} + \sum_j^n{\mu_j
\frac{\partial \sum_i^r |L^{(k+1)}_{ij}| }{\partial L^{(k+1)}}}
\end{equation}
Note that $\mathcal{G}$ is a convex function with respect to
$L^{(k+1)}$. Hence, the KKT conditions are both necessary and sufficient conditions
for global optimality. Combining Equations (\ref{eqn:Gradient_L}) and (\ref{KKT_Optimal_L}), we obtain:
\begin{eqnarray}\label{convergence:proof:tendency1}
& &\beta B^{(k+1)T}\left(B^{(k+1)T}(L^{(k+1)} - L^{(k)})\right)\\
&=& \beta B^{(k+1)T}(W - B^{(k+1)T}L^{(k)} ) + B^{(k+1)T} \pi -
\sum_j^n{\mu_j \frac{\partial \sum_i^r |L^{(k+1)}_{ij}| }{\partial
L^{(k+1)}}} \nonumber
\end{eqnarray}
We derive the following equations according to the update rule for $B$ (at Line
\ref{step1} in Algorithm \ref{algorithm:Lagrange}) and the Lagrangian
multiplier update rule for $\pi$ (at Line \ref{step3} in Algorithm
\ref{algorithm:Lagrange}), respectively:
\begin{eqnarray}\label{convergence:proof:tendency2}
B^{(k+1)} - B^{(k)} = \left(\beta W L^{(k)T}+\pi L^{(k)T} - B^{(k)}
(\beta L^{(k)}L^{(k)T} + I ) \right) \left(\beta L^{(k)}L^{(k)T} +
I\right)^{-1}
\end{eqnarray}
\begin{eqnarray}\label{convergence:proof:tendency3}
\pi^{(k+1)} - \pi^{(k)} = - \beta^{(k+1)}
\left(W-B^{(k+1)}L^{(k+1)}\right)
\end{eqnarray}
Since $\{X^{(k)}\}_{k=1} ^{\infty}$ is bounded according to our assumption,
the sequences $\{B^{(k)}\}_{k=1} ^{\infty}$ and $\{L^{(k)}\}_{k=1}
^{\infty}$ are also bounded. Hence, $\lim_{k\rightarrow \infty}
(X^{(k+1)}-X^{(k)}) = 0$ implies that both sides of Equation (\ref{convergence:proof:tendency1}, \ref{convergence:proof:tendency2}, \ref{convergence:proof:tendency3})
converge to zero as $k$ approaches infinity. Consequently,

\begin{eqnarray}\label{convergence:proof:station}
W-B^{(k)}L^{(k)} \rightarrow 0, ~~  \pi L^{(k)T} - B^{(k)} \rightarrow 0 \nonumber\\
\exists \mu ~: - B^{(k+1)T}\pi + \sum_j^n{\mu_j \frac{\partial
\sum_i^r |L^{(k+1)}_{ij}| }{\partial L^{(k+1)}}} \rightarrow 0
\end{eqnarray}

where the first limit in Equation (\ref{convergence:proof:station}) is
used to derive other limits. Therefore, the sequence $\{X^{(k)}\}_{k=1}
^{\infty}$ asymptotically satisfies the KKT conditions in Equation (\ref{kkt:optimal}). This completes the proof.
\end{proof}

Next we focus on the convergence rate of the proposed algorithm. The following
theorem states that it converges linearly.

\begin {theorem} \label{LowRankDPConvergenceRate}
\textbf{Convergence Rate of Algorithm \ref{algorithm:Lagrange}.}  Let
$X\triangleq(B,L,\pi)$ and $\{X^{(k)}\}_{k=1} ^{\infty}$ be the intermediate results of Algorithm \ref{algorithm:Lagrange} after the $k$-th iteration. Assume that $\{X^{(k)}\}_{k=1} ^{\infty}$
is bounded and $\lim_{k\rightarrow \infty} (X^{(k+1)}-X^{(k)}) = 0$. Let
$(B^{(k)},L^{(k)})$ be the solution obtained after the $k$-th
iteration and $(B^*,L^*)$ be the optimal solution to Program
(\ref{eqn:relaxed-problem}), we have


\begin{equation}
\min_{i=1,2,...,k} | \mbox{tr}\left({B^{(i)}}^TB^{(i)}\right) -  \mbox{tr}\left({B^*}^TB^*\right) |  \leq
\mathcal{O}\left(\frac{1}{k}\right)
\end{equation}
In other words, Algorithm \ref{algorithm:Lagrange} converges to the stationary point linearly.
\end {theorem}

\begin{proof}
Let ${B^{(k)}}$ denote the solution of the
Lagrangian sub-problem in the $k^{th}$ iteration. The following
inequality holds on the sequence of the Lagrangian subproblems:

\begin{eqnarray}\label{eqn:proof1222}
& & \mathcal{J}({B^{(k+1)}},{L^{(k+1)}},{\pi^{(k)}},\beta^{(k)}) \nonumber\\
&\leq& \min_{W=BL, \atop\forall j \sum_i |L_{ij}|\leq 1} \mathcal{J}(B,L,{\pi^{(k)}},\beta^{(k)}) \nonumber \\
&\leq& \min_{W=BL, \atop\forall j \sum_i |L_{ij}|\leq 1} \mathcal{J}(B,L,\pi^{*},\beta^{(k)}) \nonumber \\
&=& \min_{W=BL, \atop\forall j \sum_i |L_{ij}|\leq 1}
\frac{1}{2}\mbox{tr}(B^TB) =\frac{1}{2}\mbox{tr}(B^{*T}B^*) 
\end{eqnarray}

By the definition of $\mathcal{J}(\cdot)$ and the inequality above, we derive the following inequality:
\begin{eqnarray} \label{convergence:rate1}
& & \frac{1}{2}\mbox{tr}({B^{(k+1)}}^TB^{(k+1)})\nonumber\\
&=& \mathcal{J}({B^{(k+1)}},{L^{(k+1)}},{\pi^{(k)}},\beta^{(k)}) - \langle \pi^{(k)}, W-B^{(k+1)} L^{(k+1)}\rangle + \frac{\beta^{(k)}}{2}\|W-B^{(k+1)}L^{(k+1)}\|_F^2\nonumber\\
&=&\mathcal{J}({B^{(k+1)}},{L^{(k+1)}},{\pi^{(k)}},\beta^{(k)}) -
\frac{1}{2\beta^{(k)}}\left(\|\pi^{(k)} + \beta^{(k)} (W-B^{(k+1)}
L^{(k+1)})\|_F^2 - \|\pi^{(k)}\|_F^2\right)\nonumber\\
&=& \mathcal{J}({B^{(k+1)}},{L^{(k+1)}},{\pi^{(k)}},\beta^{(k)}) - \frac{1}{2\beta^{(k)}} \left(\|\pi^{(k+1)}\|_F^2   - \|\pi^{(k)}\|_F^2\right) \nonumber\\
&\leq& \frac{1}{2}\mbox{tr}(B^{*T}B^*) - \frac{1}{2\beta^{(k)}}
\left(\|{\pi^{(k+1)}}\|_F^2 - \|{\pi^{(k)}}\|_F^2\right)
\end{eqnarray}

The third equality holds because of the Lagrangian multiplier update rule:
$$W-{B^{(k+1)}}{L^{(k+1)}} =
\frac{1}{\beta^{(k)}}\left( {\pi^{(k+1)}}-{\pi^{(k)}}\right).$$

By the non-negativity of norms, we have:
\begin{eqnarray} \label{convergence:rate2}
\frac{1}{2}\mbox{tr}(B^{(k+1)^{T}}B^{(k+1)}) &\geq& \frac{1}{2}\mbox{tr}(B^{(k+1)^{T}}B^{(k+1)}) - \|W-B^{(k+1)}L^{(k+1)}\|_F^2 \nonumber \\
&\geq& \frac{1}{2}\mbox{tr}(B^{*T}B^{^*}) - \|W-B^{(k+1)}L^{(k+1)}\|_F^2 \nonumber \\
&=& \frac{1}{2}\mbox{tr}(B^{*T}B^{^*}) - \frac{1}{2\beta^{(k)}}
\left(\|{\pi^{(k+1)}}\|_F^2 - \|{\pi^{(k)}}\|_F^2\right)
\end{eqnarray}

Combining Equations (\ref{convergence:rate1}) and (\ref{convergence:rate2}), we obtain:
\begin{equation}
\beta^{(i+1)} \left(\mbox{tr}\left({B^{(i+1)}}^TB^{(i+1)}\right) - \mbox{tr}\left({B^*}^TB^*\right) \right) =  \|{\pi^{(i+1)}}\|_F^2 - \|{\pi^{(i)}}\|_F^2,~\forall i  \nonumber
\end{equation}

Summing this equality above over $i=0,1...,k-1$, we have:
\begin{equation}
\sum_{i=0}^{k-1}  \beta^{(i+1)} \left(\mbox{tr}\left({B^{(i+1)}}^TB^{(i+1)}\right) - \mbox{tr}\left({B^*}^TB^*\right) \right) = \|{\pi^{(k)}}\|_F^2 - \|{\pi^{(0)}}\|_F^2
\end{equation}

Since $\beta^{(k)}$ is non-decreasing, we have:

\begin{equation}
\min_{i=0,1,...,k-1} | \mbox{tr}\left({B^{(i+1)}}^TB^{(i+1)}\right) -  \mbox{tr}\left({B^*}^TB^*\right) |  \leq \frac{\left(\|{\pi^{(k)}}\|_F^2 - \|{\pi^{(0)}}\|_F^2\right)/\beta^{(0)}}{k}
\end{equation}

By the boundedness of $\|{\pi^{(k)}}\|_F^2 - \|{\pi^{(0)}}\|_F^2$, we complete the proof.
\end{proof}

Note that although our convergence proof assumes that each subproblem is solved exactly, this is not required in practise, because the inexact augmented Lagrange multipliers method has been shown to converge practically as fast as the exact augmented Lagrange multipliers \cite{Lin2010}. Meanwhile, inexact augmented Lagrange multipliers require significantly fewer iterations when solving the subproblem, leading to much higher efficiency.

\textbf{Complexity Analysis:} Each update on $B$ in Equation (\ref{UpdateB}) takes
$\mathcal{O}(r^2m)$ time, while each update on $L$ consumes $\mathcal{O}(r^2 n)$ time. Assuming that Algorithm \ref{algorithm:Lagrange} converges to a local minimum within
$N_{in}$ inner iterations (at line 3 in Algorithm \ref{algorithm:Lagrange}) and $N_{out}$ outer iterations (line 2 in Algorithm \ref{algorithm:Lagrange}), the overall complexity of
Algorithm \ref{algorithm:Lagrange} is $\mathcal{O}(N_{in} \times N_{out}\times(r^2m + r^2n))$.

%

%% file: ext.tex
\section{LRM under ($\epsilon$, $\delta$)-Differential Privacy}\label{sec:approx}

This section extends LRM to ($\epsilon$, $\delta$)-differential privacy. Section 6.1 formulates the workload decomposition as an optimization program. Section 6.2 analyzes the utility of LRM. Section 6.3 discusses the algorithm for solving optimal workload decomposition.

\subsection{Workload Decomposition}

Similar to the case of $\epsilon$-differential privacy described in section \ref{sec:formulation}, LRM decomposes the workload matrix $W$ into $W=BL$. Then, LRM applies the Gaussian mechanism to the intermediate queries corresponding to $L$ to enforce ($\epsilon$, $\delta$)-differential privacy. Finally, LRM combines the noisy results of the intermediate queries according to $B$, to obtain the results of $Q$. Formally, let $\Theta(L)$ be the $\mathcal{L}_2$ sensitivity of $L$, i.e., $\Theta(L)=\max_j \left(\sum_i L_{ij}^2\right)^{1/2}$. LRM under ($\epsilon$, $\delta$)-differential privacy is defined as follows.

\begin{equation}\label{eqn:part_mech_app}
M_{LRM, (\epsilon, \delta)}(Q,D)=B\left(LD+Gau\left(\frac{\Theta(L)}{h(\epsilon,\delta)}\right)^r\right)
\end{equation}
where $h(\epsilon,\delta)=\frac{\epsilon}{\sqrt{8 \ln (2/ \delta )}}$.

Let $\Phi(B)$ be scale of the decomposition as defined in Definition \ref{def:scale}, i.e., $\Phi(B)=\sum_{i,j}B_{ij}^2$. The following lemma shows that the error of LRM is linear to $\Phi(B)$, and quadratic to $\Theta(L)$.

\begin{lemma}\label{lem:decomp_error_app}
The expected squared error of $M_{LRM, (\epsilon, \delta)}(Q,D)$ with respect to the decomposition $W=BL$ is $ 8 \ln (2/ \delta )\Phi_B\Theta(L)^2/\epsilon^2$.
\end{lemma}

\begin{proof}
According to Equation (\ref{eqn:part_mech_app}), $Q(D)-M_{LRM(\epsilon, \delta)}(Q,D)=B\left(Gau\left(\frac{\Theta(L)}{h(\epsilon,\delta)}\right)^r\right)$. The expected
squared error of LRM is thus
$\sum_{ij}B^2_{ij}\frac{2(\Theta(L))^2}{ {h(\epsilon,\delta)}^2}$. Since
$\Phi_B=\sum_{ij}B^2_{ij}$ and $h(\epsilon,\delta)=\frac{\epsilon}{\sqrt{8 \ln (2/ \delta )}}$, the error can be rewritten as $8 \ln (2/ \delta )\Phi_B\Theta(L)^2/\epsilon^2$.
\end{proof}

Therefore, the best decomposition is the one that minimizes $\Phi_B\Theta(L)^2$. Similar to the case of $\epsilon$-differential privacy, the particular value of $\Theta(L)$ is not important, as stated in the following lemma.

\begin{lemma}\label{lem:rescale2}
Given a workload decomposition $W=BL$, we can always construct another decomposition $W=B'L'$ satisfying (i) $\Theta(L')=1$ and (ii) ($B'$, $L'$) lead to the same expected squared error of $M_{LRM, (\epsilon, \delta)}$ as ($B$, $L$).
\end{lemma}

The proof is similar to that of Lemma \ref{lem:rescale}, and omitted for brevity. Based on Lemma \ref{lem:rescale2}, we formulate the following optimization program for finding the best decomposition for $M_{LRM, (\epsilon, \delta)}$:

\begin{equation}\label{eqn:opt-problem_app}
\begin{split}
\min_{B,L}~\frac{1}{2}\mbox{tr}(B^TB)\\
\mbox{s.t.~~}  W=BL\\
\forall j \sum_i^r L_{ij}^2 \leq 1
\end{split}
\end{equation}

\subsection{Utility Analysis and Budget Selection}
This subsection analyzes the utility $M_{LRM, (\epsilon, \delta)}$, as well as the choice of the privacy parameters ($\epsilon$, $\delta$) given a user-specified utility constraint. We use ($\xi$, $\eta$)-usefulness (Definition \ref{def:usefulness}) as the utility measure. The result is stated in the following theorem.

\begin{theorem} \textbf{Utility of LRM under ($\epsilon$, $\delta$)-differential privacy.} Given database $D$ and workload $W$, for any $\xi>0$ and $0<\eta<1$, mechanism $M_{LRM, (\epsilon, \delta)}$ using the optimal decomposition $W=BL$ solved from Program (\ref{eqn:opt-problem_app}) has the following utility guarantees:
(i) when $\epsilon \geq  \sqrt {6 \cdot \ln{\frac{2}{\delta}\cdot \left(\frac{r}{2}\ln3-\ln \eta\right)}}|||B|||_2 / \xi$, the output of $M_{LRM, (\epsilon, \delta)}$ is $(\xi,\eta)$-useful under the $\|\cdot\|_{2}$-norm;
(ii) when $\epsilon \geq \sqrt { (6\ln r - 3\ln3)(\ln 2 - \ln \delta ) / \eta } |||B|||_{\infty}  / \xi$, the output of $M_{LRM, (\epsilon, \delta)}$ is $(\xi,\eta)$-useful under the $\|\cdot\|_{\infty}$-norm.
\end{theorem}

\begin{proof}
(i) Let $X$ be the Gaussian noise vector injected to the intermediate results in LRM. According to Equation (\ref{eqn:part_mech_app}), we have:
\begin{eqnarray}
& \|M_{LRM, (\epsilon, \delta)}(Q,D)-Q(D)\|_2^2 = \|B(LD+X) - WD\|_2^2 = \|B \cdot X\|_2^2  \leq |||B|||_2^2 \cdot \|X\|_2^2\nonumber
\end{eqnarray}
The inequality above is due to the fact that $\|Rx\|_2\leq |||R|||_2 \cdot \|x\|_2$, for any matrix $R$ and vector $x$.
Accordingly, we derive the following:
\begin{eqnarray}
& \|M_{LRM, (\epsilon, \delta)}(Q,D)-Q(D)\|_2^2 \leq  |||B|||_2^2 \cdot  \|X\|_2^2 \nonumber\\
& \Rightarrow \forall \xi, \Pr( \|M_{LRM, (\epsilon, \delta)}(Q,D)-Q(D)\|_2^2 \geq
\xi^2)  \leq \Pr( \|X\|_2^2 \cdot |||B|||_2^2 \geq \xi^2 )  \nonumber \\
& \Rightarrow
\forall \xi, \Pr( \|M_{LRM, (\epsilon, \delta)}(Q,D)-Q(D)\|_2 \geq \xi) \leq  \Pr(
\|X\|_2^2 \geq \frac{\xi^2}{|||B|||_2^2} ) \nonumber
\end{eqnarray}

Next we focus on properties of $X$. According to the Gaussian mechanism, the elements of $X$, i.e., $X_1,X_2,\cdots,X_r$ follow i.i.d. zero-mean Gaussian distribution with scale $\sigma = \frac{\Theta(L)}{h(\epsilon,\delta)}$. Since the decomposition $W=BL$ is solved from Program \ref{eqn:opt-problem_app}, we have $\Theta(L)=1$. Thus, $\sigma = \frac{1}{h(\epsilon,\delta)} = \frac{ \sqrt{2 \ln (2/\delta )}  }{\epsilon}$.

Let $t$, $c$ be any positive number, we have:
$$\Pr(\|X\|_2^2 \geq c)=\Pr\left(\frac{\|X\|_2^2}{t\sigma^2}>\frac{c}{t\sigma^2}\right) = \Pr\left(e^{\frac{\|X\|_2^2}{t\sigma^2}} > e^{\frac{c}{t\sigma^2}}\right) \leq \frac{\mathds{E}\left[e^{\frac{\|X\|_2^2}{t\sigma^2}}\right]}{e^{\frac{c}{t\sigma^2}}} $$
where the last inequality holds due to Markov's inequality.

Consider the random variable $Y_i=\exp{\left(\frac{X_i^2}{t\sigma^2}\right)}$, where $t$ is an arbitrary
positive number such that $\mathds{E}[Y_i]$ exists. According to the probability density function of the Gaussian distribution (Equation (\ref{pdf:gaussian})), we have:
\begin{equation}
\mathds{E}[Y_i] = \int_{-\infty}^{\infty}{  g(x) e^{\left({\frac{x^2}{t\sigma^2}} \right)} } dx = \int_{-\infty}^{\infty}{\sqrt{ \frac{1}{2\pi
\sigma^2}} e^{\left(-\frac{x^2}{2\sigma^2}\right) }
e^{{\frac{x^2}{t\sigma^2}}} } dx = \sqrt{\frac{t}{t-2}},\forall t > 2\nonumber
\end{equation}

Based on the above derivations, and the fact that $X_i$'s are independent variables, we obtain:
\begin{eqnarray}
\Pr(\|X\|_2^2 \geq c) \leq \frac{ \prod_{i=1}^{r} ( \mathds{E} e^{ \frac{X_i^2}{t\sigma^2}} )
}{e^{\frac{c}{t\sigma^2}}}\nonumber = \frac{ \prod_{i=1}^{r}
 \mathds{E} [Y_i]   }{e^{\frac{c}{t\sigma^2}}} =  \frac{  (\frac{t}{t-2} )^{r/2} }{e^{\frac{c}{t\sigma^2}}}
\end{eqnarray}

With the choice of $t = 3$, $c=\frac{\xi^2}{|||B|||_2^2}$, and $\sigma= \frac{ \sqrt{2 \ln (2/\delta )}  }{\epsilon}$, this leads to:
\begin{eqnarray}
\Pr( \|M_{\epsilon, \delta}(Q,D)-Q(D)\|_2 \geq \xi) \leq   \frac{
(\frac{t}{t-2} )^{r/2} }{e^{\frac{c}{t\sigma^2}}}  =  \frac{ 3^{r/2}
}{e^{\frac{\xi^2\epsilon^2 }{6\ln(2/\delta)|||B|||_2^2}}} \nonumber
\end{eqnarray}
When $\epsilon \geq { \sqrt {6 \cdot \ln{\frac{2}{\delta}\cdot \left(\frac{r}{2}\ln3-\ln \eta\right)}}|||B|||_2} /
\xi$, the above probability is bound by $\eta$.

(ii) Let $X$ be the Gaussian noise vector injected to the intermediate results as in part (i) of the proof. We have:
$$\|M_{\epsilon, \delta}(Q,D)-Q(D)\|_{\infty}^2  = \|B \cdot X\|_{\infty}^2 \leq  |||B|||_{\infty}^2 \cdot \|X \|_{\infty}^2 =  |||{\sigma}B|||_{\infty}^2 \cdot \|\frac{1}{\sigma}X\|_{\infty}^2$$

The above inequality holds due to the fact that $\|Rx\|_{\infty} \leq |||R|||_{\infty}\cdot \|x\|_{\infty}$ for any matrix $R$ and vector $x$. Let $Z= \|\frac{1}{\sigma}X\|_{\infty}^2 = \left(\max(\frac{1}{\sigma}X_1, \cdots \max(\frac{1}{\sigma}X_r)\right)^2$. We derive:
\begin{eqnarray}
& \|M_{\epsilon, \delta}(Q,D)-Q(D)\|_{\infty}^2 \leq  |||{\sigma}B|||_{\infty}^2 \cdot
\|\frac{1}{\sigma}X\|_{\infty}^2\nonumber\\
& \Rightarrow \forall \xi, \Pr( \|M_{\epsilon, \delta}(Q,D)-Q(D)\|_{\infty}^2 \geq
\xi^2)  \leq \Pr(     |||{\sigma}B|||_{\infty}^2   \cdot  Z   \geq \xi^2 )  \nonumber \\
& \Rightarrow
\forall \xi, \Pr( \|M_{\epsilon, \delta}(Q,D)-Q(D)\|_{\infty} \geq \xi) \leq  \Pr(
Z \geq \frac{\xi^2}{|||\sigma B|||_{\infty}^2} ) \nonumber
\end{eqnarray}
By Markov's inequality, we obtain:
\begin{eqnarray}
\Pr( Z \geq \frac{\xi^2}{|||\sigma B|||_{\infty}^2})
\leq \frac{\mathds{E}[Z]}{\frac{\xi^2}{|||\sigma B|||_{\infty}^2}}  \nonumber
\end{eqnarray}
Note that the above bound is tight, even though Chernoff bound can not be applied here.

Next we derive an upper bound for the expected value of $Z$. Let $Y=\frac{1}{\sigma}X$. Clearly, $Y_1, Y_2, \cdots,Y_r$ are independent, standard normal random variables. Hence, $Y_i^2$'s ($1 \leq i \leq r$) are i.i.d. $\chi^2_1$ variables, i.e., Chi-square random variables with 1-degree of freedom. The probability density function $f_i$ for $Y_i$ is thus:
$$f_i(x)= \frac{1}{\sqrt{2\pi}} x^{-\frac{1}{2}} e^{-\frac{x}{2}} $$
Since the function exp($\cdot$) is convex and positive, by Jensen's inequality, for any $t$ such that $\mathds{E} [e^{tZ}]$ exists, we have:
\begin{equation} \label{max_expected_value_of_chisquare_variables}
e^{t\mathds{E}[Z]} \leq \mathds{E} [e^{tZ}] = \mathds {E} [\max_i e^{tY_i^2}] \leq \sum_{i=1}^r \mathds{E} [ e^{tY_i^2}]
\end{equation}
Meanwhile, for any $t < \frac{1}{2}$, we have $\mathds{E} [ e^{tY_i^2}] = \int_{0}^{+\infty}e^{tx} \frac{1}{\sqrt{2\pi}} x^{-\frac{1}{2}} e^{-\frac{x}{2}} dx = (1-2t)^{-\frac{1}{2}}$. Combine this with Equation (\ref{max_expected_value_of_chisquare_variables}), we obtain an upper bound of the expected value of $Z$:
$$ \mathds{E}[Z] \leq \frac{\ln r}{t} -  \frac{1}{2t}\ln (1-2t)$$
With the choice of $t=\frac{1}{3}$, we have $\mathds{E}[Z] \leq 3\ln r + \frac{3}{2}\ln 3$.  Since $\forall j \sum_i L_{ij}^2\leq 1$, the sensitivity over the batch query workload $Q$ is 1. Since $\sigma= \frac{ \sqrt{2 \ln
(2/\delta )}  }{\epsilon}$, we obtain the following:

\begin{eqnarray}
\forall \xi, \Pr( \|M_{\epsilon, \delta}(Q,D)-Q(D)\|_{\infty} \geq \xi) &\leq& {\mathds{E}[Z]\cdot |||\sigma B|||_{\infty}^2 }/{\xi^2} \nonumber\\
&\leq&   { \left( 3\ln r + \frac{3}{2}\ln 3 \right)\cdot |||\sigma B|||_{\infty}^2 }/{\xi^2}\nonumber\\
& = & { \left( 3\ln r + \frac{3}{2}\ln 3 \right)\cdot (2 \ln (2/\delta )) \cdot ||| B|||_{\infty}^2 }/{(\xi\epsilon)^2}\nonumber
\end{eqnarray}

When $\epsilon \geq {\sqrt { (6\ln r - 3\ln3)(\ln 2 - \ln \delta ) / \eta } |||B|||_{\infty}} /
\xi$, the above probability is bound by $\eta$.
\end{proof}

\subsection{Solving for the Optimal Workload Decomposition}

The optimization program (i.e., Program (\ref{eqn:opt-problem_app})) for workload decomposition under ($\epsilon$, $\delta$)-differential privacy is identical to the one under $\epsilon$-differential privacy (Program (\ref{eqn:opt-problem})), except that the former uses $\mathcal{L}_2$ sensitivity in the constraints $\forall j \sum_i^r L_{ij}^2 \leq 1$ whereas the latter uses $\mathcal{L}_1$ sensitivity. Hence, to solve Program (\ref{eqn:opt-problem_app}), we simply adapt Algorithm \ref{algorithm:Lagrange} by modifying the parts related to these constraints.

The only major modification of Algorithm \ref{algorithm:Lagrange} lies in the projection step, which now needs to projects every column in $L$ onto the $\mathcal{L}_2$ ball of radius 1, instead of the $\mathcal{L}_1$ unit ball as in Section \ref{sec:opt}. Specifically, the $\mathcal{L}_2$ ball projection is performed by solving the following optimization program:

\begin{equation}\label{eqn:L2Projection}
\min_{\bar{L}\in \mathbb{R}^{r \times n}} \|\bar{L}-L\|_F^2 , s.t. ~~\forall j \sum_i \bar{L}_{ij}^2\leq 1
\end{equation}

The above program can be decoupled into $n$ independent $\mathcal{L}_2$ regularized sub-problems:
$$  \arg\min_{\bar{l}\in \mathbb{R}^{r \times 1}} \|\bar{l}-l\|_2^2 , s.t. ~~ \sum_i \bar{l}_{i}^2\leq 1 $$
where $l=L^{(t)}_{j}, j=1,2,...,n$, $L^{(t)}_{j}$ is the $j^{th}$ column of $L^{(t)}$. Such a projection can be computed by $\bar{l} = \frac{l}{\max(1,\|l\|_2)}$. Therefore, the projection can be computed efficiently in linear time. Finally, by adapting the proofs in section \ref{sect:convergence}, we can draw the conclusion that the modified Algorithm \ref{algorithm:Lagrange} for optimizing workload decomposition for LRM under under ($\epsilon$, $\delta$)-differential privacy also converges to the a local KKT optimal point linearly. We omit the complete proofs for brevity.

%% file: exp.tex
\section{Experiments}\label{sec:exp}

This section experimentally evaluate the effectiveness of LRM under $\epsilon$- and ($\epsilon$, $\delta$)- differential privacy definitions. For $\epsilon$-differential privacy, we compares LRM against six state-of-the-art methods: Laplace mechanism (LM) \cite{DMNS06}, Privlet (WM) \cite{XWG10}, hierarchical mechanism (HM) \cite{HRMS10}, exponential smoothing (ESM) \cite{yuan2012low} (an implementation of the approximate matrix mechanism \cite{LHR+10}, described in Appendix \ref{sec:appedinex:esm}), adaptive mechanism (AM) \cite{li2012adaptive} (another implementation of the approximate matrix mechanism \cite{LHR+10,li2012adaptive}, described in Appendix \ref{sec:appedinex:am}) and the exponential mechanism with multiplicative weights update (MWEM) \cite{hardt2012simple}, whose performance depends on the dataset.  For ($\epsilon$, $\delta$)-differential privacy, we compare LRM against WM, HM, ESM, AM, and the Gaussian mechanism (GM) \cite{MM09}.

\textbf{Implementations}: For AM, we employ the Python implementation that can be obtained from the authors' website (\url{http://cs.umass.edu/~chaoli}). We use the default stopping criterion provided by the authors. For MWEM, we used Hardt et al's C\# code listed in the Appendix of \cite{hardt2012simple}). Note that MWEM needs to tune an additional parameter $T$ which denotes the number of iterations in order to ensure its performance. We follow the experimental setting in \cite{hardt2012simple}. Specifically, we choose $T\in\{10, 12, 14, 16\}$ in our experiments and reported the values for the best setting of $T$ in each case (Strictly speaking, such parameter tuning violates differential privacy; hence, the reported results are in favor on MWEM). For all remaining methods, we implemented them in Matlab, and published all code online (\url{http://yuanganzhao.weebly.com/}). We performed all experiments on a desktop PC with an Intel quad-core 2.50 GHz CPU and 4GBytes RAM. In each experiment, every algorithm is executed 20 times and the average performance is reported.

\textbf{Datasets}: We use four real-world data sets in our experiments \cite{HRMS10,XZXYY13,hardt2012simple}: \emph{Search Log}, \emph{Net Trace}, \emph{Social Network} and \emph{UCI Adult}. \emph{Search Log} includes search keyword statistics collected from \emph{Google Trends} and \emph{American Online} between 2004 and 2010. Each unit count is the number of appearances of a particular keyword. \emph{Social Network} contains information about users in a social network, where each unit count is the number of users with a specific degree in the social graph. \emph{Net Trace} is collected from a university intranet, where each unit count is the number of TCP packets related to a particular IP address. The total number of unit counts in \emph{Search Logs}, \emph{Net Trace} and \emph{Social Network} are $65,536$, $32,768$ and $11,342$ respectively. The UCI Adult data was extracted from the census bureau database in the U.S. Department of Commerce, it contains 14 features, among which six are continuous and eight are categorical. We use the following strategies to generate the sensitive data with varying domain size $n$. For the $\{$\emph{Search Log}, \emph{Net Trace}, \emph{Social Network}$\}$ data sets, we transform the original counts into a vector of fixed size $n$ (domain size), by merging consecutive counts in order. For the \emph{UCI Adult} data set, we only consider the combined $\{$workclass, education, occupation, race$\}$ attributes (with their total corresponding domain of size $\{8\times16 \times 14 \times5=8960\}$) and uniformly choose $n$ domains. The counting numbers of their corresponding records are used as the domain data. We observed that all the data sets $\{$\emph{Search Log}, \emph{Net Trace}, \emph{Social Network}$\}$ are dense with their sparsity exactly equals to $100\%$, while the \emph{UCI Adult} data set is sparse with its sparsity roughly $12\%\sim17\%$.

\textbf{Workloads}: We generated four different types of workloads, namely \emph{WDiscrete}, \emph{WRange}, \emph{WMarginal} and \emph{WRelated}.
In \emph{WDiscrete}, for each $W_{ij}$ (i.e., the coefficient of the $i$-th query on the $j$-th unit count), we set $W_{ij}=1$ with probability 0.02 and $W_{ij}=-1$ otherwise.
In \emph{WRange}, each query $q_i$ sums the unit counts in a range $[s_i, t_i] \subset [1, n]$, i.e., $W_{ij} = 1$ for $s_i \leq j \leq t_i$, and $W_{ij}=0$ otherwise. The start and end points $s_i$ and $t_i$ of each query $q_i$ is randomly generated, following the uniform distribution.
\emph{WMarginal} is used in \cite{li2012adaptive}, which contain queries that are uniformly sampled from the set of all 2-way marginals.
For \emph{WRelated}, we generate $s$ independent linear counting queries (called \emph{base queries}) with random weights following $(0,1)$-normal distribution. Let $A$ (of size $s \times n$) denote the workload matrix of the $s$ queries. We also generate another matrix $C$ of size $m\times s$ in a similar way The workload matrix $W$ is then the product of $C$ and $A$, i.e., the linear combination of base queries according to $C$.

\textbf{Parameters}: We test the impact of five parameters in our experiments:
$\gamma$, $r$, $n$, $m$ and $s$. $\gamma$ is the relaxation factor
defined in Program (\ref{eqn:relaxed-problem}). $r$ is the
number of intermediate queries in LRM, i.e., the number of columns in $B$ (and also the number of rows in $L$). $n$
is the number of unit counts and $m$ is the number of queries in the batch. Finally, $s$ is the number of base queries during the generation of \emph{WRelated}. The ranges and defaults (shown in bold) of the parameters are summarized in Table \ref{tab:exp:parameters}. Moreover, we test three different values of the privacy budget: $\epsilon=1$, $0.1$ and $0.01$. For ($\epsilon$, $\delta$)-differential privacy, following \cite{li2012adaptive}, we set $\delta=0.0001$.

\begin{table}[!h]
\tbl{Parameters used in the experiments.\label{tab:exp:parameters}}{%
\begin{tabular}{|c|c|}
\hline
$\gamma $ & $  0.0001, 0.001, \textbf{0.01}, 0.1, 1, 10$\\
\hline
$ r $ & $\{0.8, \textbf{1.0}, \textbf{1.2}, 1.4, 1.7, 2.1, 2.5, 3.0, 3.6\} \times rank(W) $\\
\hline
$ n $ & $ 128, 256, 512, \textbf{1024}, 2048, 4096, 8192$\\
\hline
$ m$ & $64, 128, \textbf{256}, 512, 1024$\\
\hline
$ s $ (during the generation of \emph{WRelated} ) & $ \{0.1, 0.2, 0.3, 0.4, \textbf{0.5}, 0.6, 0.7, 0.8, 0.9, 1.0\} \times min(m,n)$\\
\hline
\end{tabular}}
\end{table}%

In the experiments, we measure average squared error and computation time of the methods. Specifically, the
average squared error is the average squared $\mathcal{L}_2$
distance between the exact query answers and the noisy answers. In
the following, Section \ref{vary_gamma} examines the impact of $\gamma$ and $r$,
which are only used in LRM. The results provide important insights on how
to set these two parameters to maximize the utility of LRM. Then, Sections \ref{sec:vary_n} to \ref{sec:exp:scalability} compare LRM against existing methods.

\subsection{Impact of $\gamma$ and $r$ on LRM}\label{vary_gamma}

In LRM, the relaxation factor $\gamma$ controls the difference between $BL$ and $W$. In our first set of
experiments, we investigate the impact of $\gamma$ on the accuracy
and efficiency of LRM. Figure \ref{fig:exp:gamma} and Figure \ref{fig:exp:gamma:app} report
the performance of LRM with varying values for $\gamma$ under $\epsilon$-differential privacy and ($\epsilon$, $\delta$)-differential privacy respectively, using the \emph{Search Logs} dataset. Results on other datasets lead to similar conclusions, and are omitted for brevity.

The results in the Figure \ref{fig:exp:gamma} and Figure \ref{fig:exp:gamma:app} show that when $\epsilon$ is relatively low (meaning strong privacy),
the error of LRM is not sensitive to $\gamma$ regardless of the workload, for all values of $\gamma$ tested in the experiments ($(10^{-4}$ to 10). Only when $\epsilon$ reaches 1 does large values of $\gamma$ (e.g., $\gamma>1$) show negative impact on the performance of LRM. This negative effect is relatively small under $\epsilon$-differential privacy; it is more pronounced under ($\epsilon$, $\delta$)-differential privacy. The reason is that the error of LRM comes from two sources: the added noise and the difference between the decomposition $BL$ and the original workload $W$. When the privacy requirement is strong (i.e., when $\epsilon$ is relatively low, or when $\epsilon$-differential privacy is used), the error introduced by inexact decomposition is negligible compared to the noise added to satisfy differential privacy. Conversely, with looser privacy requirement (high $\epsilon$ and ($\epsilon$, $\delta$)-differential privacy definition), the noise level becomes low, and the error in decomposition becomes more evident. Nevertheless, when $\gamma \leq 0.1$, its impact is insignificant in all settings. Meanwhile, LRM runs much faster with a larger $\gamma$. Overall, $\gamma \leq 0.1$ is a safe choice, and a larger value of $\gamma$ is recommended for applications with strong privacy requirements. In the following experiments, we fix $\gamma$ to 0.01.


\begin{figure*}[!t]
\centering \subfigure[WDiscrete]
{\includegraphics[width=0.244\textwidth]{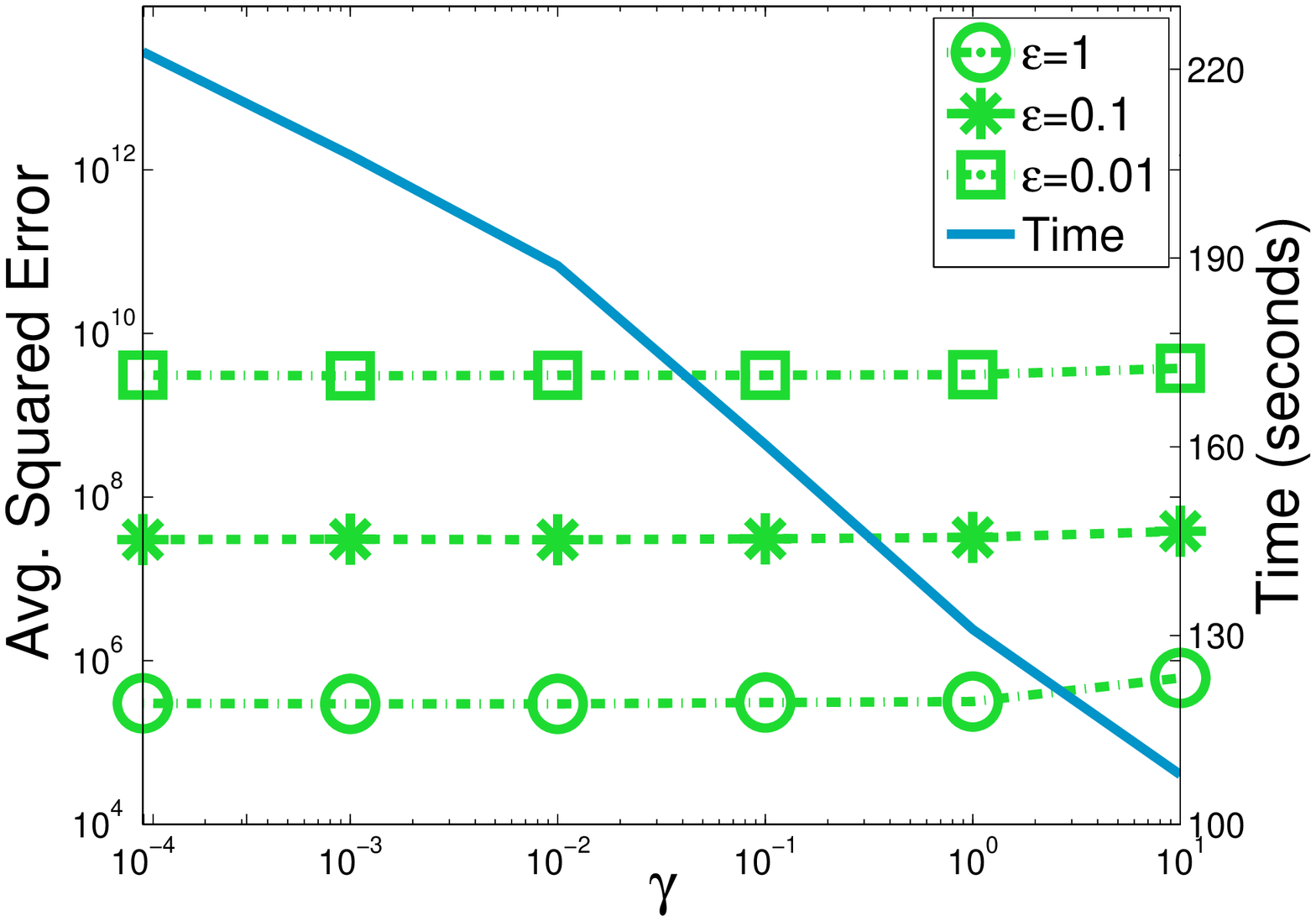}}
\centering \subfigure[WRange]
{\includegraphics[width=0.244\textwidth]{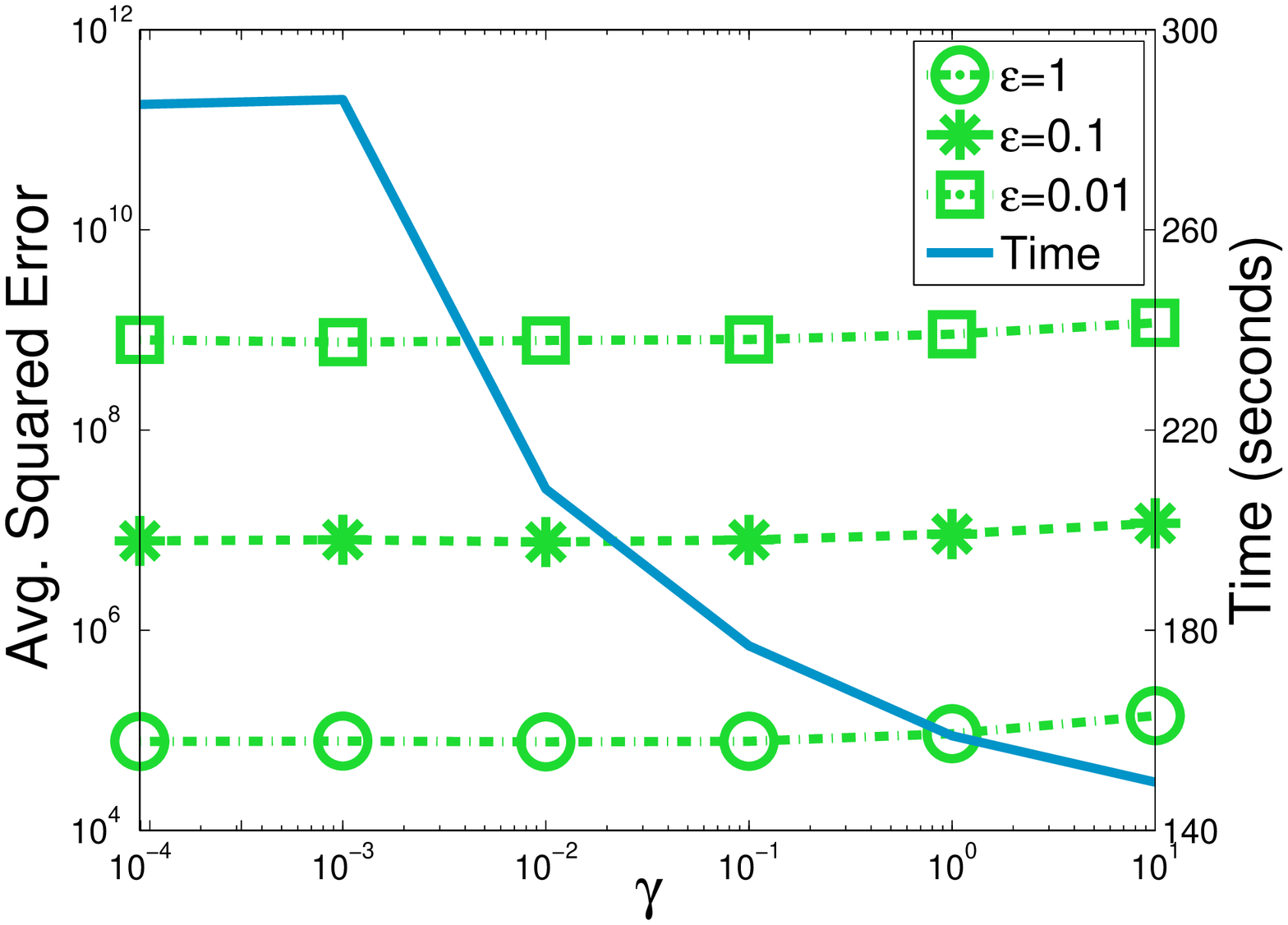}}
\centering \subfigure[WMarginal]
{\includegraphics[width=0.244\textwidth]{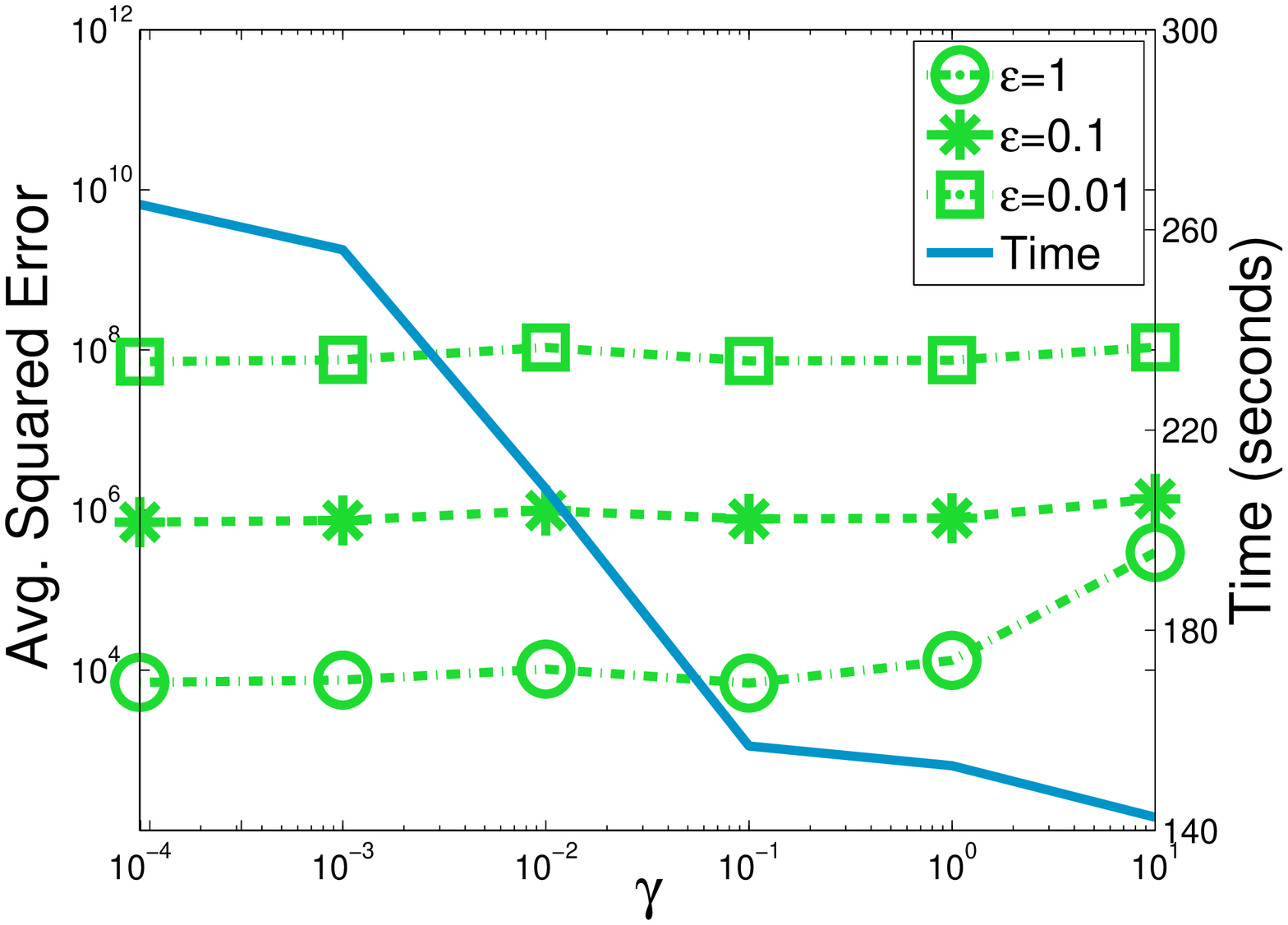}}
\centering \subfigure[WRelated]
{\includegraphics[width=0.244\textwidth]{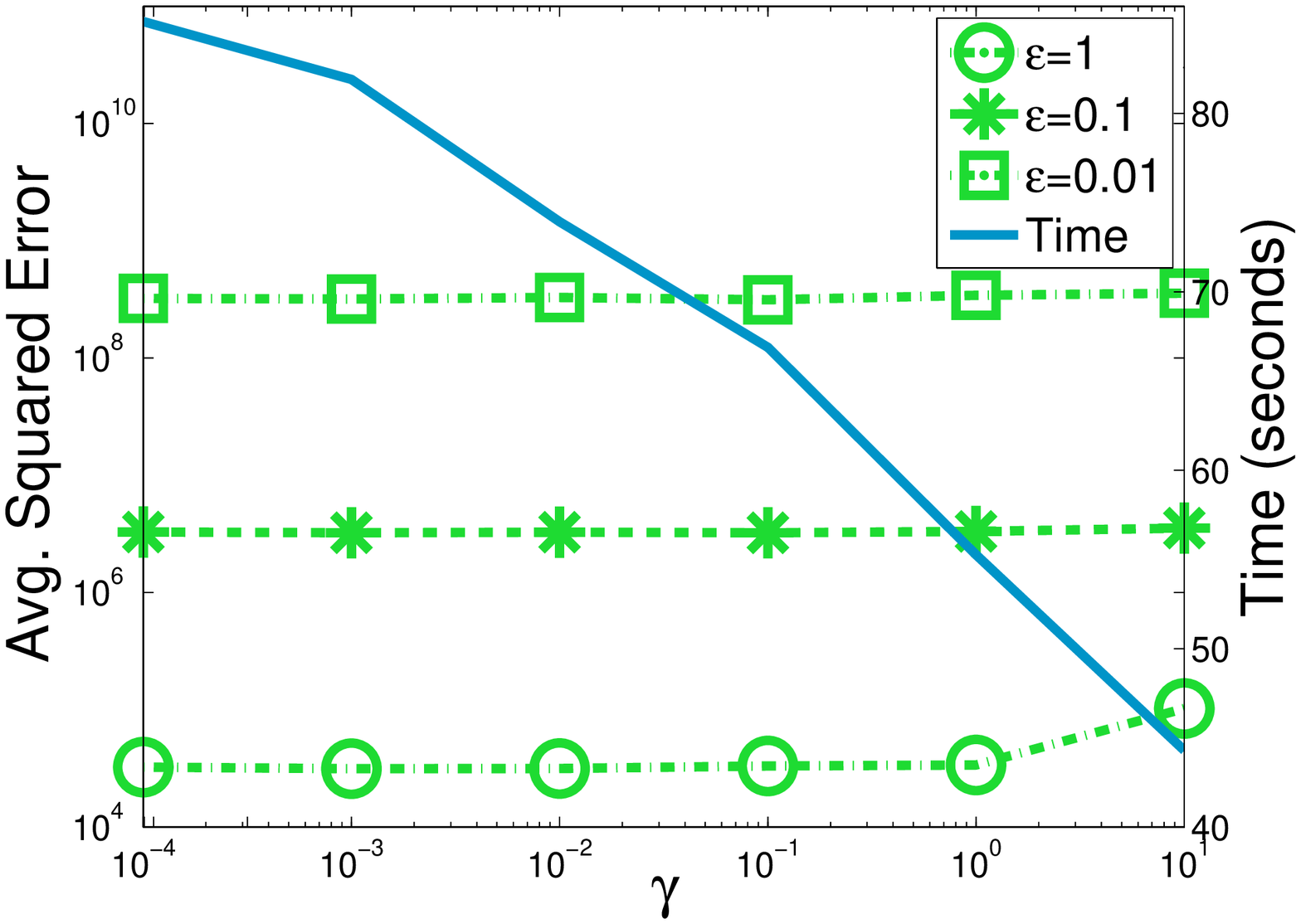}}
\vspace{-5pt} \caption{Effect of relaxation parameter
$\gamma$ on \emph{Search Logs} under $\epsilon$-differential privacy} \label{fig:exp:gamma}
\end{figure*}

\begin{figure*}[!t]
\centering \subfigure[WDiscrete]
{\includegraphics[width=0.244\textwidth]{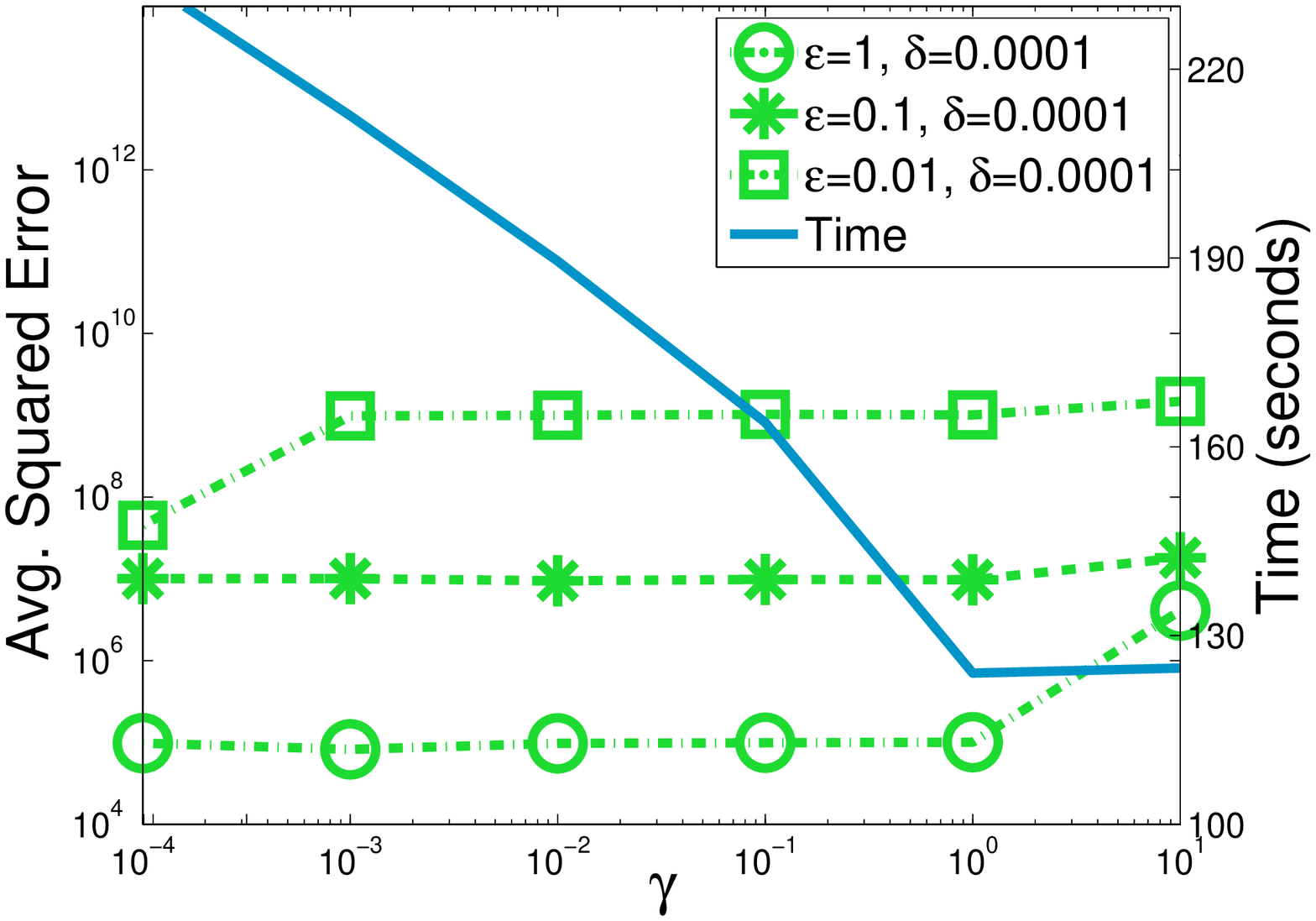}}
\centering \subfigure[WRange]
{\includegraphics[width=0.244\textwidth]{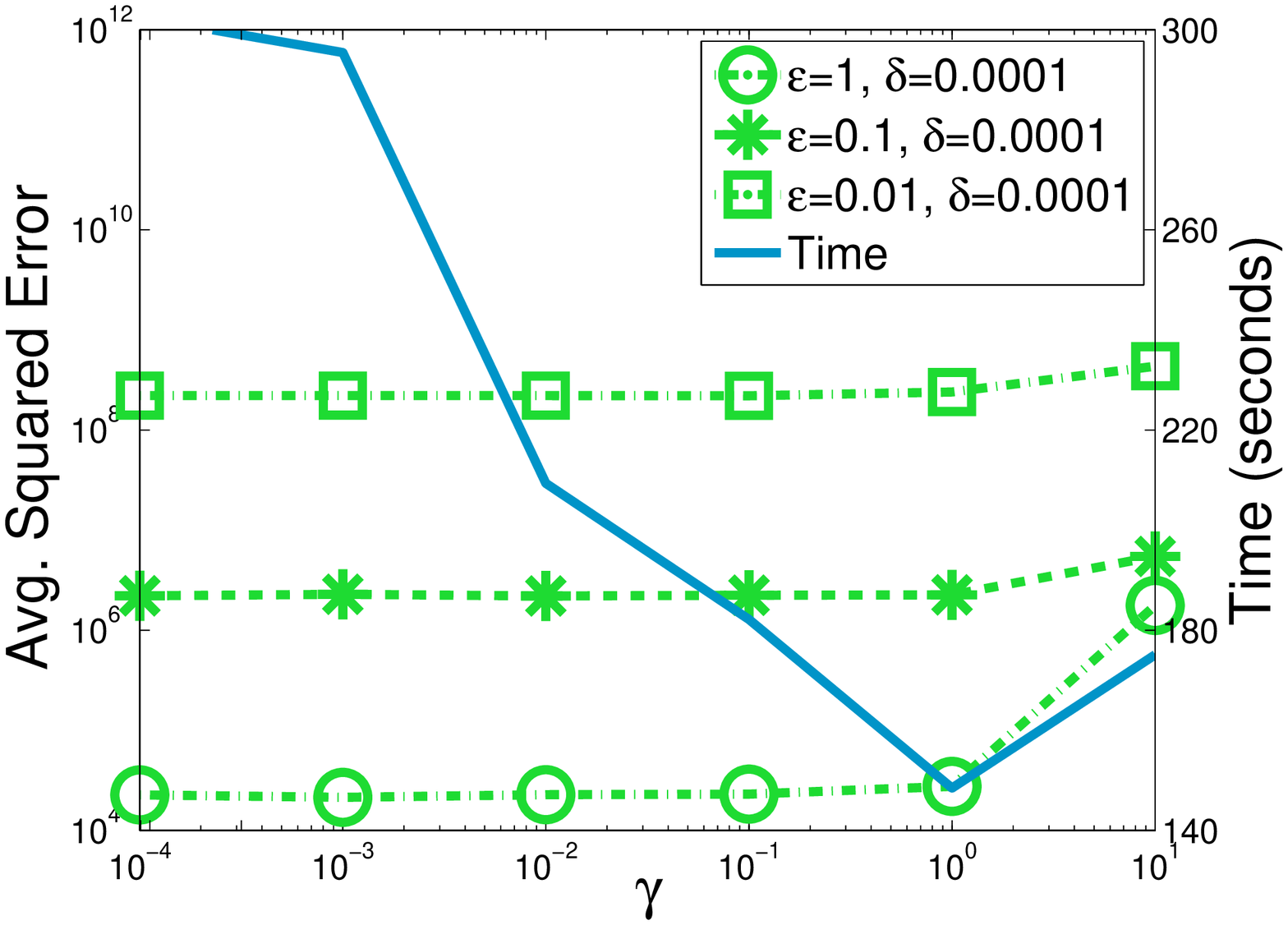}}
\centering \subfigure[WMarginal]
{\includegraphics[width=0.244\textwidth]{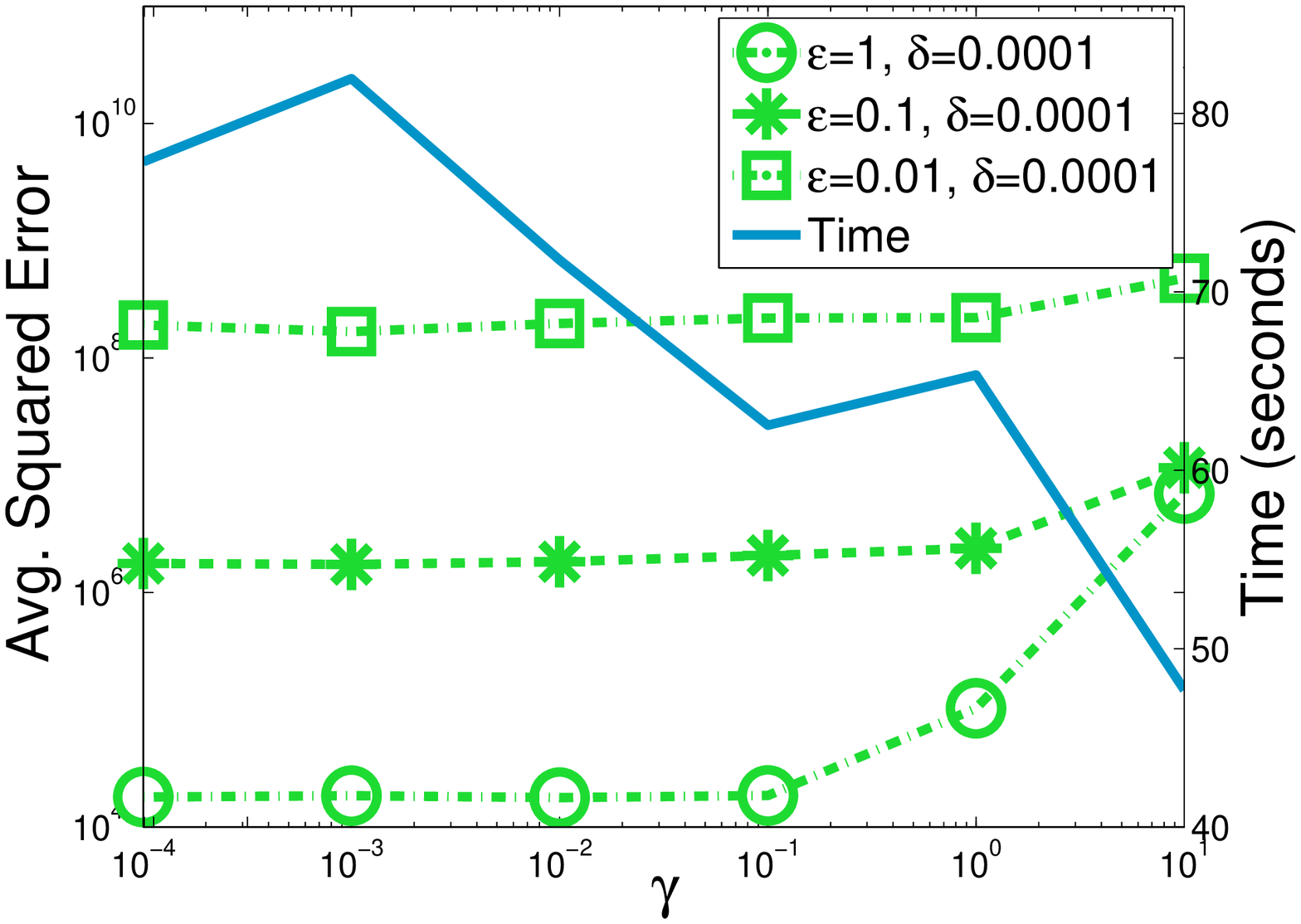}}
\centering \subfigure[WRelated]
{\includegraphics[width=0.244\textwidth]{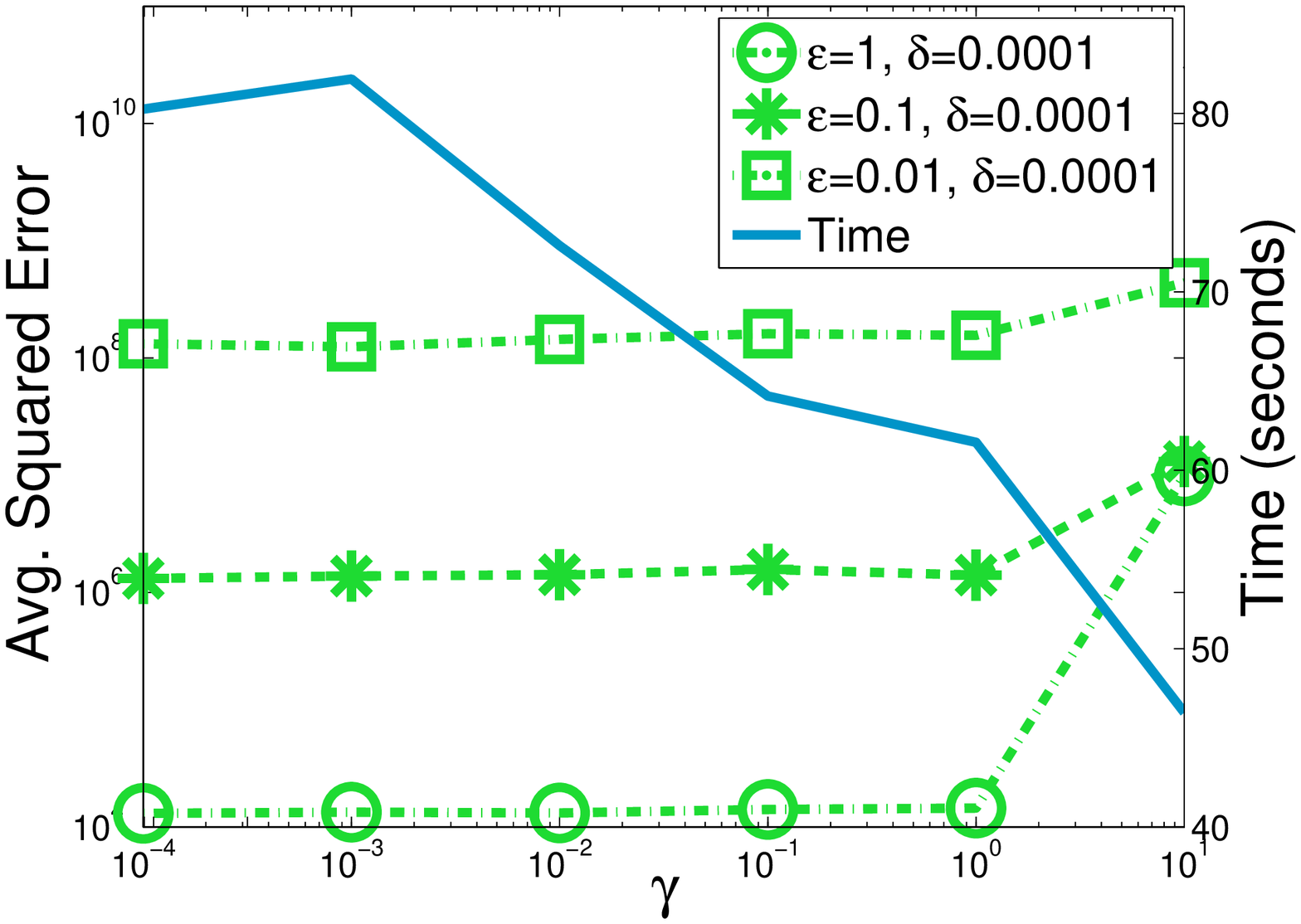}}
\vspace{-5pt} \caption{Effect of relaxation parameter
$\gamma$ on \emph{Search Logs} under ($\epsilon$, $\delta$)-differential privacy} \label{fig:exp:gamma:app}
\end{figure*}

\begin{figure*}[!t]
\centering \subfigure[WDiscrete]
{\includegraphics[width=0.244\textwidth]{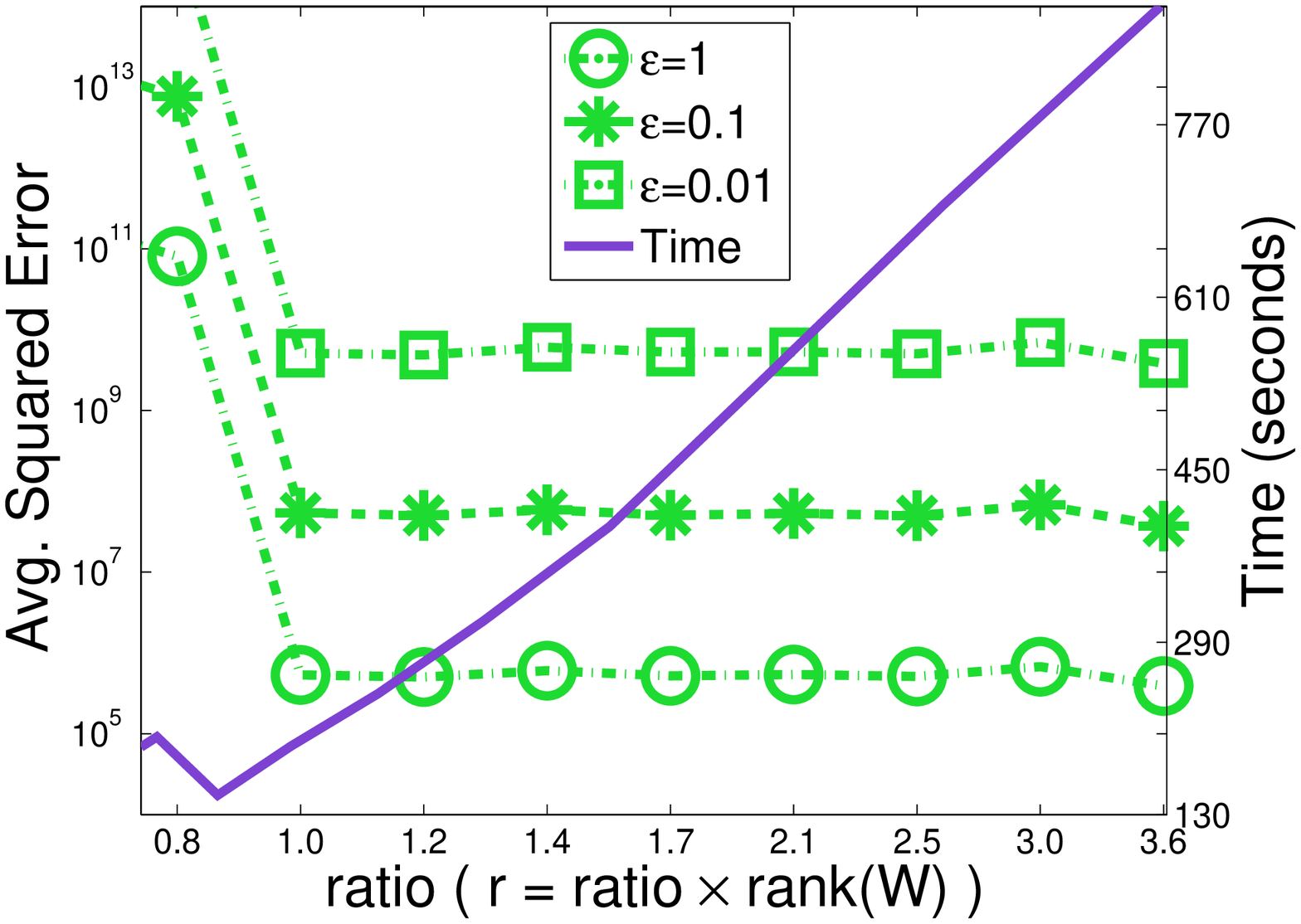}}
\centering \subfigure[WRange]
{\includegraphics[width=0.244\textwidth]{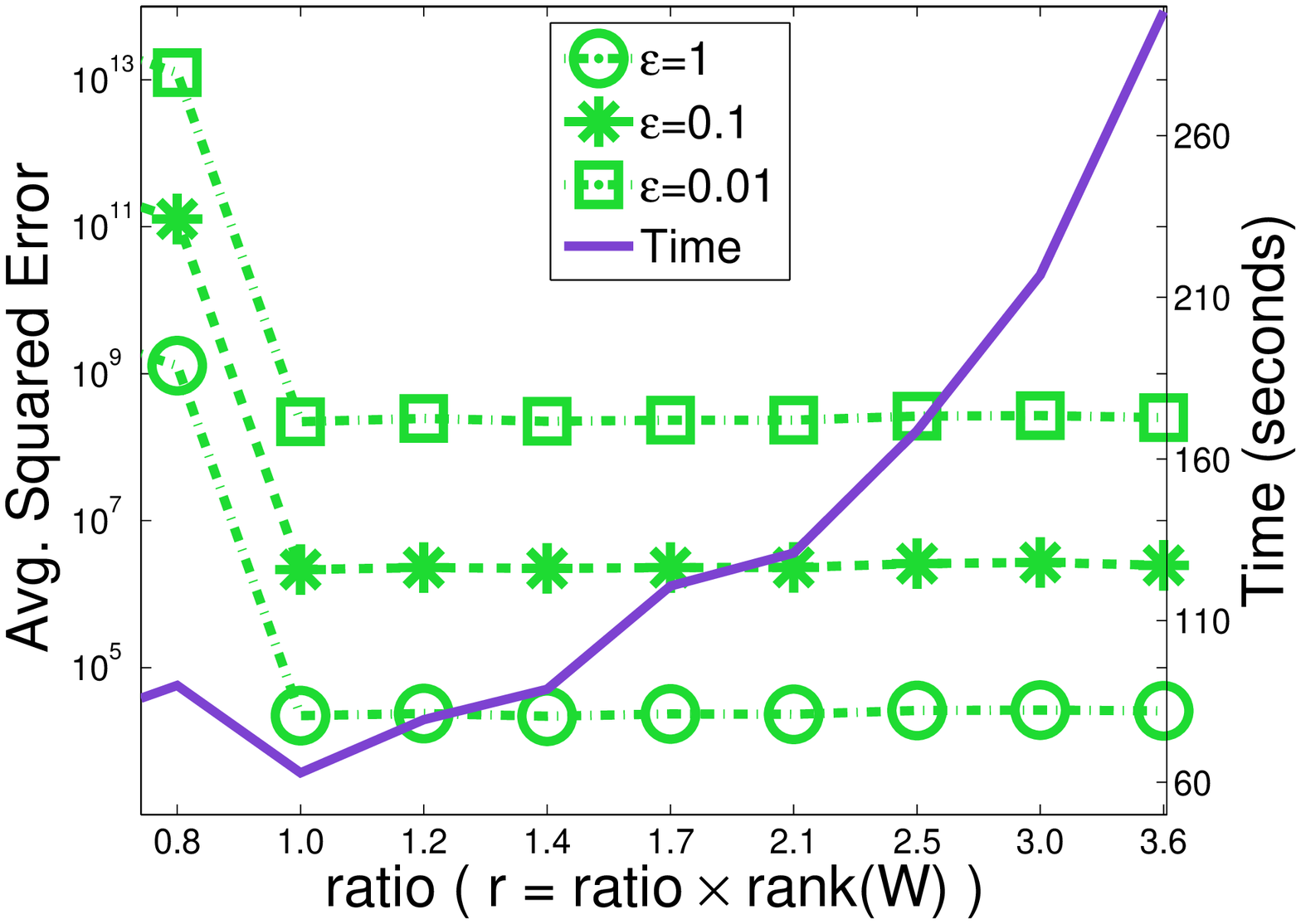}}
\centering \subfigure[WMarginal]
{\includegraphics[width=0.244\textwidth]{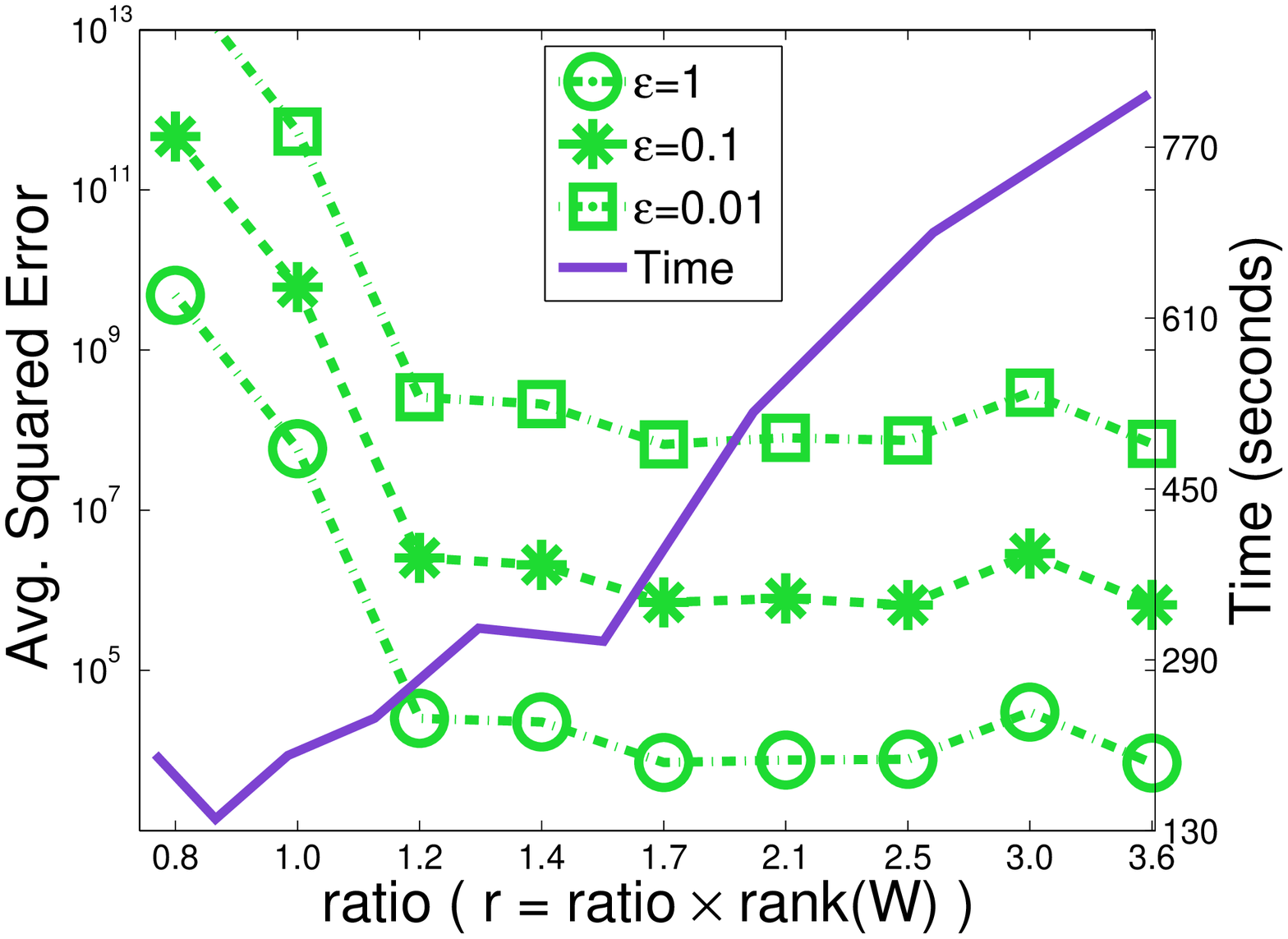}}
\centering \subfigure[WRelated]
{\includegraphics[width=0.244\textwidth]{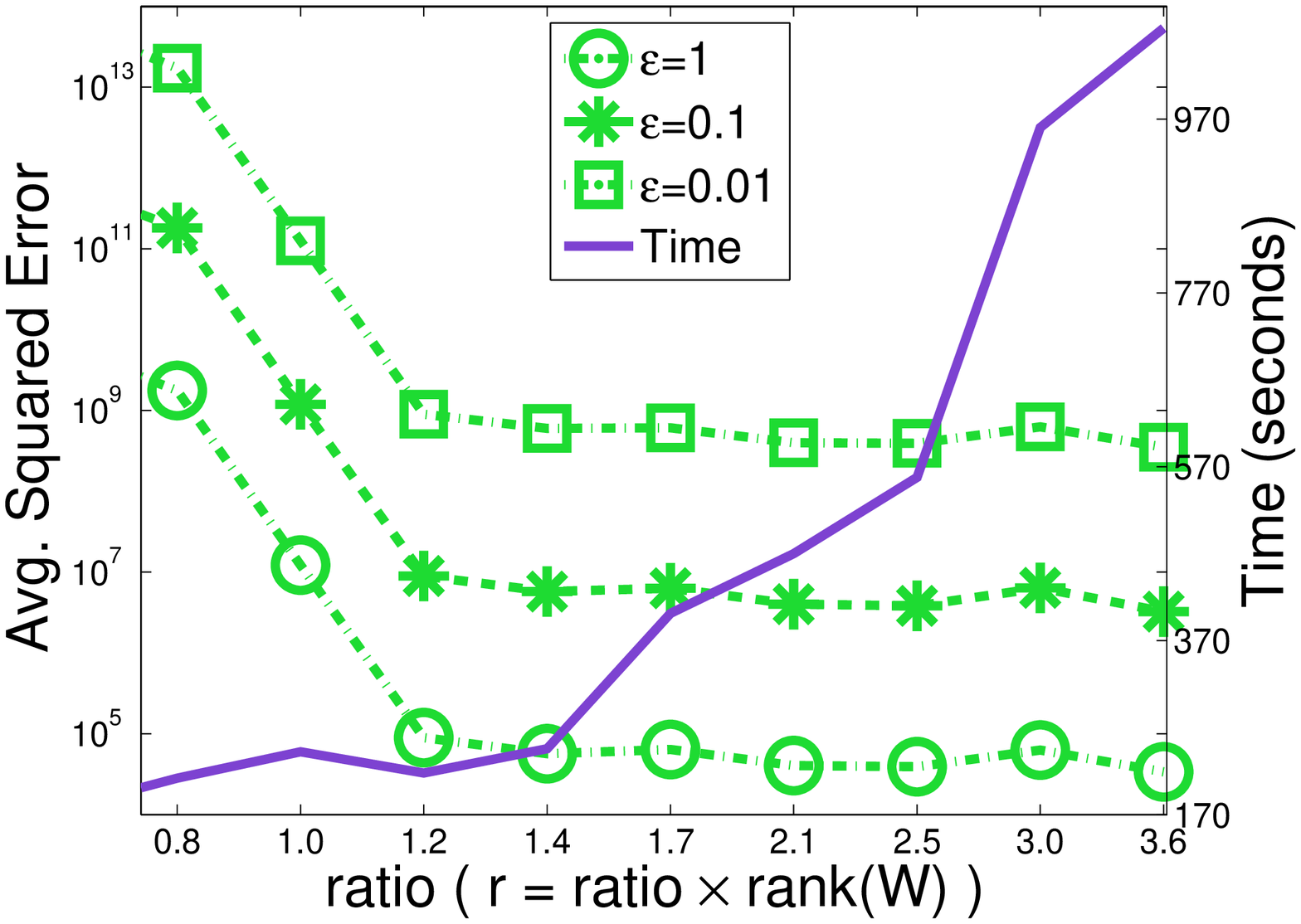}}
\vspace{-5pt} \caption{Effect of $r$ on \emph{Search Logs} under $\epsilon$-differential privacy}
\label{fig:exp:r}
\end{figure*}

\begin{figure*}[!t]
\centering \subfigure[WDiscrete]
{\includegraphics[width=0.244\textwidth]{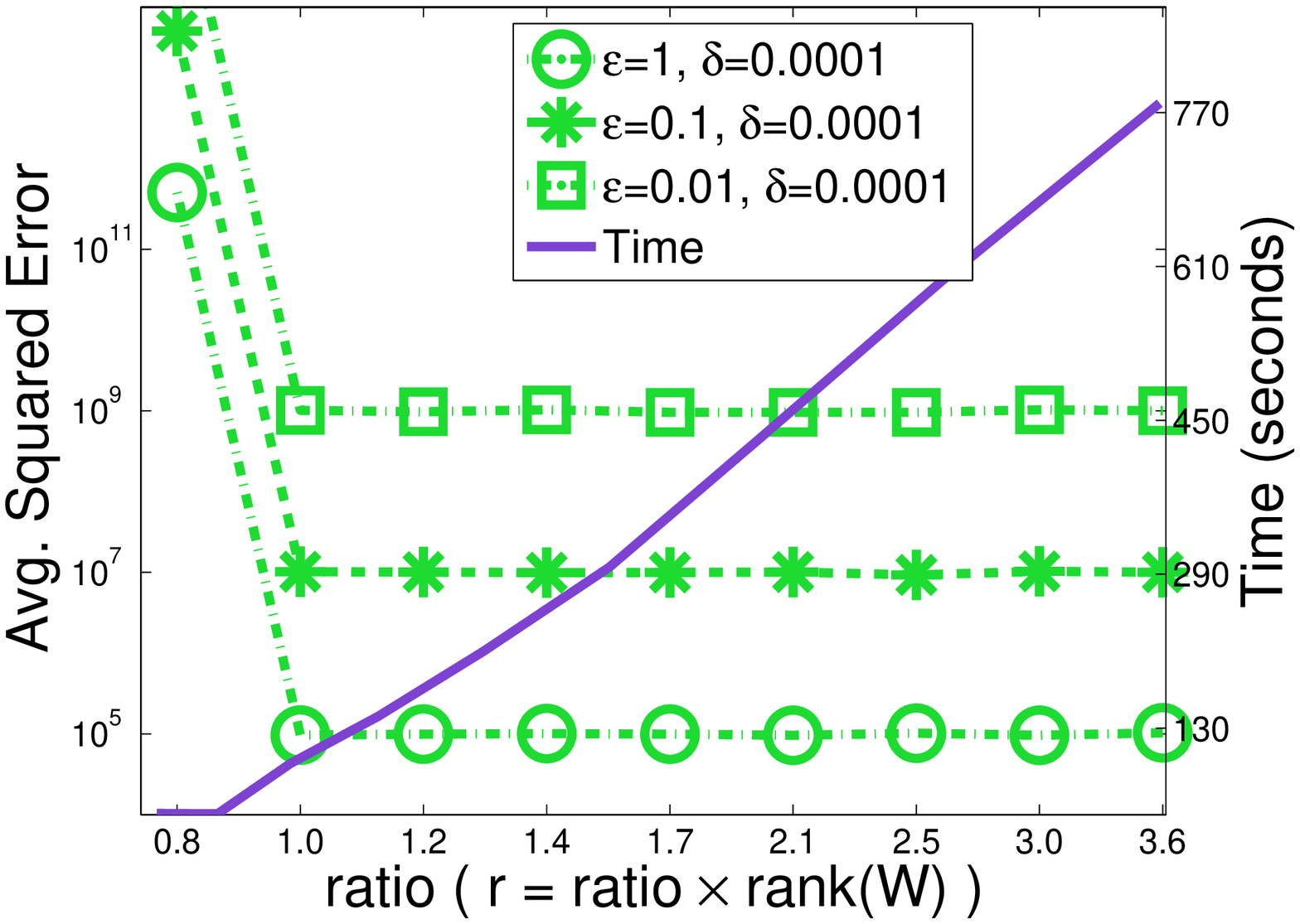}}
\centering \subfigure[WRange]
{\includegraphics[width=0.244\textwidth]{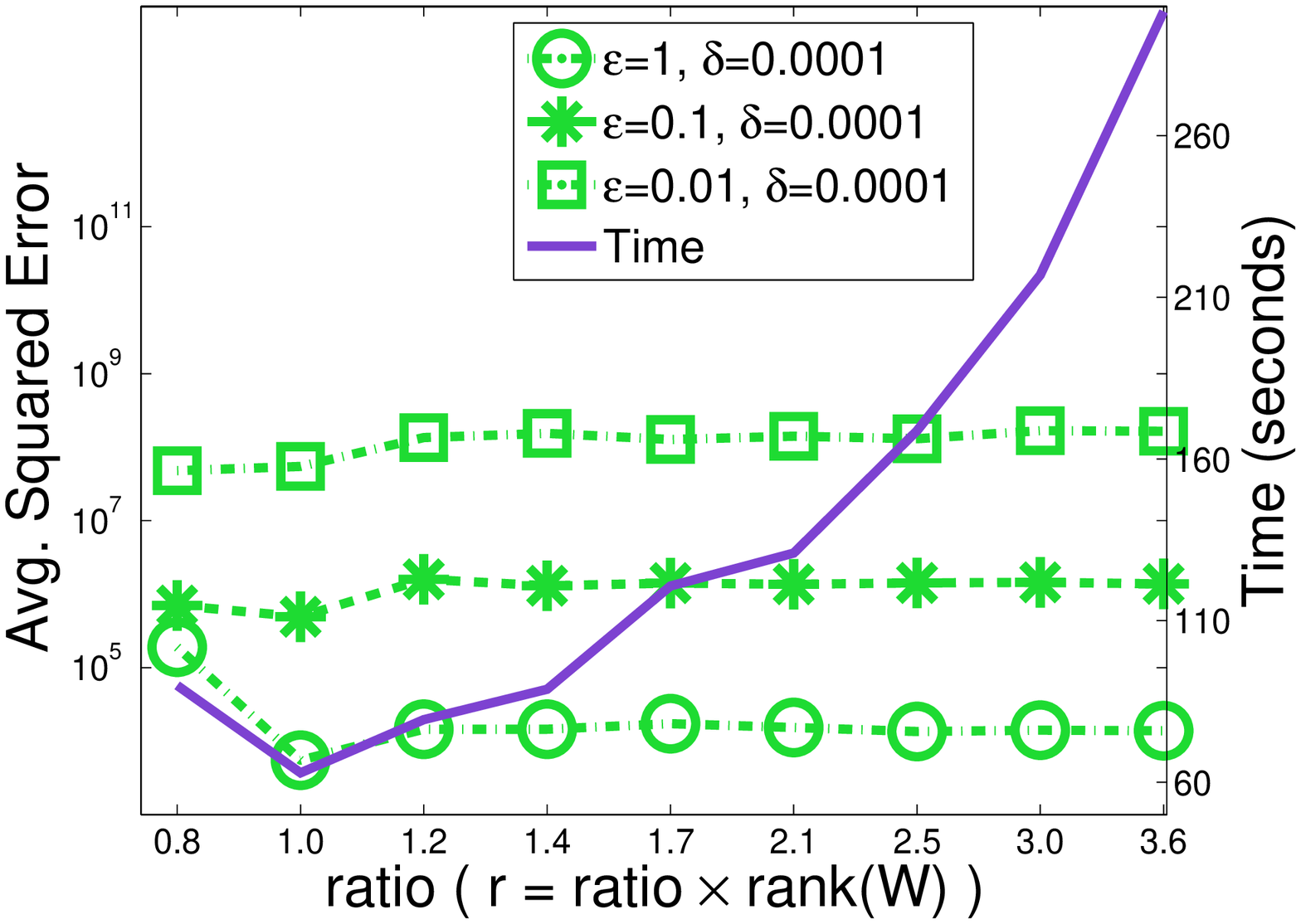}}
\centering \subfigure[WMarginal]
{\includegraphics[width=0.244\textwidth]{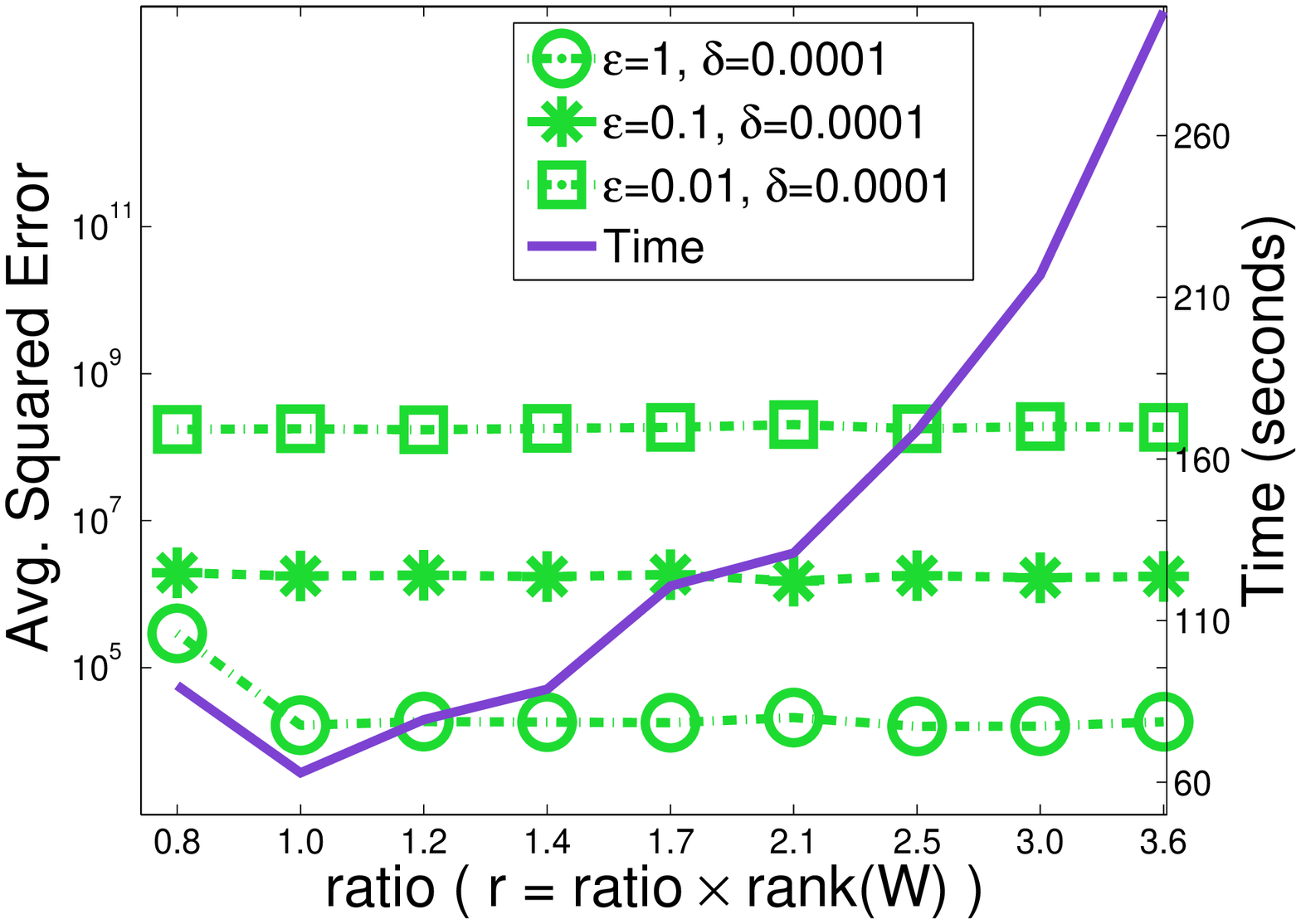}}
\centering \subfigure[WRelated]
{\includegraphics[width=0.244\textwidth]{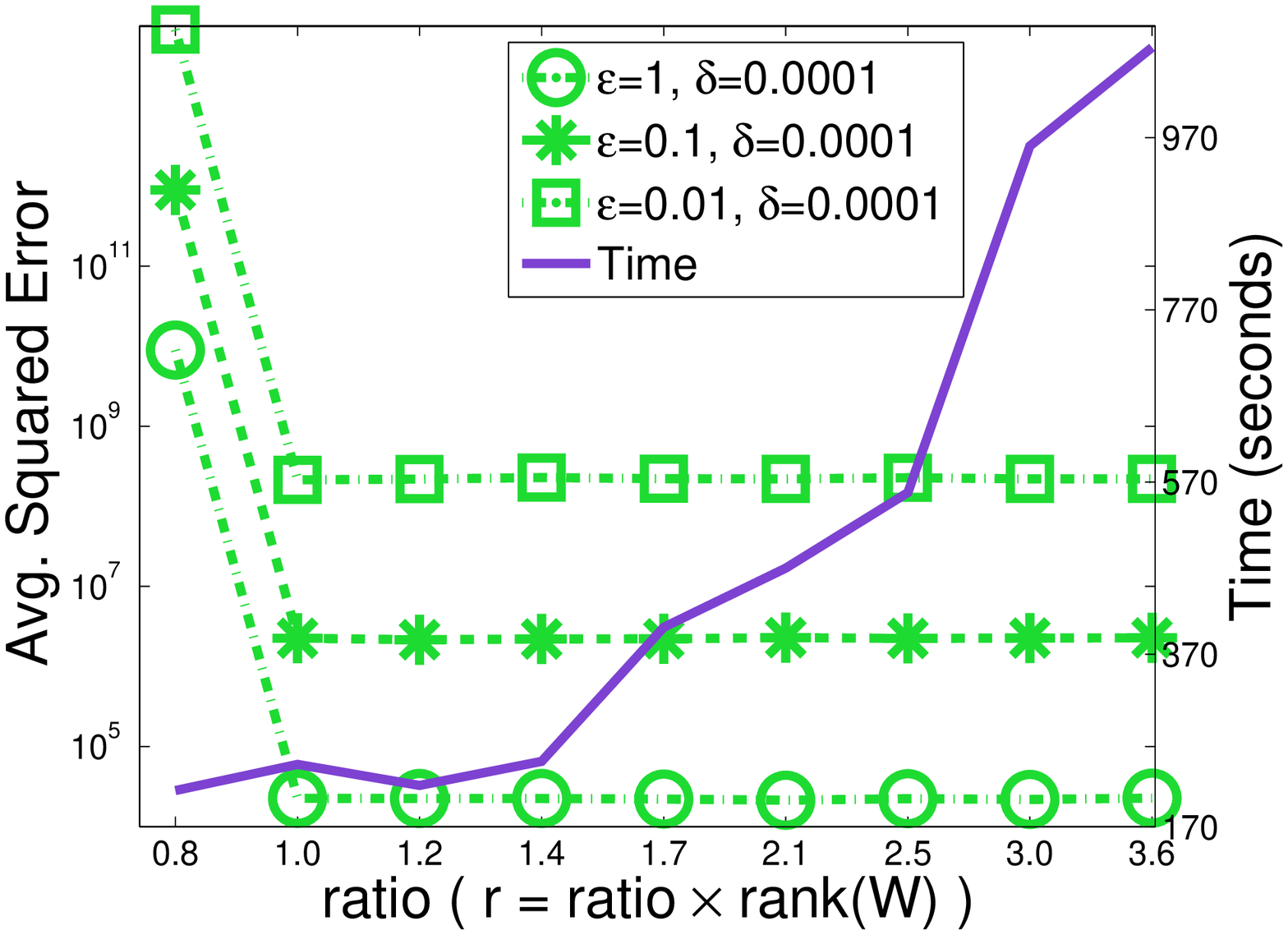}}
\vspace{-5pt} \caption{Effect of $r$ on \emph{Search Logs} under ($\epsilon$, $\delta$)-differential privacy}
\label{fig:exp:r:app}
\end{figure*}

$r$ is another important parameter in LRM that determines the rank of
the matrix $BL$ that approximates the workload $W$. $r$ affects both the
approximation accuracy and the optimization speed. When $r$ is too small, e.g., when $r<rank(W)$, our optimization formulation may fail to find a good approximation, leading to suboptimal
accuracy for the query batch. On the other hand, an overly large $r$ leads to poor efficiency, as the search space expands dramatically. We thus test LRM with varying $r$, by controlling the ratio of $r$ to the actual rank $rank(W)$, on the \emph{Search Log} dataset. We record the average
squared error and running time of LRM for all the workloads under $\epsilon$ and ($\epsilon$, $\delta$)-differential privacy, and report them in Figure
\ref{fig:exp:r} and Figure \ref{fig:exp:r:app} respectively.

There are several important observations in Figure \ref{fig:exp:r} and Figure \ref{fig:exp:r:app}. First, a value of $r$ below $rank(W)$ leads to far worse accuracy (up to two orders of magnitude) compared to settings with higher values of $r$. Second, the performance of LRM becomes stable when $r$ exceeds $1.2\cdot rank(W)$ for $\epsilon$-differential privacy, and $1.0\cdot rank(W)$ for ($\epsilon$, $\delta$)-differential privacy. This is because the optimization formulation has enough freedom to find the optimal decomposition when $r>rank(W)$. For ($\epsilon$, $\delta$)-differential privacy, this result is expected, because any decomposition $W=BL$ with $r>rank(W)$ can be transformed into a decomposition $B'L'$ with $r = rank(W)$, by projecting the columns of $L$ and the rows of $B$ onto the range of $L$, which does not affect the $\mathcal{L}_2$-sensitivity of $B$. Finally, the amount of computations for workload decomposition increases linearly with $r$ (note that both axes are in logarithmic scale). Thus, to balance the efficiency and effectiveness of LRM, a good value for $r$ is between $rank(W)$ and $1.2\cdot rank(W)$. In subsequent experiments, we set $r=1.2\cdot rank(W)$ and $r=1.0\cdot rank(W)$ for $\epsilon$ and ($\epsilon$, $\delta$)-differential privacy, respectively.

\subsection{Impact of Varying Domain Size $n$}\label{sec:vary_n}

\begin{figure*}[!t]
\centering
\subfigure[\emph{Search Logs}]
{\includegraphics[width=0.244\textwidth]{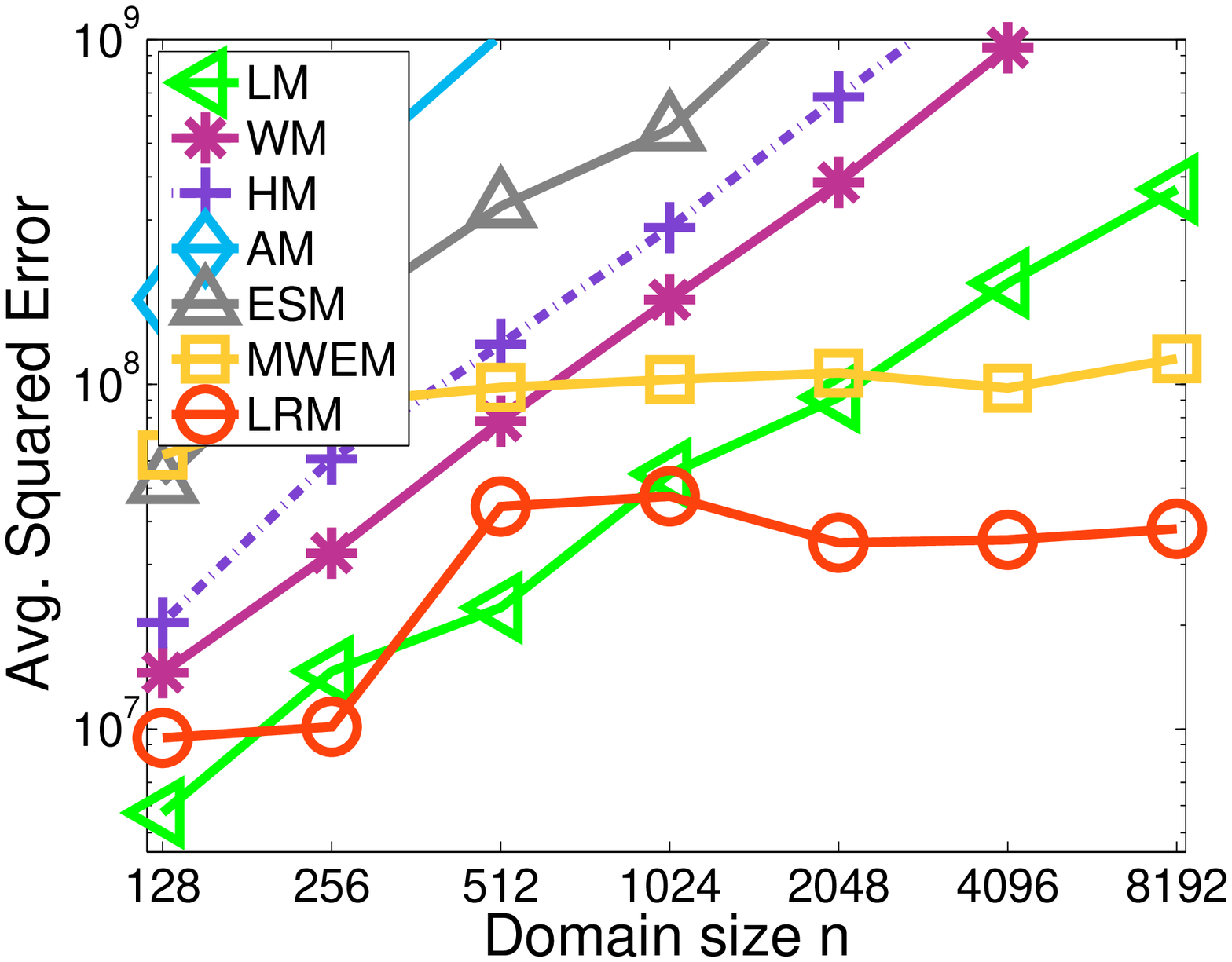}}
\subfigure[\emph{Net Trace}]
{\includegraphics[width=0.244\textwidth]{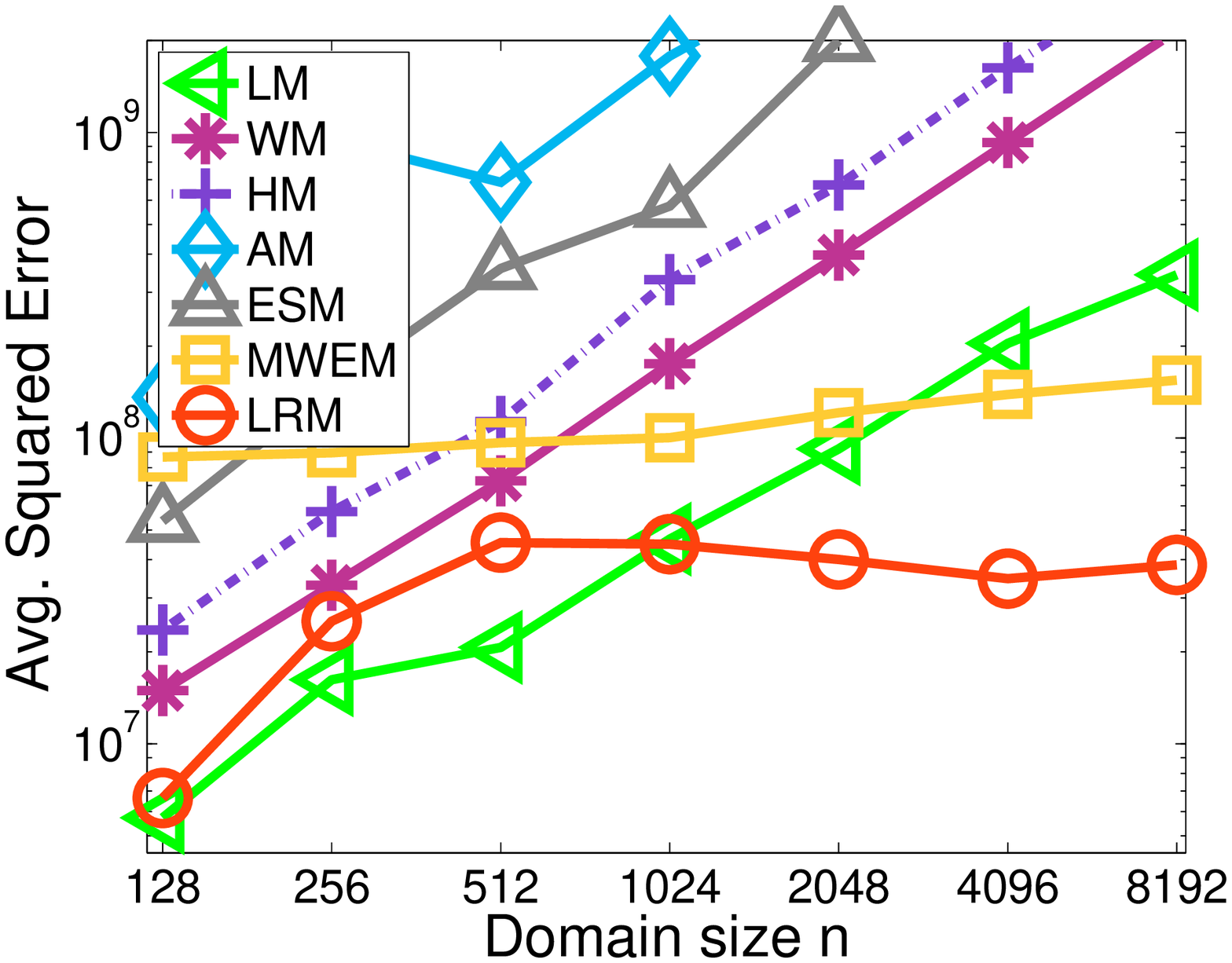}}
\centering
\subfigure[\emph{Social Network}]
{\includegraphics[width=0.244\textwidth]{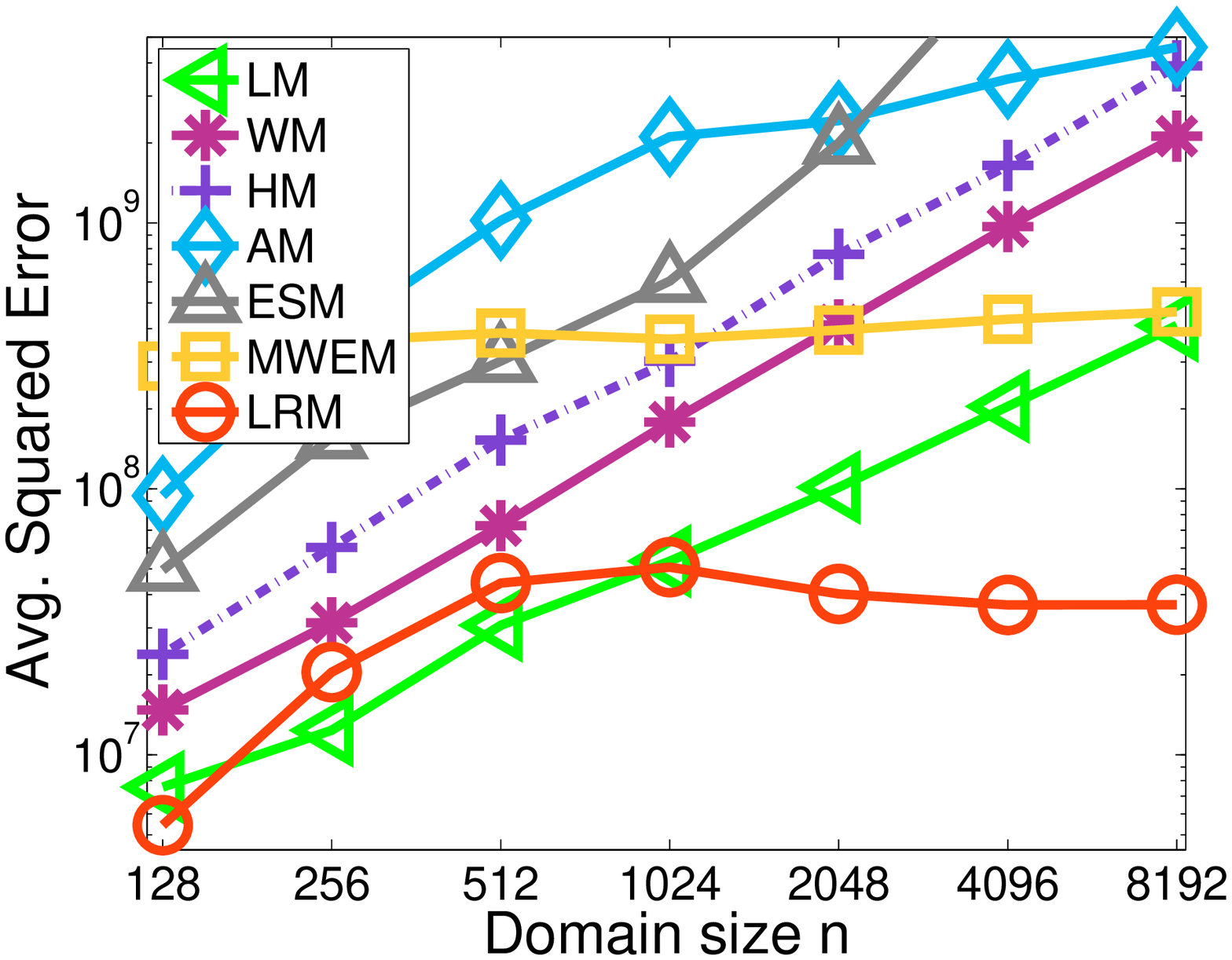}}
\centering
\subfigure[\emph{UCI Adult}]
{\includegraphics[width=0.244\textwidth]{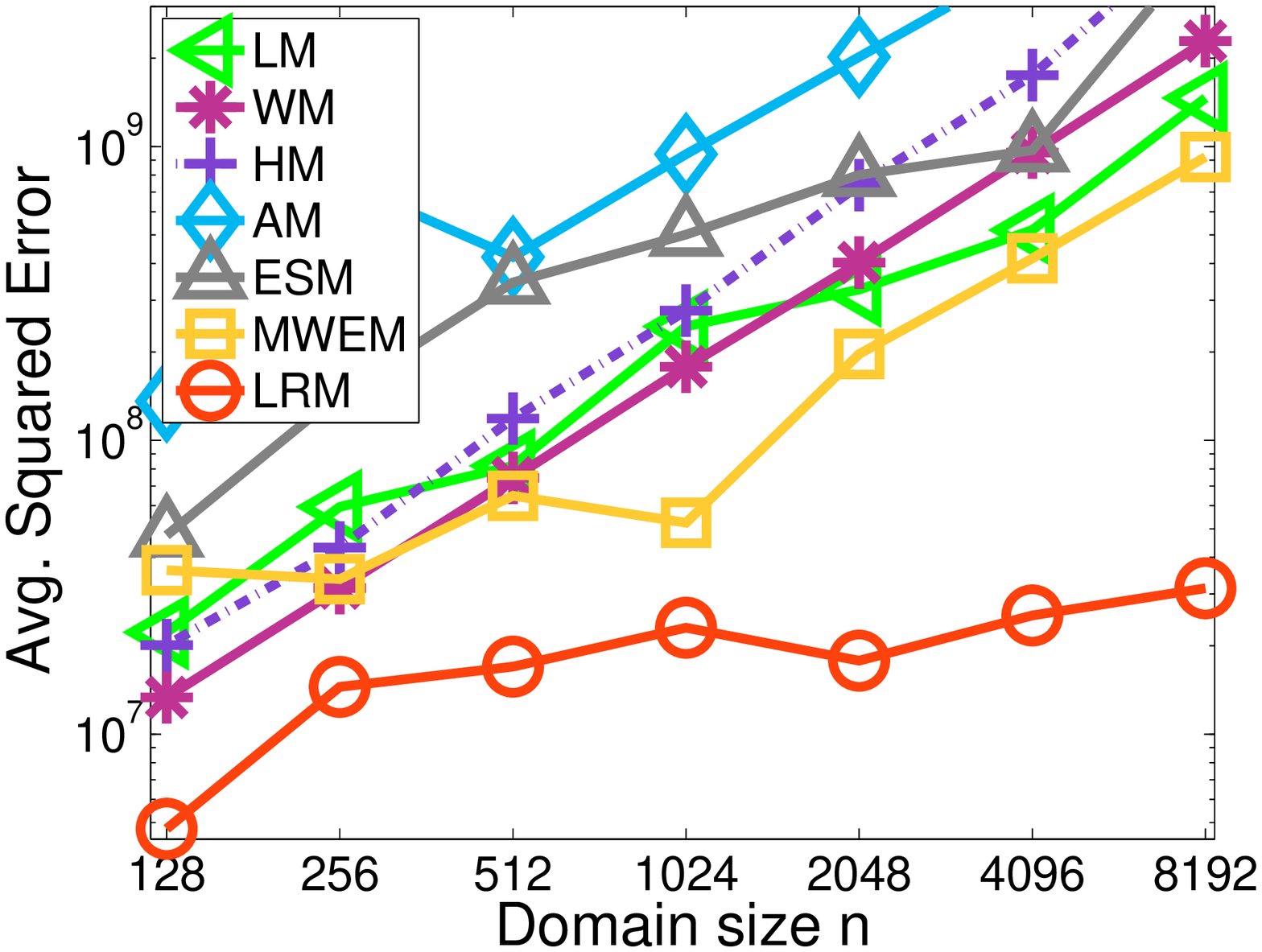}}
\caption{Effect of domain size $n$ on workload \emph{WDiscrete}
under $\epsilon$-differential privacy with $\epsilon=0.1$}\label{fig:exp:n:WDiscrete}
\end{figure*}

\begin{figure*}[!t]
\centering
\subfigure[\emph{Search Logs}]
{\includegraphics[width=0.244\textwidth]{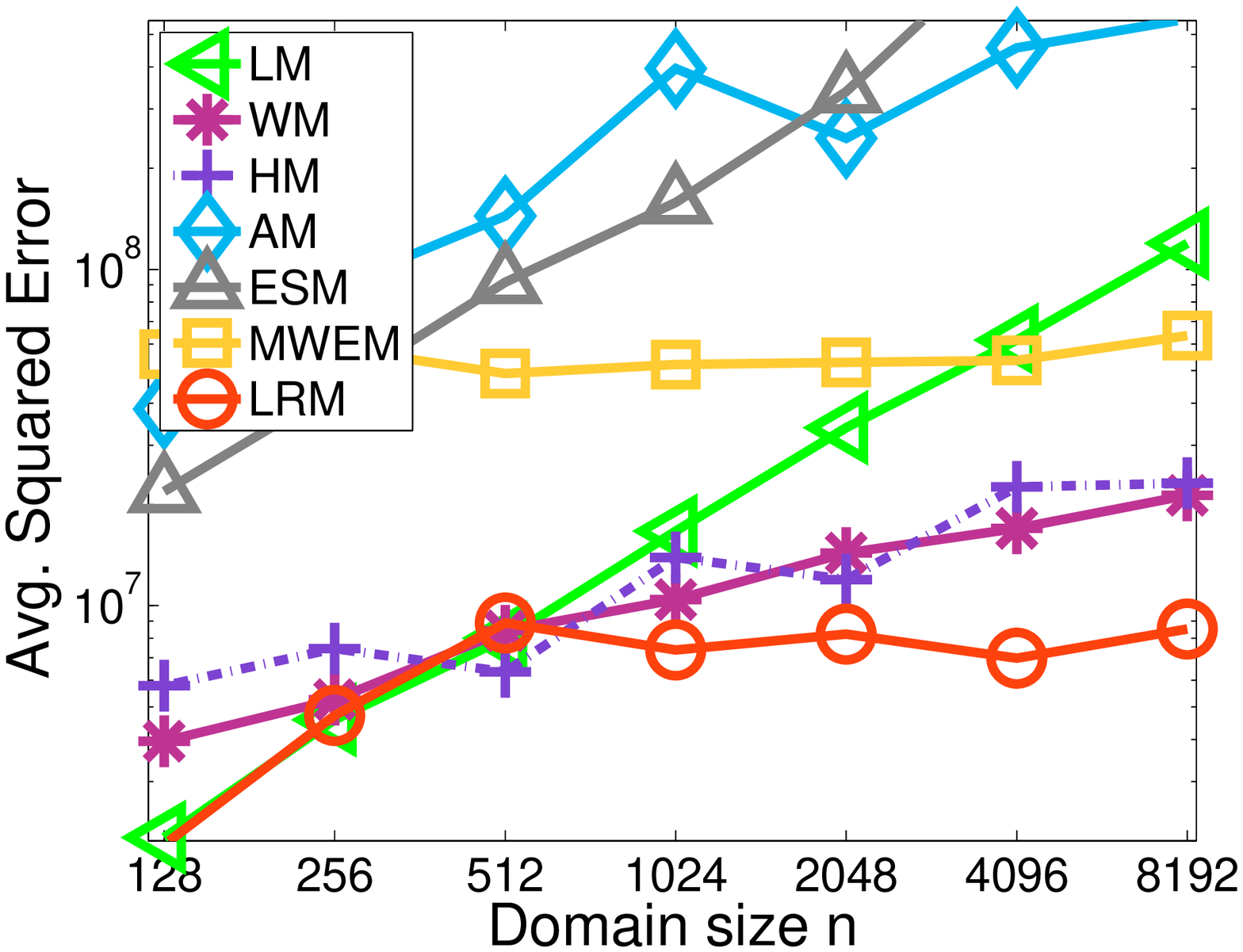}}
\subfigure[\emph{Net Trace}]
{\includegraphics[width=0.244\textwidth]{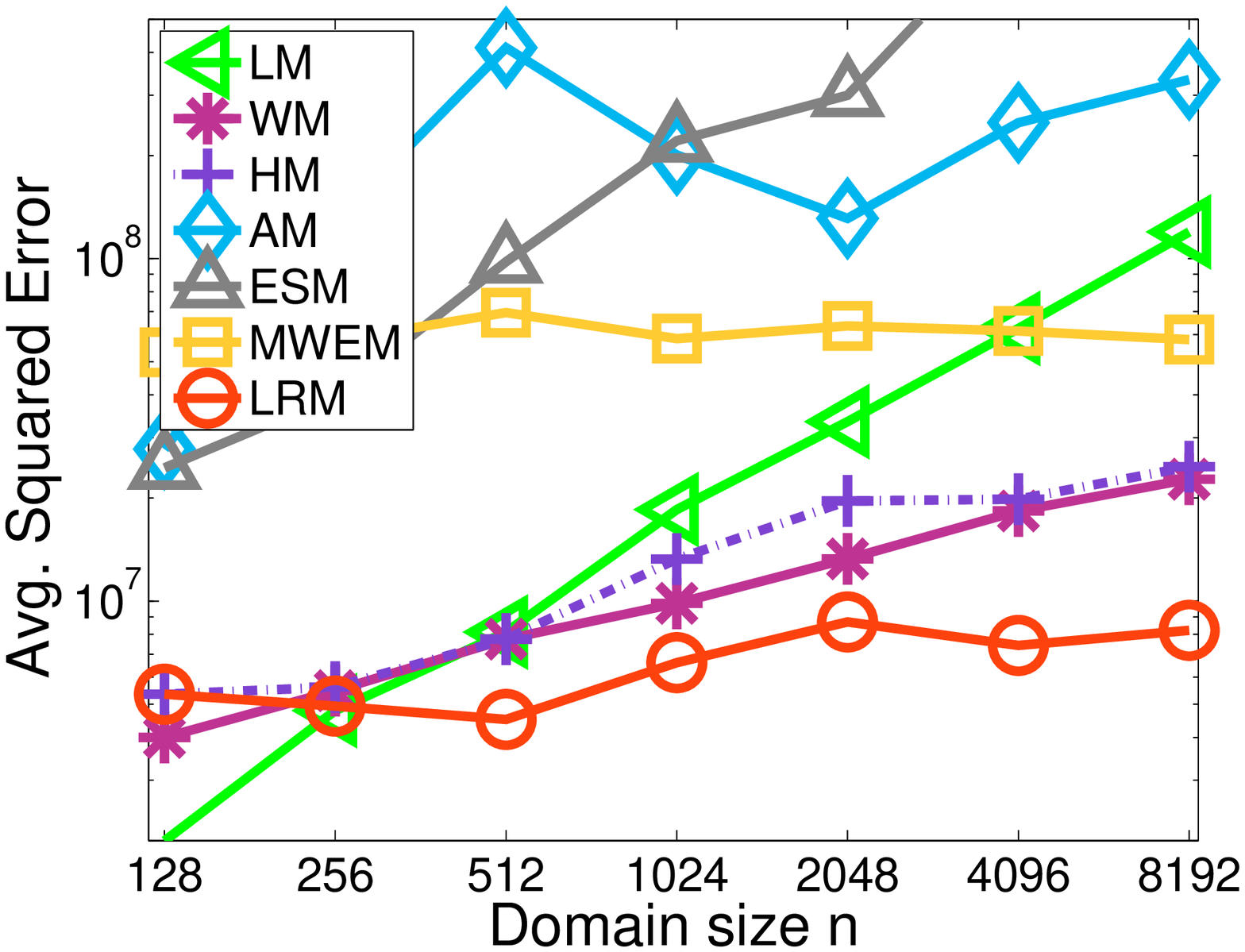}}
\centering
\subfigure[\emph{Social Network}]
{\includegraphics[width=0.244\textwidth]{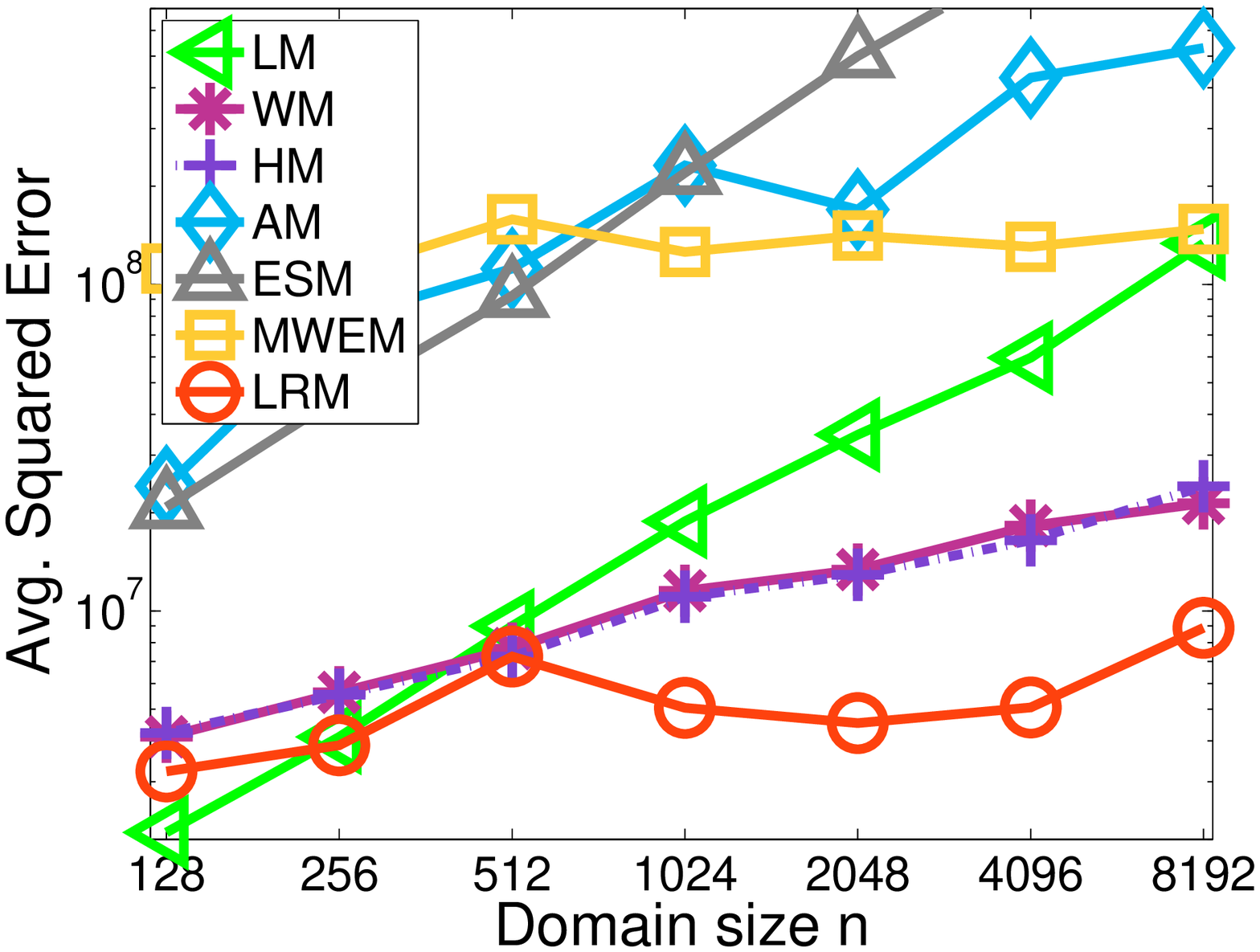}}
\centering
\subfigure[\emph{UCI Adult}]
{\includegraphics[width=0.244\textwidth]{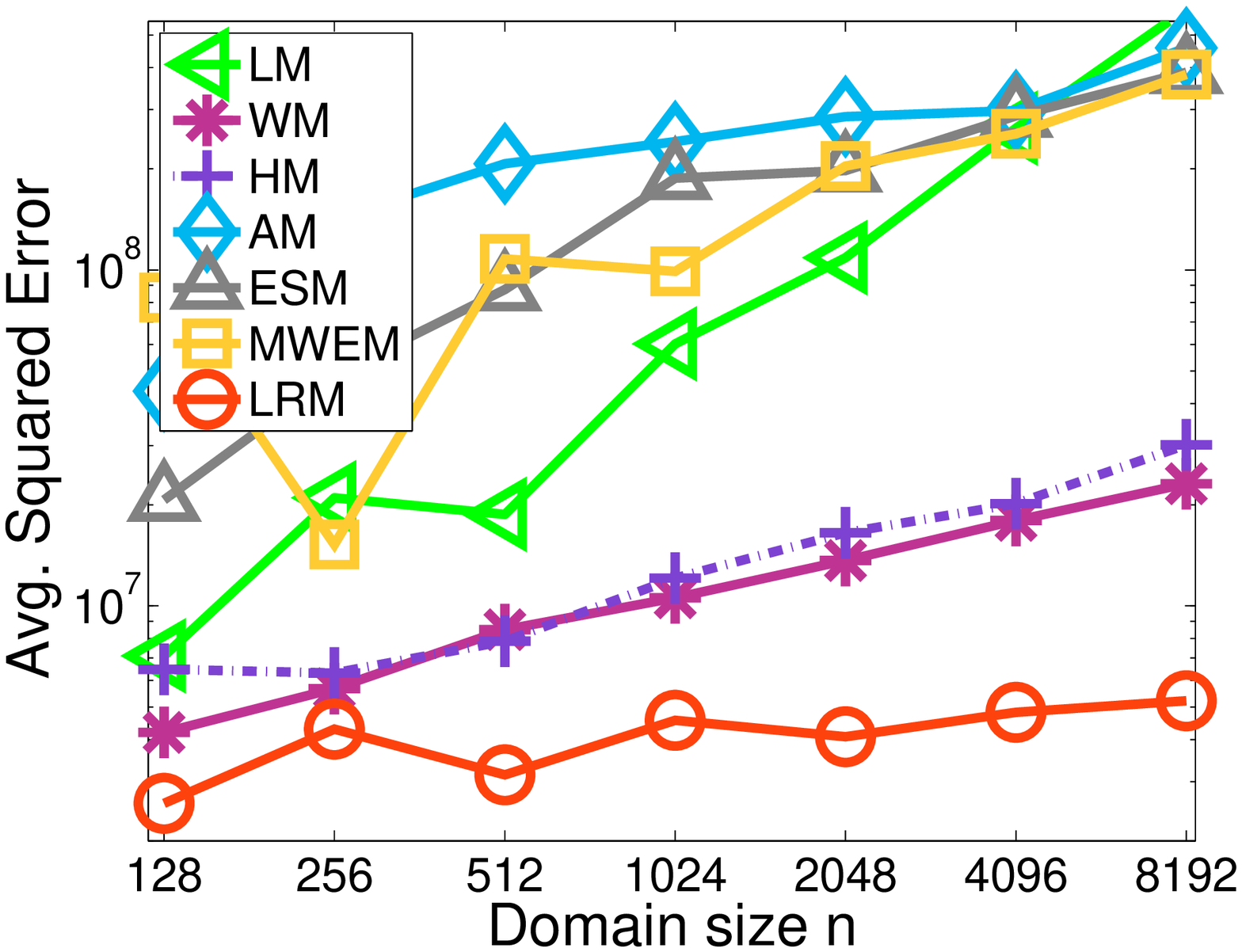}}
\caption{Effect of domain size $n$ on workload \emph{WRange} under
$\epsilon$-differential privacy with $\epsilon=0.1$}\label{fig:exp:n:WRange}
\end{figure*}

\begin{figure*}[!t]
\centering \subfigure[\emph{Search Logs}]
{\includegraphics[width=0.244\textwidth]{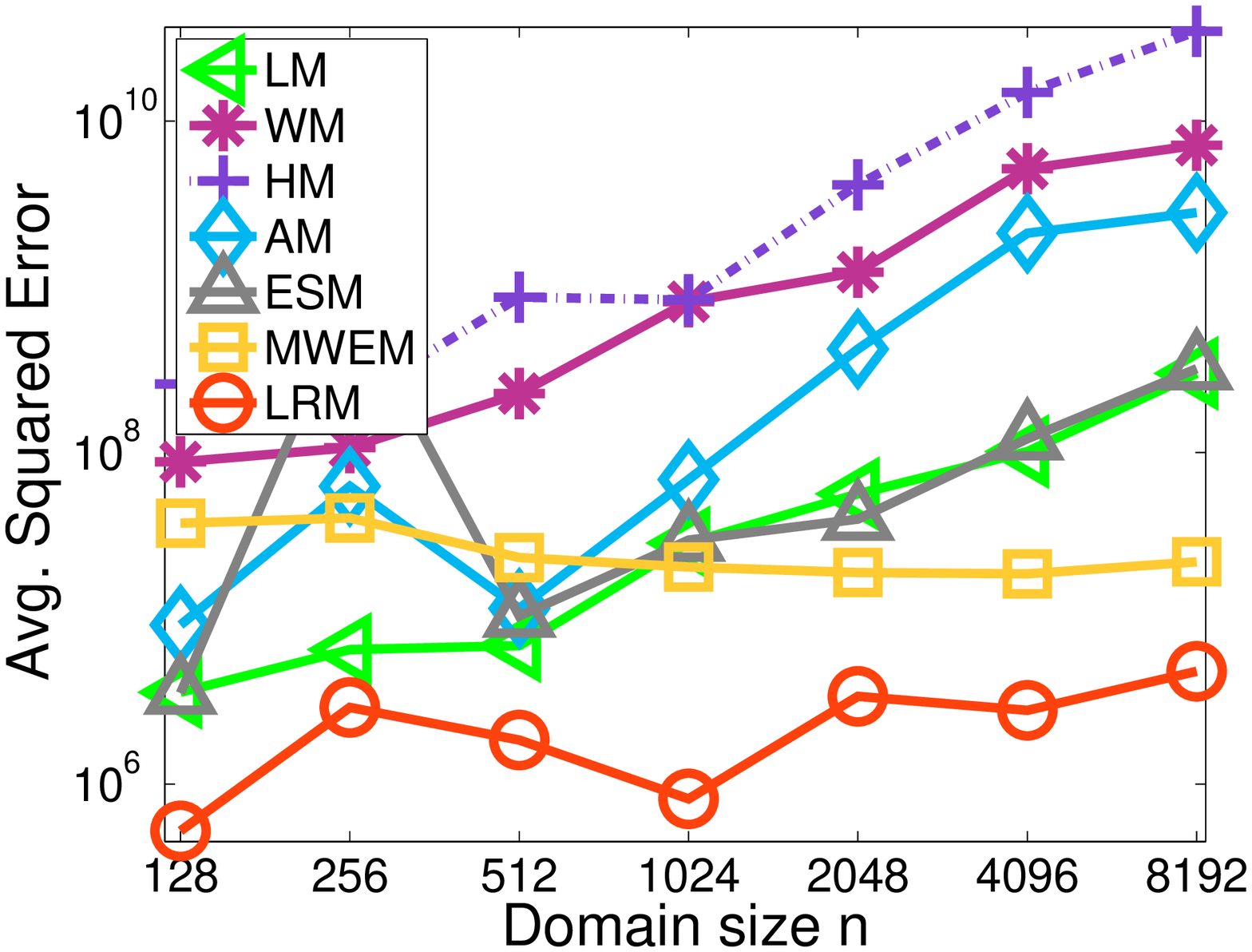}}
\subfigure[\emph{Net Trace}]
{\includegraphics[width=0.244\textwidth]{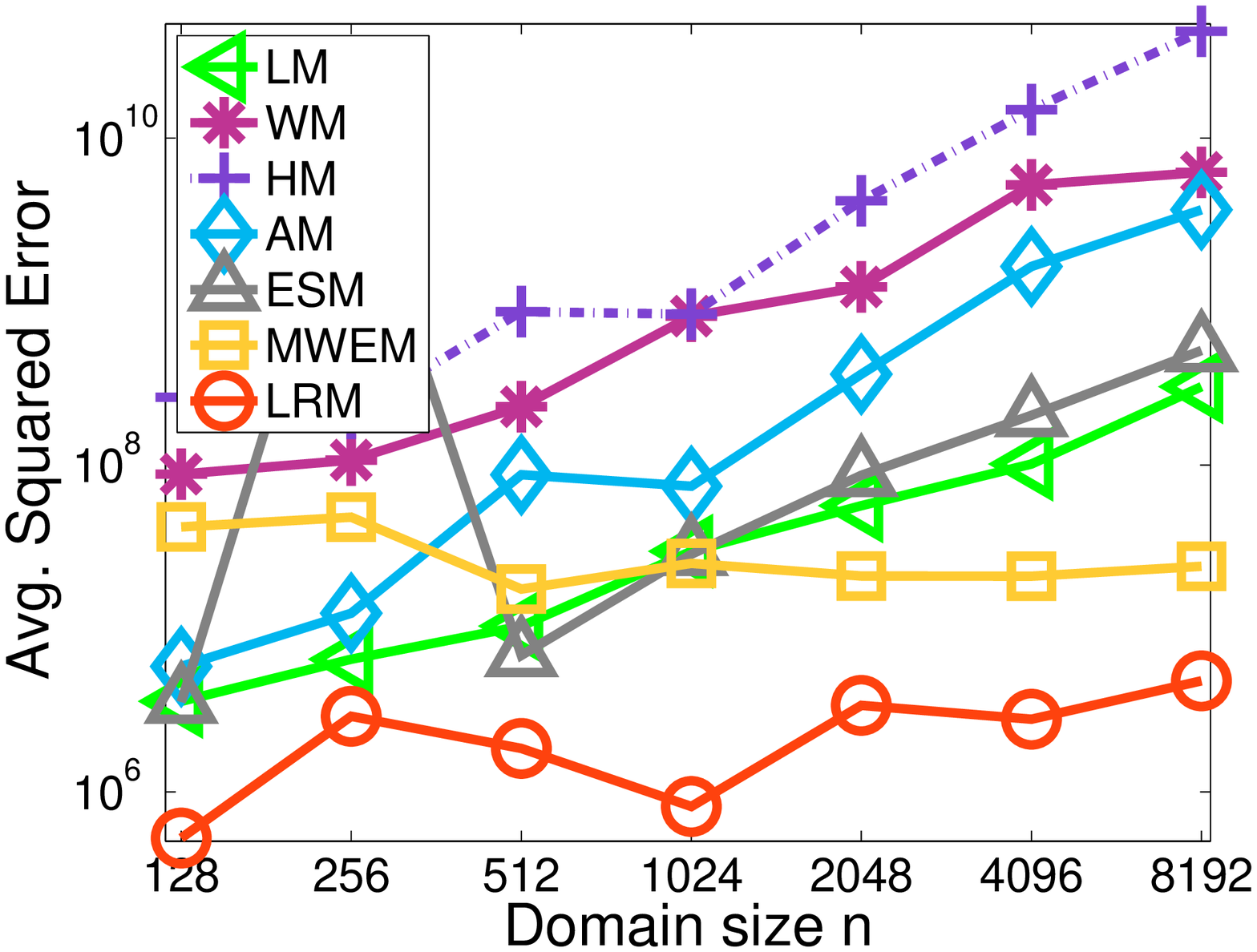}}
\centering \subfigure[\emph{Social Network}]
{\includegraphics[width=0.244\textwidth]{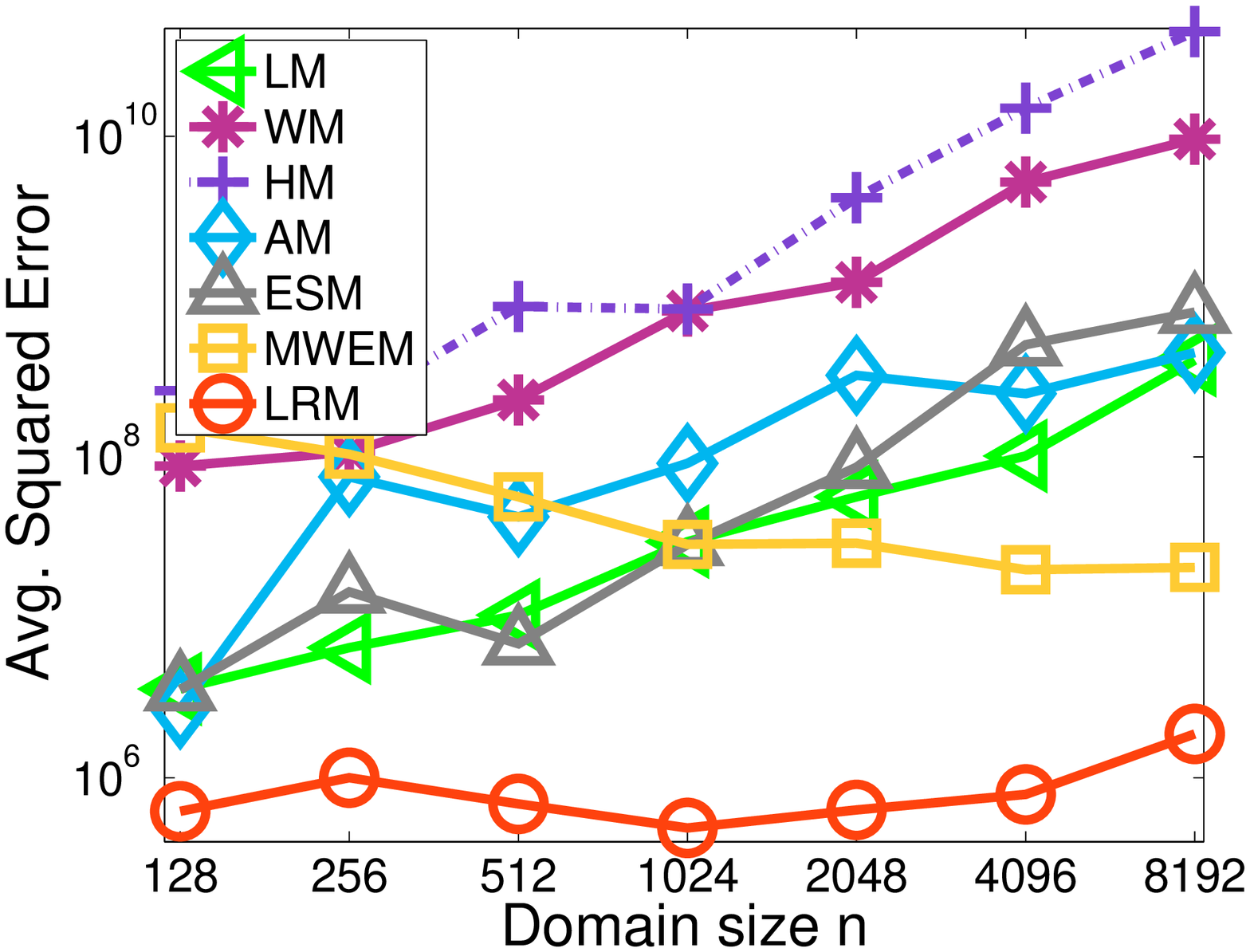}}
\centering \subfigure[\emph{UCI Adult}]
{\includegraphics[width=0.244\textwidth]{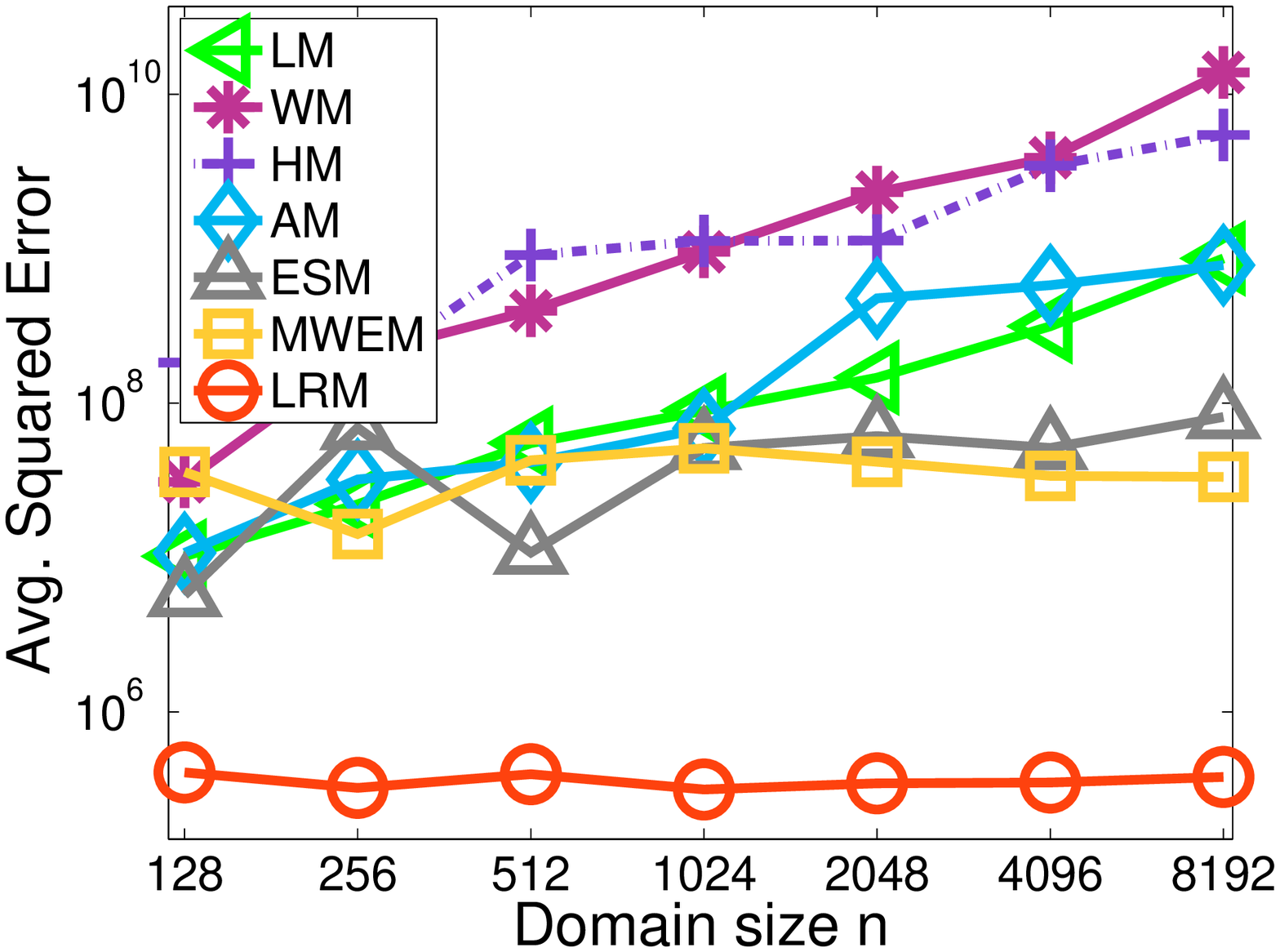}}
\caption{Effect of domain size $n$ on workload \emph{WMarginal} under $\epsilon$-differential privacy with $\epsilon=0.1$}\label{fig:exp:n:WMarginal}
\end{figure*}

\begin{figure*}[!t]
\centering \subfigure[\emph{Search Logs}]
{\includegraphics[width=0.244\textwidth]{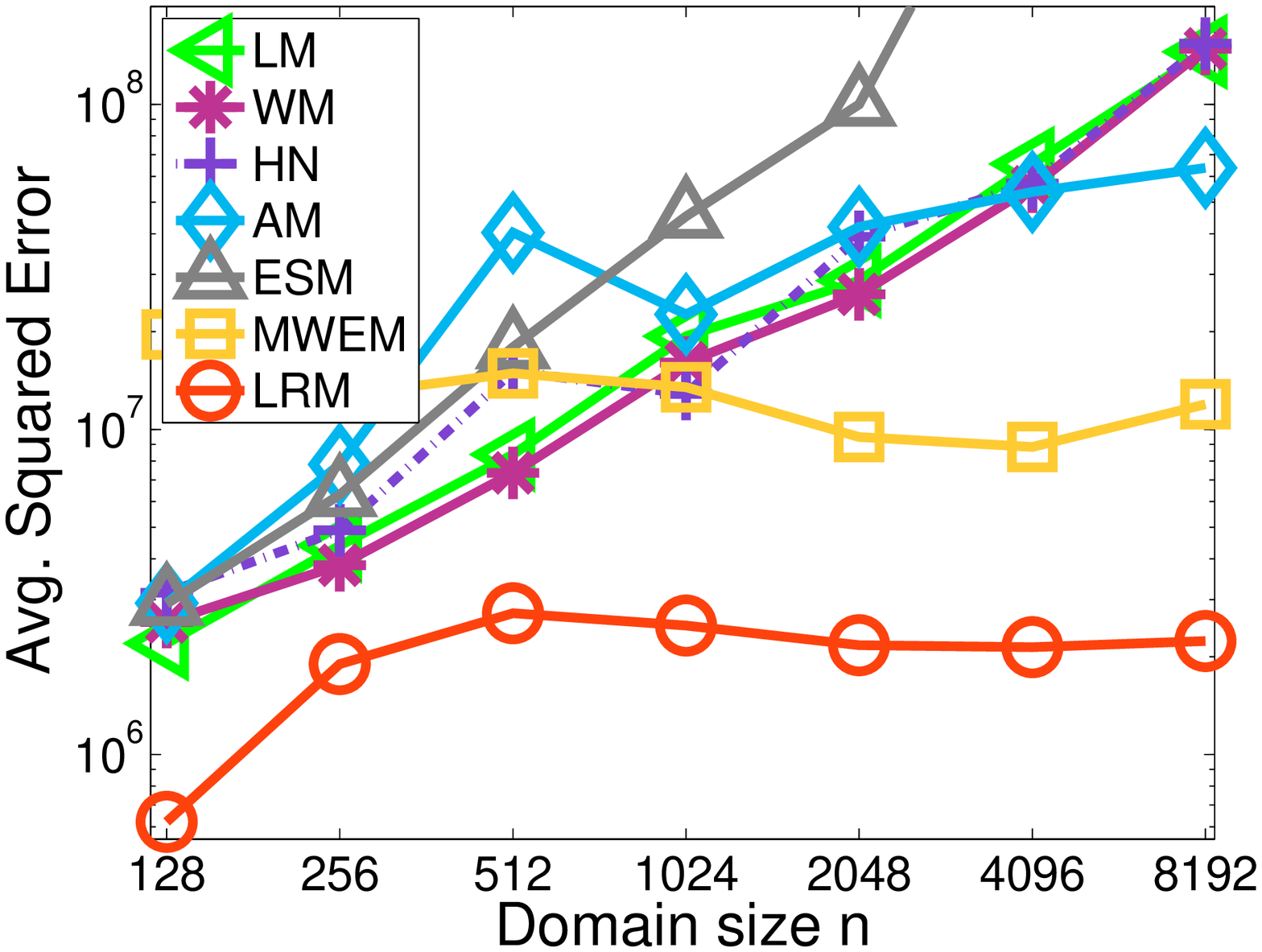}}
\subfigure[\emph{Net Trace}]
{\includegraphics[width=0.244\textwidth]{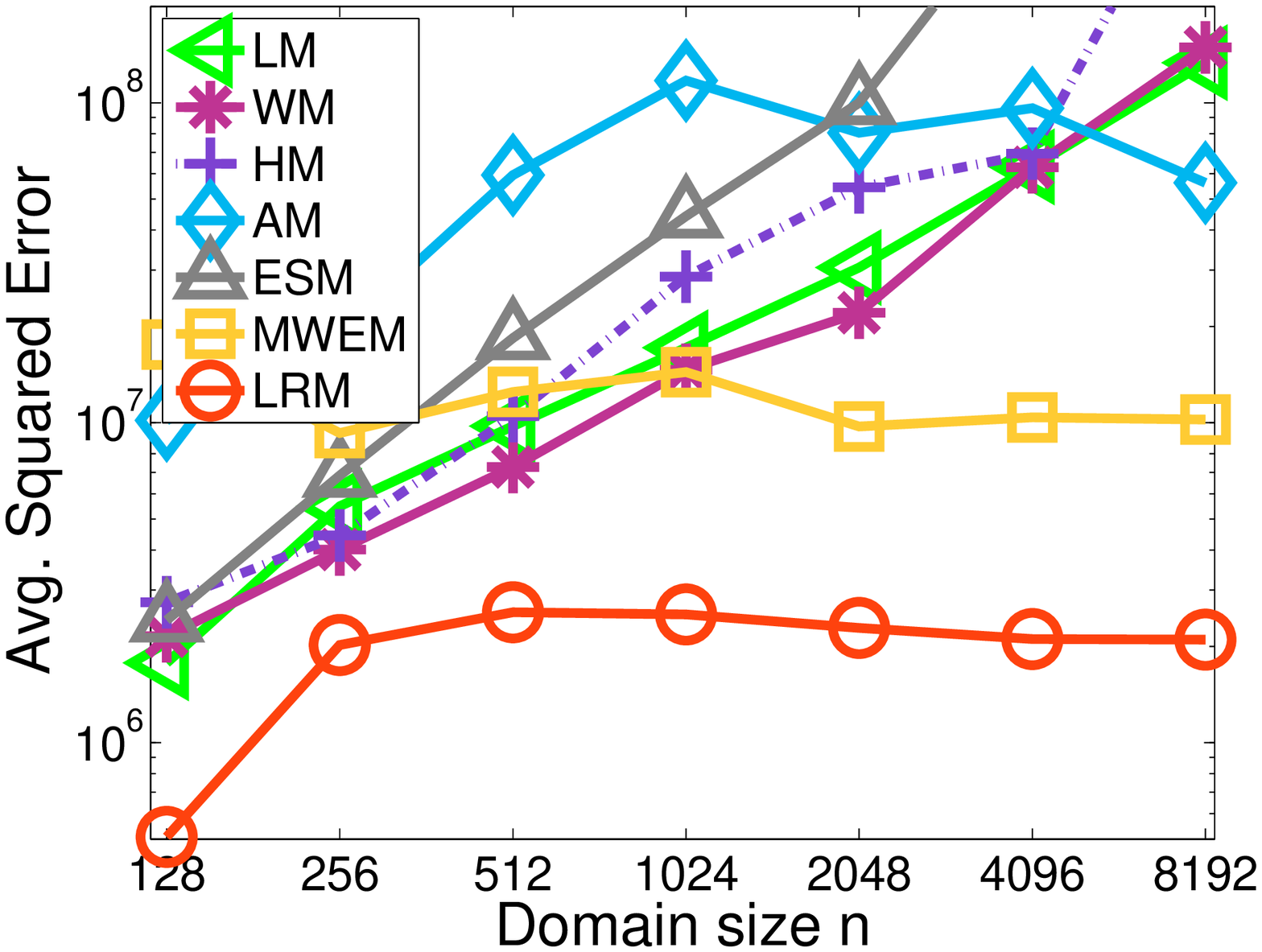}}
\centering \subfigure[\emph{Social Network}]
{\includegraphics[width=0.244\textwidth]{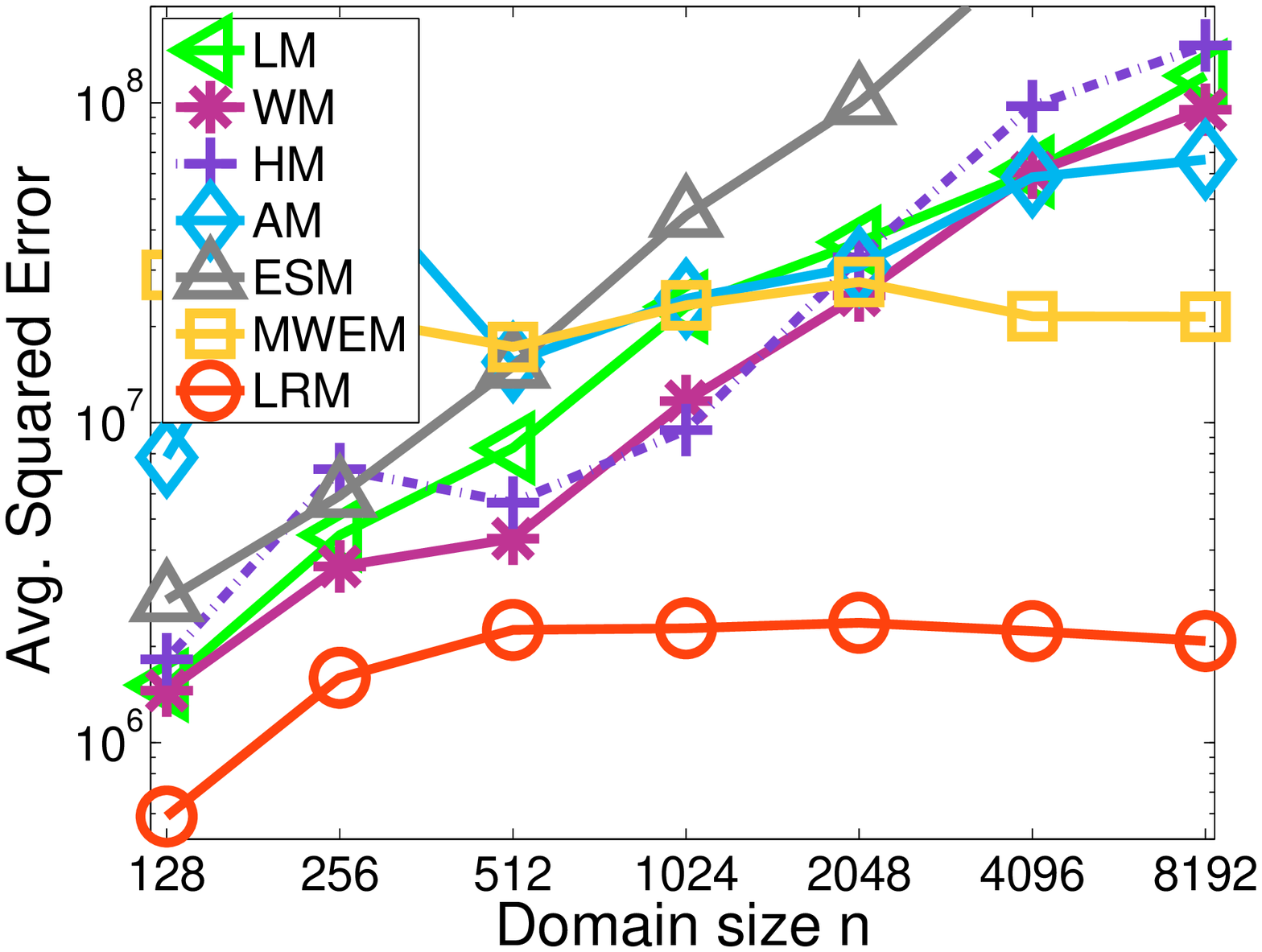}}
\centering \subfigure[\emph{UCI Adult}]
{\includegraphics[width=0.244\textwidth]{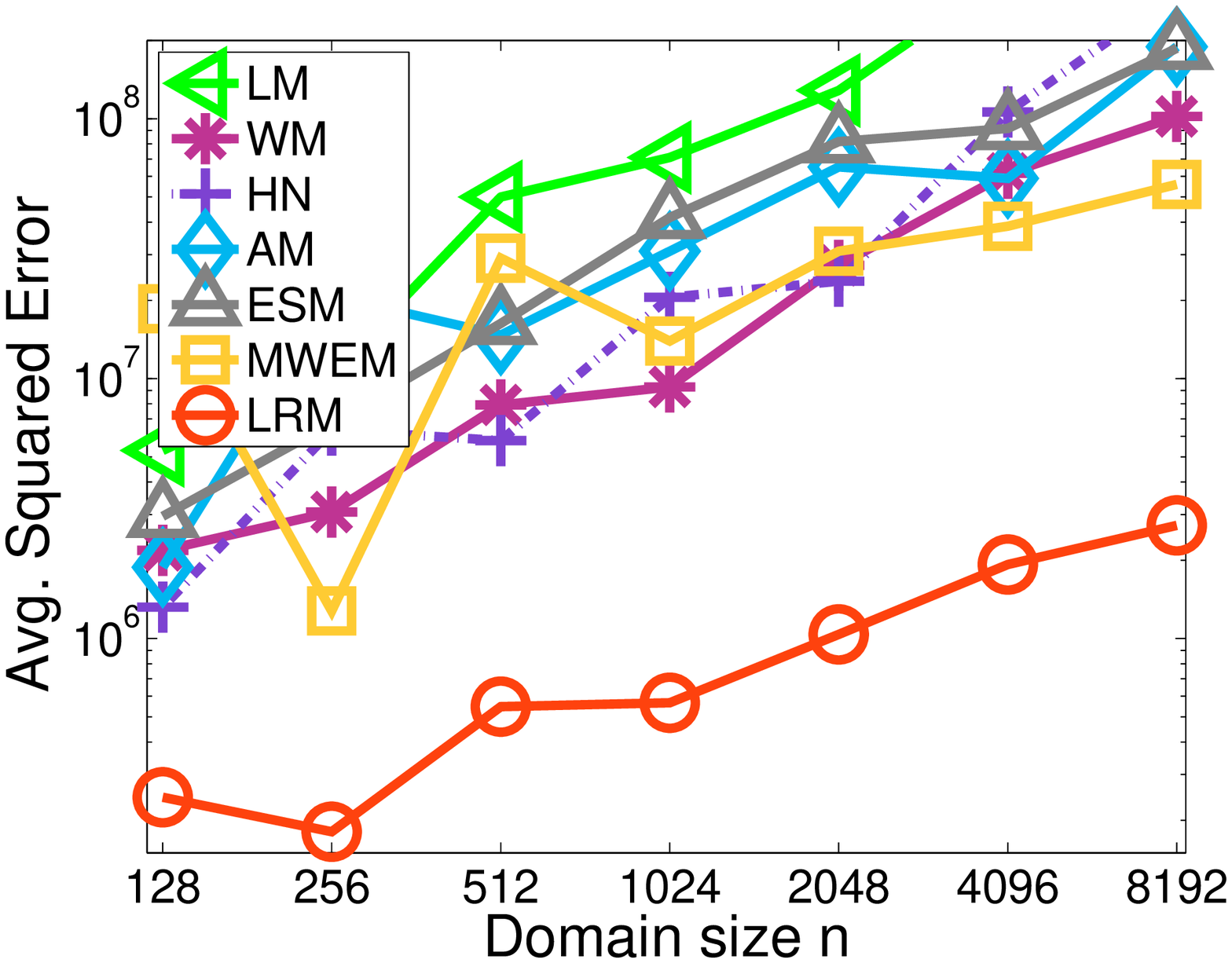}}
\caption{Effect of domain size $n$ on workload \emph{WRelated} under $\epsilon$-differential privacy with $\epsilon=0.1$}\label{fig:exp:n:WRelated}
\end{figure*}

\begin{figure*}[!t]
\centering \subfigure[\emph{Search Logs}]
{\includegraphics[width=0.244\textwidth]{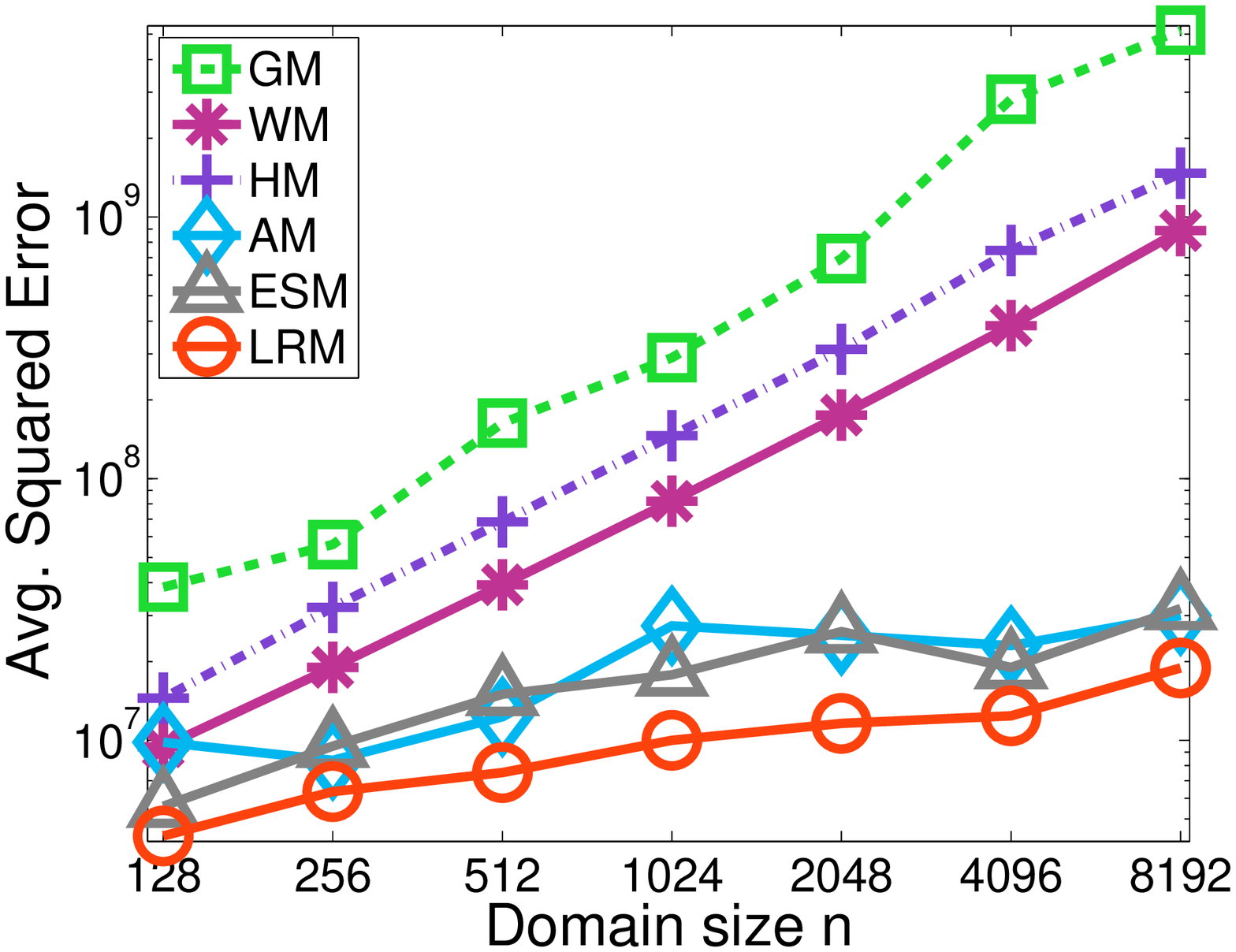}}
\subfigure[\emph{Net Trace}]
{\includegraphics[width=0.244\textwidth]{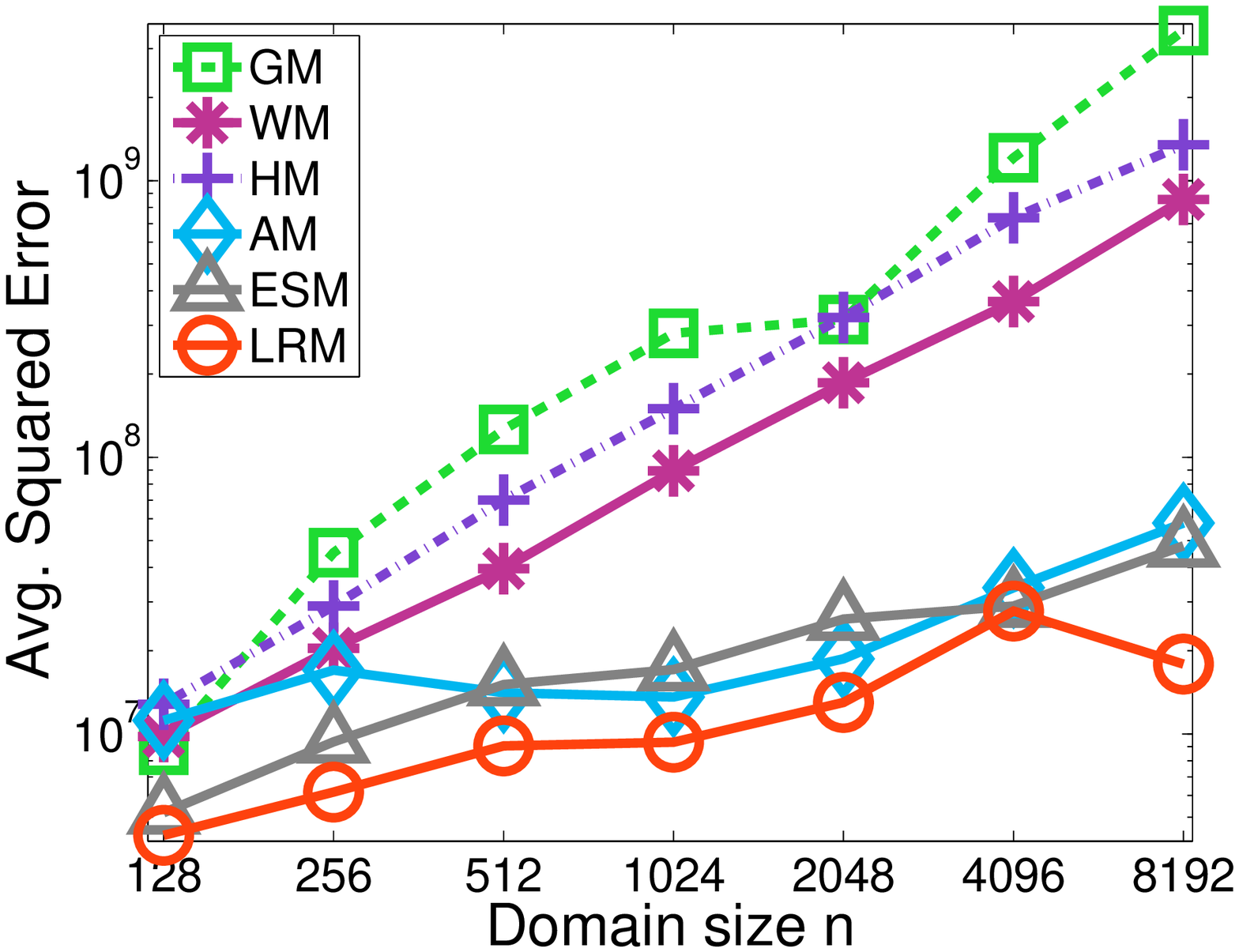}}
\centering \subfigure[\emph{Social Network}]
{\includegraphics[width=0.244\textwidth]{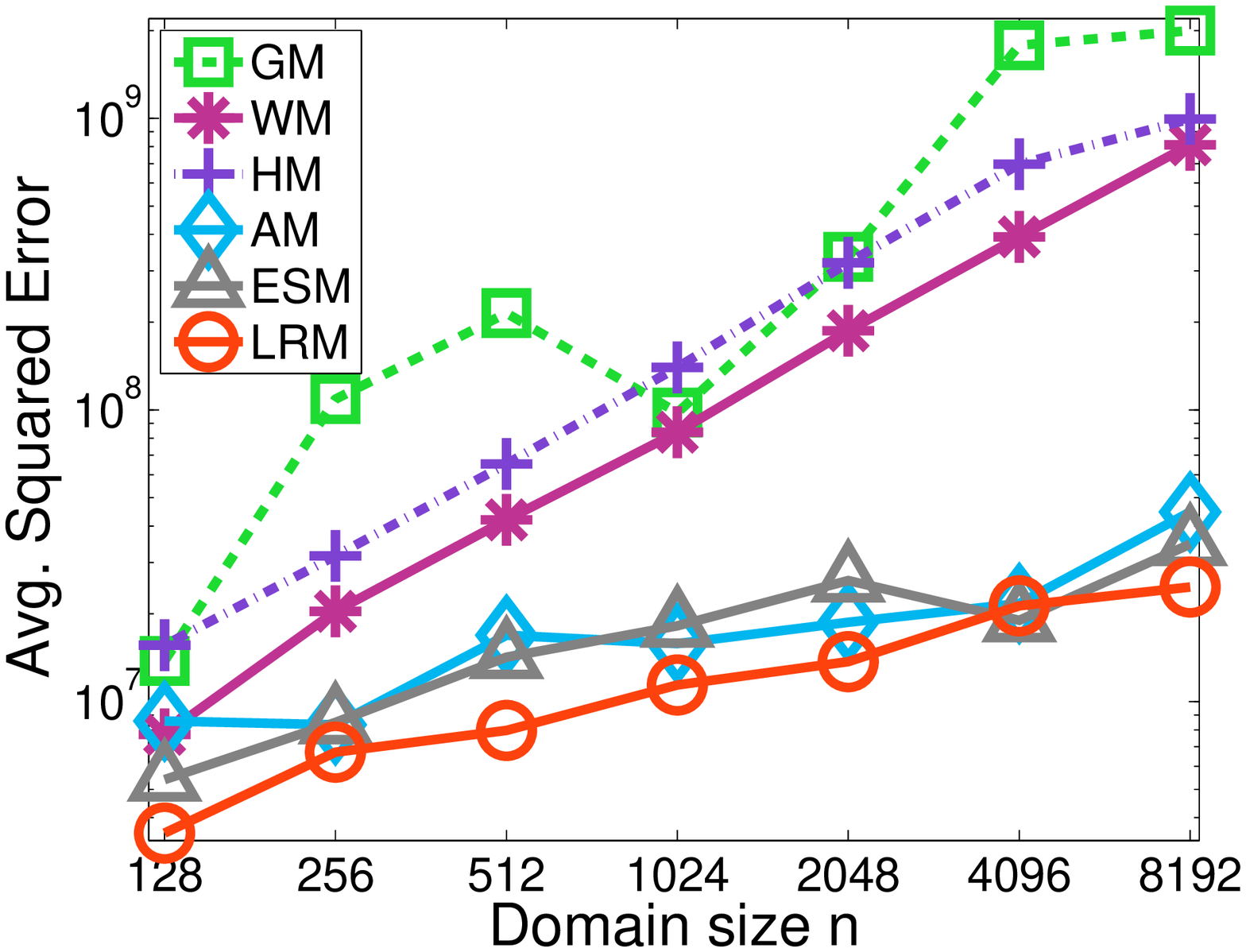}}
\centering \subfigure[\emph{UCI Adult}]
{\includegraphics[width=0.244\textwidth]{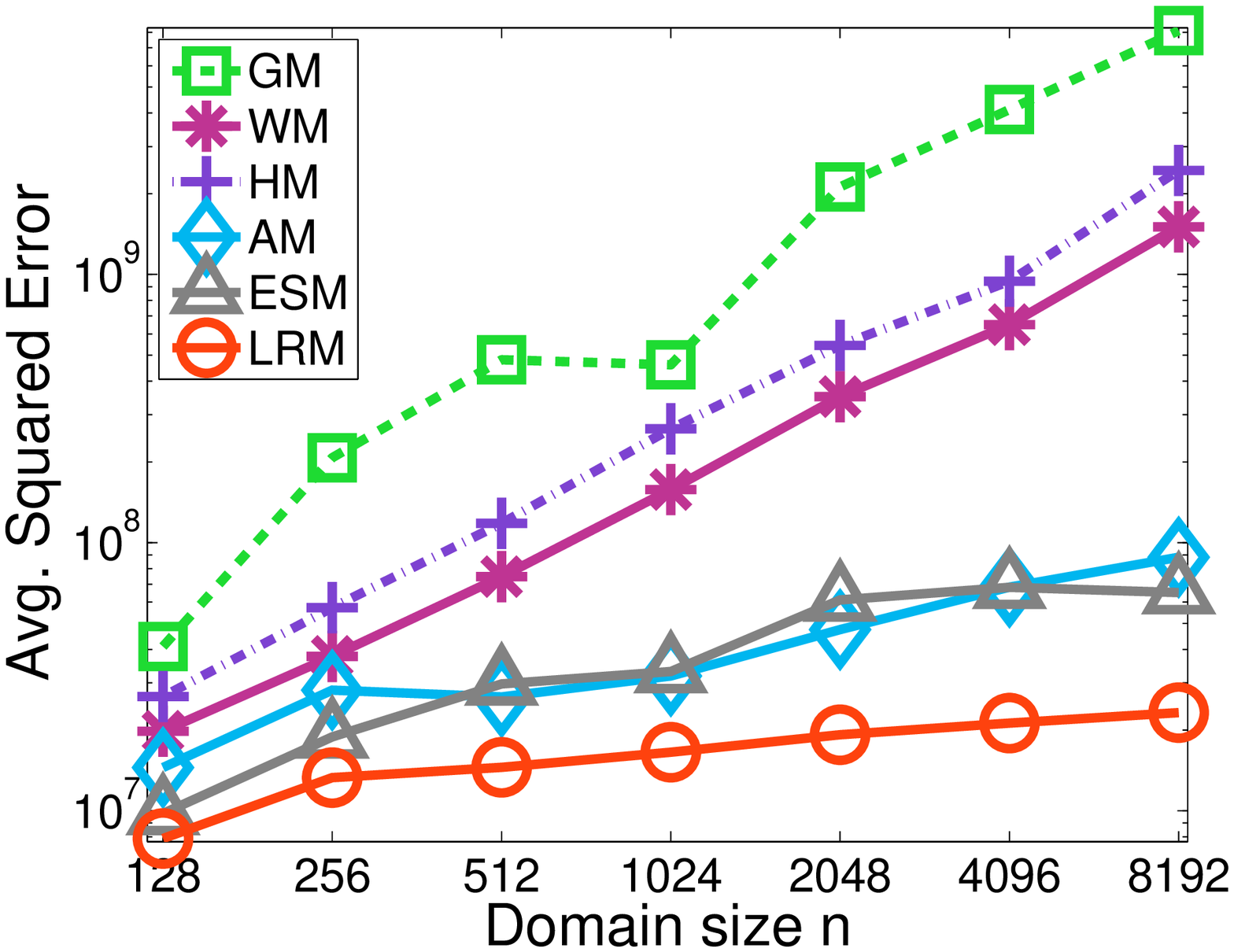}}
\caption{Effect of domain size $n$ on workload WDiscrete under ($\epsilon$, $\delta$)-differential privacy with $\epsilon=0.1$
and $\delta=0.0001$}\label{fig:exp:n:WDiscrete:app}
\end{figure*}

\begin{figure*}[!t]
\centering \subfigure[\emph{Search Logs}]
{\includegraphics[width=0.244\textwidth]{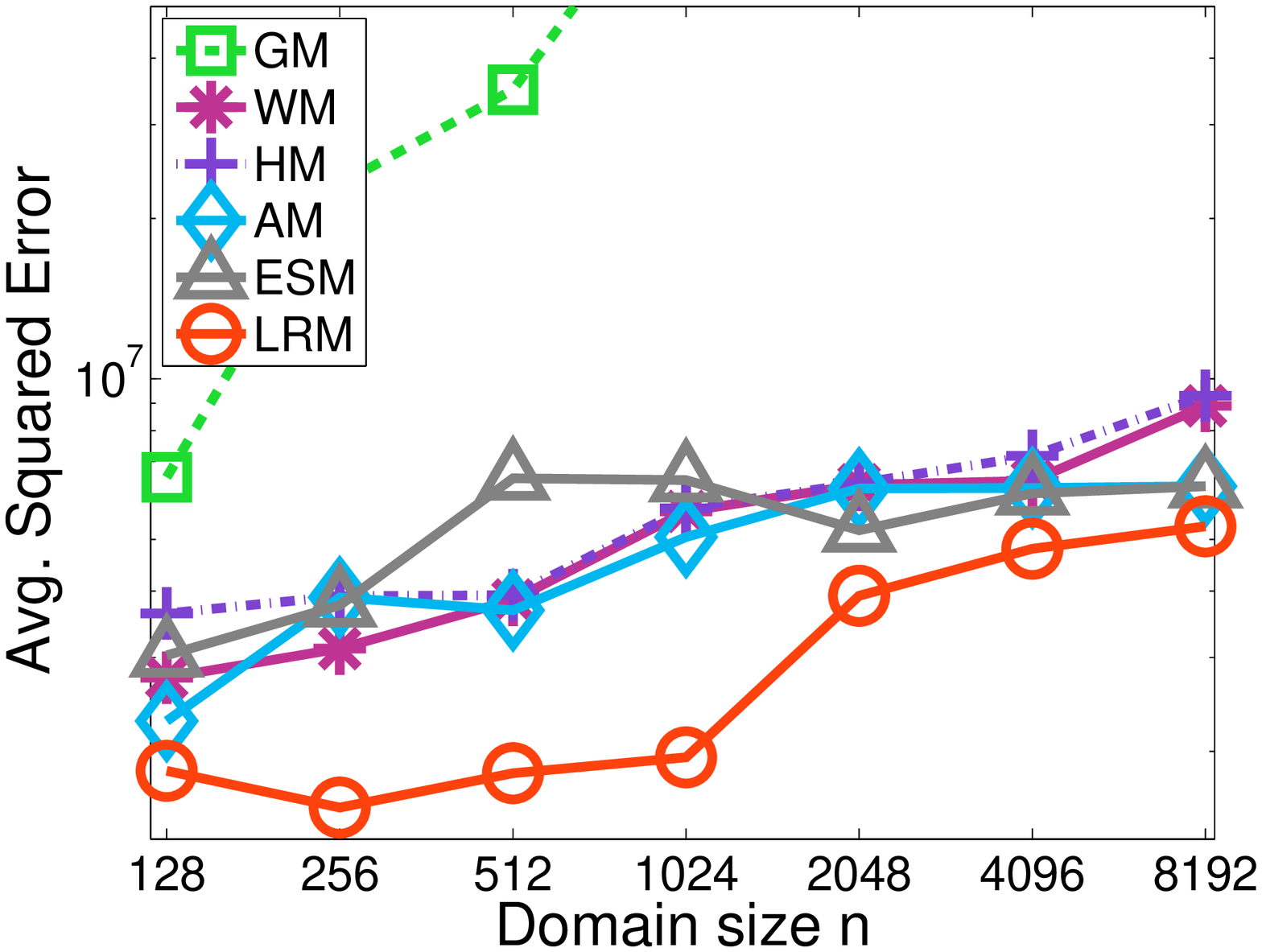}}
\subfigure[\emph{Net Trace}]
{\includegraphics[width=0.244\textwidth]{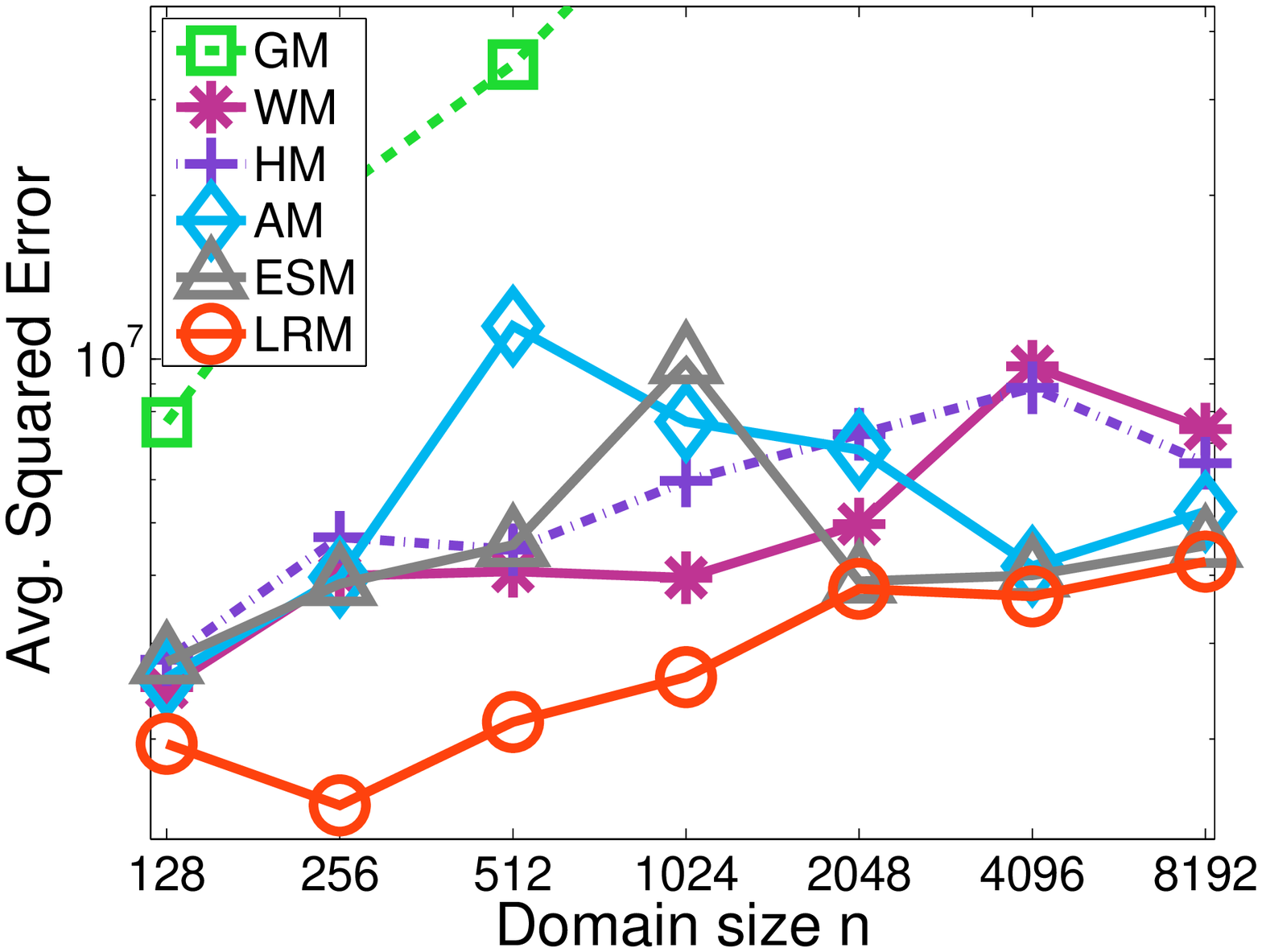}}
\centering \subfigure[\emph{Social Network}]
{\includegraphics[width=0.244\textwidth]{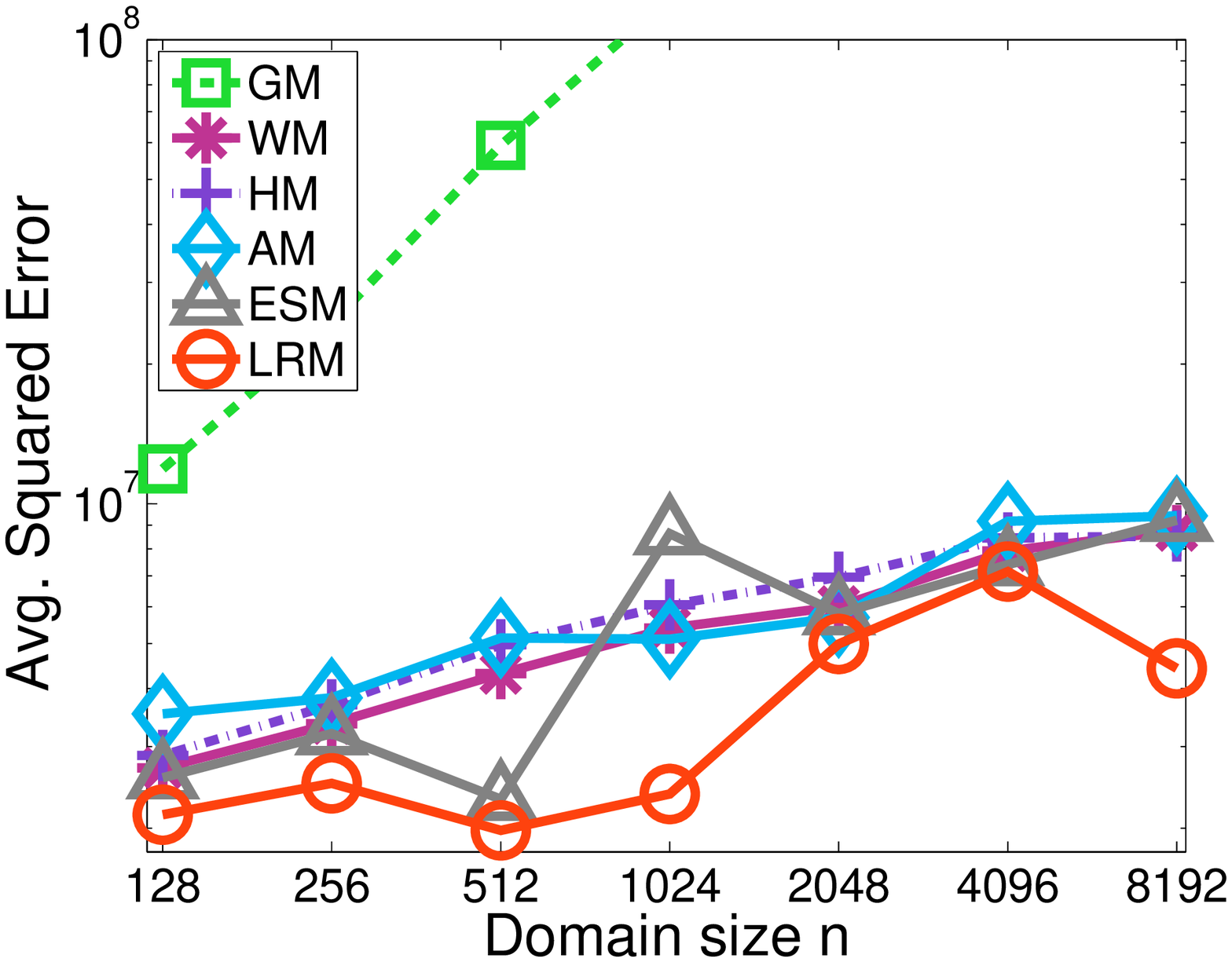}}
\centering \subfigure[\emph{UCI Adult}]
{\includegraphics[width=0.244\textwidth]{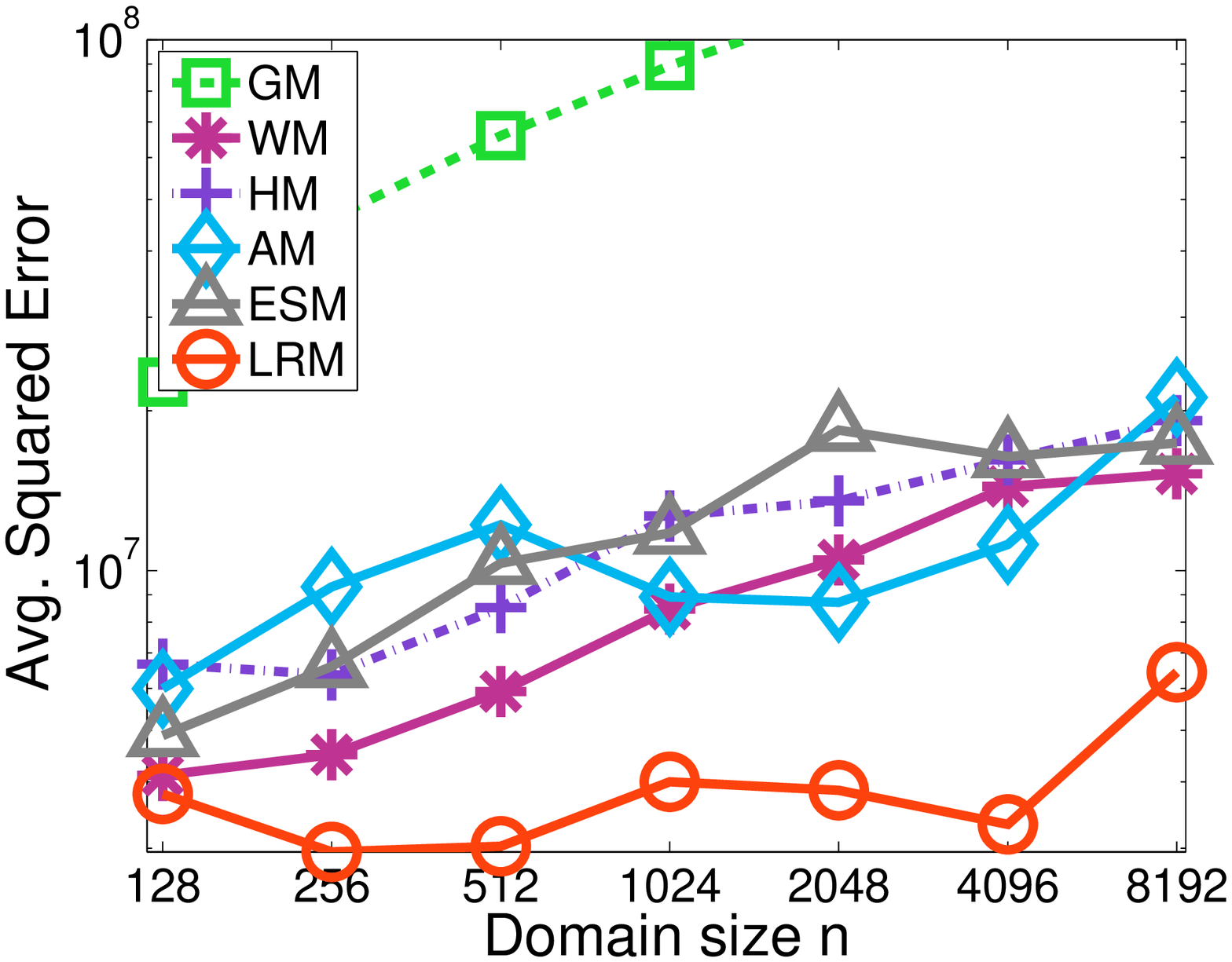}}
\caption{Effect of domain size $n$ on workload \emph{WRange} under
($\epsilon$, $\delta$)-differential privacy with $\epsilon=0.1$ and
$\delta=0.0001$}\label{fig:exp:n:WRange:app}
\end{figure*}
\begin{figure*}[!t]
\centering \subfigure[\emph{Search Logs}]
{\includegraphics[width=0.244\textwidth]{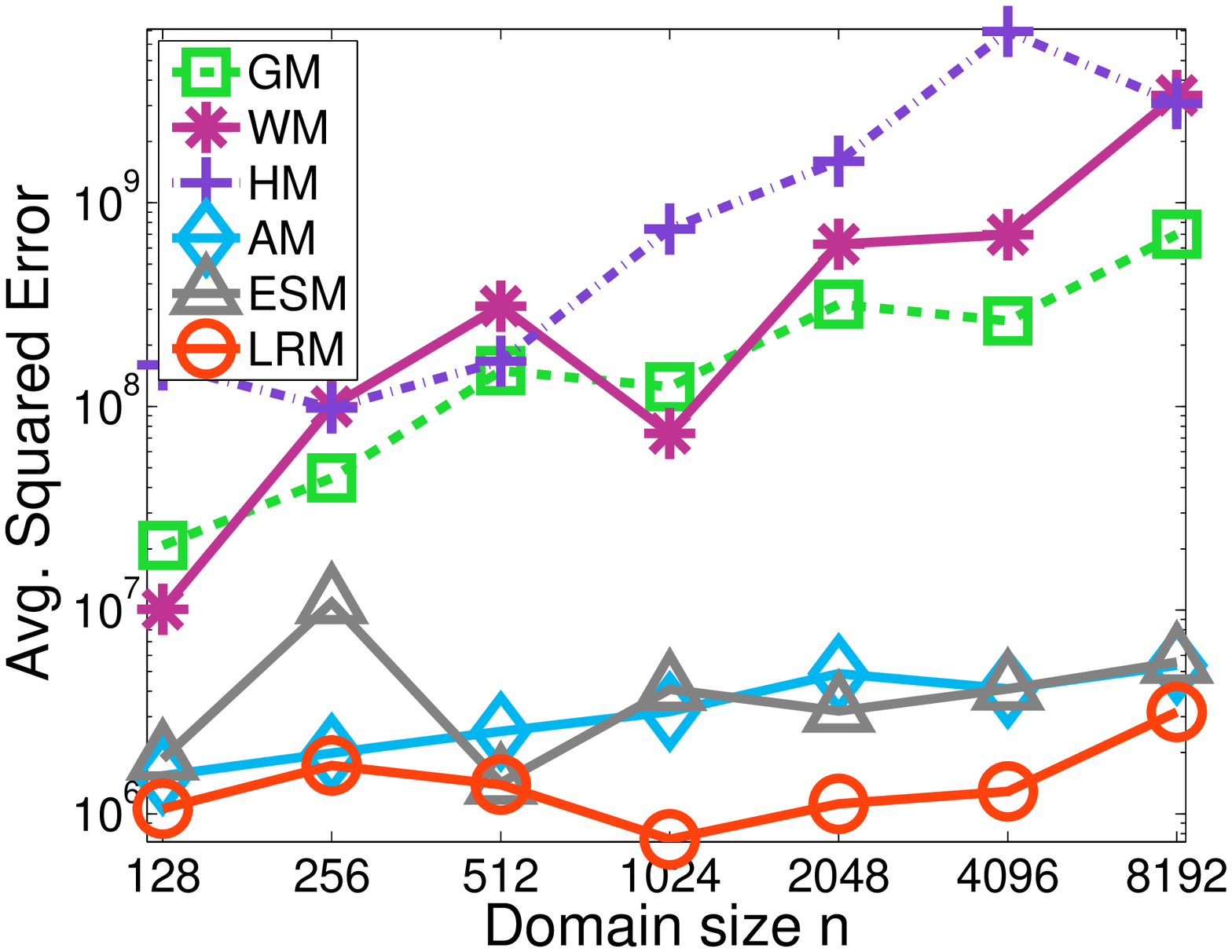}}
\subfigure[\emph{Net Trace}]
{\includegraphics[width=0.244\textwidth]{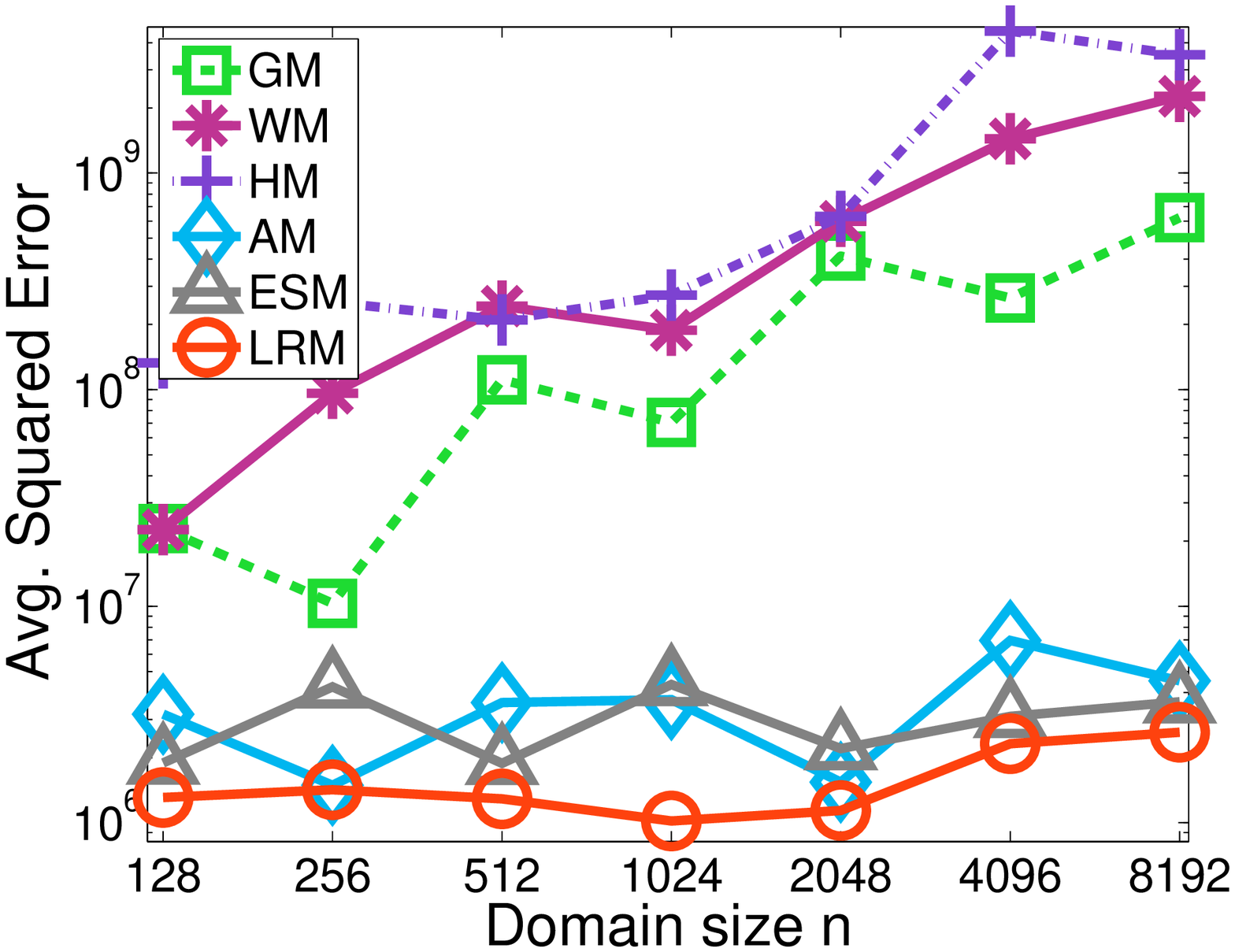}}
\centering \subfigure[\emph{Social Network}]
{\includegraphics[width=0.244\textwidth]{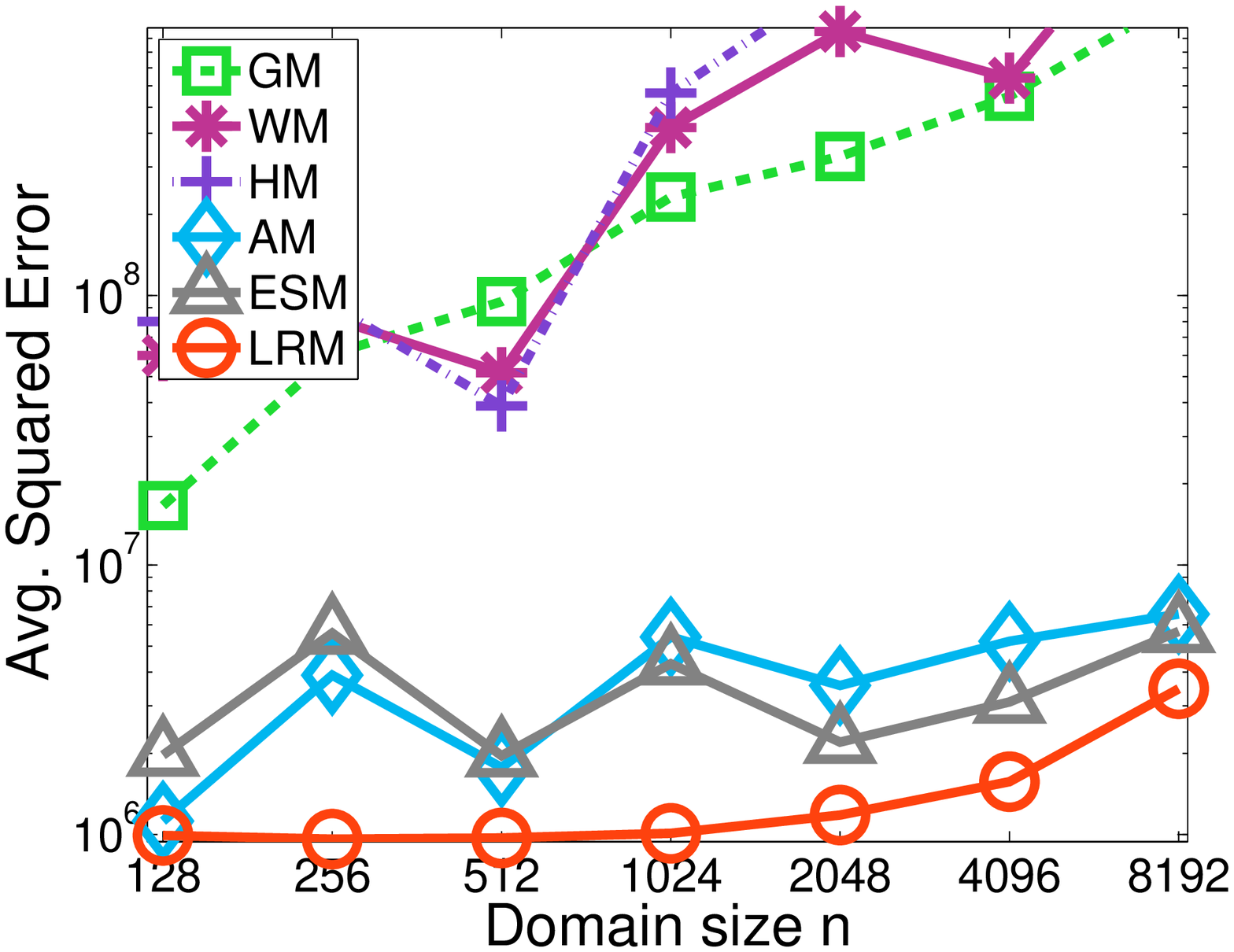}}
\centering \subfigure[\emph{UCI Adult}]
{\includegraphics[width=0.244\textwidth]{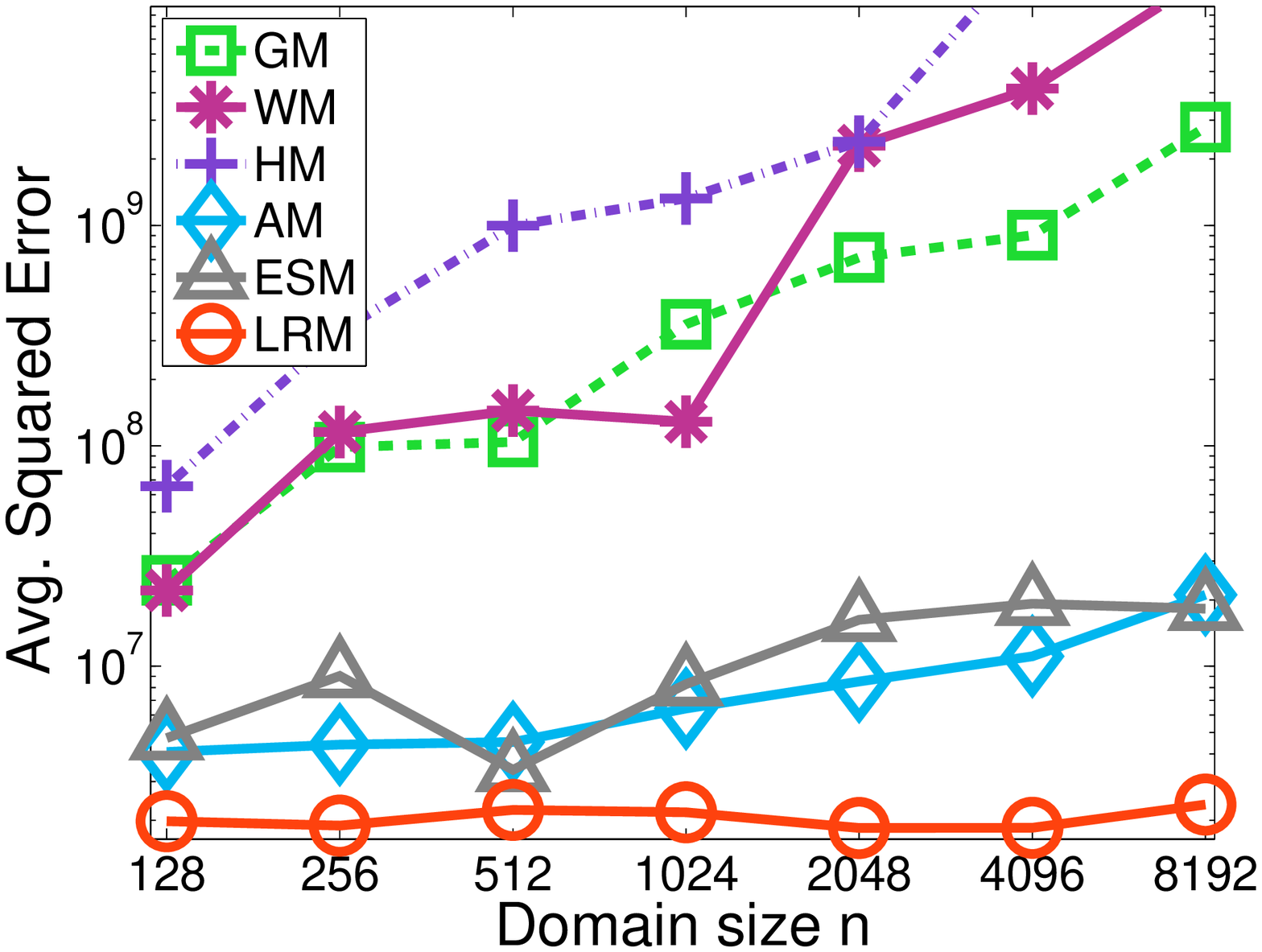}}
\caption{Effect of domain size $n$ on workload \emph{WMarginal}
under ($\epsilon$, $\delta$)-differential privacy with $\epsilon=0.1$
and $\delta=0.0001$}\label{fig:exp:n:WMarginal:app}
\end{figure*}

\begin{figure*}[!t]
\centering \subfigure[\emph{Search Logs}]
{\includegraphics[width=0.244\textwidth]{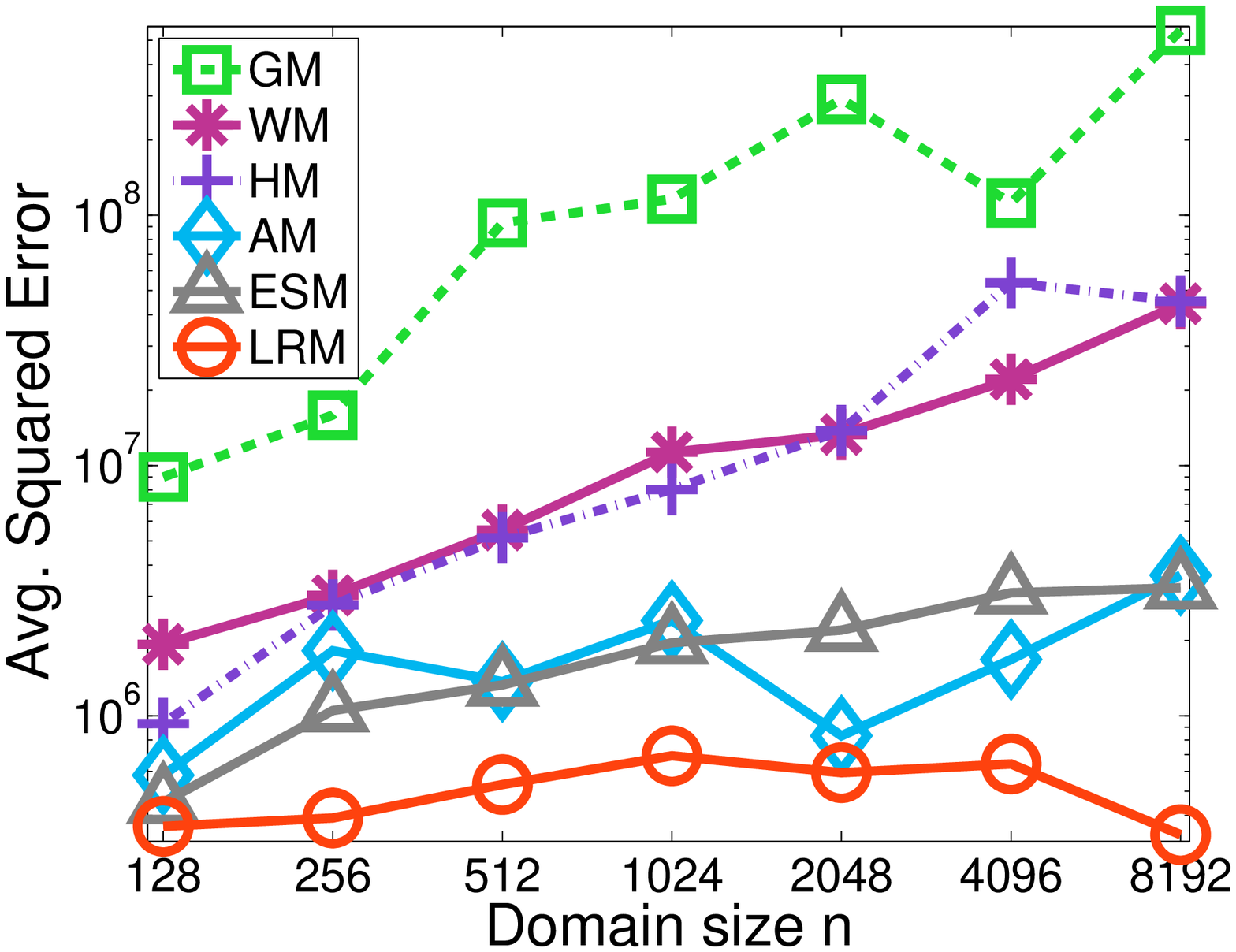}}
\subfigure[\emph{Net Trace}]
{\includegraphics[width=0.244\textwidth]{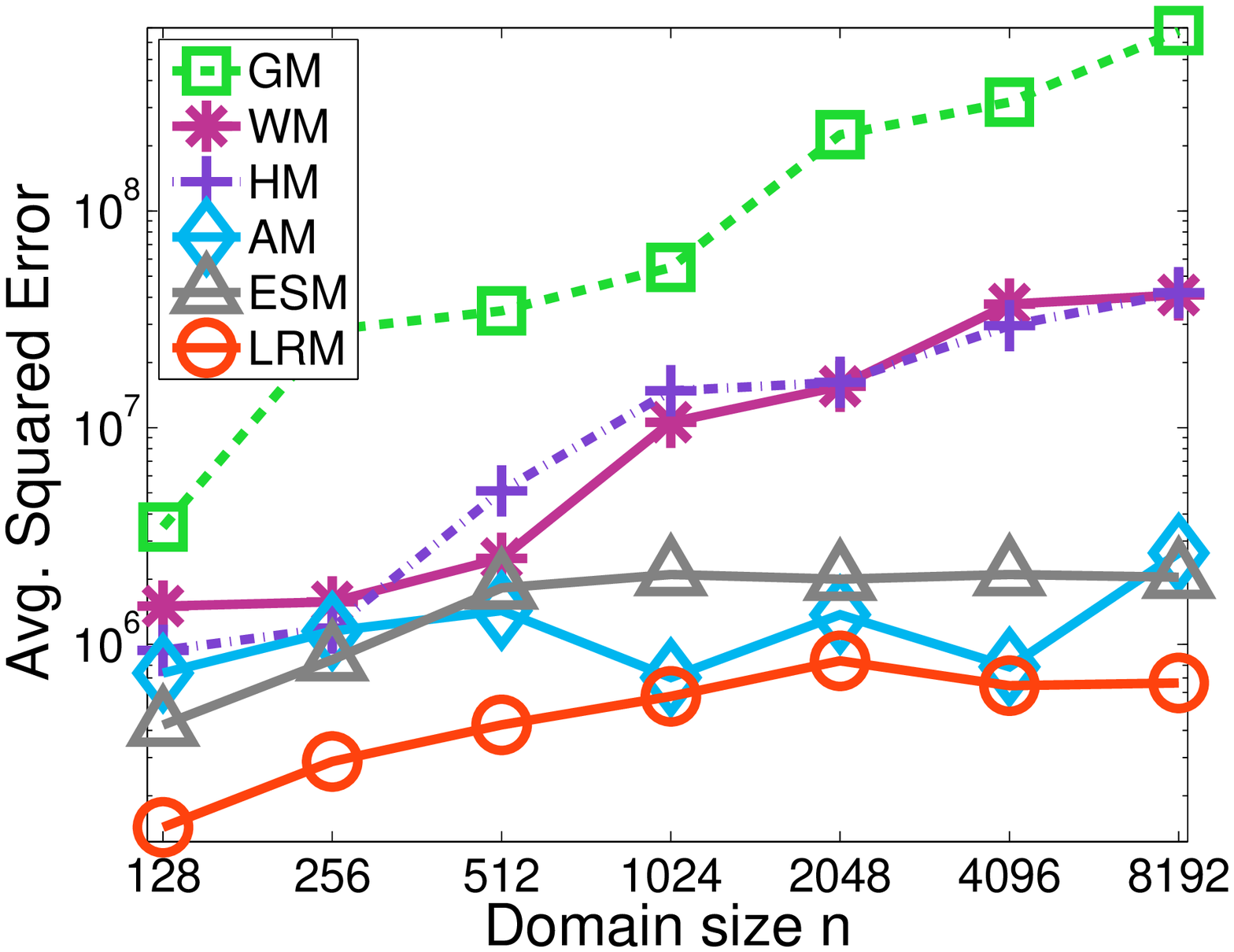}}
\centering \subfigure[\emph{Social Network}]
{\includegraphics[width=0.244\textwidth]{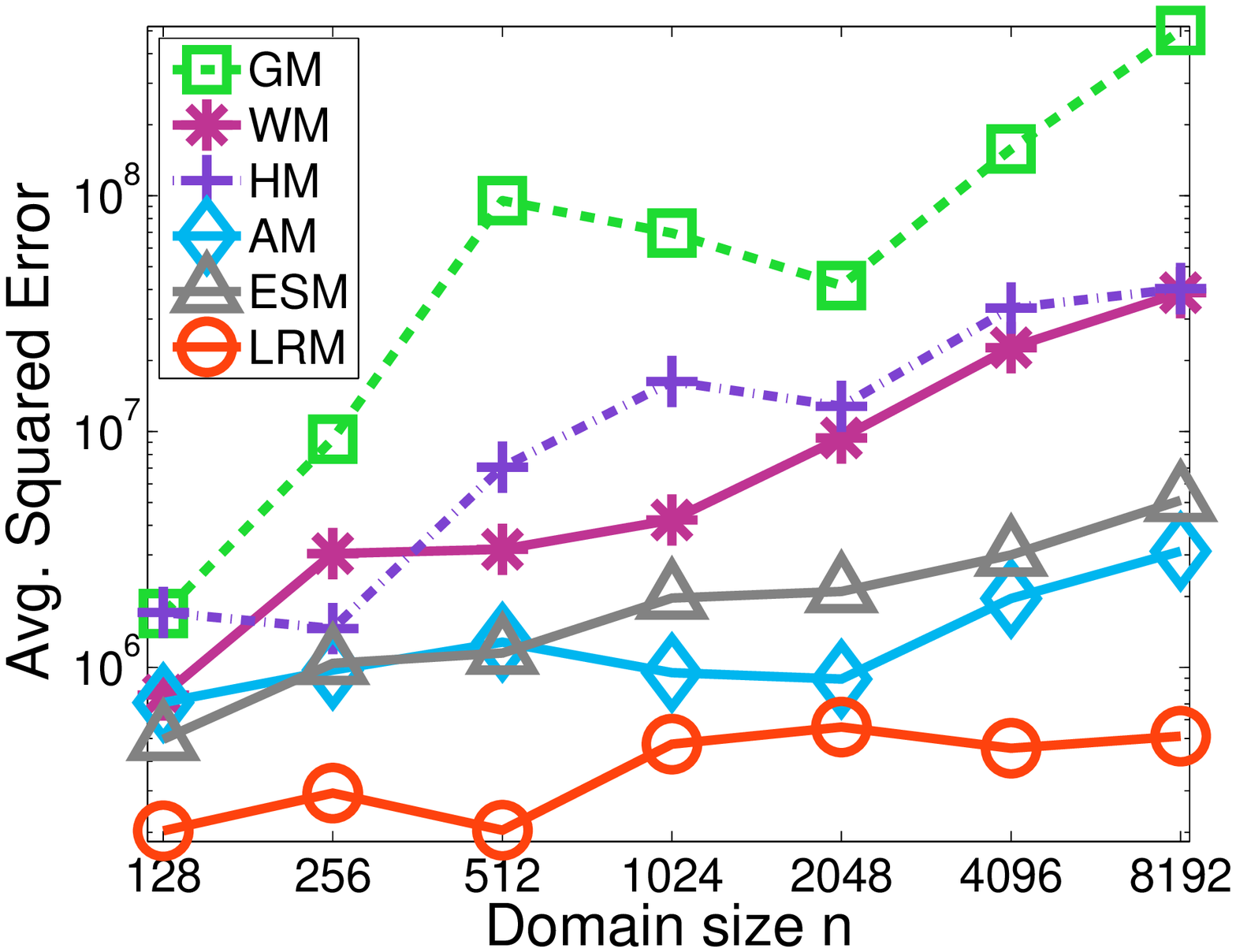}}
\centering \subfigure[\emph{UCI Adult}]
{\includegraphics[width=0.244\textwidth]{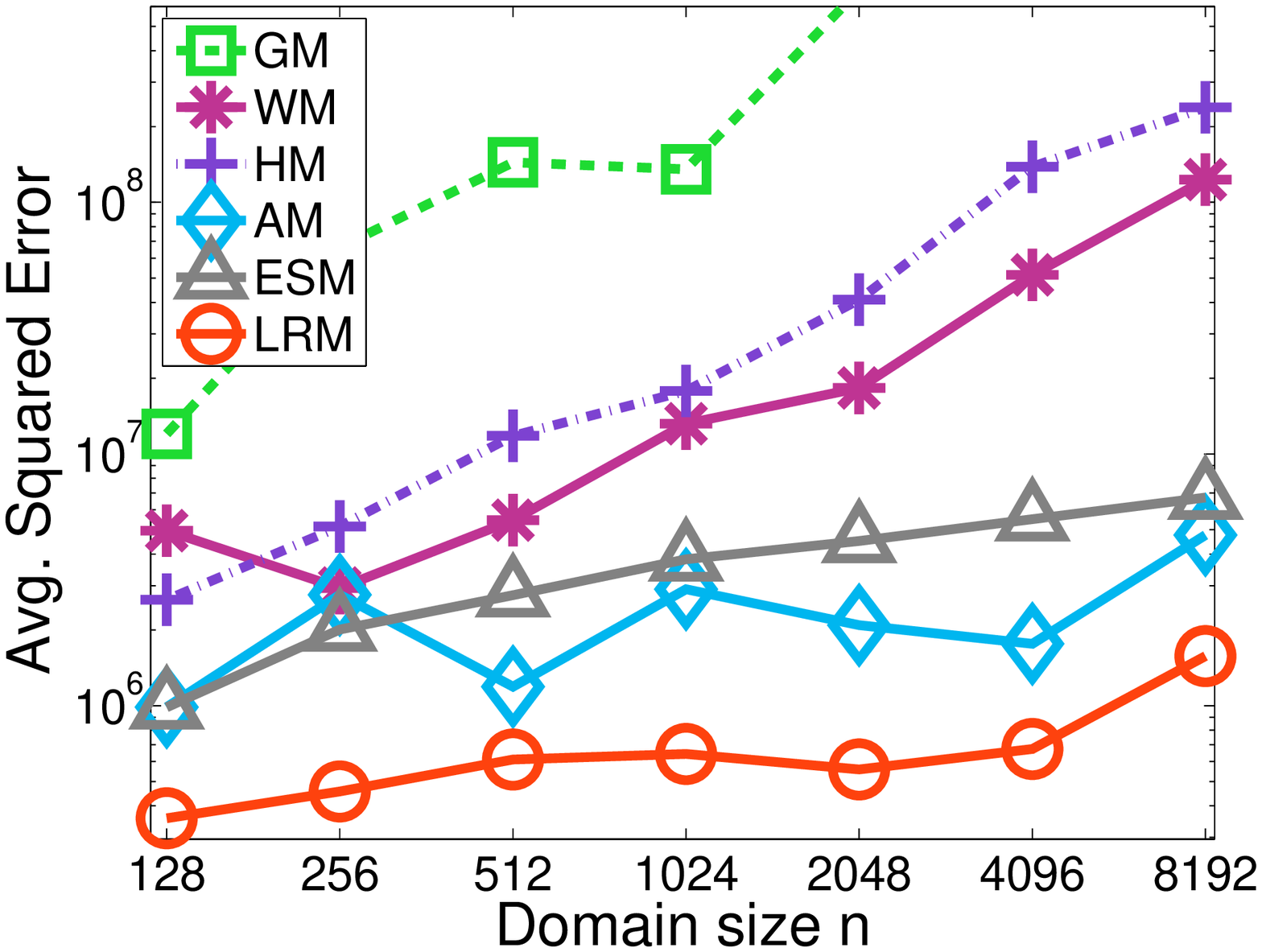}}
\caption{Effect of domain size $n$ on workload \emph{WRelated} under
($\epsilon$, $\delta$)-differential privacy with $\epsilon=0.1$ and
$\delta=0.0001$}\label{fig:exp:n:WRelated:app}
\end{figure*}

We now evaluate the accuracy performance of all mechanisms with varying domain size $n$. We perform all experiments with $\epsilon=0.1$, since the specific value of $\epsilon$ has negligible impact on the relative performance of different mechanisms. For $\epsilon$-differential privacy, we report the results of all mechanisms on the 4 different workloads in Figures \ref{fig:exp:n:WDiscrete}, \ref{fig:exp:n:WRange}, \ref{fig:exp:n:WMarginal} and \ref{fig:exp:n:WRelated}, respectively. On workloads \emph{WMarginal} and \emph{WRelated}, the performance of AM and ESM is comparable to the naive Laplace mechanism, and significantly worse than the other methods, sometimes by more than an order of magnitude. This is mainly because the $\mathcal{L}_2$ approximation used by AM and ESM does not lead to a good optimization of the actual objective function formulated using $\mathcal{L}_1$ sensitivity. On \emph{WDiscrete}, the Laplace mechanism outperforms all other mechanisms when the data is non-sparse and domain size is relatively small. This is in part due to the fact that the queries in \emph{WDiscrete} are generally independent when $m\geq n$. Since the other mechanisms do not gain from correlations among queries, Laplace mechanism is optimal in such a situation. Whereas all other data-independent mechanisms incur an error linear to the domain size $n$, LRM's error stops increasing when the domain size reaches 512. This is because LRM's error rate depends on the rank of the workload matrix $W$, which is no larger than $\min(m,n)$. This explains the excellent performance of LRM in larger domains. On \emph{WRange}, the errors of WM and HM are smaller than that of the Laplace mechanism when the domain size is no smaller than 512. Moreover, WM and HM perform better on \emph{WRange} than on the others workloads, since they are designed to optimize mainly for range queries. Nonetheless, LRM's performance is significantly better than any of them, since it fully utilizes the correlations between the range queries on large domains. On \emph{WMarginal} and \emph{WRelated}, LRM achieves the best performance in all settings. The performance gap between LRM and other methods is over two orders of magnitude when the domain size reaches 8192. Since \emph{WRelated} naturally leads to a low rank workload matrix $W$, this result verifies LRM's vast benefit from exploiting the low-rank property of the workload. Finally, we observe some interesting behaviors of the data-dependent method MWEM. The error incurred by MWEM does not scale well with the domain size $n$ on non-sparse data sets. Moreover, MWEM performs comparably to LRM on \emph{Search Logs} and \emph{Net Trace} when the $n$ is very large ($n\geq 4096$). However, the performance of MWEM is rather unstable; it incurs much larger error than LRM on \emph{Social Network} and \emph{UCI Adult}, in some cases by more than two order of magnitude.

Regarding ($\epsilon$, $\delta$)-differential privacy, we report the accuracy of all methods in Figures \ref{fig:exp:n:WDiscrete:app}, \ref{fig:exp:n:WRange:app}, \ref{fig:exp:n:WMarginal:app} and \ref{fig:exp:n:WRelated:app}. LRM obtains the best performance in all settings, especially when $n$ is large. Its improvement over the naive Gaussian mechanism is over two orders of magnitude. AM and ESM have similar accuracy. For range queries, the performance of ESM and AM is comparable to that of WM and HM, which are optimized for range counts. However, the accuracy of AM and ESM is rather unstable on workloads \emph{WRange} and \emph{WMarginal}. For ESM, this instability is caused by numerical errors in the matrix inverse operations, which can be high when the final solution matrix is low-rank. For AM, the problem is with its post-processing step, which gives approximation solutions with unstable quality. The performance of LRM, on the other hand, is consistently good in all settings.

\subsection{Impact of Number of Queries $m$}\label{sec:vary_m}

In this subsection, we test the impact of the query set cardinality $m$ on the performance of the mechanisms. We mainly focus on settings when the number of queries $m$ is no larger than the domain size $n$. For $\epsilon$-differential privacy, the accuracy results are reported in Figures \ref{fig:exp:m:WDiscrete}, \ref{fig:exp:m:WRange}, \ref{fig:exp:m:WMarginal} and \ref{fig:exp:m:WRelated}. On \emph{WRange} and \emph{WMarginal}, LRM outperforms all other mechanisms, when $m$ is significantly smaller than $n$. As $m$ grows, the performance of all mechanisms on \emph{WRange} tends to converge. The degeneration in performance of LRM is due to the lack of low rank property when the batch contains too many random range queries.  When $m$ is no less than 256, both the WM and HM achieve comparable accuracy to LRM, since they are optimized for range queries. On \emph{WDiscrete}, MWEM is comparable to LRM on \emph{UCI Adult} data set, one possible reason is that MWEM can make use of the sparsity of the data on \emph{WDiscrete} workload. On \emph{WRelated} workload, the accuracy of LRM is dramatically higher than the other methods, for all values of $m$. This is because the rank of the \emph{WRelated} workload is fixed to $s$, regardless of the number of queries. Finally, we observe that on \emph{WDiscrete} and \emph{WRange}, while the performance of other mechanisms does not differ much from data to data, the data-dependent method MWEM generally performs better on the \emph{UCI Adult} dataset compared to on other datasets, due to the high sparsity of \emph{UCI Adult}.

For ($\epsilon$, $\delta$)-differential privacy, we report the results in Figures \ref{fig:exp:m:WDiscrete:app}, \ref{fig:exp:m:WRange:app}, \ref{fig:exp:m:WMarginal:app} and \ref{fig:exp:m:WRelated:app}. We have the following observations from these results. On \emph{WDiscrete}, \emph{WRange} and \emph{WRelated} workload, WM and HM improve upon the naive Gaussian mechanism; however, on \emph{WMarginal}, WM and HM incur higher errors than GM. AM and ESM again exhibit similar performance, which is often better than that of WM, HM, and GM. LRM consistently outperforms its competitors in all test cases.

\begin{figure*}[!t]
\centering \subfigure[\emph{Search Logs}]
{\includegraphics[width=0.244\textwidth]{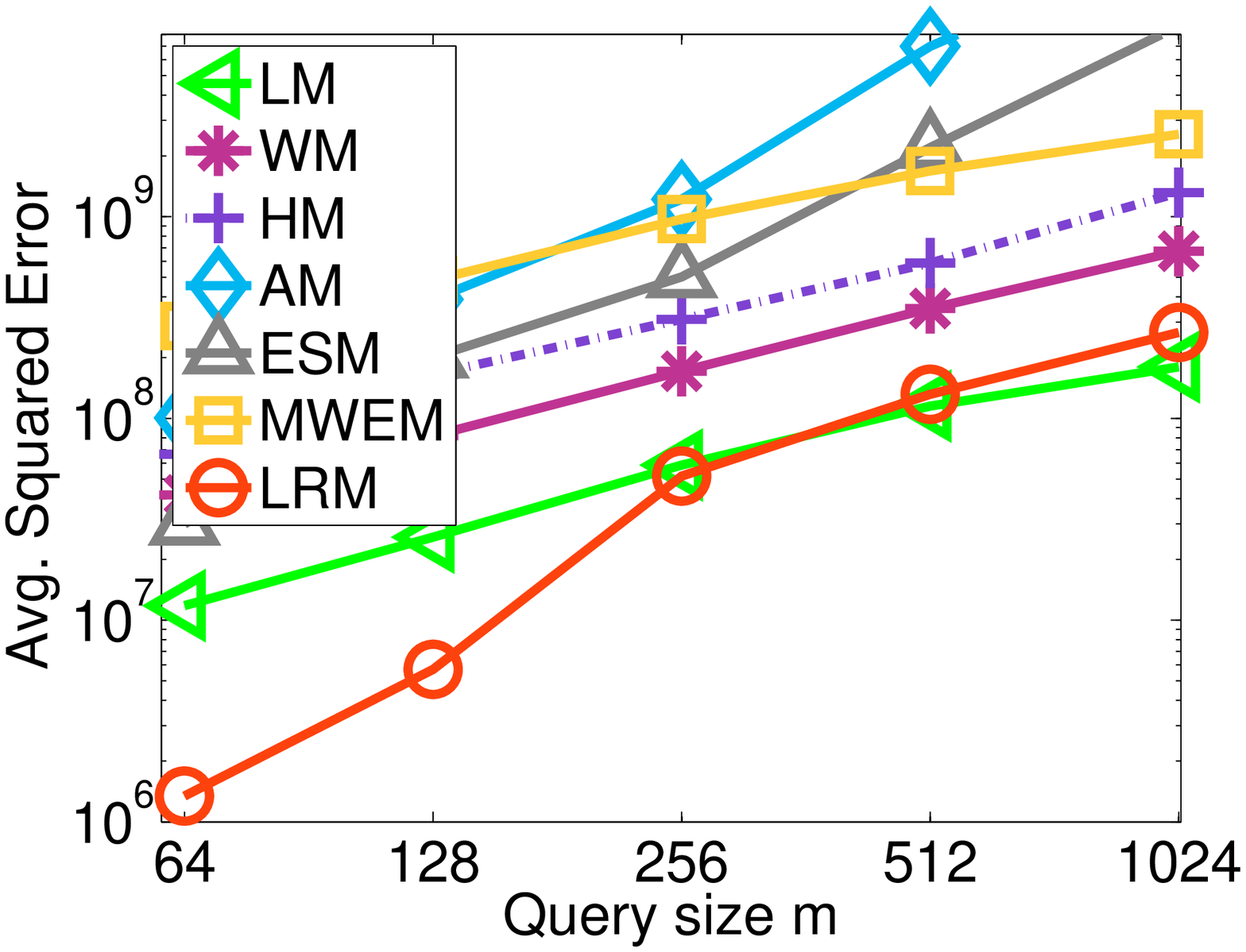}}
\subfigure[\emph{Net Trace}]
{\includegraphics[width=0.244\textwidth]{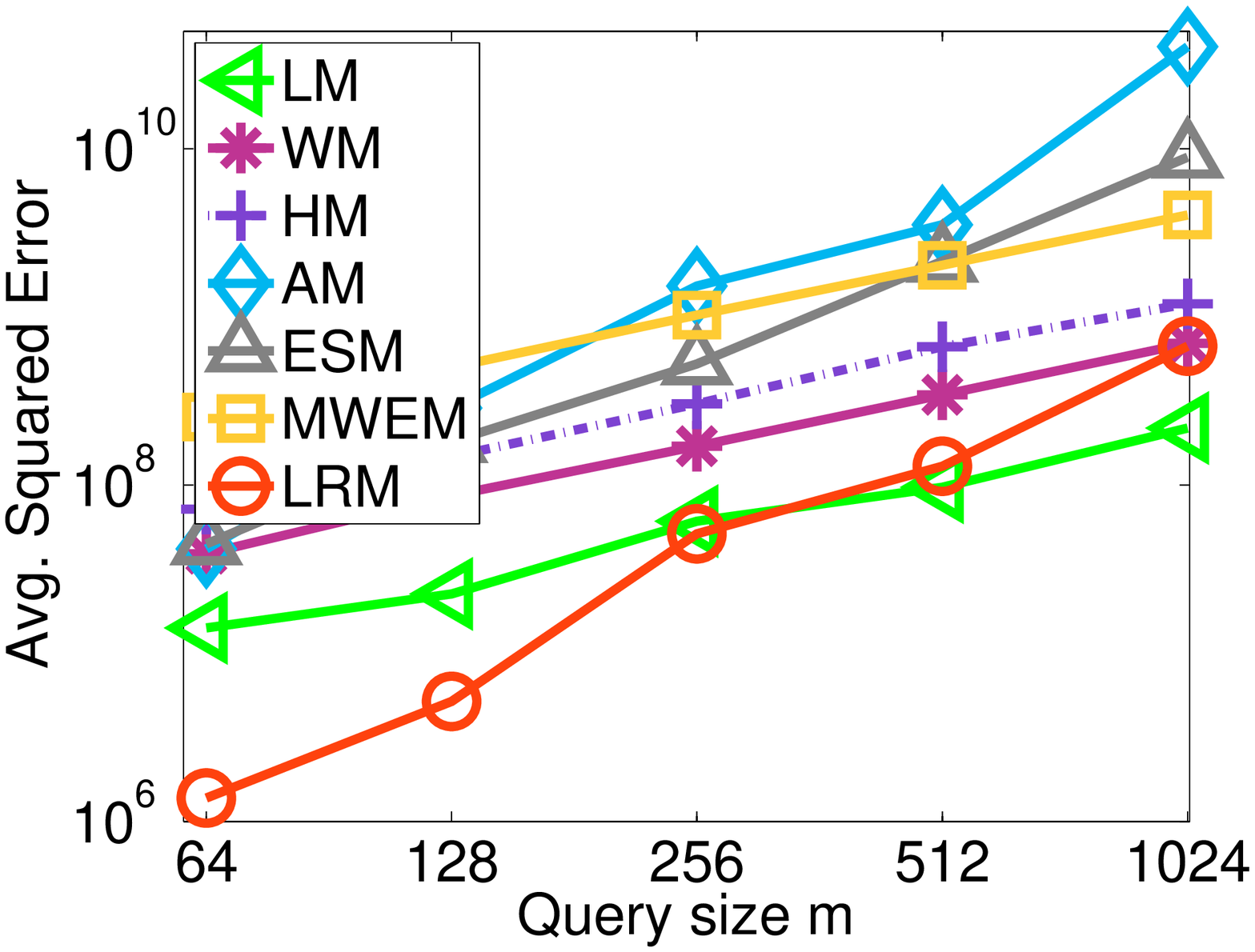}}
\centering \subfigure[\emph{Social Network}]
{\includegraphics[width=0.244\textwidth]{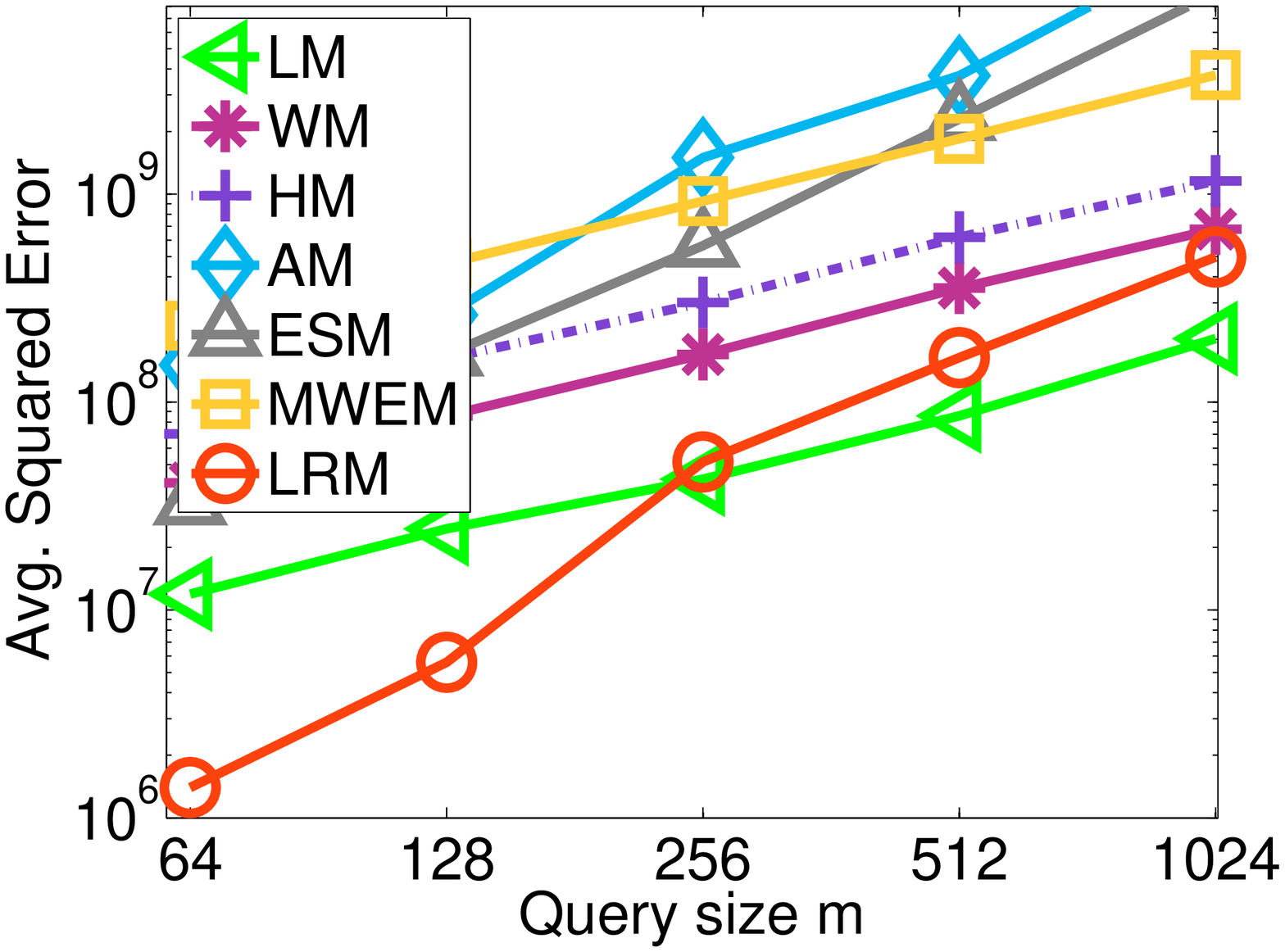}}
\centering \subfigure[\emph{UCI Adult}]
{\includegraphics[width=0.244\textwidth]{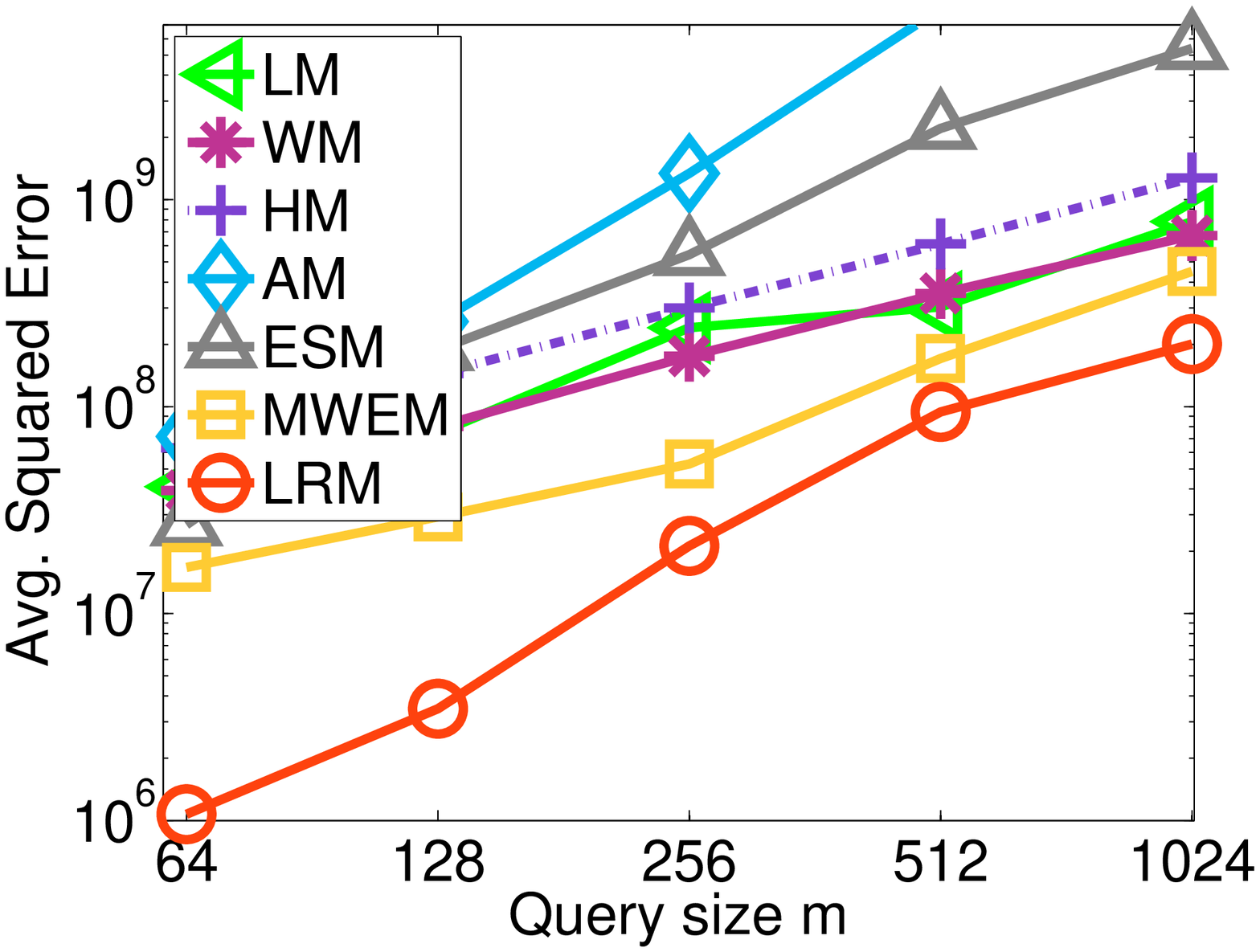}}
\caption{Effect of number of queries $m$ on workload
\emph{WDiscrete} under $\epsilon$-differential privacy with $\epsilon=0.1$}
\label{fig:exp:m:WDiscrete}
\end{figure*}

\begin{figure*}[!t]
\centering \subfigure[\emph{Search Logs}]
{\includegraphics[width=0.244\textwidth]{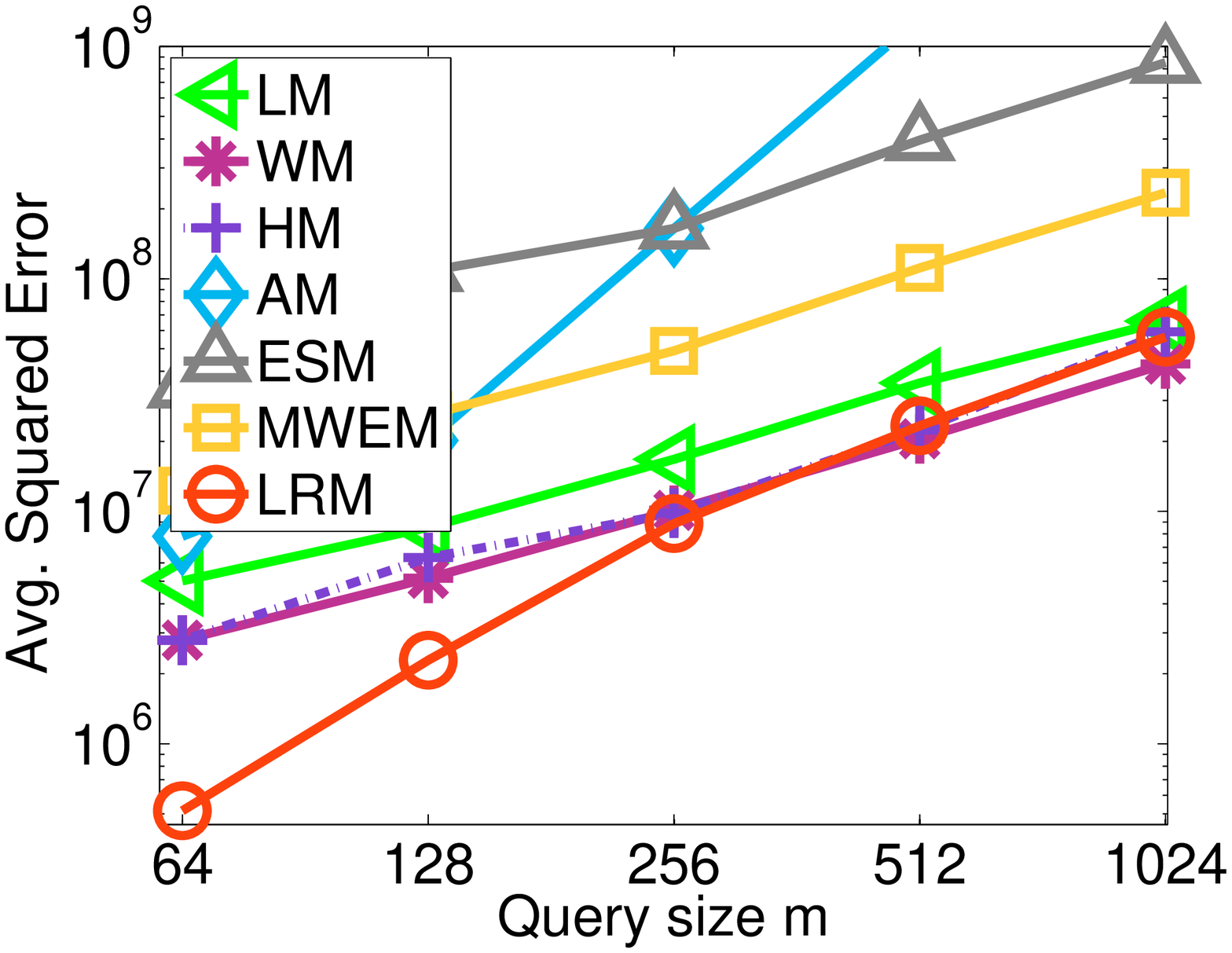}}
\subfigure[\emph{Net Trace}]
{\includegraphics[width=0.244\textwidth]{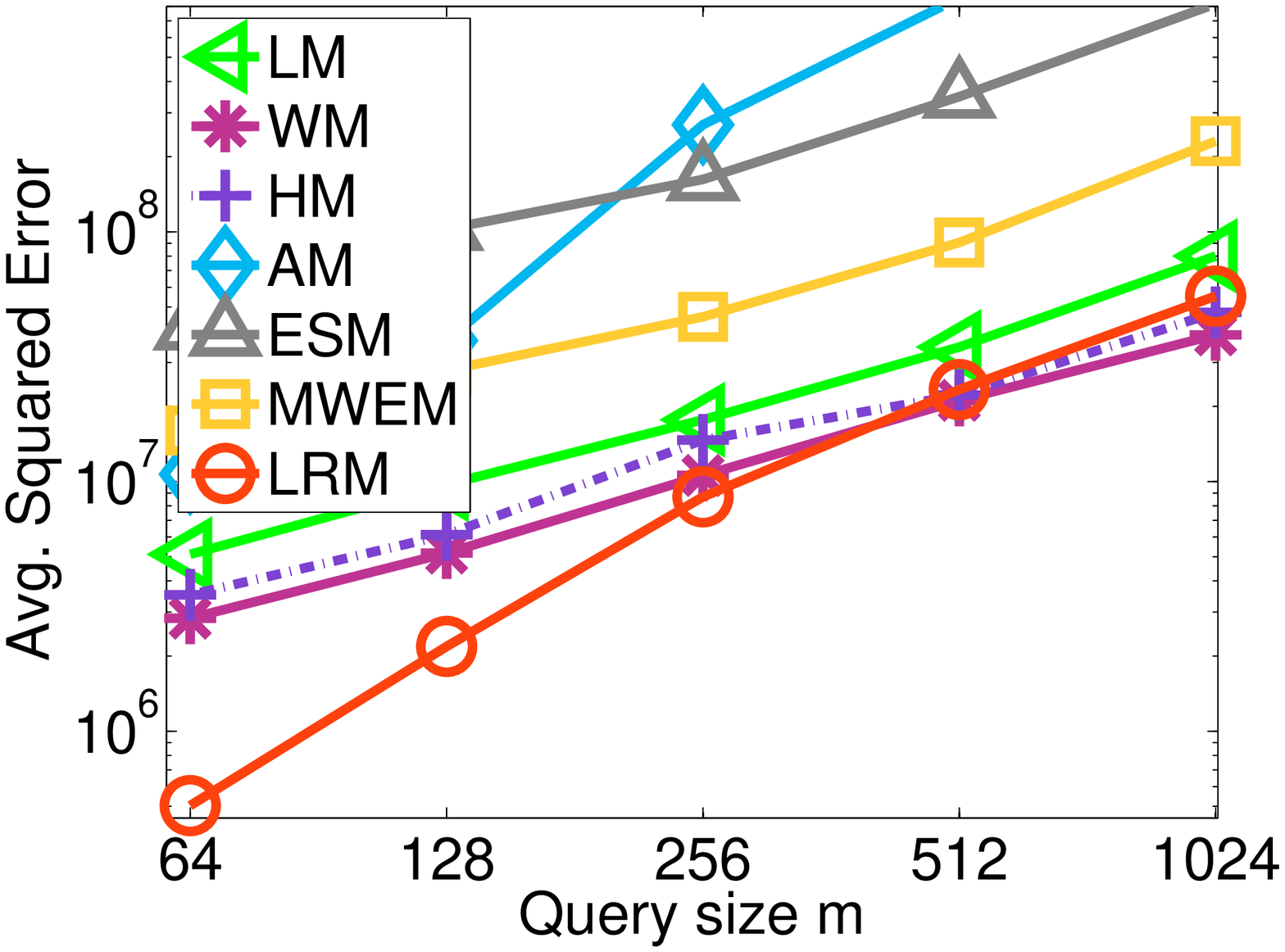}}
\centering \subfigure[\emph{Social Network}]
{\includegraphics[width=0.244\textwidth]{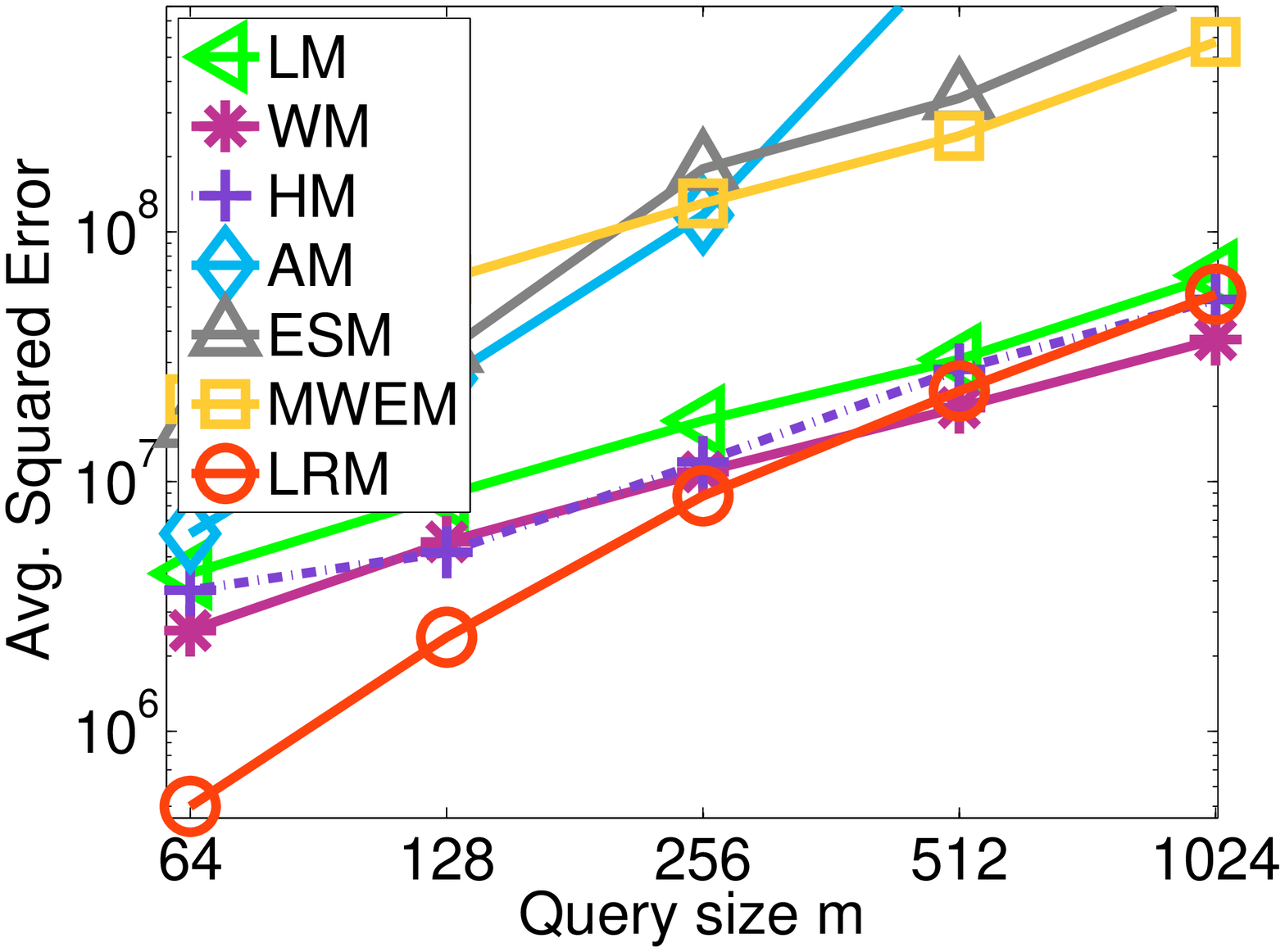}}
\centering \subfigure[\emph{UCI Adult}]
{\includegraphics[width=0.244\textwidth]{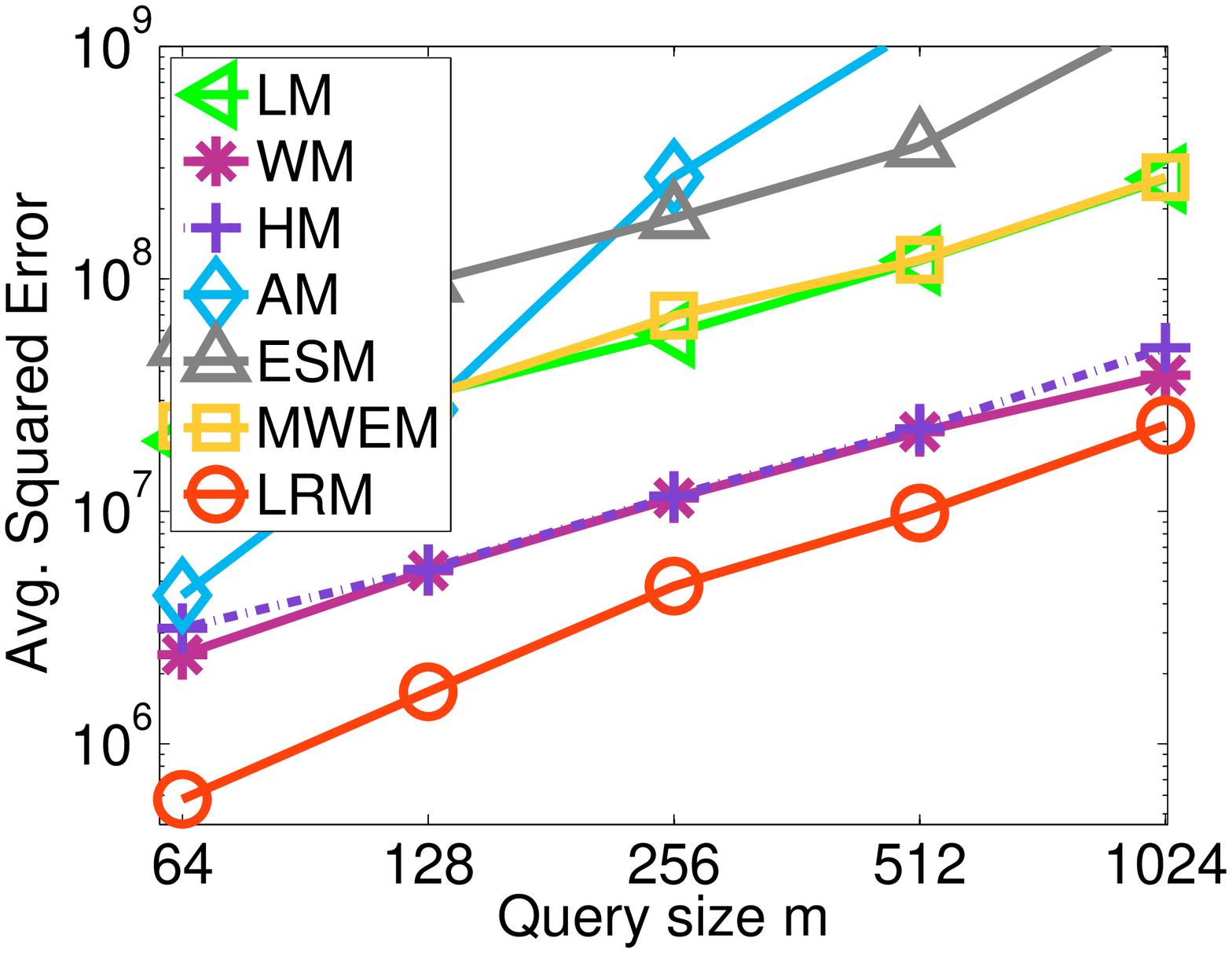}}
\caption{Effect of number of queries $m$ on workload \emph{WRange} under $\epsilon$-differential privacy with $\epsilon=0.1$}
\label{fig:exp:m:WRange}
\end{figure*}

\begin{figure*}[!t]
\centering \subfigure[\emph{Search Logs}]
{\includegraphics[width=0.244\textwidth]{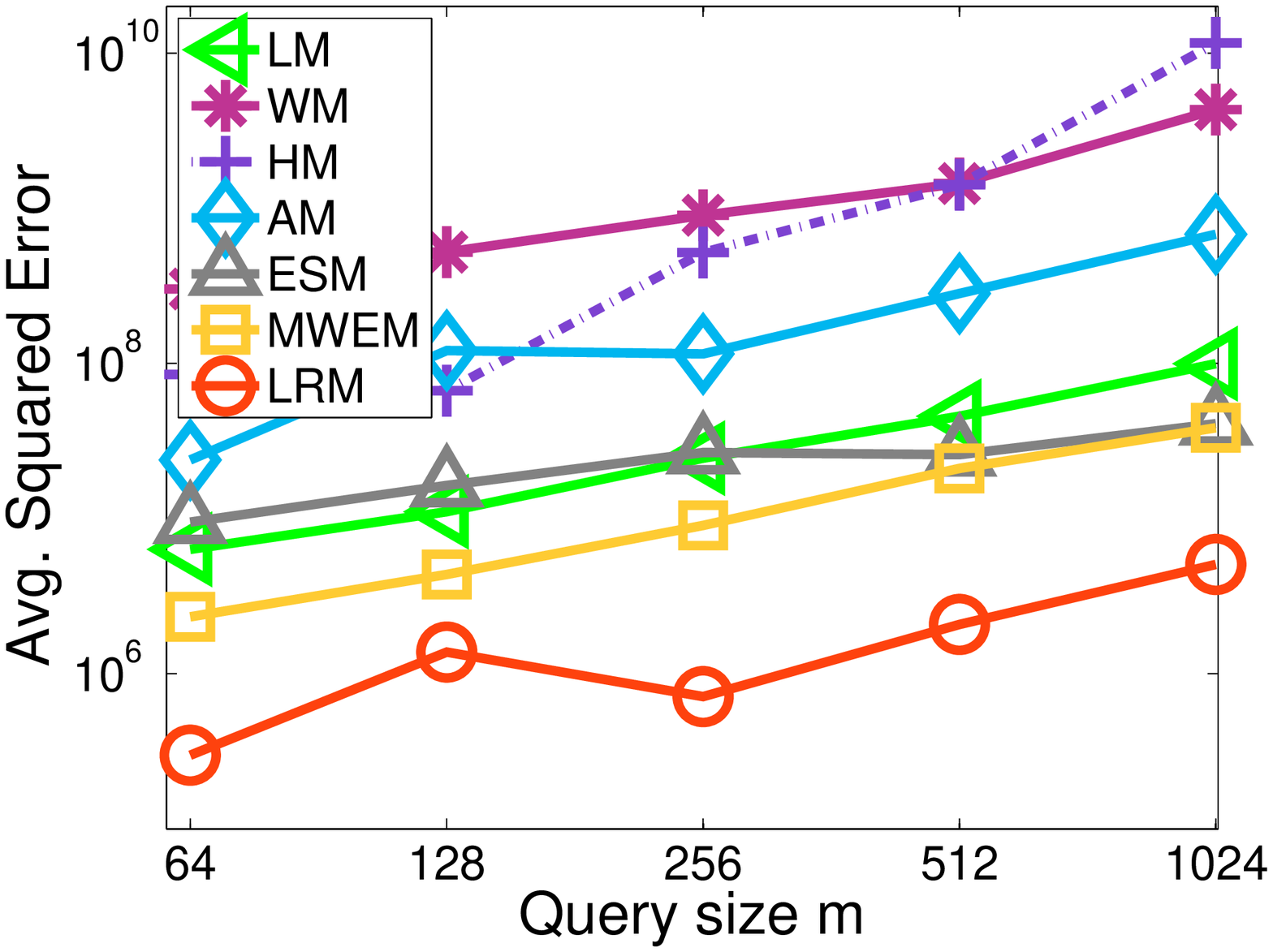}}
\subfigure[\emph{Net Trace}]
{\includegraphics[width=0.244\textwidth]{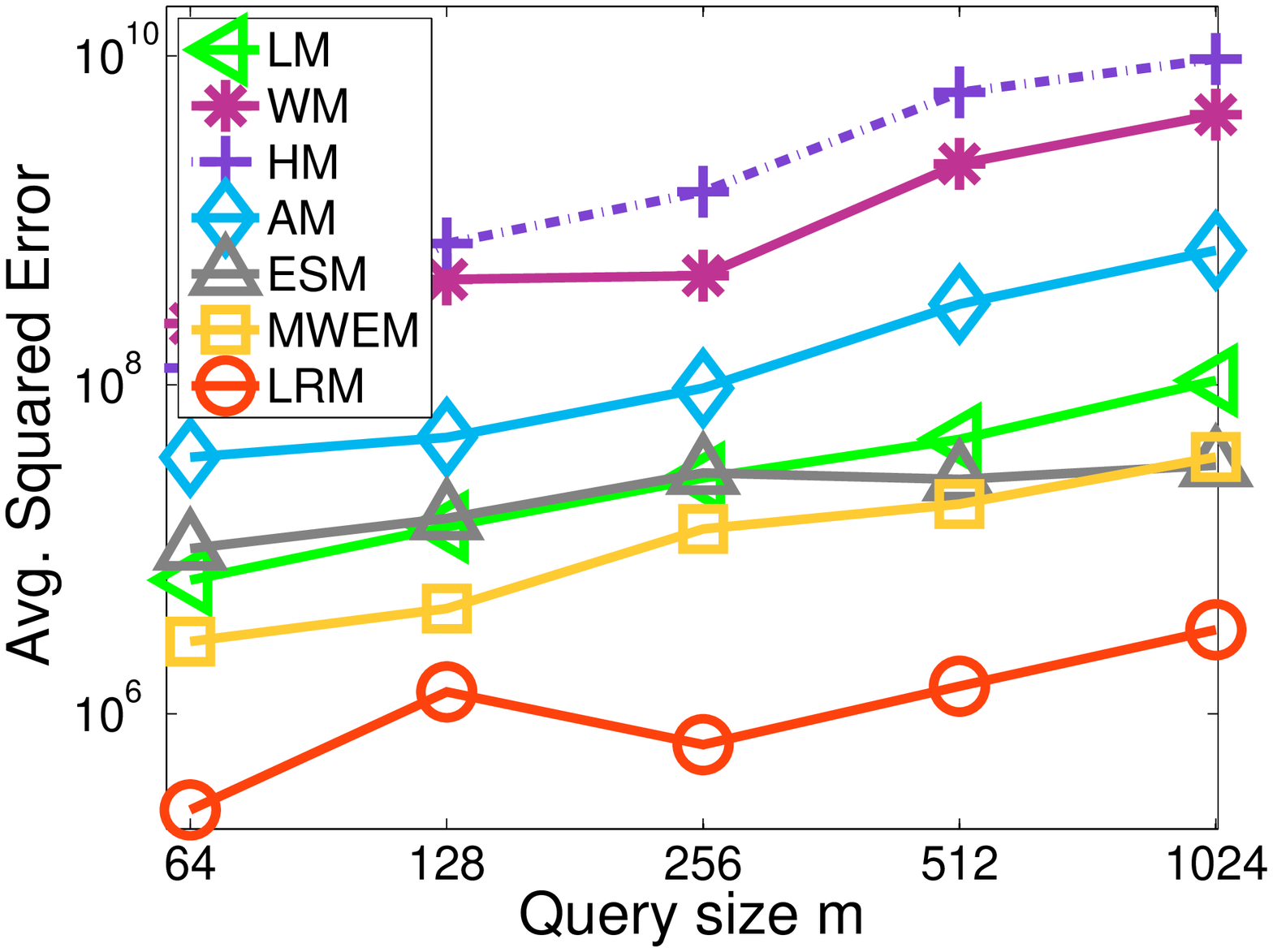}}
\centering \subfigure[\emph{Social Network}]
{\includegraphics[width=0.244\textwidth]{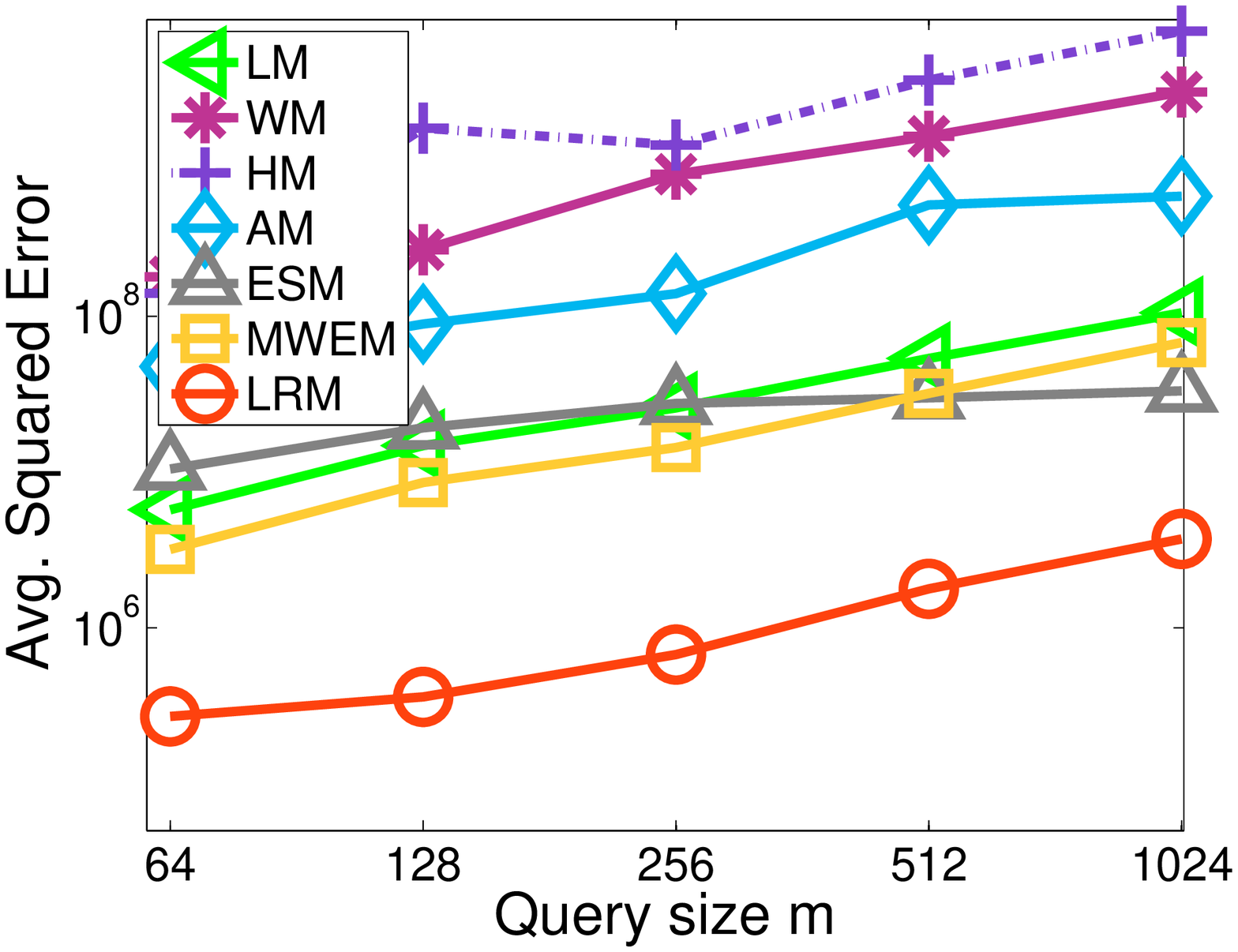}}
\centering \subfigure[\emph{UCI Adult}]
{\includegraphics[width=0.244\textwidth]{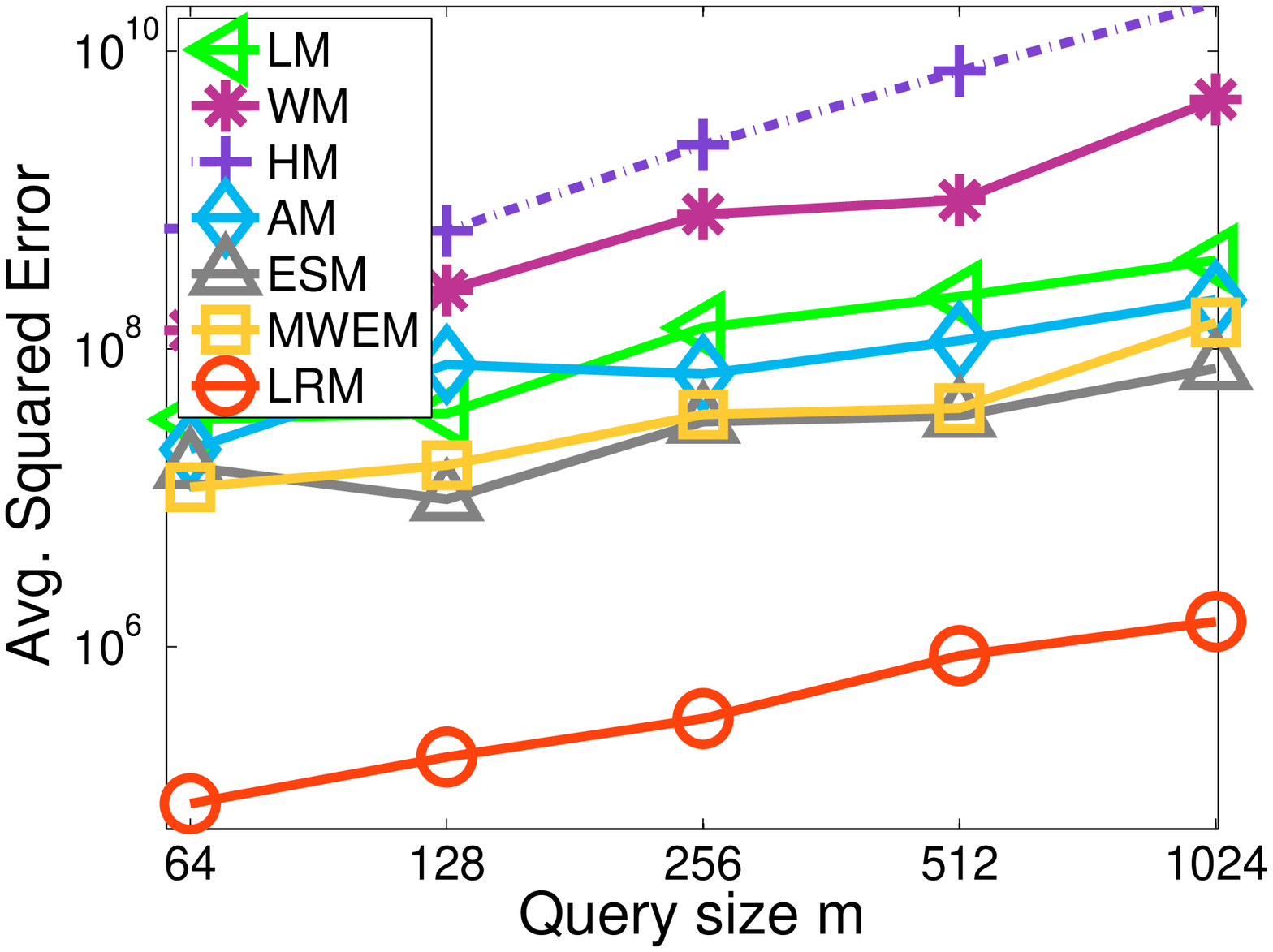}}
\caption{Effect of number of queries $m$ on workload \emph{WMarginal} under $\epsilon$-differential privacy with $\epsilon=0.1$}
\label{fig:exp:m:WMarginal}
\end{figure*}

\begin{figure*}[!t]
\centering \subfigure[\emph{Search Logs}]
{\includegraphics[width=0.244\textwidth]{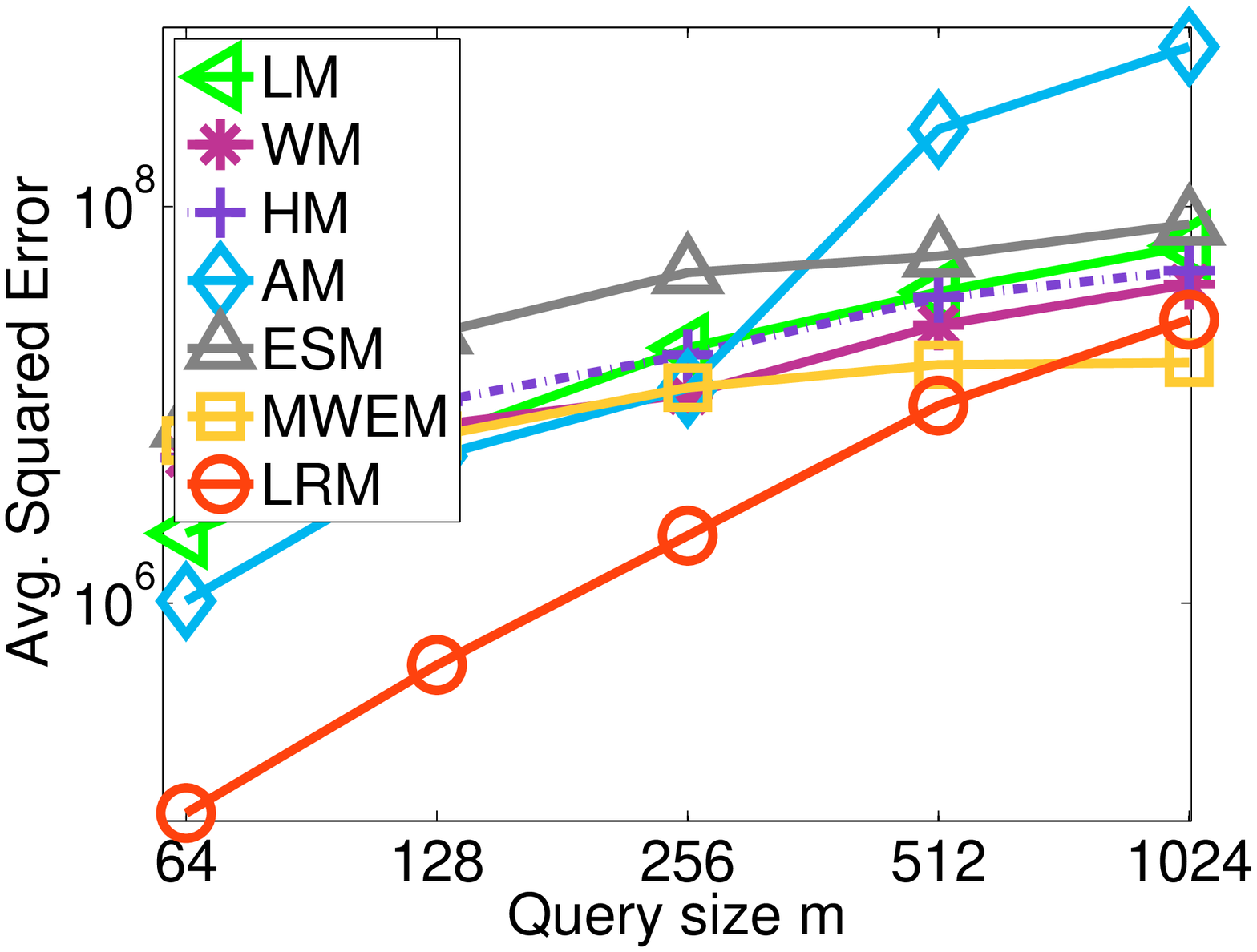}}
\subfigure[\emph{Net Trace}]
{\includegraphics[width=0.244\textwidth]{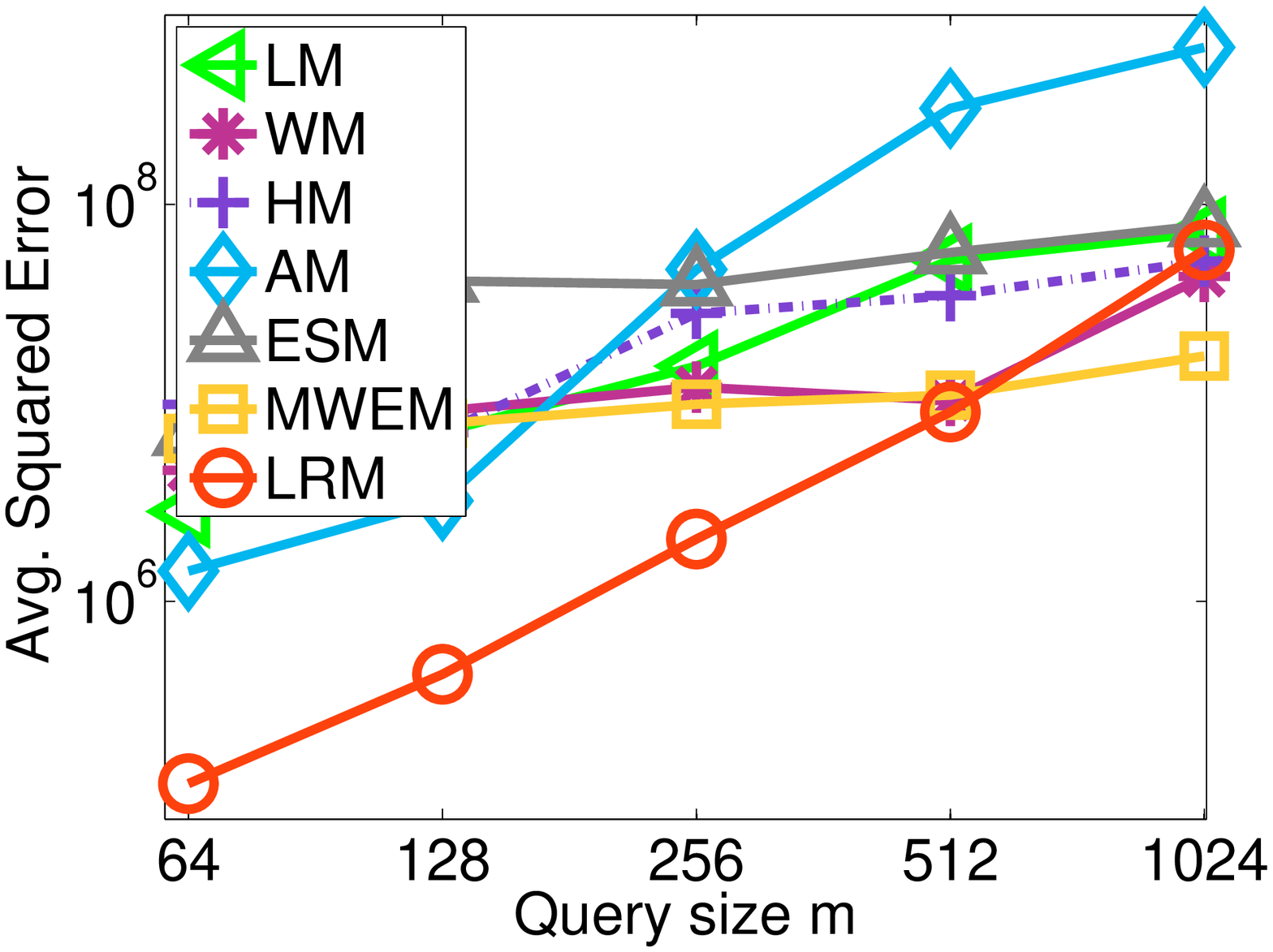}}
\centering \subfigure[\emph{Social Network}]
{\includegraphics[width=0.244\textwidth]{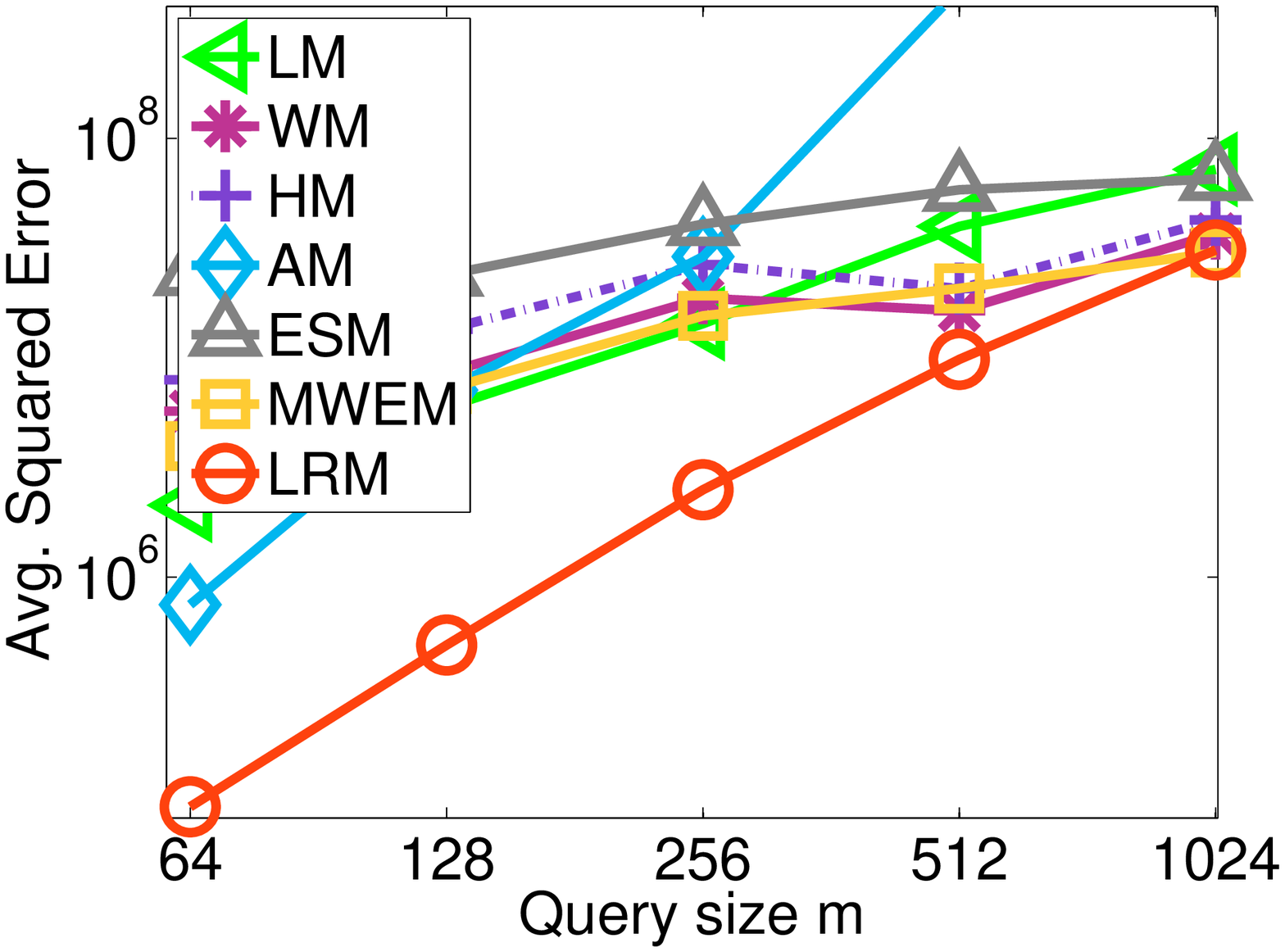}}
\centering \subfigure[\emph{UCI Adult}]
{\includegraphics[width=0.244\textwidth]{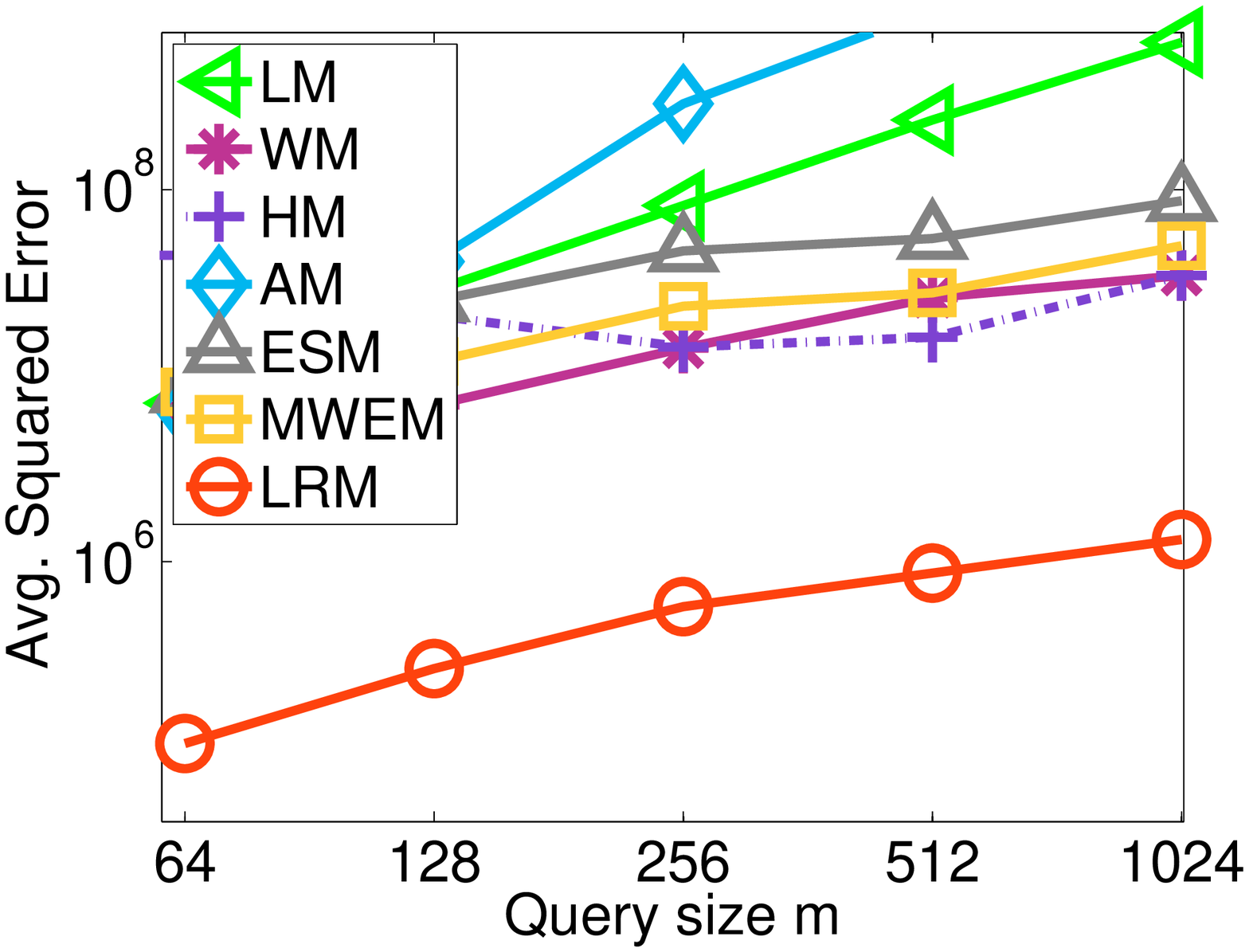}}
\caption{Effect of number of queries $m$ on workload \emph{WRelated} under $\epsilon$-differential privacy with $\epsilon=0.1$}
\label{fig:exp:m:WRelated}
\end{figure*}

\subsection{Impact of Varying Query Rank $s$}\label{sec:vary_s}

The previous experiments demonstrate LRM's substantial performance
advantages when the workload matrix has low rank. In this set of experiments, we manually control the rank of workload $W$ to verify
this observation. Recall that the parameter $s$ determines
the size of the matrix $C_{m\times s}$ and the size of the matrix
$A_{s\times n}$ during the generation of the \emph{WRelated} workload. When
$C$ and $A$ contain only independent rows/columns, $s$ is exactly
the rank of the workload matrix $W=CA$. In Figure \ref{fig:exp:s} and \ref{fig:exp:s:app}, we vary $s$ from $0.1 \times \min(m,n)$ to $1 \times \min(m,n)$.

For $\epsilon$-differential privacy, LRM outperforms all other methods by at least one order of magnitude when $s$ is low. With increasing $s$, the performance gap gradually closes. This phenomenon confirms that the low rank property is the main reason behind LRM's advantages. For ($\epsilon$, $\delta$)-differential privacy, LRM also gives the best performance in all test cases; its performance advantage decreases with $s$, though at a much slower rate compared to the case of $\epsilon$-differential privacy.

\begin{figure*}[!t]
\centering \subfigure[\emph{Search Logs}]
{\includegraphics[width=0.244\textwidth]{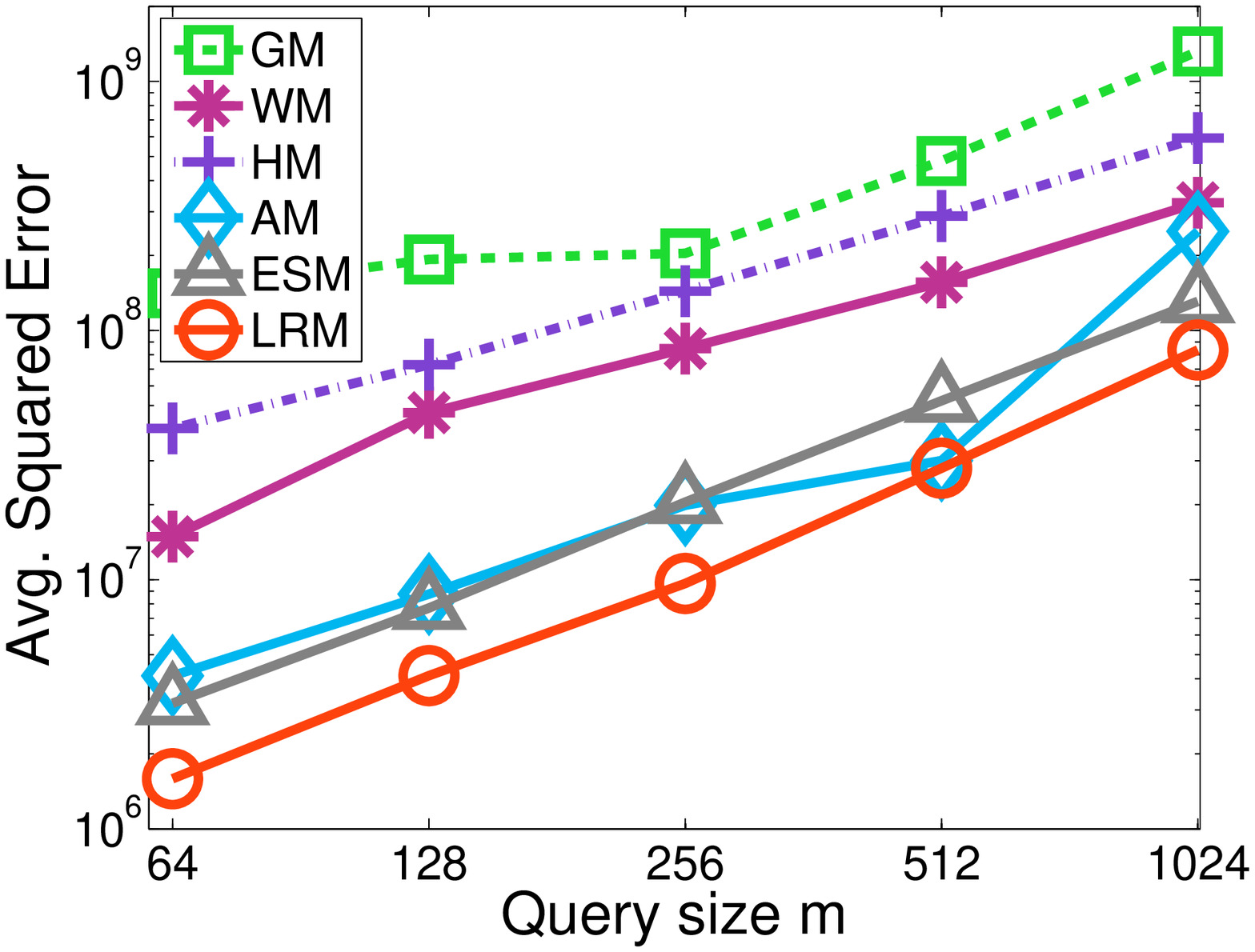}}
\subfigure[\emph{Net Trace}]
{\includegraphics[width=0.244\textwidth]{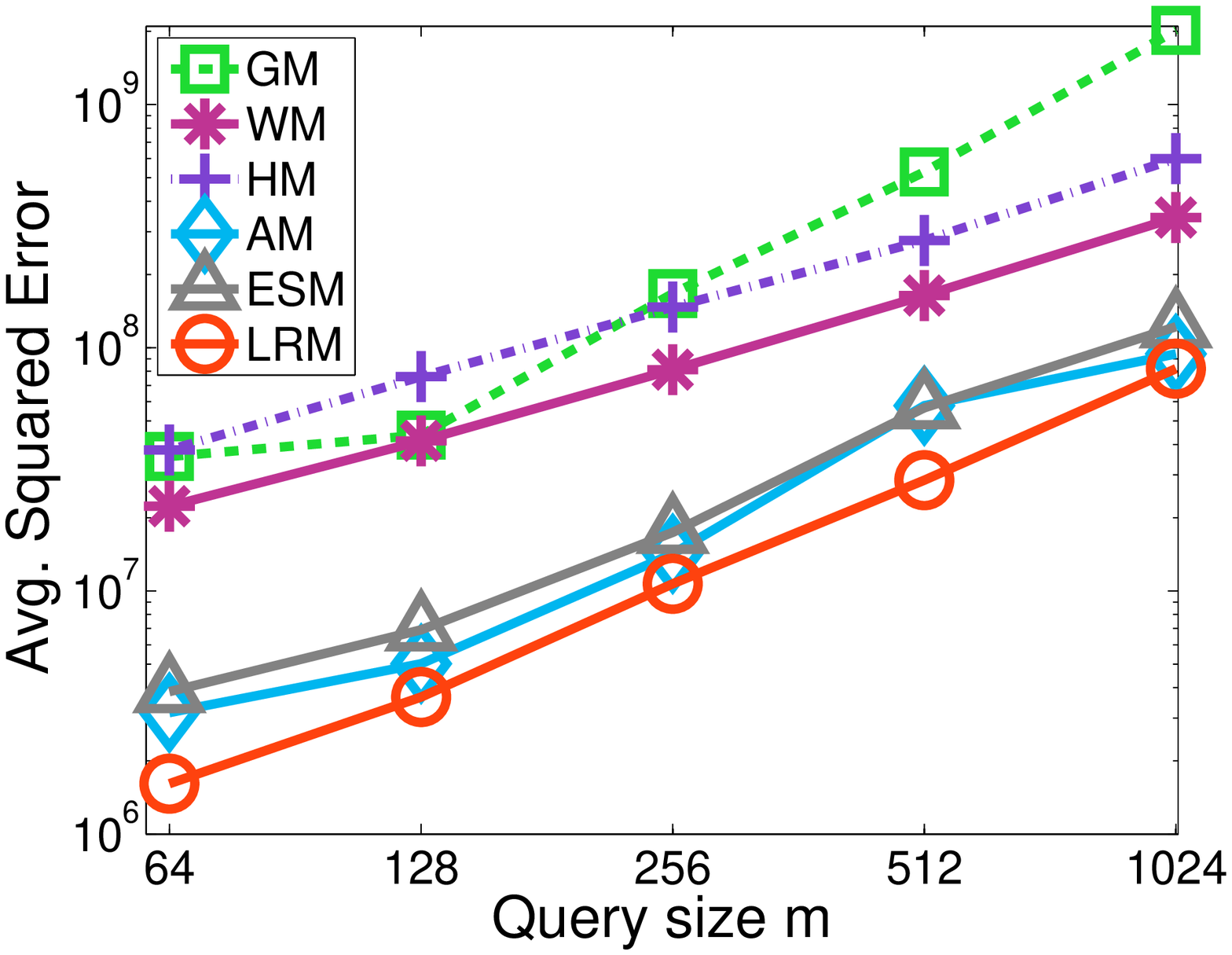}}
\centering \subfigure[\emph{Social Network}]
{\includegraphics[width=0.244\textwidth]{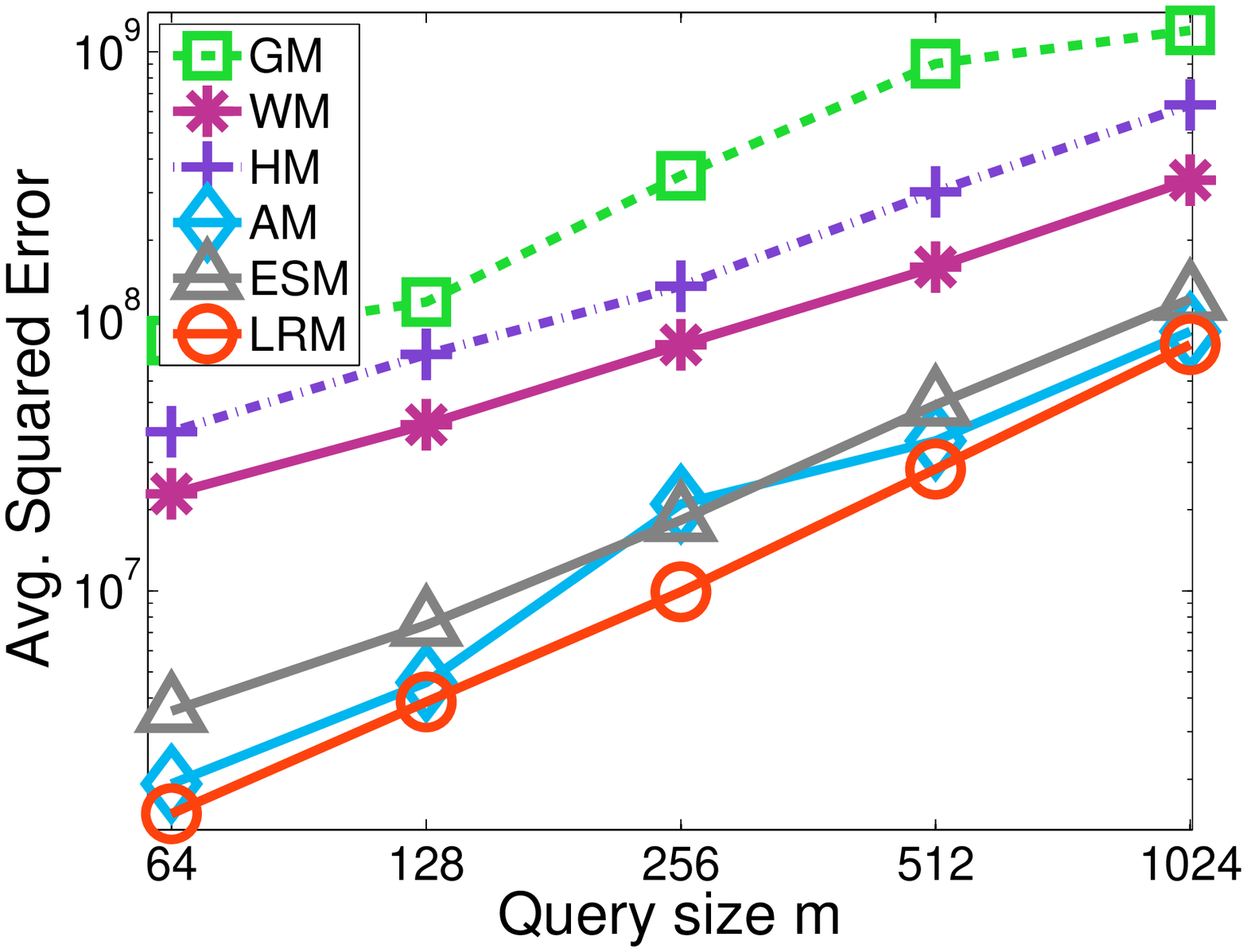}}
\centering \subfigure[\emph{UCI Adult}]
{\includegraphics[width=0.244\textwidth]{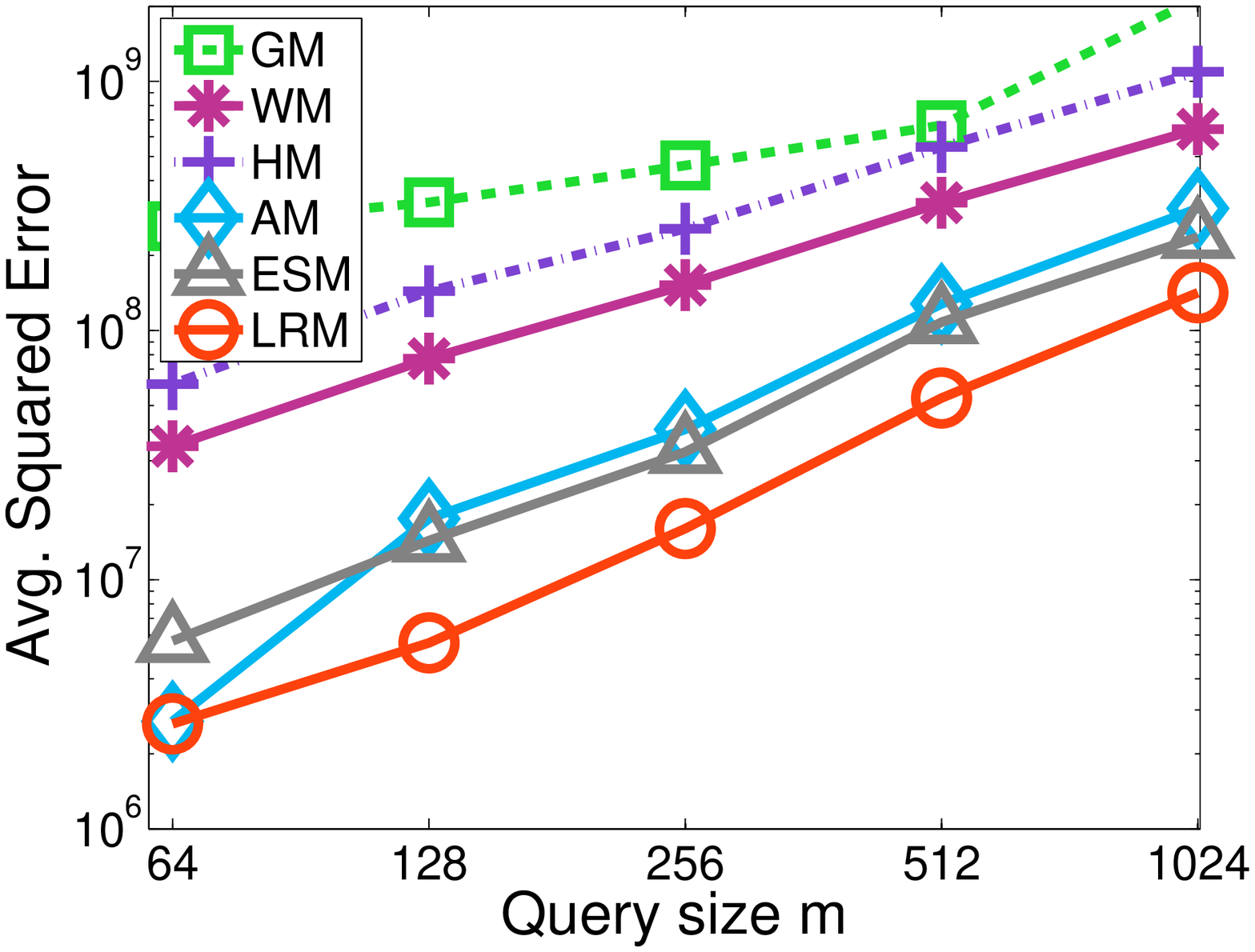}}
\caption{Effect of number of queries $m$ on workload
\emph{WDiscrete} under ($\epsilon$, $\delta$)-differential privacy
with $\epsilon=0.1$ and $\delta=0.0001$} \label{fig:exp:m:WDiscrete:app}
\end{figure*}

\begin{figure*}[!t]
\centering \subfigure[\emph{Search Logs}]
{\includegraphics[width=0.244\textwidth]{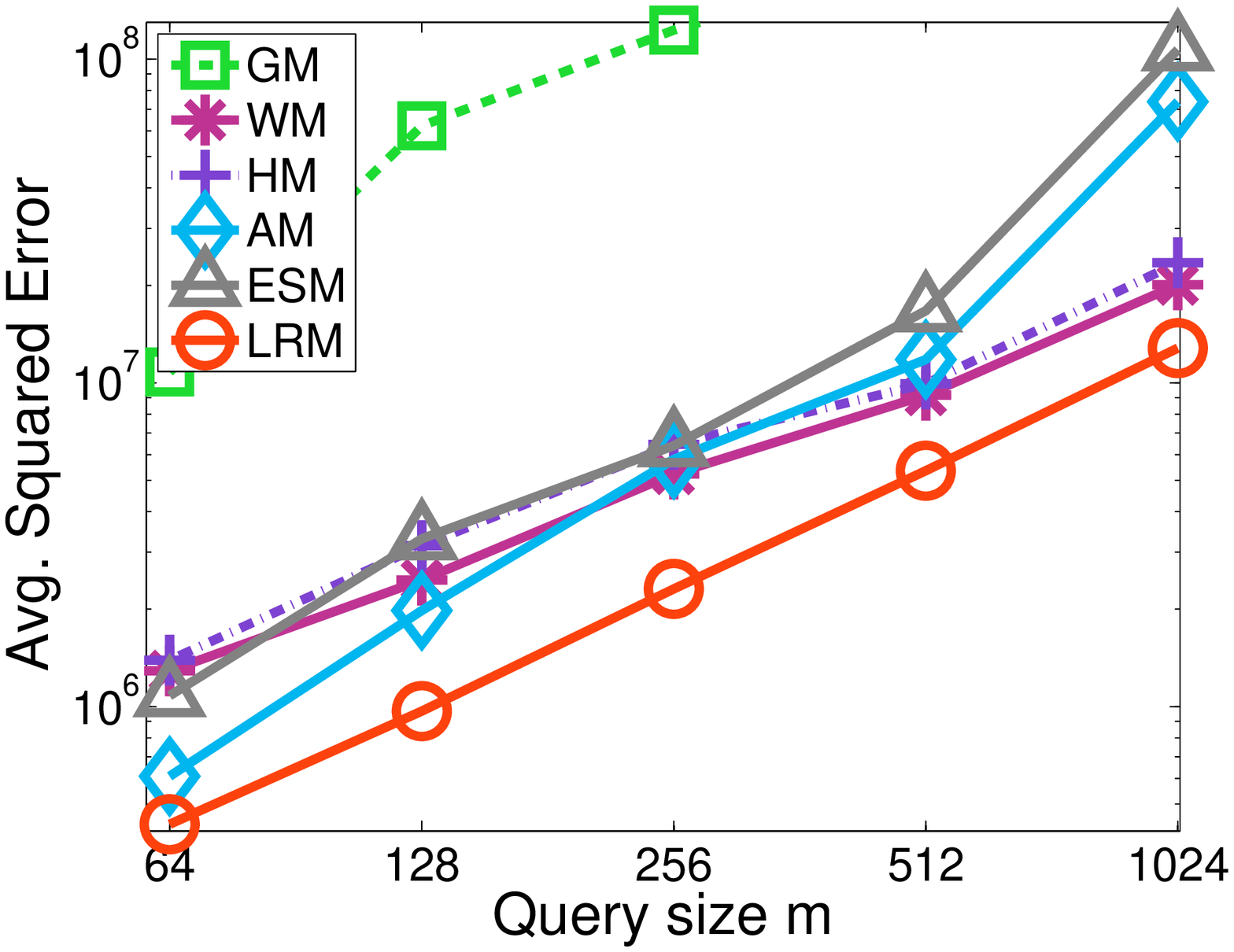}}
\subfigure[\emph{Net Trace}]
{\includegraphics[width=0.244\textwidth]{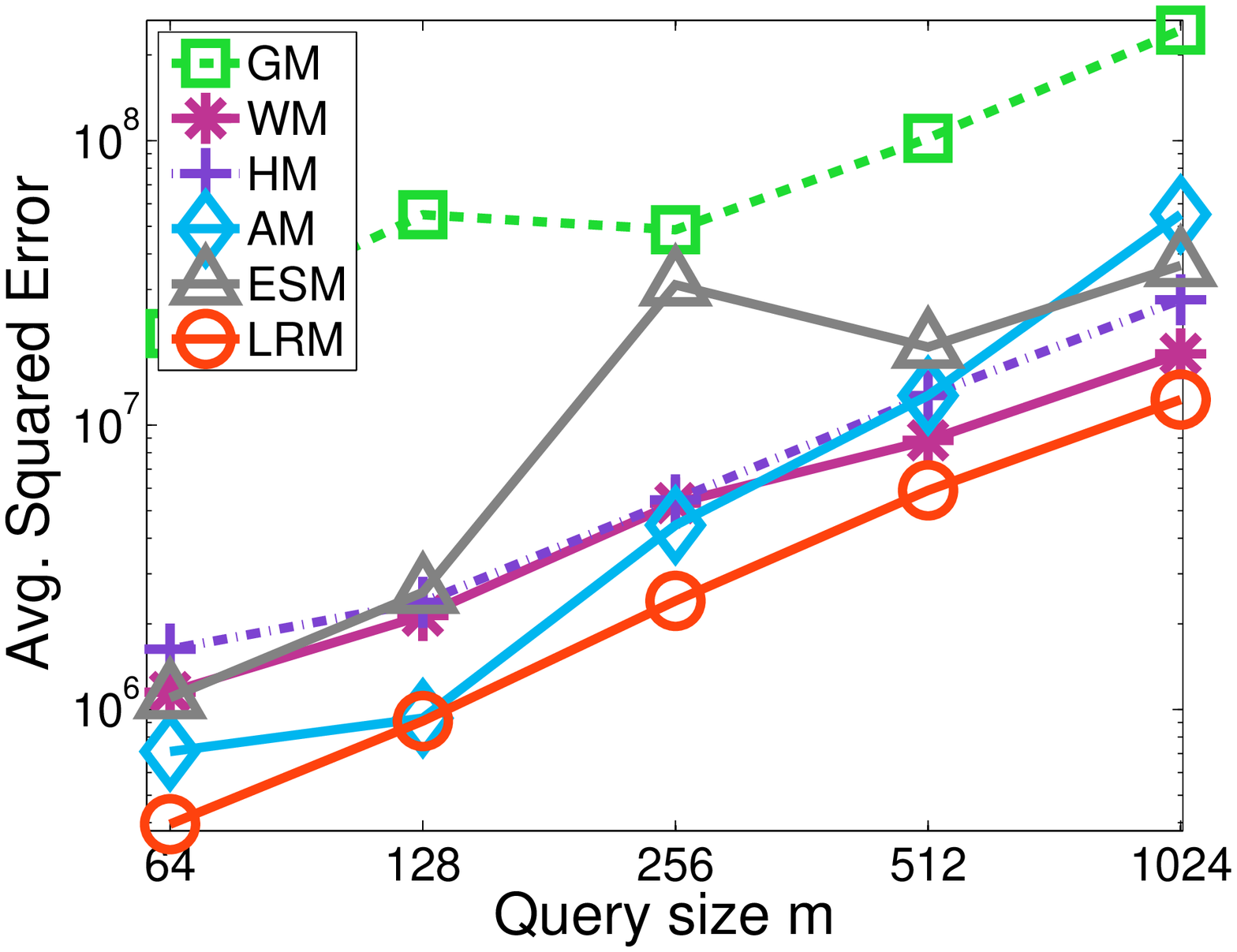}}
\centering \subfigure[\emph{Social Network}]
{\includegraphics[width=0.244\textwidth]{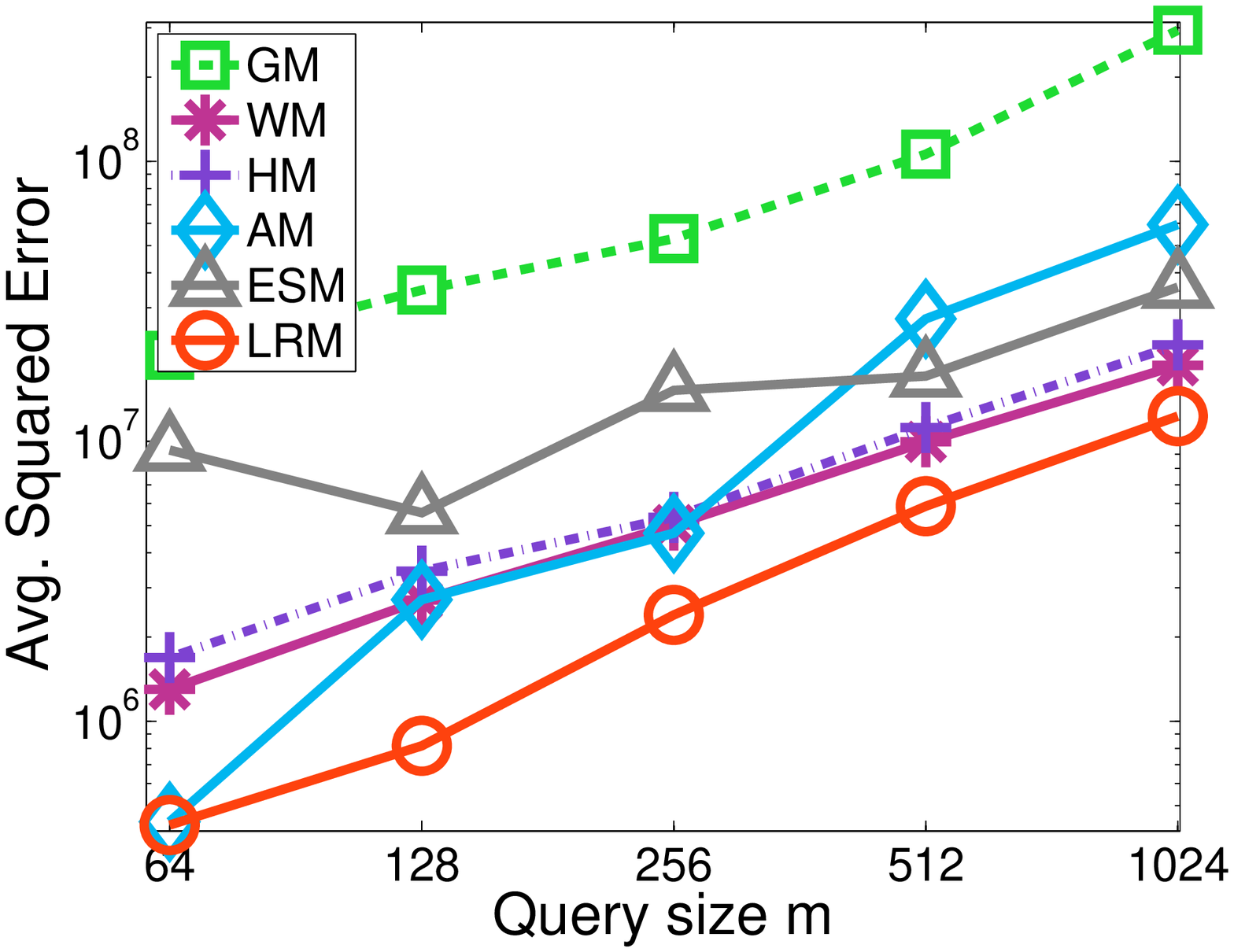}}
\centering \subfigure[\emph{UCI Adult}]
{\includegraphics[width=0.244\textwidth]{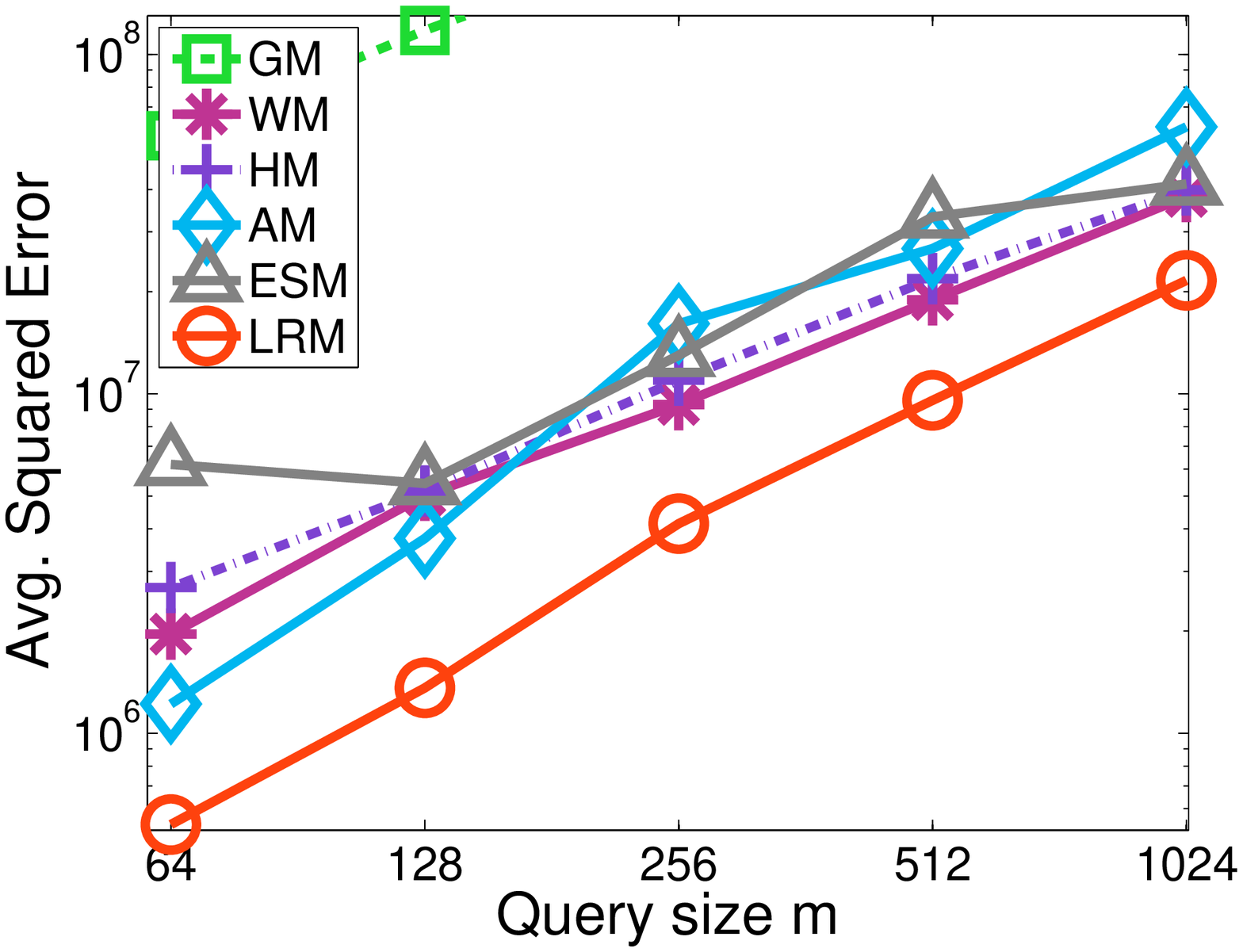}}
\caption{Effect of number of queries $m$ on workload \emph{WRange}
under ($\epsilon$, $\delta$)-differential privacy with $\epsilon=0.1$
and $\delta=0.0001$} \label{fig:exp:m:WRange:app}
\end{figure*}

\begin{figure*}[!t]
\centering \subfigure[\emph{Search Logs}]
{\includegraphics[width=0.244\textwidth]{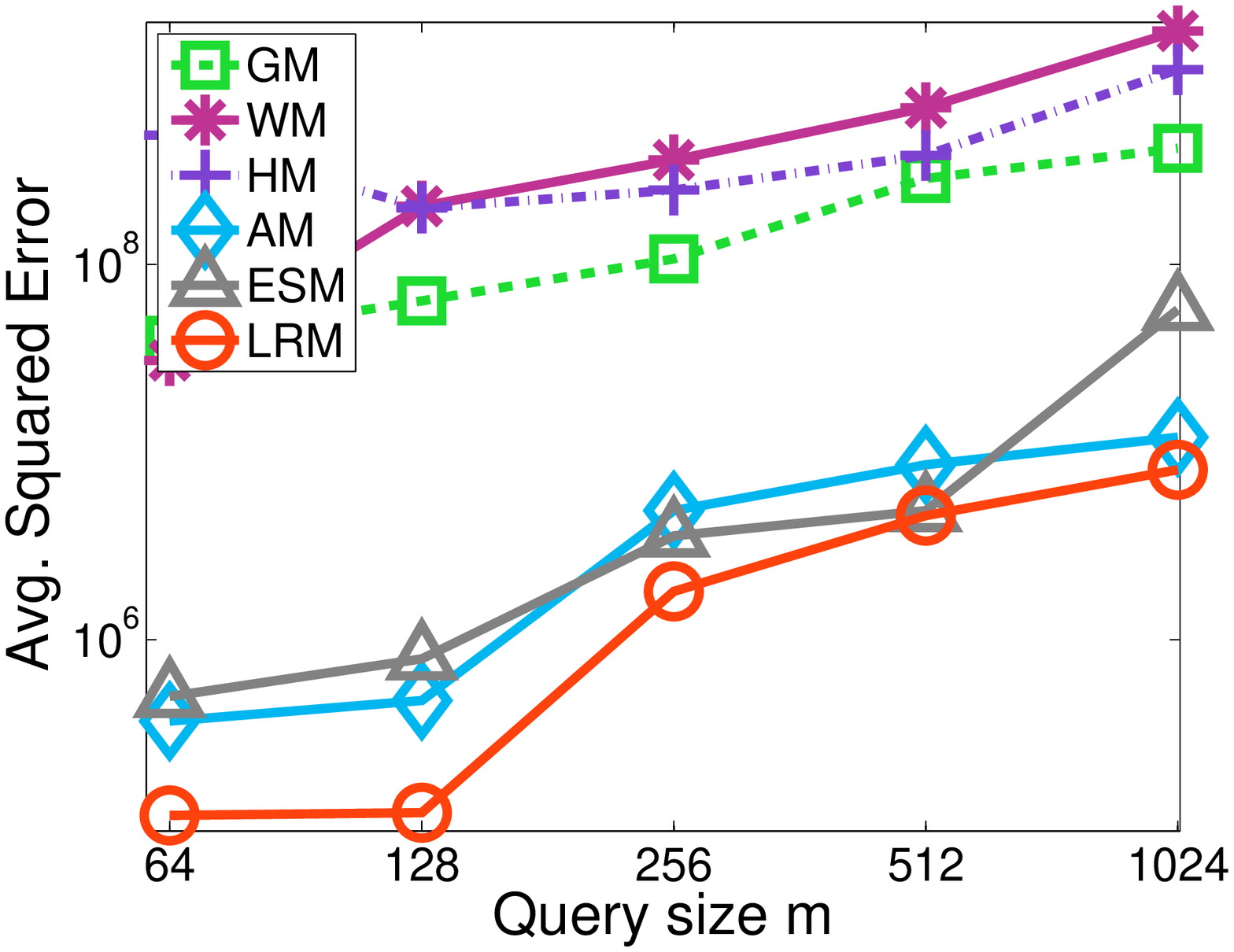}}
\subfigure[\emph{Net Trace}]
{\includegraphics[width=0.244\textwidth]{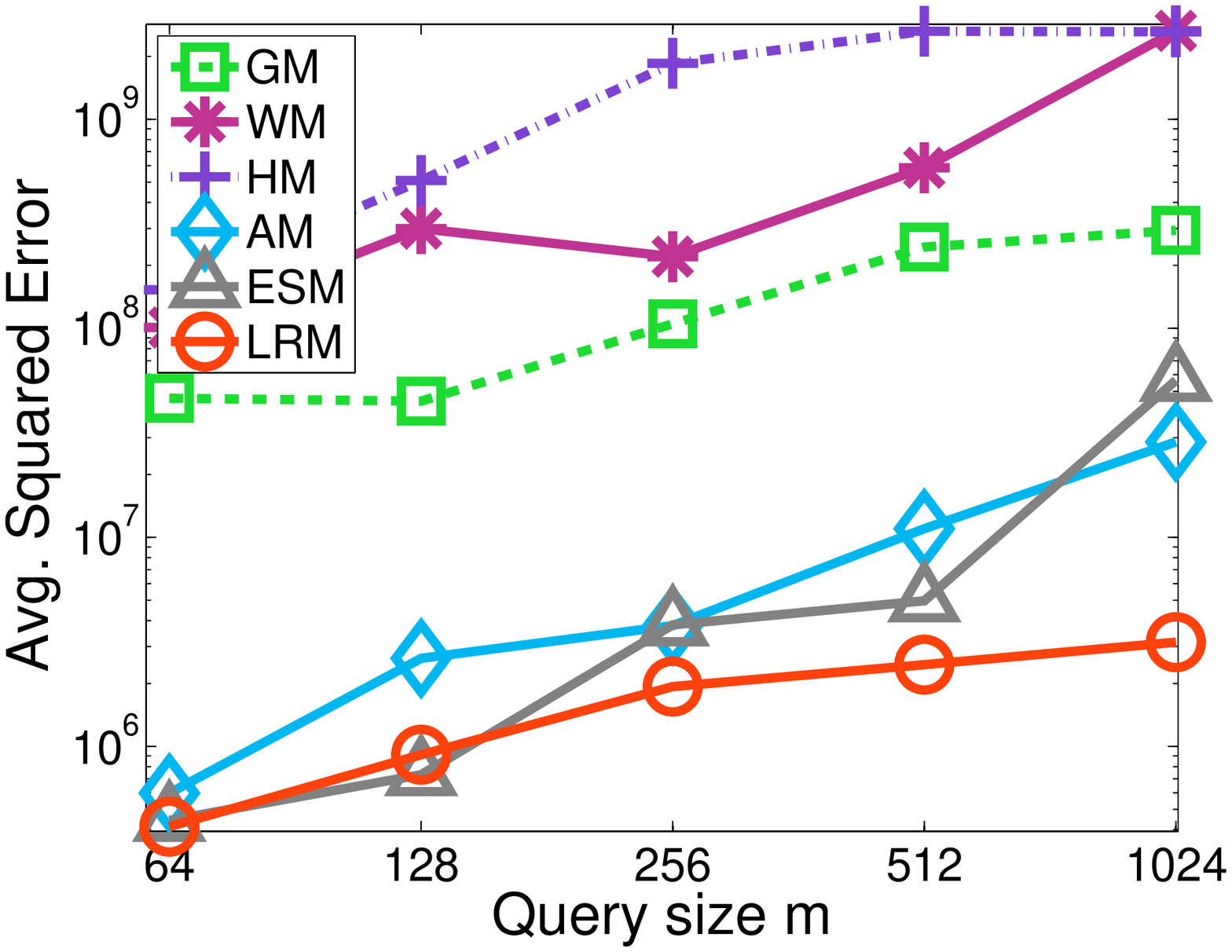}}
\centering \subfigure[\emph{Social Network}]
{\includegraphics[width=0.244\textwidth]{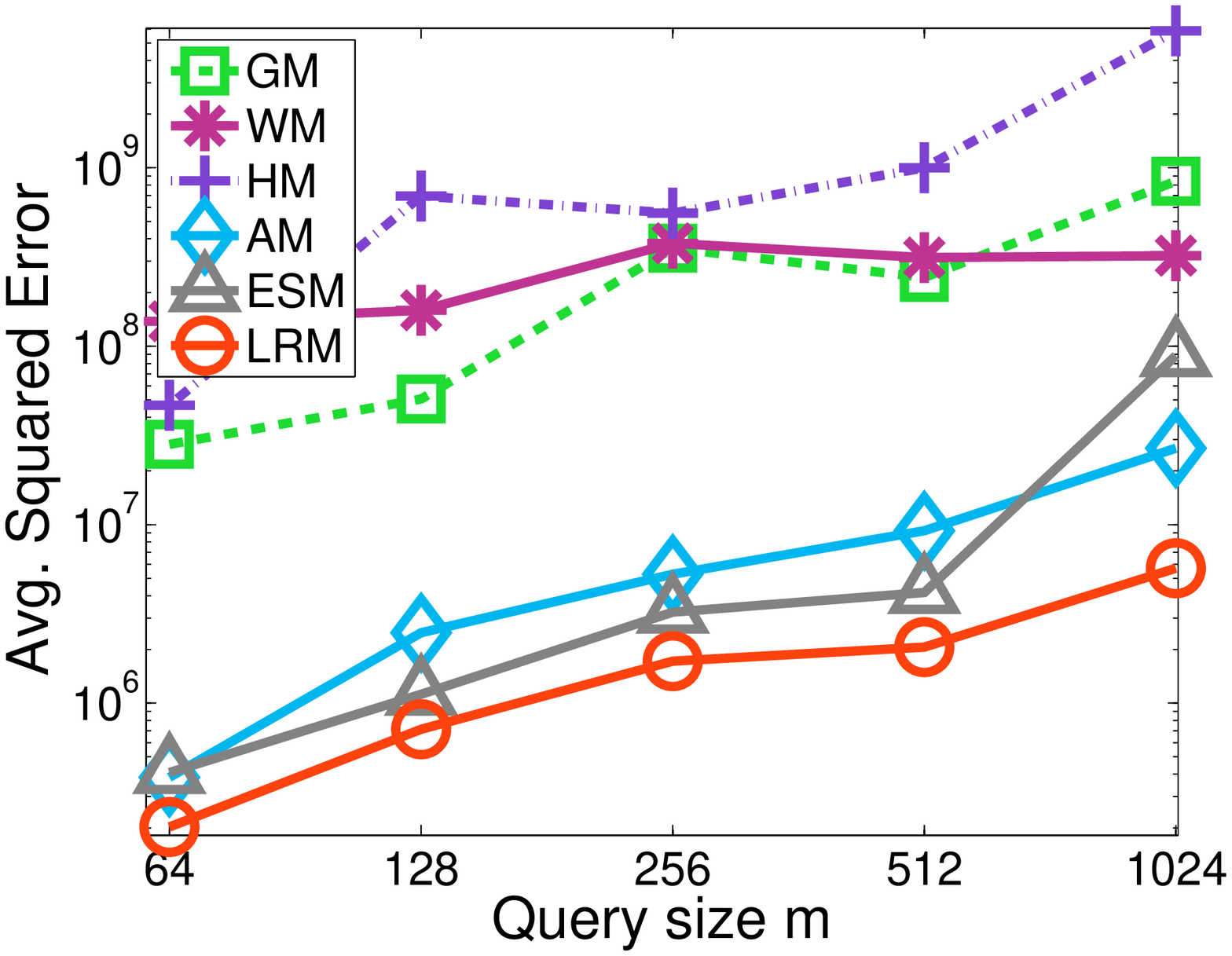}}
\centering \subfigure[\emph{UCI Adult}]
{\includegraphics[width=0.244\textwidth]{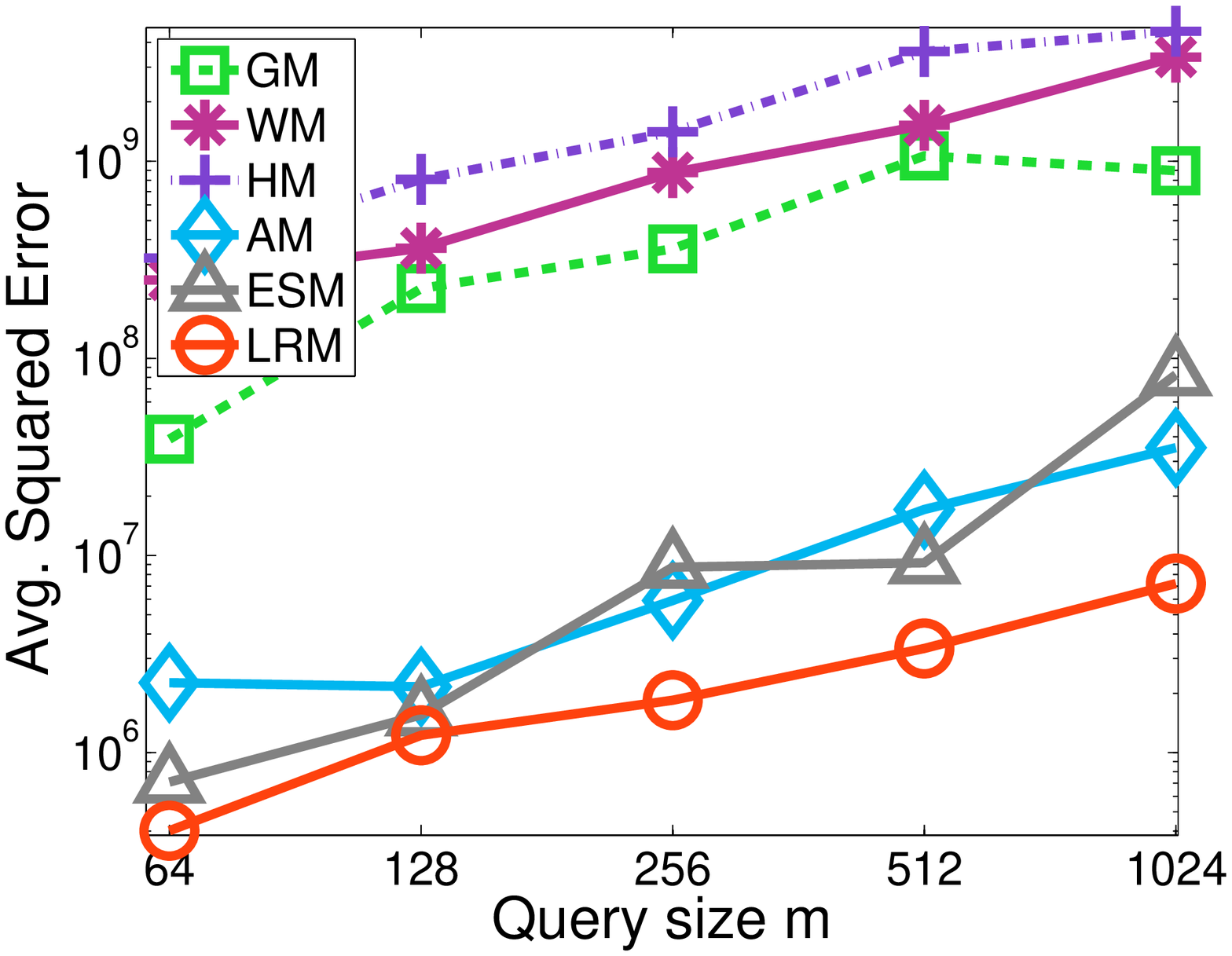}}
\caption{Effect of number of queries $m$ on workload
\emph{WMarginal} under ($\epsilon$, $\delta$)-differential privacy
with $\epsilon=0.1$ and $\delta=0.0001$} \label{fig:exp:m:WMarginal:app}
\end{figure*}

\begin{figure*}[!t]
\centering \subfigure[\emph{Search Logs}]
{\includegraphics[width=0.244\textwidth]{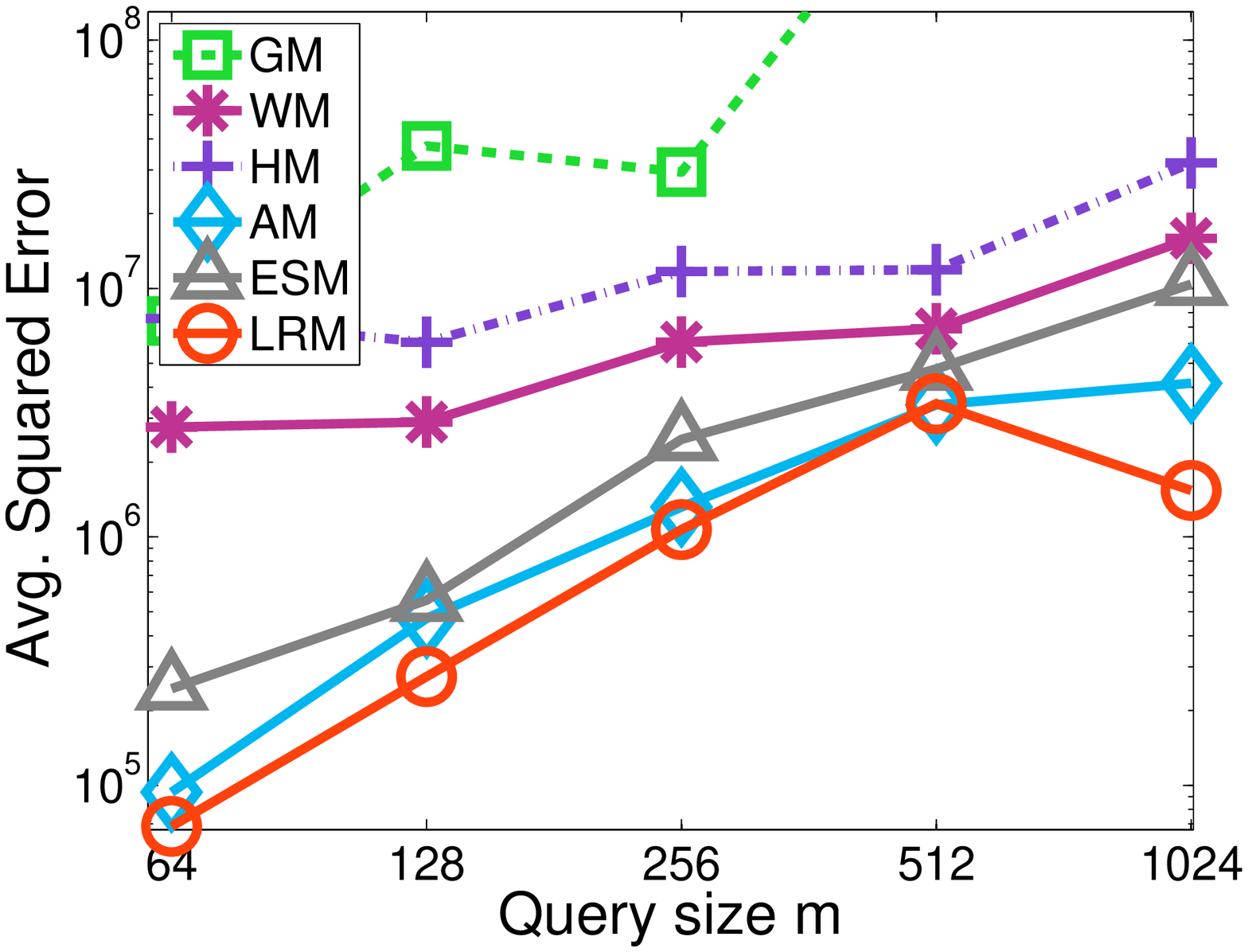}}
\subfigure[\emph{Net Trace}]
{\includegraphics[width=0.244\textwidth]{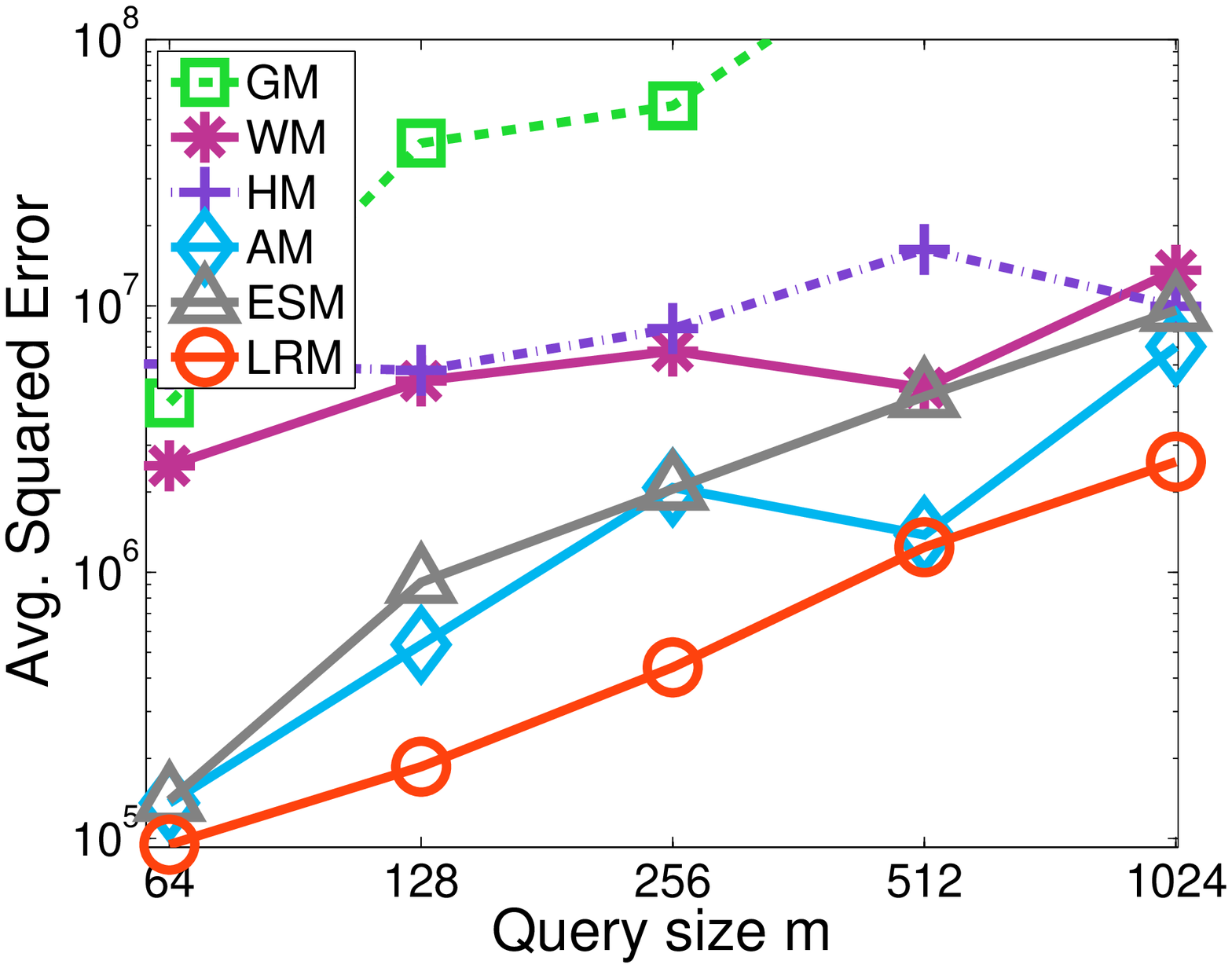}}
\centering \subfigure[\emph{Social Network}]
{\includegraphics[width=0.244\textwidth]{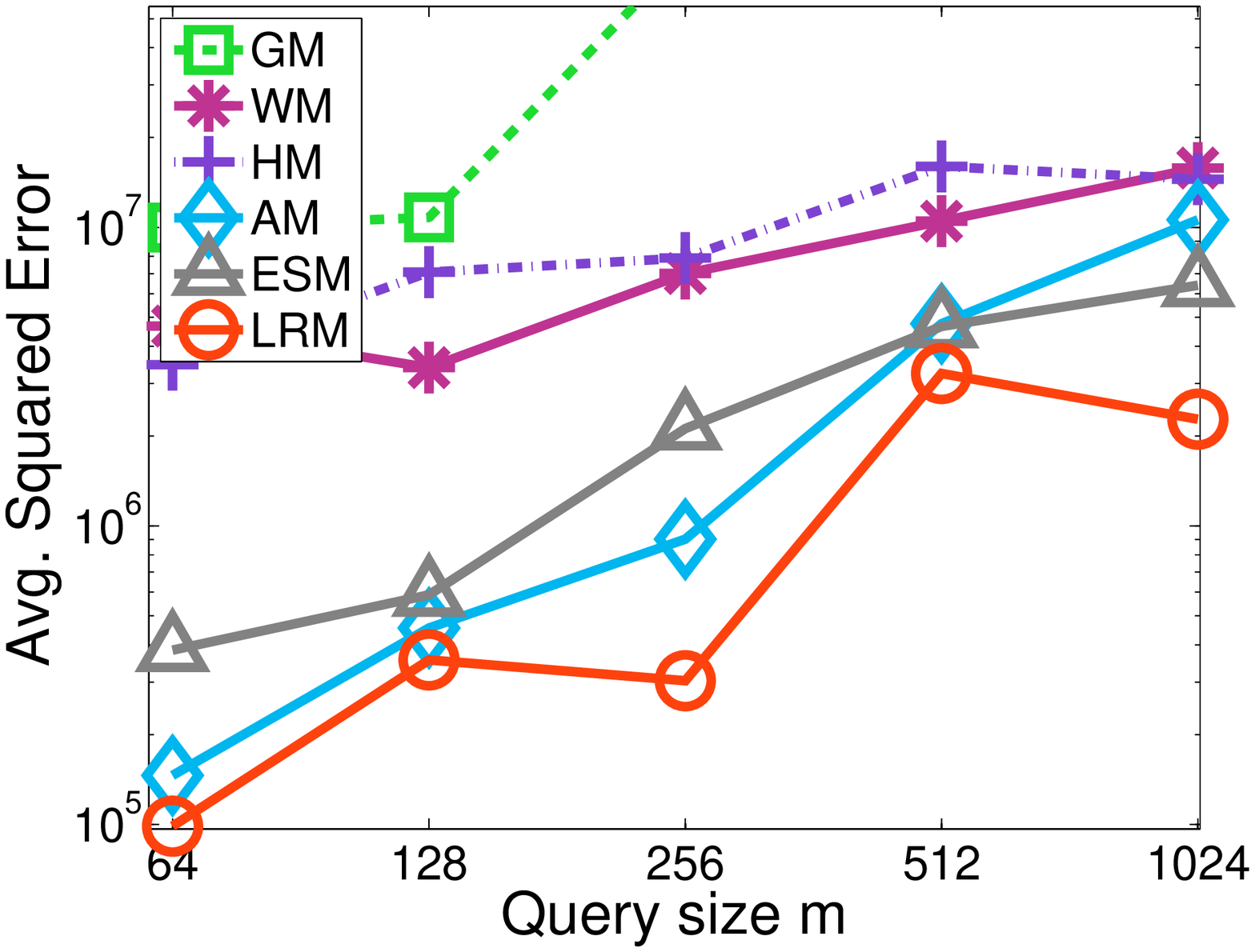}}
\centering \subfigure[\emph{UCI Adult}]
{\includegraphics[width=0.244\textwidth]{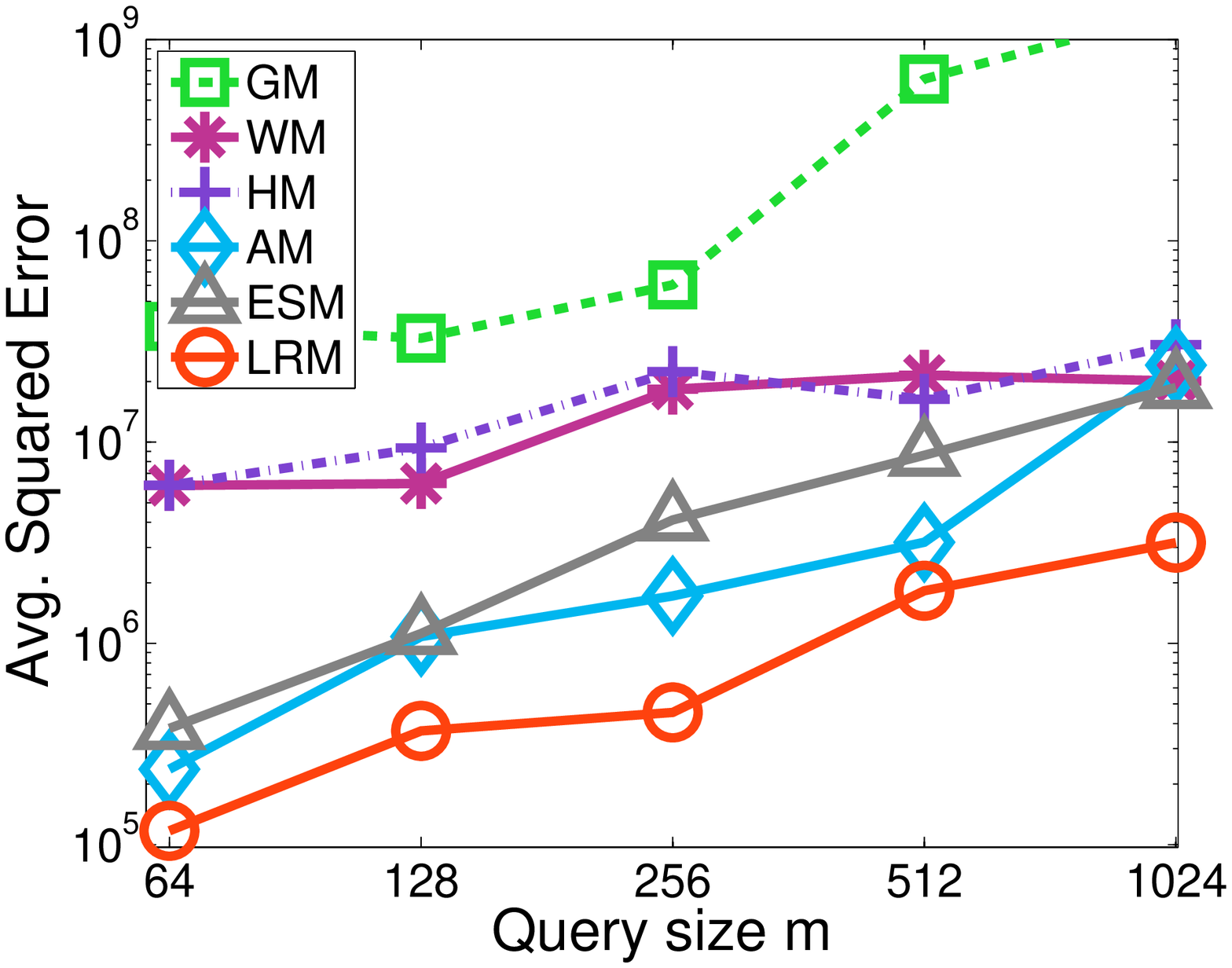}}
\caption{Effect of number of queries $m$ on workload \emph{WRelated}
under ($\epsilon$, $\delta$)-differential privacy with $\epsilon=0.1$
and $\delta=0.0001$} \label{fig:exp:m:WRelated:app}
\end{figure*}

\subsection{Scalability of the Low-Rank Mechanism}\label{sec:exp:scalability}
Finally, we demonstrate the efficiency and scalability of LRM under $\epsilon$- and ($\epsilon$, $\delta$)-differential privacy. The running time of LRM is dominated by the optimization module that solves the best workload decomposition, which is independent of the dataset. In Figure \ref{fig:exp:scal} and Figure \ref{fig:exp:scal:app}, we vary the domain size $n$ from 128 to 8192 and the number of queries $m$ from 64 to 256, respectively, and report the total running time of LRM for the 4 different types of workloads in our experiments. LRM scales roughly linearly with the domain size $n$ and the number of queries $m$ (note that both axes are in logarithmic scale). Moreover, we observe that for workload \emph{WRelated}, LRM runs faster when the rank $s$ of the workload is lower, given the same values of $n$ and $m$. LRM under ($\epsilon$, $\delta$)-differential privacy is slightly more efficient than under $\epsilon$-differential privacy. This is expected, since we set a smaller value of $r$ for ($\epsilon$, $\delta$)-differential privacy. In all settings, LRM always terminates within 20 minutes for each experiment. In practice, this computation time pays off as LRM achieves significantly higher accuracy than existing methods.

\begin{figure*}[!t]
\centering \subfigure[\emph{Search Logs}]
{\includegraphics[width=0.244\textwidth]{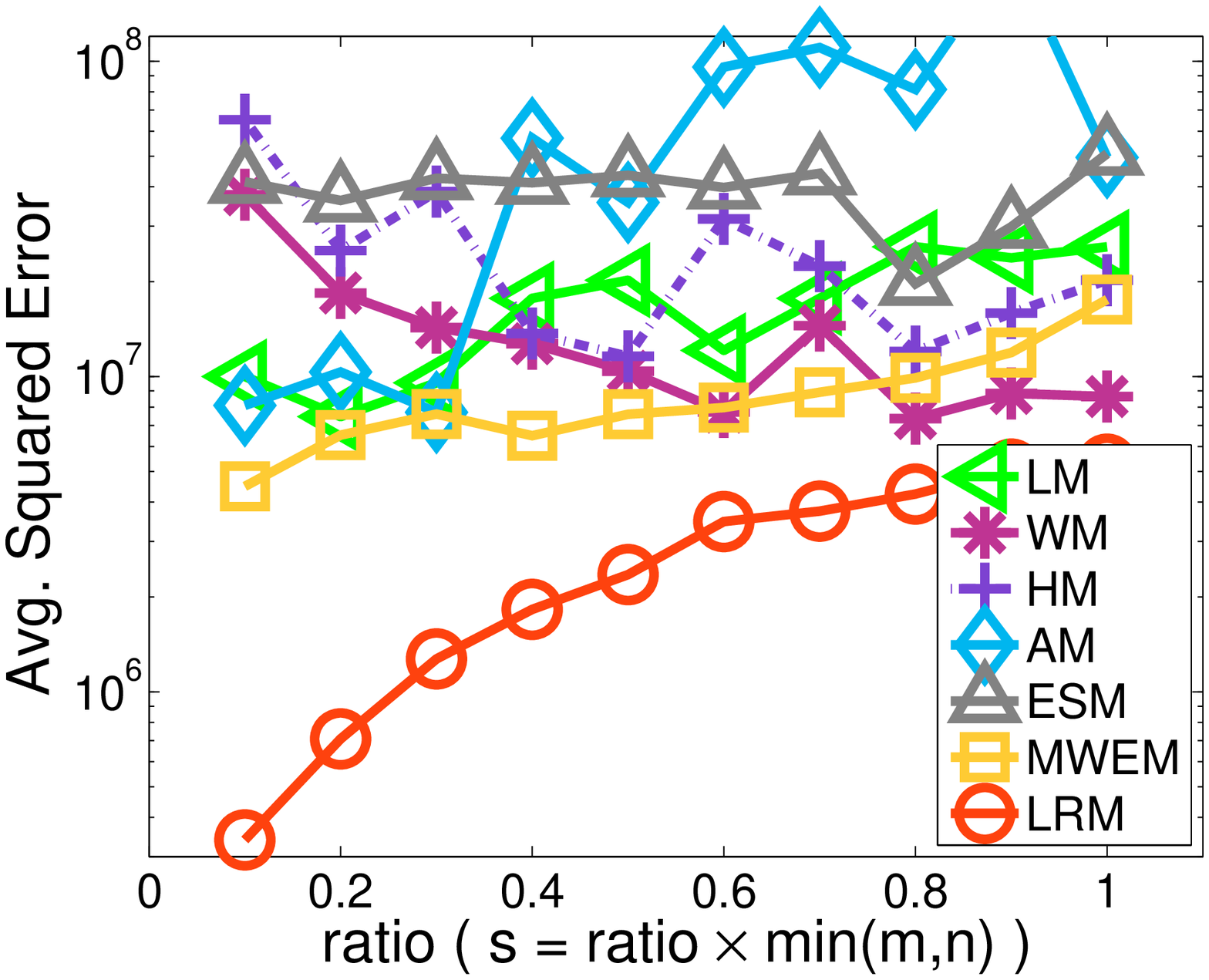}}
\subfigure[\emph{Net Trace}]
{\includegraphics[width=0.244\textwidth]{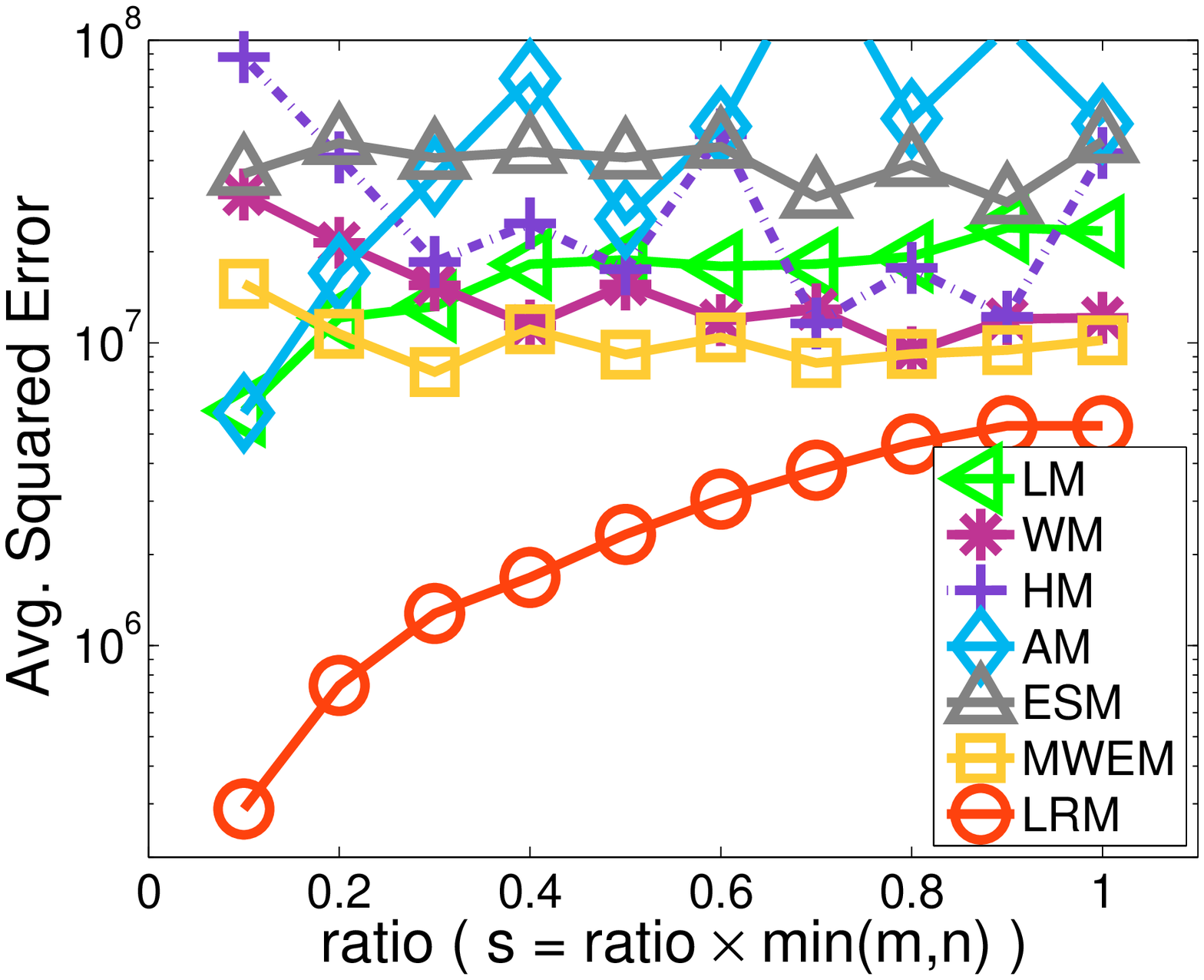}}
\centering \subfigure[\emph{Social Network}]
{\includegraphics[width=0.244\textwidth]{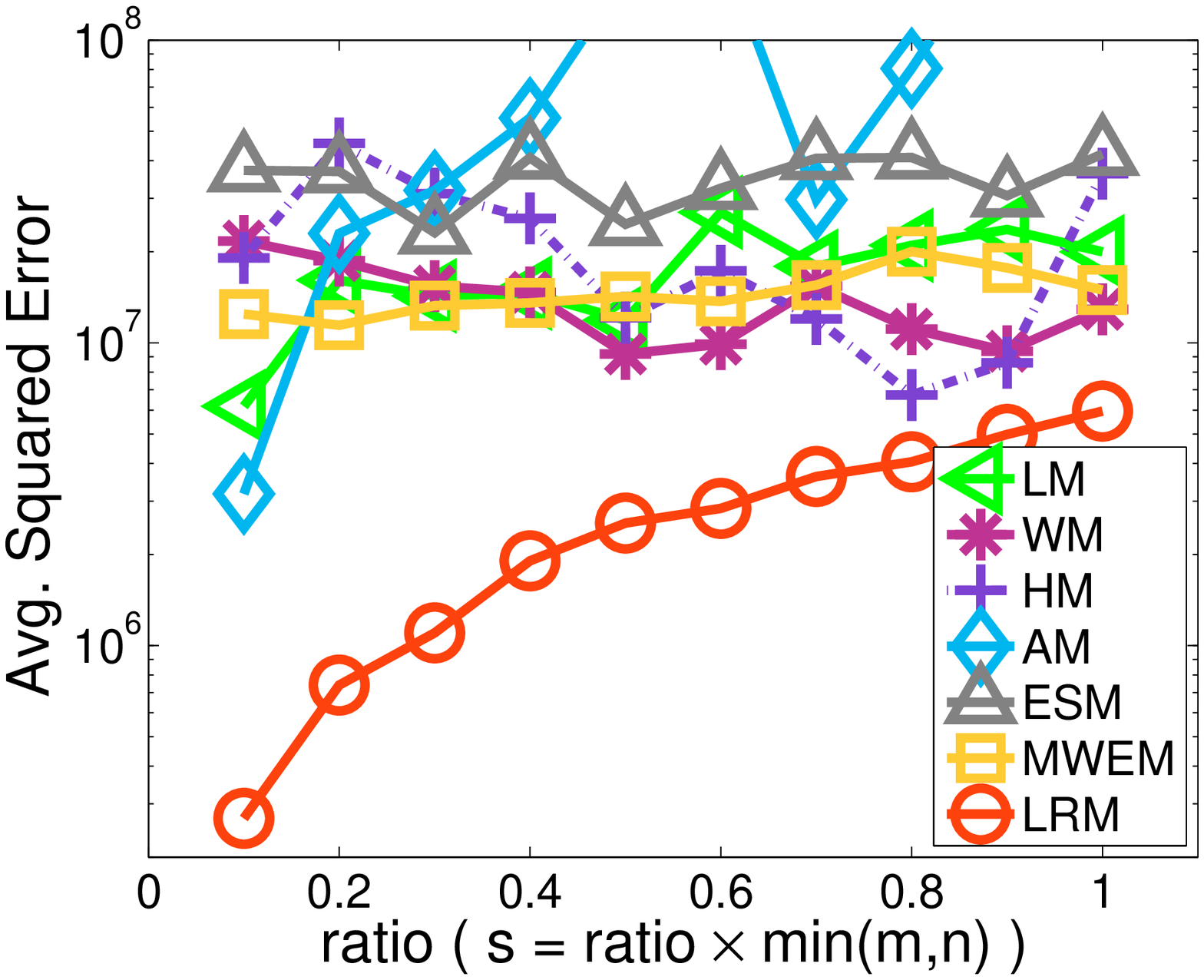}}
\centering \subfigure[\emph{UCI Adult}]
{\includegraphics[width=0.244\textwidth]{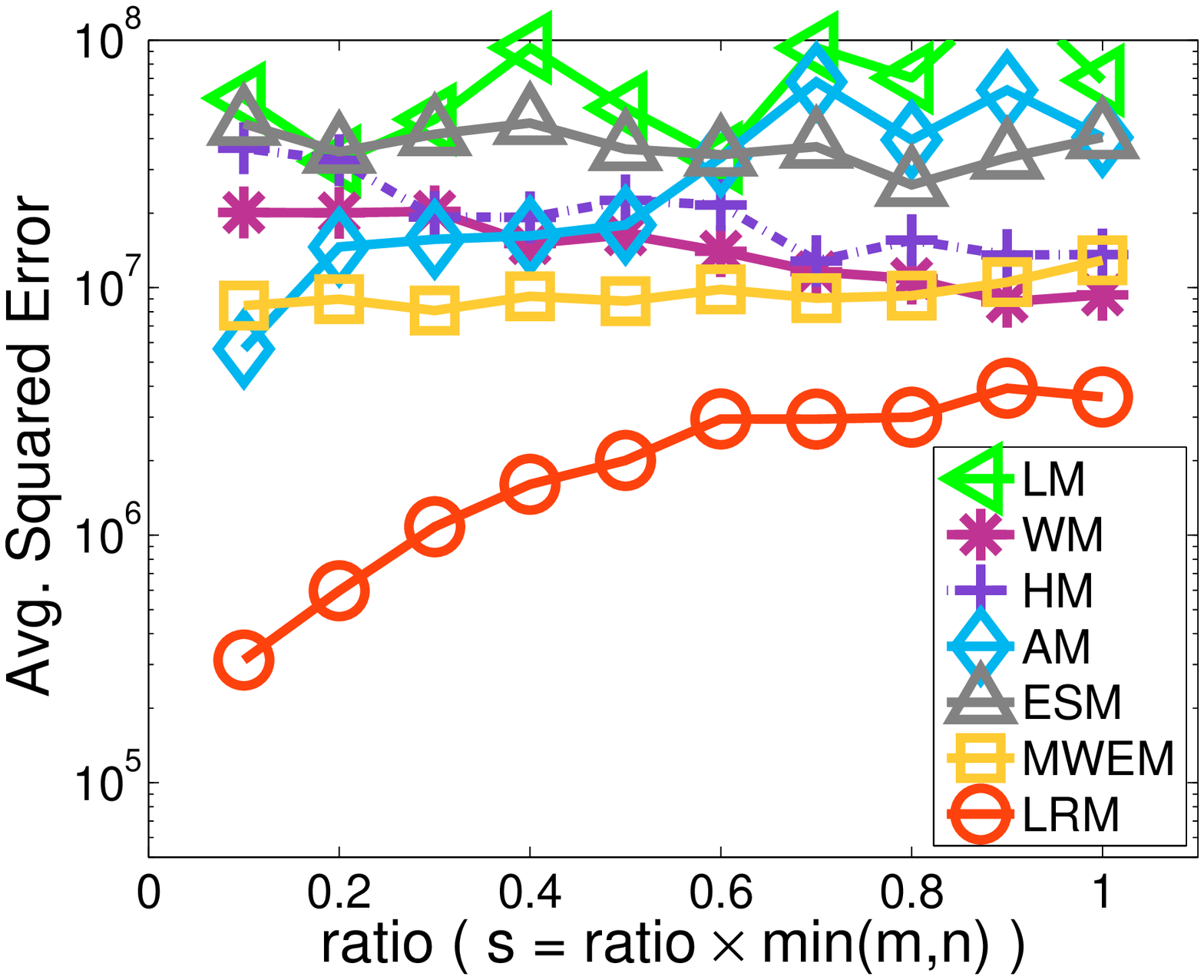}}
\caption{Effect of parameter $s$ with under $\epsilon$-differential privacy with $\epsilon=0.1$}
\label{fig:exp:s}
\end{figure*}

\begin{figure*}[!t]
\centering \subfigure[\emph{Search Logs}]
{\includegraphics[width=0.244\textwidth]{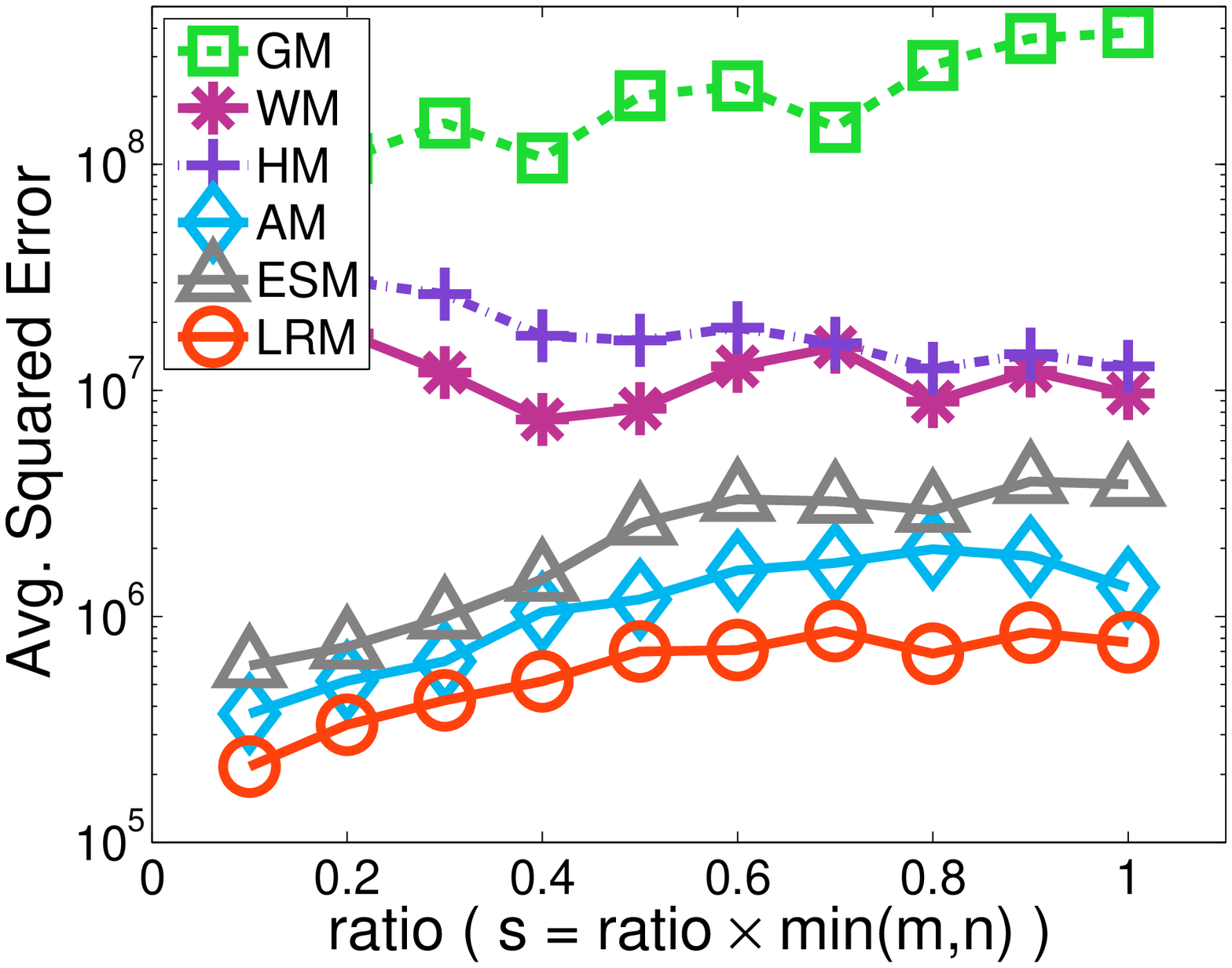}}
\subfigure[\emph{Net Trace}]
{\includegraphics[width=0.244\textwidth]{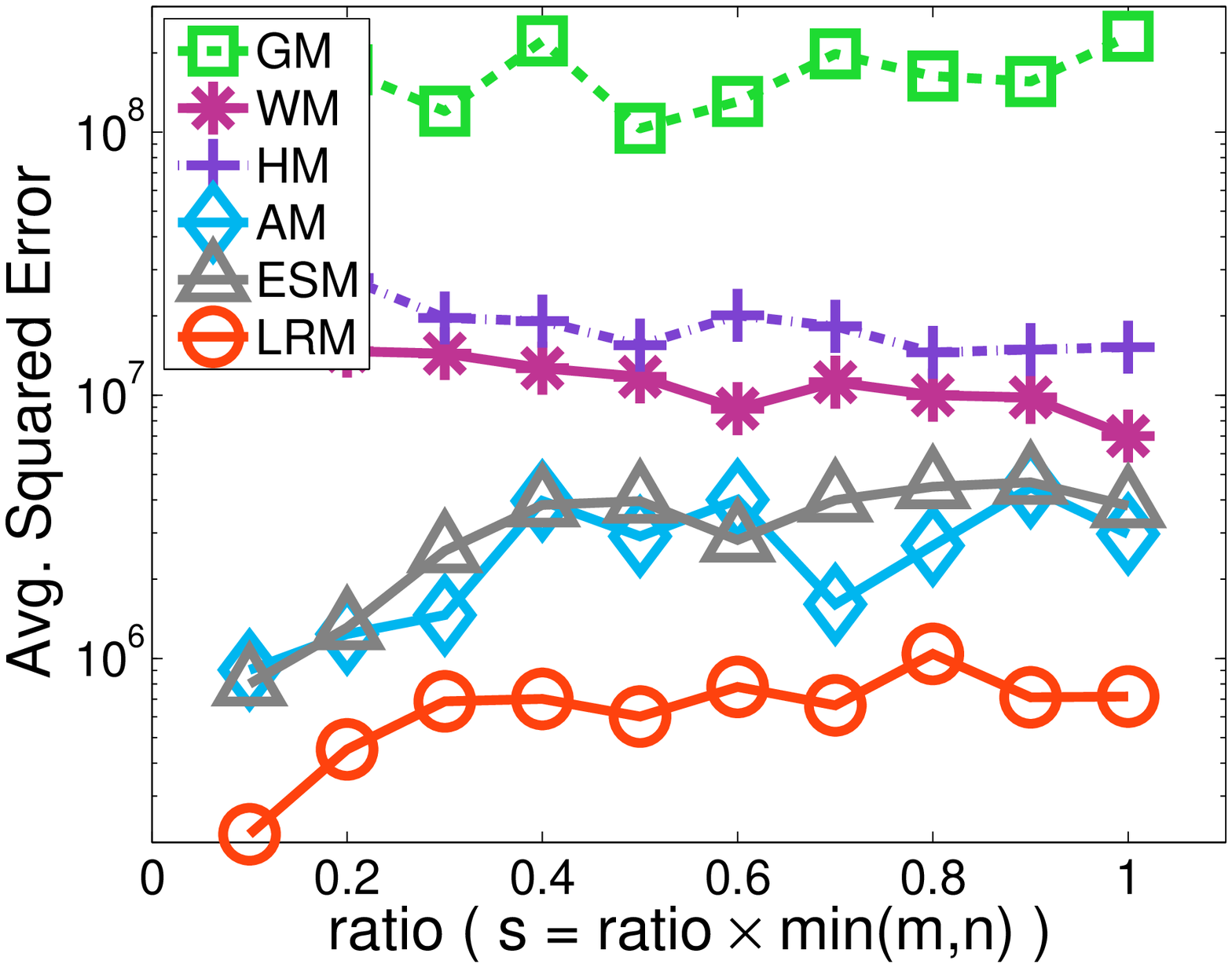}}
\centering \subfigure[\emph{Social Network}]
{\includegraphics[width=0.244\textwidth]{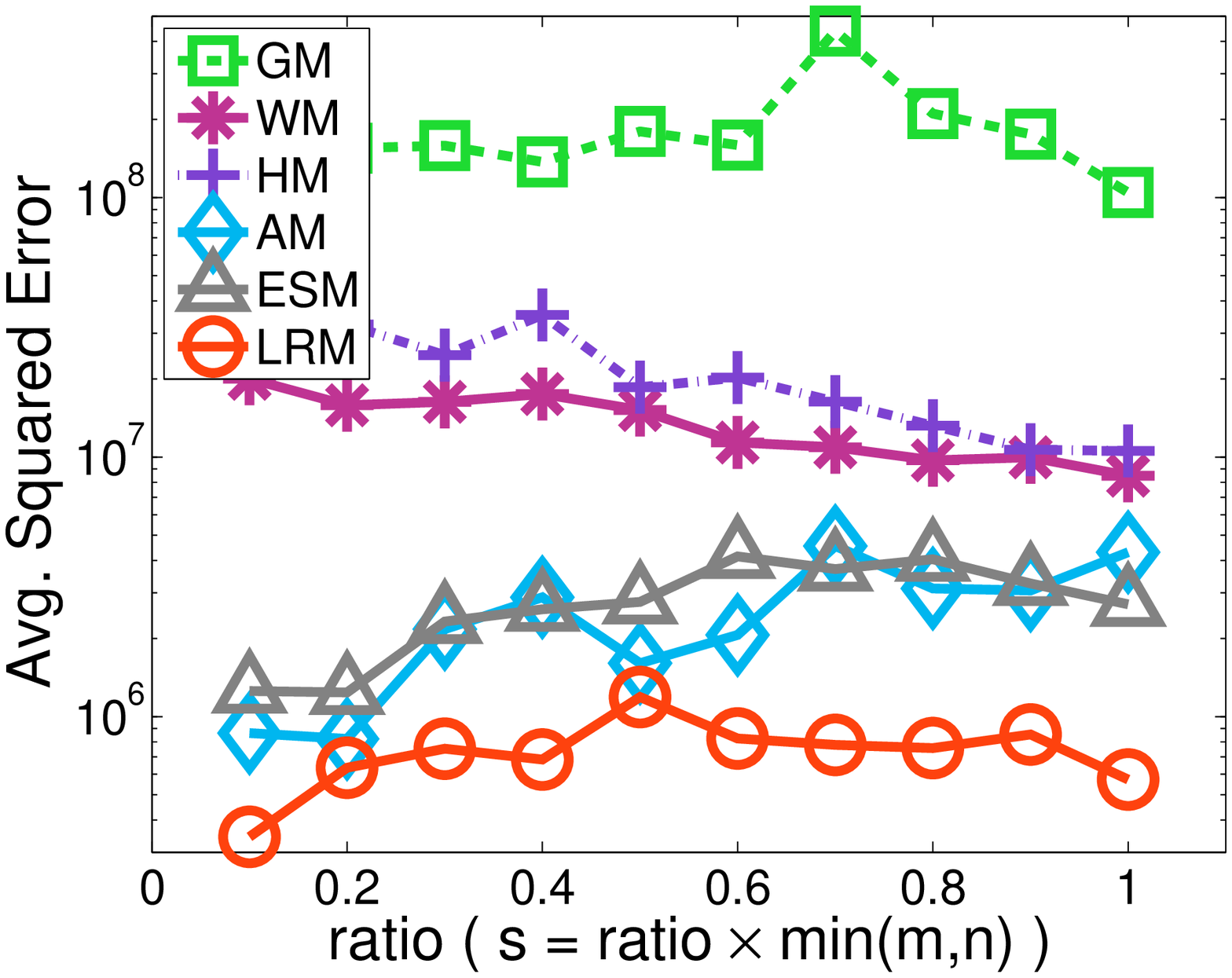}}
\centering \subfigure[\emph{UCI Adult}]
{\includegraphics[width=0.244\textwidth]{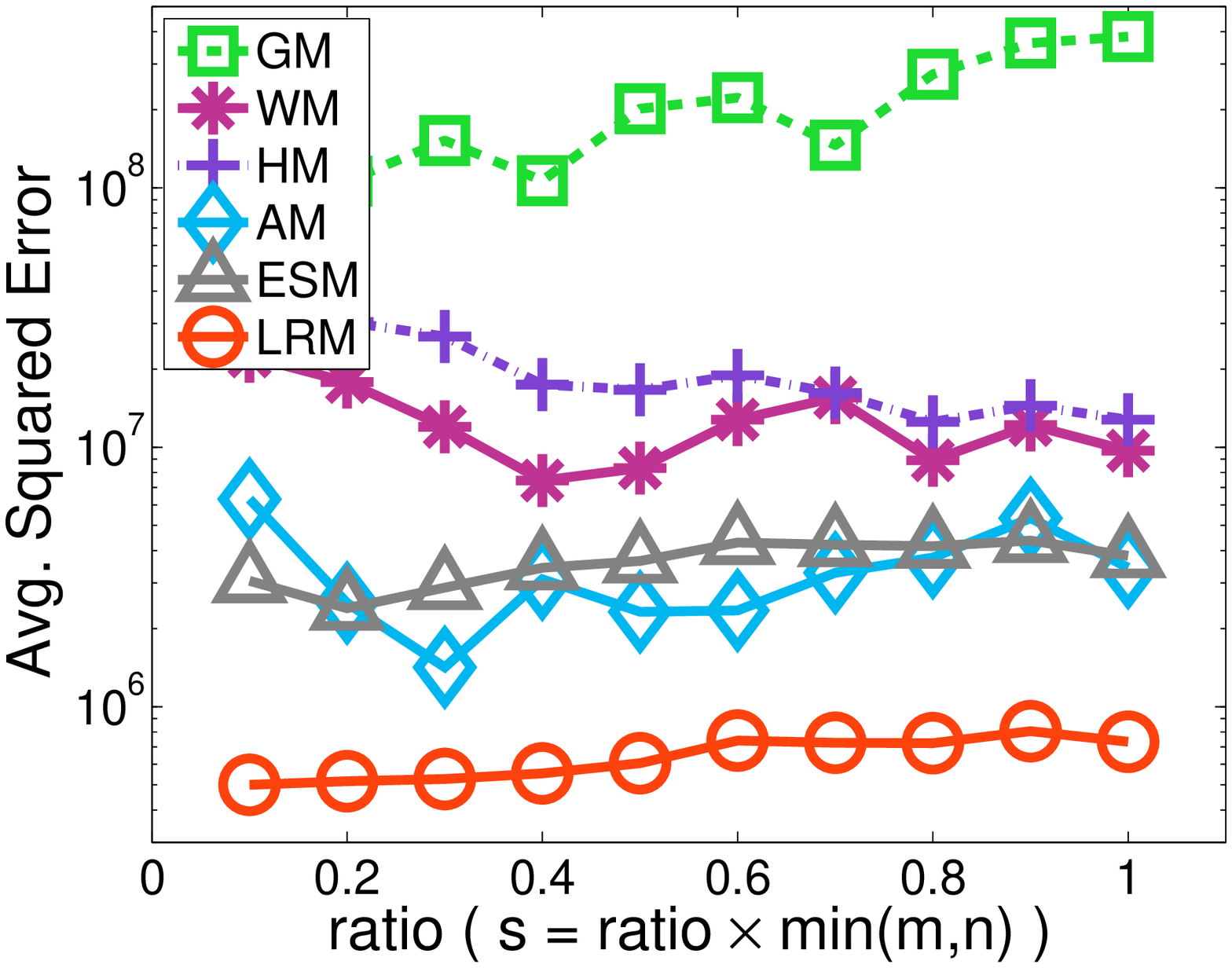}}
\caption{Effect of parameter $s$ under ($\epsilon$, $\delta$)-differential privacy with $\epsilon=0.1$ and $\delta=0.0001$}
\label{fig:exp:s:app}
\end{figure*}

\begin{figure*}[!t]
\centering \subfigure[on workload \emph{WDiscrete}]
{\includegraphics[width=0.244\textwidth]{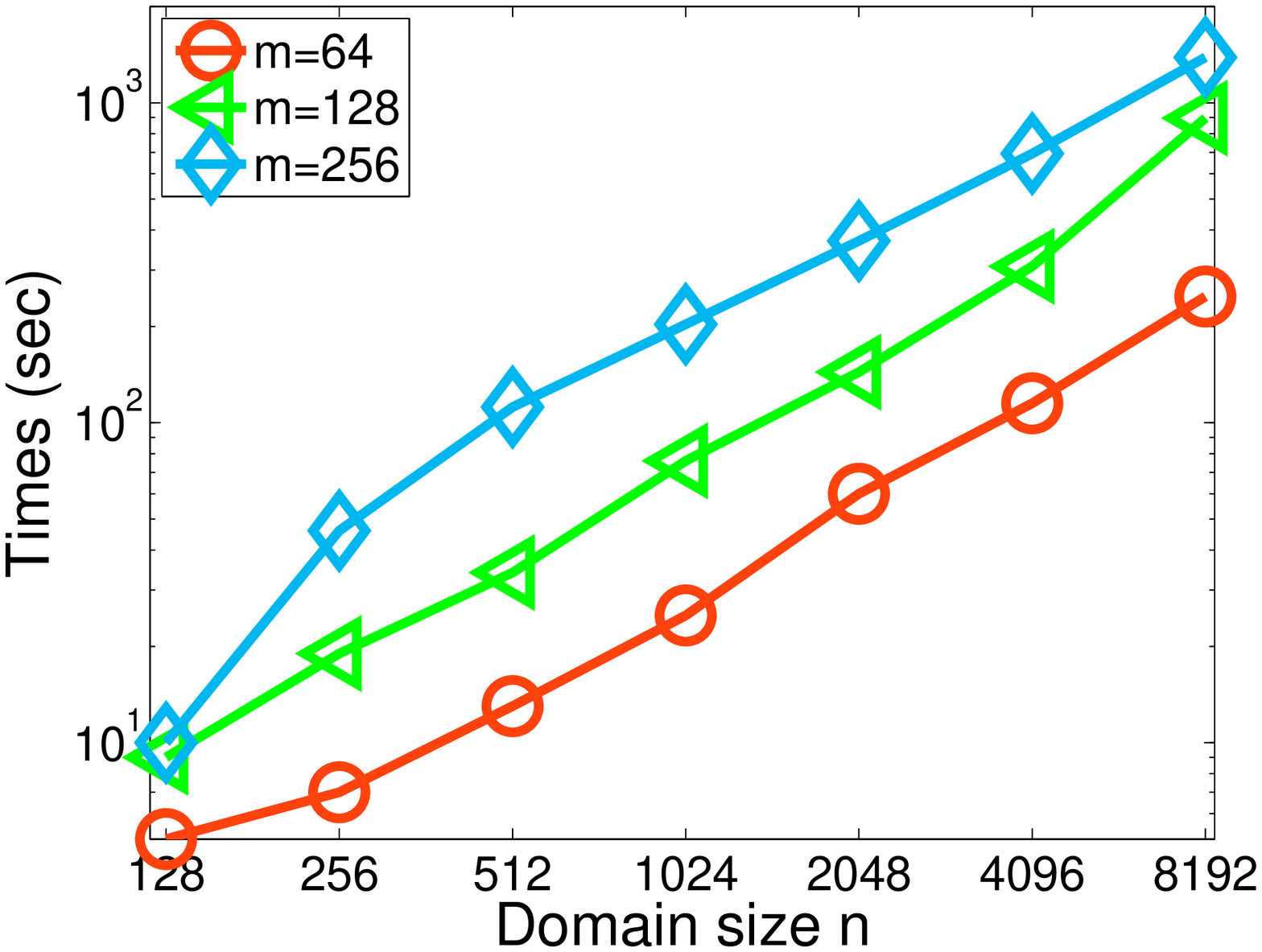}}
\centering \subfigure[on workload \emph{WRange}]
{\includegraphics[width=0.244\textwidth]{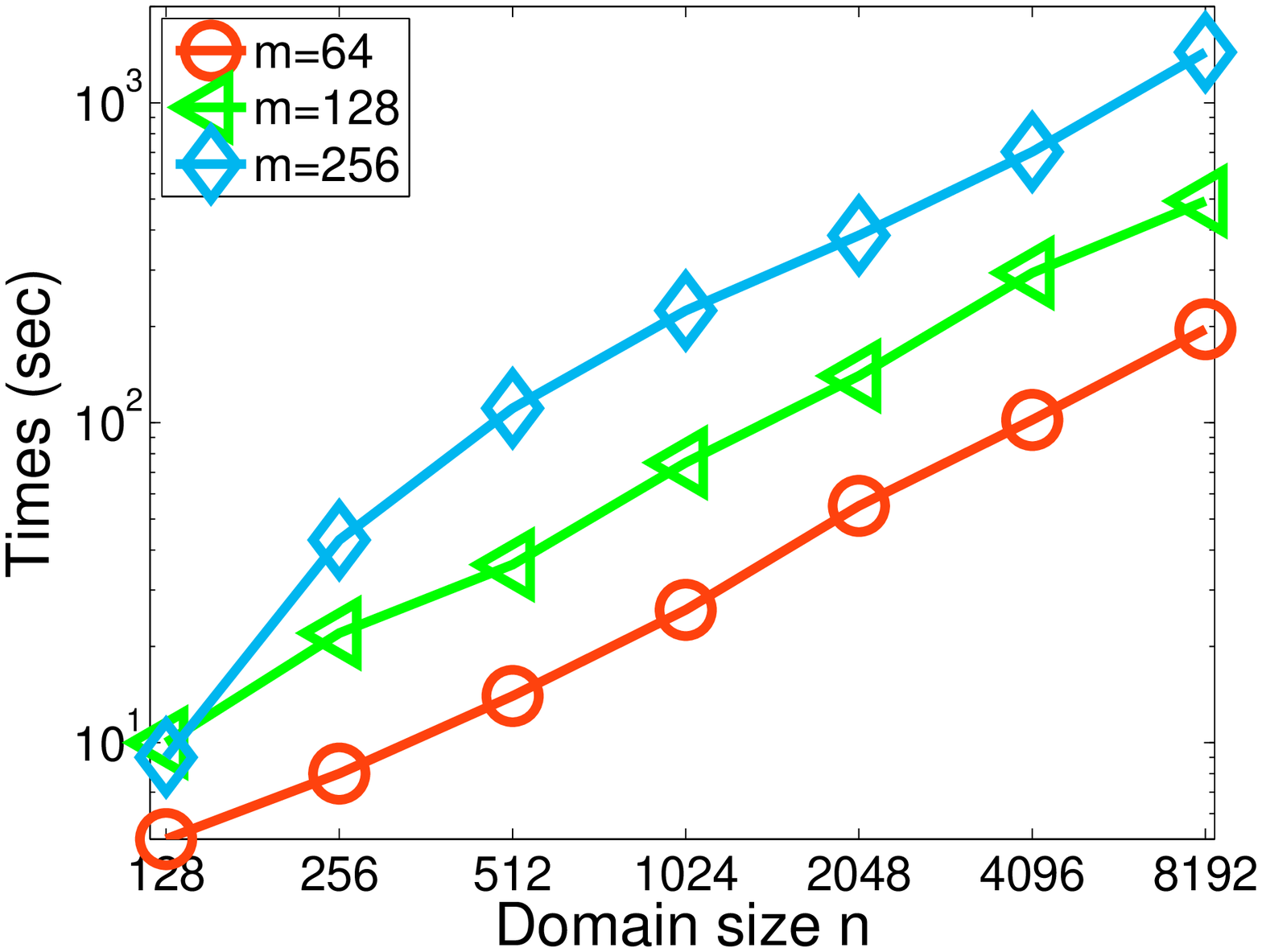}}
\centering \subfigure[on workload \emph{WMarginal}]
{\includegraphics[width=0.244\textwidth]{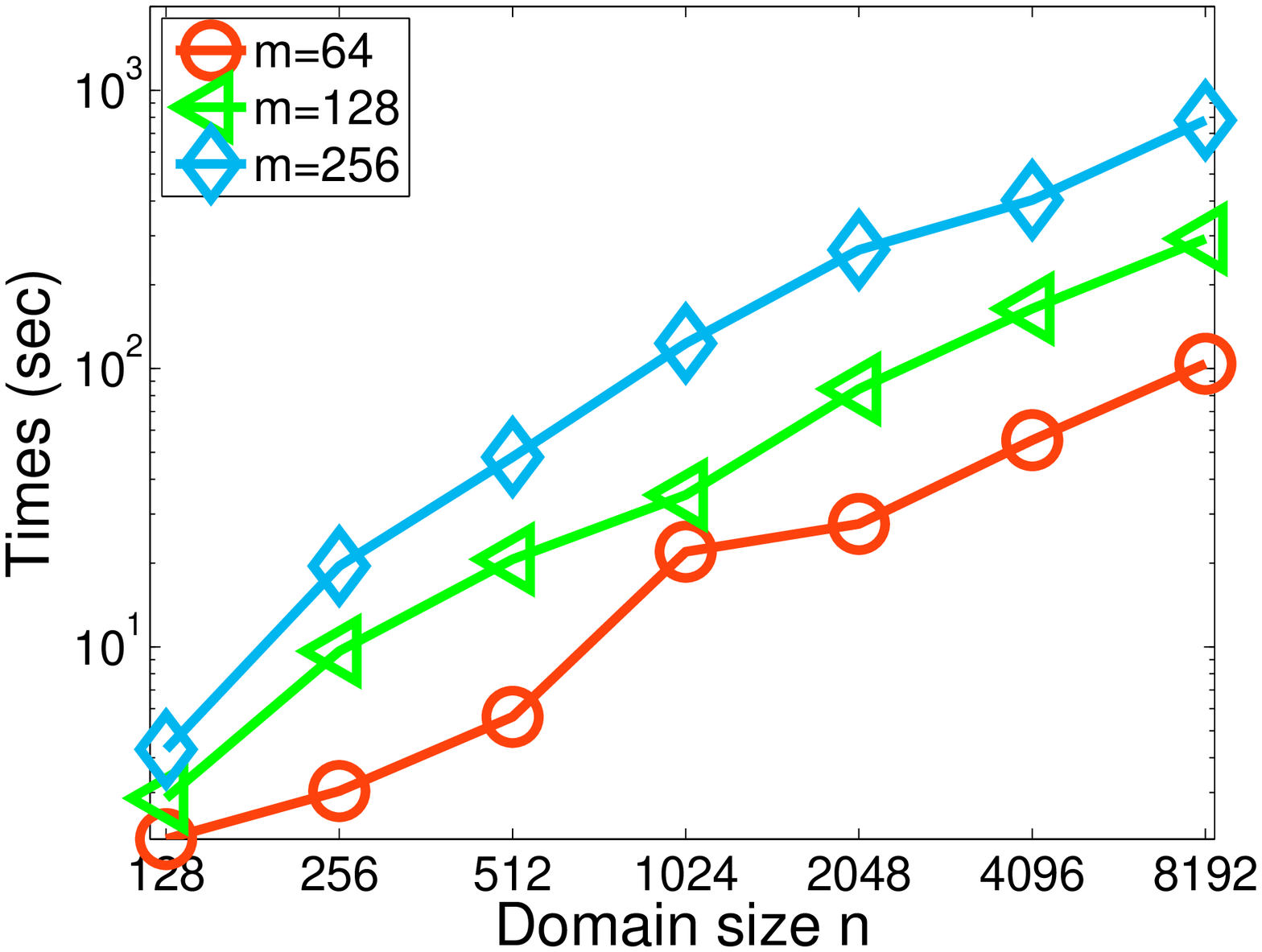}}
\centering \subfigure[on workload \emph{WRelated}]
{\includegraphics[width=0.244\textwidth]{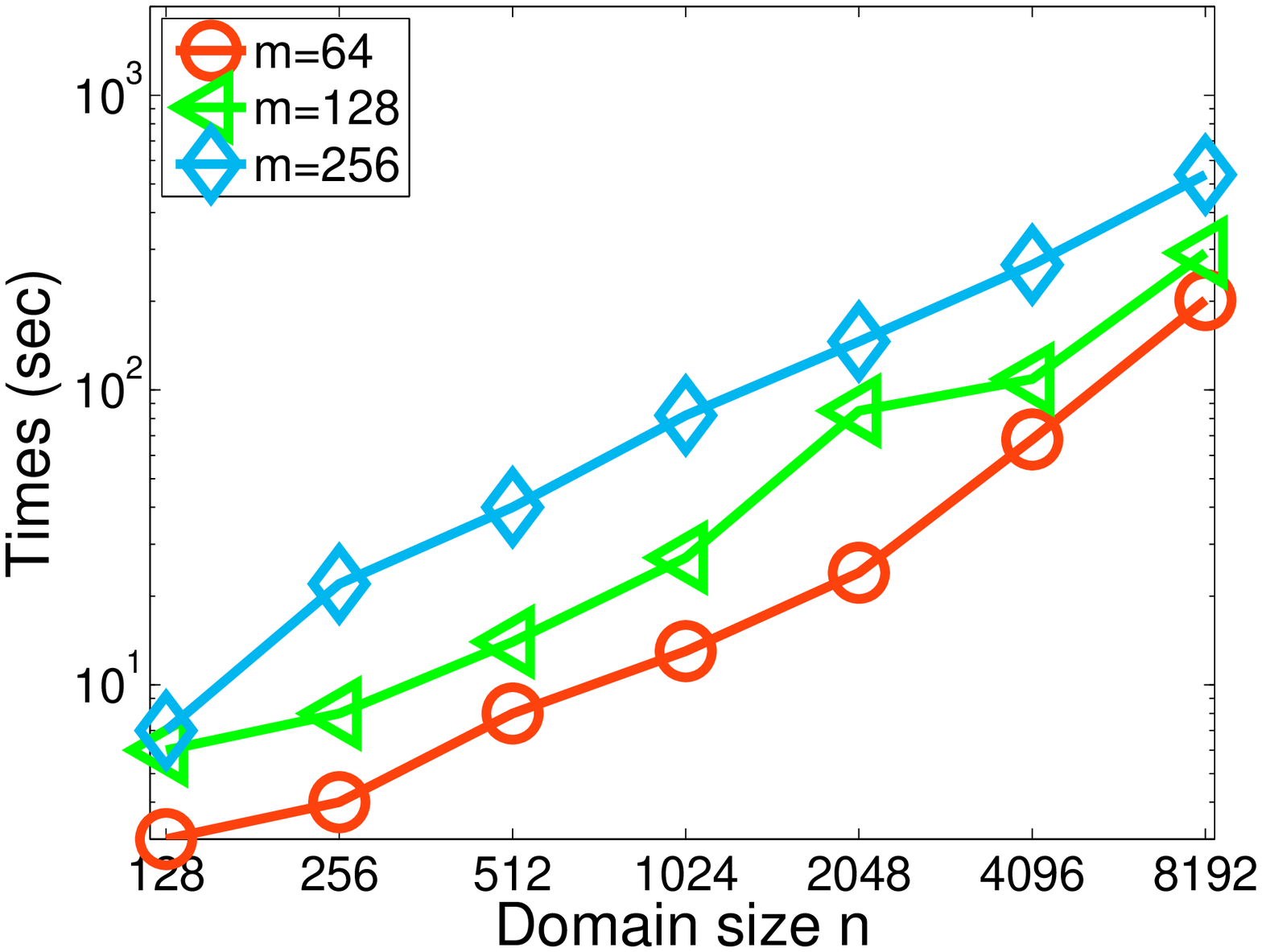}}
\caption{Scalability of LRM under $\epsilon$-differential privacy} \label{fig:exp:scal}
\end{figure*}

\begin{figure*}[!t]
\centering \subfigure[on workload \emph{WDiscrete}]
{\includegraphics[width=0.244\textwidth]{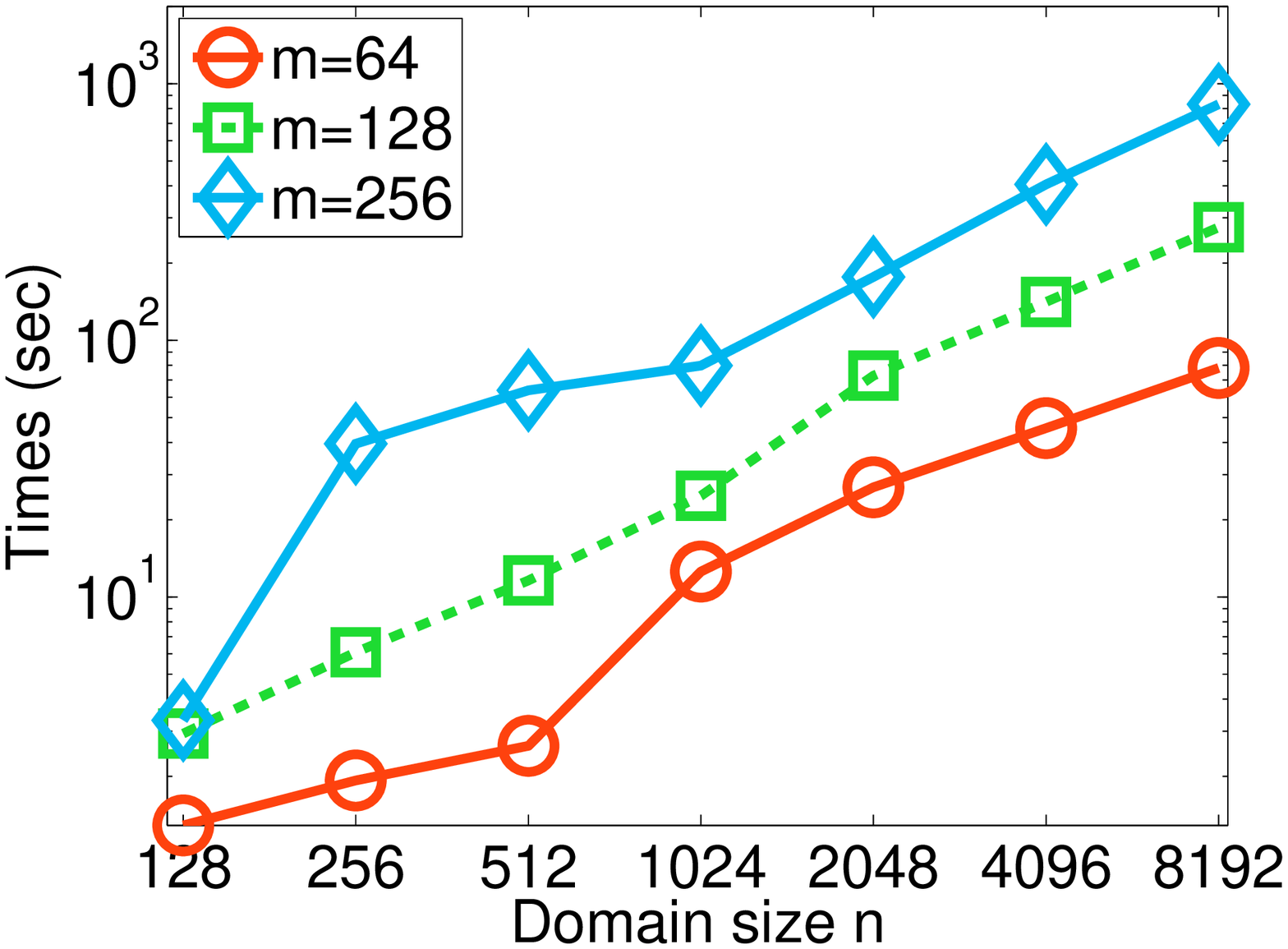}}
\centering \subfigure[on workload \emph{WRange}]
{\includegraphics[width=0.244\textwidth]{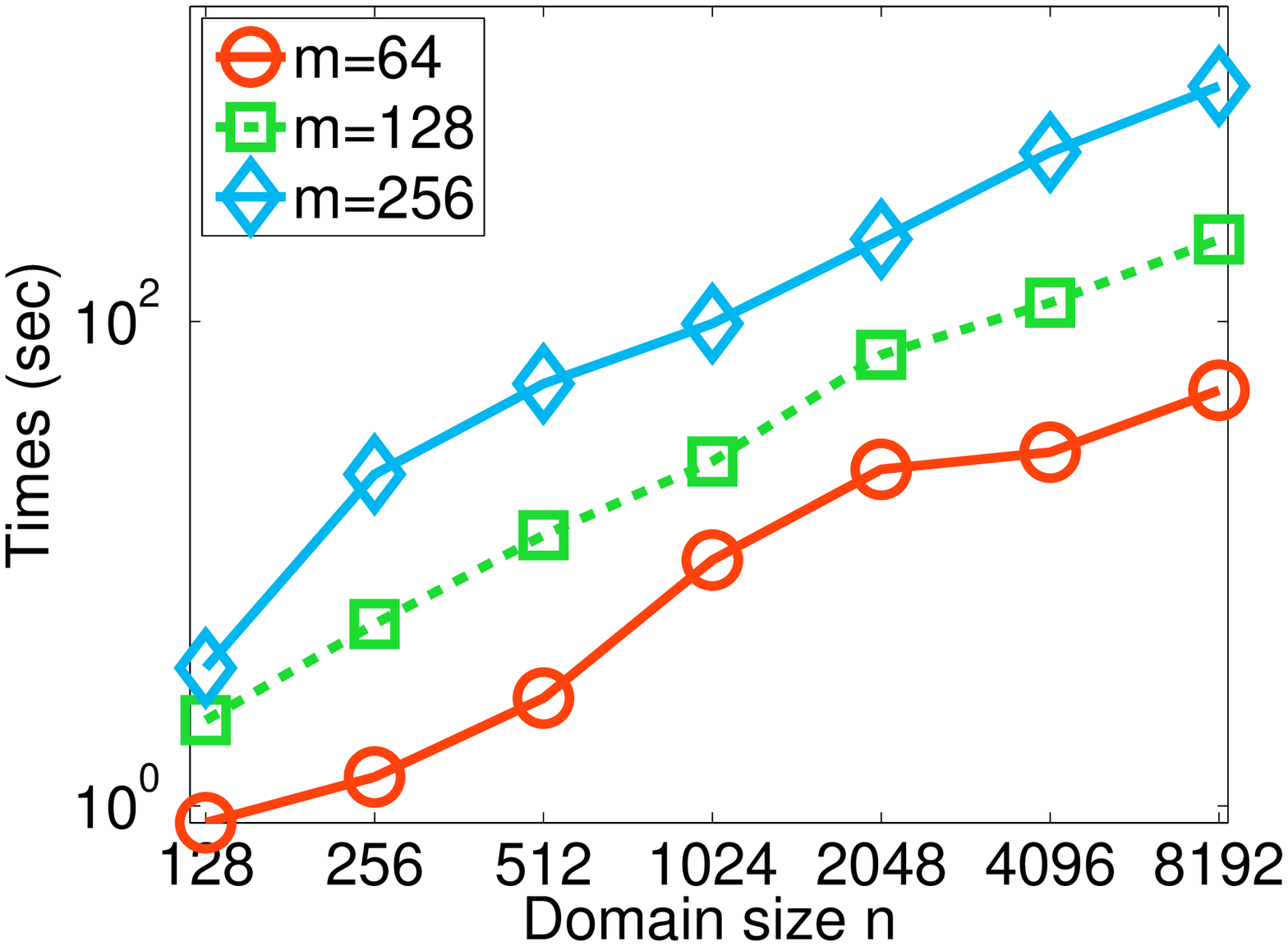}}
\centering \subfigure[on workload \emph{WMarginal}]
{\includegraphics[width=0.244\textwidth]{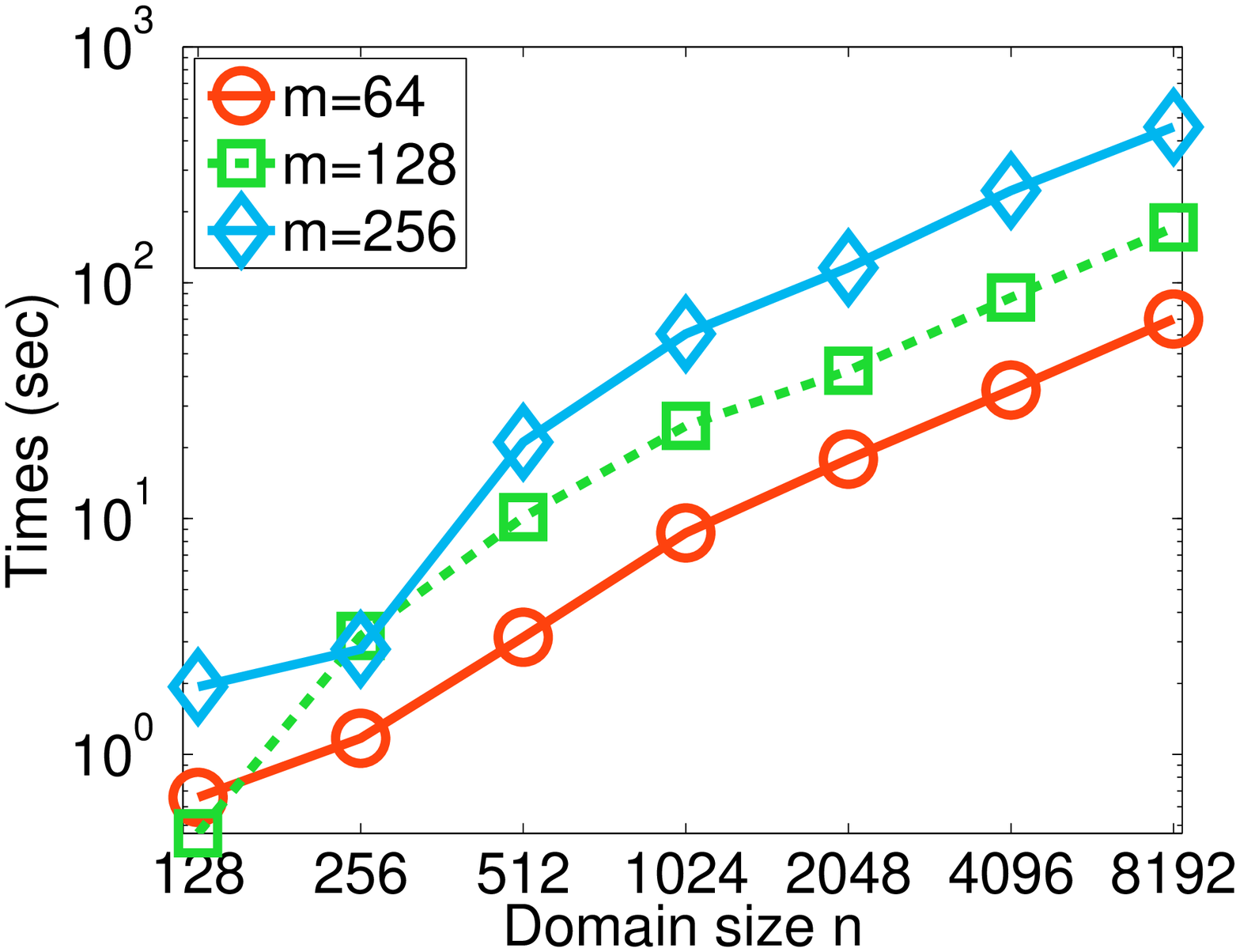}}
\centering \subfigure[on workload \emph{WRelated}]
{\includegraphics[width=0.244\textwidth]{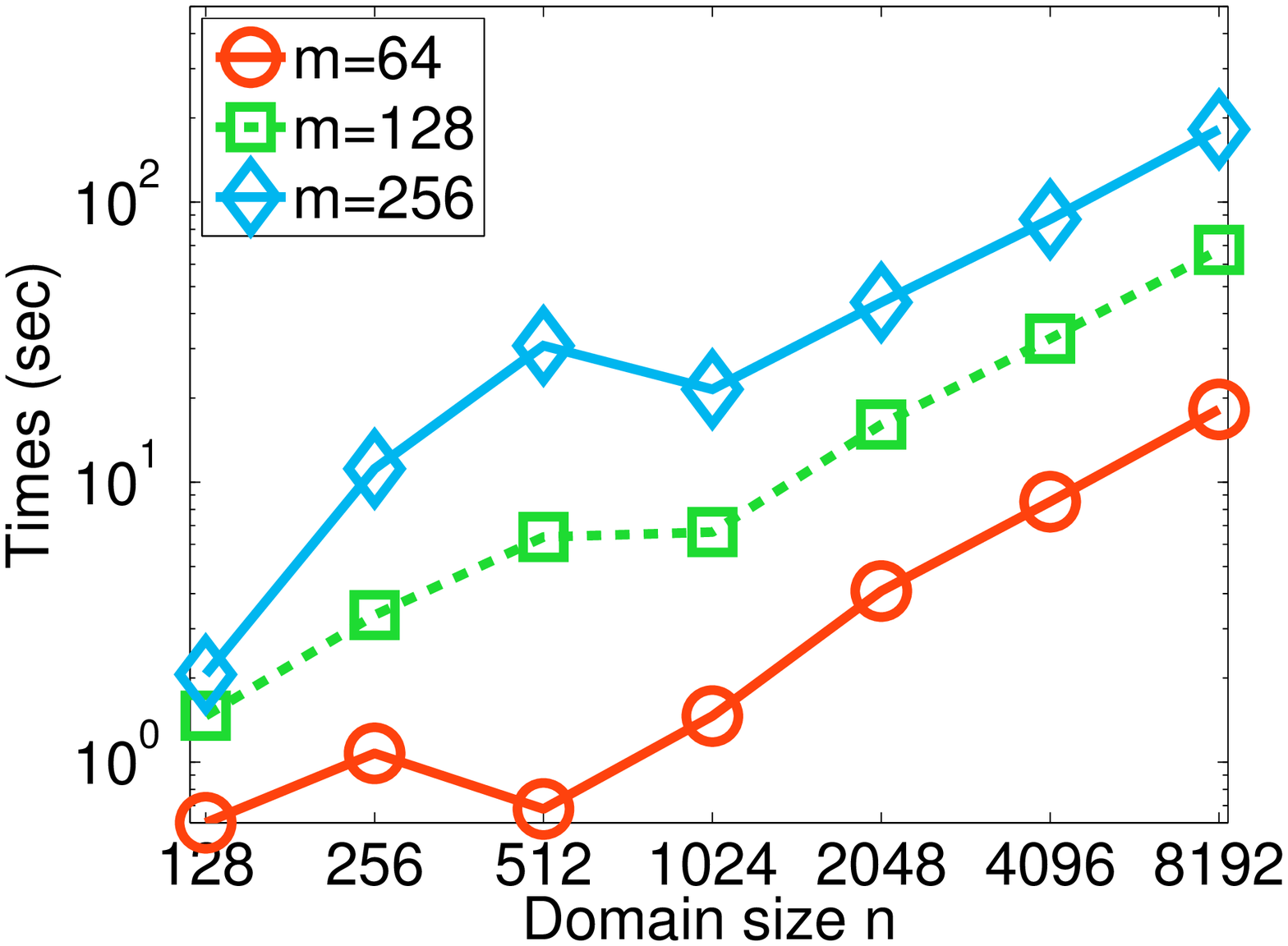}}
\caption{Scalability of LRM under
($\epsilon$, $\delta$)-differential privacy} \label{fig:exp:scal:app}
\end{figure*}

%% file: conc.tex
\section{Conclusions and Future Work}\label{sec:concl}

This paper presents the low rank mechanism (LRM), an optimization framework that minimizes the overall error of the results for a batch of linear queries under differential privacy. The proposed method is the first practical method for a large number of linear
queries, with an efficient and effective implementation using well
established optimization techniques. Experiments show that LRM
significantly outperforms other state-of-the-art differentially
private query processing mechanisms, often by orders of magnitude.
The current design of LRM focuses on exploiting the correlations
between different queries. One interesting direction for future work
is to further optimize LRM by utilizing also the correlations
between data values, e.g., as is done in \cite{XZXYY13,RN10,LZMY11}.

%% file: app.tex
\section{Implementation of the Approximate Matrix Mechanism}\label{sec:appedinex:matrix}

Li et al. \cite{LHR+10} describes two implementations of the Matrix Mechanism, which optimizes the accuracy of a batch of linear counting queries under $\epsilon$-differential privacy. The first directly solves the optimization program of the matrix mechanism which can be formulated as follows:
\begin{equation}\label{l1_app}
\min_{A\in \mathbb{R}^{r\times n}} \|A\|_{1,\infty}^2\mbox{tr}\left(WA^\dag A^{\dag T}W^T\right)
\end{equation}
where $A^\dag$ denotes the pseudo-inverse of matrix $A$, and $\|A\|_{1,\infty}$ is the maximum $\mathcal{L}_1$ norm of column vectors of $A$. It is shown that this problem can be formulated as a semidefinite program with rank constraint and solved by a sequence of semidefinite programs. However, it incurs high computational overhead, which is prohibitively expensive even for moderate-sized workload. The second implementation solves an approximate version of Program (\ref{l1_app}), as follows:
\begin{equation}\label{l2_app}
\min_{A\in \mathbb{R}^{r\times n}} \|A\|_{2,\infty}^2\mbox{tr}\left(WA^\dag
A^{\dag T}W^T\right)
\end{equation}
where $\|A\|_{2,\infty}$ is the maximum $\mathcal{L}_2$ norm of column vectors of $A$. Under $\epsilon$-differential privacy, Program (\ref{l2_app}) is essentially the $\mathcal{L}_2$ approximate of the original matrix mechanism formulation. The solution to Program (\ref{l2_app}) presented in \cite{LHR+10}, however, is rather complicated, and incurs high computational costs. In the following two subsections, we describe two implementation of the approximate matrix mechanism, the \emph{exponential smoothing mechanism} (\emph{ESM}) \cite{yuan2012low} and the adaptive mechanism (AM) \cite{li2012adaptive} for solving Program (\ref{l2_app}).

\subsection{Exponential Smoothing Mechanism}\label{sec:appedinex:esm}
In this subsection, we present a simpler and more efficient solution, referred to as the \emph{exponential smoothing mechanism} (\emph{ESM}), based on the methodology of \emph{exponential smoothing}. Observe that $\|A\|_{2,\infty}^2 = \max(\mbox{diag}(A^TA)) \footnote{We use the Matlab notations in this paper. When $\Delta$ is a matrix, $diag(\Delta)$ denotes a column vector formed from the main diagonal of $\Delta$, when $\Delta$ is a vector, $diag(\Delta)$
denotes a diagonal matrix with $\Delta$ in the main diagonal entries. Moreover, $\max(\cdot)$ retrieves the largest element of an array.}$, and $(A^TA)^{-1}=(A^TA)^{\dag}$ ($A$ has full column rank). Let
$M=A^TA$, we reformulate Program (\ref{l2_app}) as the following positive definite optimization problem:
\begin{equation}
\min_{M \in R^{n\times n}} G(M) = \max(\mbox{diag}(M))  \mbox{tr}(W
M^{-1}W^T) ~~~  s.t.~~ M \succ 0 \nonumber
\end{equation}

$A$ is given by $A=\sum_{i}^n \sqrt{\lambda_i} v_i v_i^T$, where
$\lambda_i, v_i$ are the $i$th eigenvalue and eigenvector of $M$,
respectively. Calculating the second term $\mbox{tr}(W M^{-1}W^T)$
is relatively straightforward. Since it is smooth, its gradient can
be computed as $- M^{-1}W^TWM^{-1}$. However, calculating the first
term $\max(\mbox{diag}(M))$ is harder since it is non-smooth.
Fortunately, inspired by \cite{dAspremont2007}, we can still use a
logarithmic and exponential function to approximate this term.

\textbf{Approximate the maximum positive number:} Since $M$ is
positive definite, $v=\mbox{diag}(M)>0$. we let $\mu>0$ be a sufficient small parameter and define:

\begin{equation}\label{MatMac:Fobj}
f_{\mu}(v)=\mu \log
\sum_i^n\left(\exp\left(\frac{v_i}{\mu}\right)\right)
\end{equation}
We then have ${\max}(v) \leq f_{\mu}\left(v\right) \leq  {\max}(v) +
\mu \log n$. The gradient of the objective function in Equation (\ref{MatMac:Fobj}) with respect
to $v$ can be computed as:

\begin{equation}\label{MatMac:Grad}
\frac{\partial f}{\partial v_i} =
\frac{\exp\left(\frac{v_i-{\max}(v)}{\mu}\right)}{\sum_j^n
\left(\exp\left(\frac{v_j-{\max}(v)}{\mu}\right) \right)} = \frac{\exp\left(\frac{v_i}{\mu}\right)}{\sum_j^n
\left(\exp\left(\frac{v_j}{\mu}\right) \right)},~~\forall i
\end{equation}
Since the second order hessian matrix of the objective function in Equation (\ref{MatMac:Fobj}) can be computed as:
\begin{equation}\label{MatMac:Hess}
\frac{\partial^2f }{\partial v\partial v}  = \frac{ diag(\exp(\frac{v}{\mu}))}{\mu \sum_j^n\left(\exp(\frac{v_j}{\mu})\right) } - \frac{ \exp(\frac{v}{\mu}) \exp(\frac{v}{\mu})^T}{\mu \left(\sum_j^n\left(\exp(\frac{v_j}{\mu})\right) \right)^2} = \mathbb{S} - \mathbb{T},\nonumber
\end{equation}
we have the upper bound of the spectral norm of the hessian: $|||\frac{\partial ^2 f}{\partial v\partial v} |||_2 =|||\mathbb{S} - \mathbb{T}|||_2 \leq |||\mathbb{S}|||_2 + |||\mathbb{T}|||_2 \leq \frac{1}{\mu} + \frac{1}{\mu} = \frac{2}{\mu}$. Therefore, the gradient of $f_{\mu}(v)$ is Lipschitz continuous with parameter $\omega = \frac{2}{\mu}$. If we set $\mu=\frac{\epsilon}{\log n}$, this becomes a uniform $\epsilon$-approximation of ${\max}(v)$ with a Lipschitz continuous gradient with constant $\omega=\frac{2}{\mu} = \frac{2\log
n}{\epsilon}$. In our experiments, we use $\mu=\frac{0.01}{\log n}$.

To mitigate the problems with large numbers, using the property of
the logarithmic and exponential functions, we can rewrite Equation (\ref{MatMac:Fobj}) and Equation (\ref{MatMac:Grad}) as:
$$f_{\mu}(v) = {\max(v)} + \mu \log \left( \sum_i^n
\exp\left(\frac{v_i-\max(v)}{\mu}\right)\right)$$
$$\frac{\partial f}{\partial v_i} = {\left(\sum_j^n
\exp\left(\frac{v_j-v_i}{\mu}\right)\right)}^{-1},~~\forall i$$

By the chain rule of differentiation in calculus, the gradient of $G(M)$ can be computed as:
$$\frac{\partial G}{\partial M} = diag(\frac{\partial f}{\partial v}) \cdot tr\left(WM^{-1}W^T\right) + f_{\mu}(v)\cdot \left(-M^{-1}W^TWM^{-1}\right)$$

Here $diag(\frac{\partial f}{\partial v})$ denotes a diagonal matrix with $\frac{\partial f}{\partial v}\in \mathbb{R}^n$ as the main diagonal entries. This formulation allows us to run the non-monotone spectral projected  gradient descent algorithm \cite{Birgin2000} on the cone of positive semidefiniteness. We use eigenvalue decomposition to trim the negative eigenvalues to maintain positive semidefiniteness of $M$, and iteratively improve the result. After the algorithm terminates, we return the final $M$ as the optimal solution to the program.

\subsection{Adaptive Mechanism}\label{sec:appedinex:am}

In this subsection, we briefly review the adaptive mechanism (AM) proposed in \cite{li2012adaptive}, a heuristic solution for the problem in Program (\ref{l2_app}). AM considers the following optimization problem:

\begin{equation} \label{adm}
\min_{\lambda \in \mathbb{R}^{n}} \sum_{i=1}^n \frac{d_i^2}{\lambda_i^2},\\
s.t.~(Q\odot Q) (\lambda\odot \lambda) \leq \textbf{1$_m$}
\end{equation}
where $Q$ is from the singular value decomposition of the workload matrix $W=QDP$ with $Q\in \mathbb{R}^{m\times n}, D\in \mathbb{R}^{n\times n}, P\in\mathbb{R}^{n\times n}$, and $d=diag(D)\in \mathbb{R}^{n}$, i.e., the diagonal values of $D$. Furthermore, $\odot$ is the Hadamard (entry-wise) product, \textbf{1$_m$} is a column vector of all entries equal to one. AM then computes the strategy matrix $A$ by
\begin{equation} \label{computeA}
A = Q diag(\lambda) \in \mathbb{R}^{m\times n}
\end{equation}
where $diag(\lambda)$ is a diagonal matrix with $\lambda$ as its diagonal values.

The optimization problem in (\ref{adm}) is non-convex since it contains quadratic term both in the objective and the constraint. By changing variable to $\lambda\odot \lambda = u$, we have the following equivalent optimization problem:
\begin{equation} \label{adm}
\min_{u \in \mathbb{R}^{n}} \sum_{i=1}^n \frac{d_i^2}{u_i},\\
s.t.~(Q\odot Q) u \leq \textbf{1$_m$},~u\geq0.
\end{equation}

\noi By introducing an auxiliary variable $v\in \mathbb{R}^{n}$, the optimization above can be reformulated as the following semidefinite program:
   \begin{eqnarray} \label{eq:sdp}
\min_{u\in \mathbb{R}^{n},u\in \mathbb{R}^{n}}~\sum_{i=1}^n v_i d_i^2,~s.t.~\left(Q\odot Q\right)u \leq \textbf{1$_m$},~
\begin{bmatrix}
u_i & 1 \\
1 & v_i
\end{bmatrix} \succeq 0,~\forall i \in [n]
\end{eqnarray}
which can be solved by off-the-shelf interior-point solvers.

 \begin{algorithm}[!h]
\caption{\label{alg:eigendesign} {\bf Adaptive Mechanism for Approximately Solving Problem (\ref{l2_app})}}
\begin{algorithmic}[1]
\STATE  Input: workload matrix $W\in \mathbb{R}^{m\times n}$
\STATE  Compute the SVD decomposition $W=QDP$ to obtain $Q\in \mathbb{R}^{m\times n}$ and $d=diag(D)$$\in \mathbb{R}^{n}$.
\STATE Solve the semidefinite program in Equation (\ref{eq:sdp}) and obtain $u$.
\STATE Compute $A'=Q diag(\sqrt{u}) \in \mathbb{R}^{m\times n}$ and $A''=diag(\sqrt{ \max(o) 1_n- o}) \in \mathbb{R}^{n\times n}$
   where $o_i=\|A_i'\|_2^2,~i=1,...n$, $o\in \mathbb{R}^{n}$.
\STATE \label{step:s4} Output the strategy matrix $A$: \begin{eqnarray}
A=\begin{bmatrix}
A' \\
A''
\end{bmatrix} \in \mathbb{R}^{(m+n)\times n}\nonumber
\end{eqnarray}
\end{algorithmic}
\end{algorithm}

The complete AM algorithm is summarized in Algorithm \ref{alg:eigendesign}. Given a workload matrix $W$, AM automatically selects a different set of ``eigen-queries'' $Q$ and use a nonnegative combination of $Q$ to compute the strategy matrix $A$ with respect to the workload matrix. First, in Step 2 the algorithm performs the SVD decomposition of $W$ to derive the eigen-queries $Q$. Based on the eigen-queries $Q$, AM aims to find the optimal linear combination $\lambda (\lambda\geq 0)$ with $\lambda=\sqrt{u}$ by solving the semidefinite program in Step 3. In Step 4, the matrix $A'$ that is constructed is a candidate strategy but may have one or more columns whose norm is less than the sensitivity. In this case, AM adds queries or completes columns in order to further reduce the expected error without raising the sensitivity. Essentially AM searches over a reduced subspace of $A$. Hence, the candidate strategy matrix $A'$ solved from the optimization problem in (\ref{l2_app}) does not guarantee to be the optimal  strategy since it is limited to a weighted nonnegative combination of the fixed eigen-queries $Q$ in Equation (\ref{computeA}).

%
%
%
%
%

%% file: app2.tex
\section{Asymptotic Error Bounds for LRM}\label{sec:appedinex:asymptotic:error:bound}

\subsection{LRM Error Bounds under $\epsilon$-Differential Privacy}\label{sec:appedinex:eps}

In this subsection, we prove the lower bound and upper bound of the error incurred by the optimal workload decomposition solved from Program (\ref{eqn:opt-problem}), and analyze the gap between the two bounds. First, we establish an error upper bound for LRM in the following lemma.

\begin{lemma}\label{lem:upper} \textbf{Error upper bound under $\epsilon$-differential privacy.} Given a workload matrix $W$ of rank $s$ with singular values $\{\lambda_1,\ldots,\lambda_s\}$, an upper bound of the expected squared error of $M_{LRM, \epsilon}(Q,D)$ w.r.t. the optimal decomposition $W=B^*L^*$ is $2\sum^s_{k=1}\lambda^2_k/\epsilon^2$.
\end{lemma}

\begin{proof}
Consider the naive method NOD, which can be considered as a special case of LRM by setting $B=W$ and $L=I$ (i.e., identity matrix). Clearly, $\Delta(L)=1$. According to Lemma \ref{lem:decomp_error}, the expected squared error of this decomposition is:

$$2 \Phi(B) \Delta(L)^2 /\epsilon^2 = 2\|W\|^2_F/\epsilon^{2} = 2\sum^s_{k=1}\lambda^2_k/\epsilon^2$$

We reach the conclusion of the lemma.
\end{proof}

Next we derive a lower bound on the squared error for linear counting queries under $\epsilon$-differential privacy, using geometric analysis under orthogonal projection \cite{HT10}. To do this, we first present the following lemma, which is used later in our geometric analysis.

\begin{lemma} \label{lemma:orthvol:bound}
For all orthogonal $V \in \mathbb{R}^{s\times n}$, we have the following inequality:
$$ \Vol(VB^n_1) \geq  \Vol(B^s_2)\cdot n^{-\frac{s}{2}}  $$
where $\Vol(B^s_2)$ denotes the volume of unit Euclidean ball, and $\Vol(VB^n_1)$ denotes the volume of unit ball of the $\mathcal{L}_1$ norm on $\mathbb{R}^n$ after the orthogonal transformation under $V$.
\end{lemma}

\begin{proof}
By Cauchy-Schwarz inequality we have $\|x\|_1 \leq \sqrt{n} \|x\|$ for all $x\in \mathbb{R}^n$, therefore, the $n$-dimensional $\ell_1$ ball contains an $\ell_2$ ball of radius $n^{-\frac{1}{2}}$, i.e. $B_1^n \supseteq n^{-\frac{1}{2}} B_2^n$. Given an orthogonal transformation $V$, we obtain $VB_1^n \supseteq n^{-\frac{1}{2}} VB_2^n$. Moreover, because the orthogonal projection of a Euclidean ball is a lower-dimensional Euclidean ball of the same radius, it holds that $n^{-\frac{1}{2}} VB_2^n = n^{-\frac{1}{2}} B_2^s$. Therefore, the volume of $VB^n_1$ is bounded from below by:
\begin{eqnarray} \label{xxxxx}
\Vol(VB^n_1) &\geq& \Vol(n^{-\frac{1}{2}} B_2^s) \nonumber \\
&=&  \Vol(B_2^s) \cdot n^{-\frac{s}{2}} .\nonumber
\end{eqnarray}
\end{proof}

We are now ready to prove the error lower bound of LRM.

\begin{lemma}\label{lem:lower} \textbf{Error Lower Bound under $\epsilon$-differential privacy.} Given a workload matrix $W$ of rank $s$ with \emph{singular values} $\{\lambda_1,\ldots,\lambda_s\}$, the expected squared error of any $\epsilon$-differential privacy mechanism is at least
$$\Omega\left(\frac{s^4}{n}\left(\frac{2^s}{s!}\prod^s_{k=1}\lambda_k\right)^{2/s}/\epsilon^2\right)$$
\end{lemma}
\begin{proof}
Corollary 3.4 in \cite{HT10} proves that any $\epsilon$-differential privacy mechanism for linear counting queries incurs expected squared
error no less than: \footnote{\cite{HT10} used absolute errors, from which which we derived the squared errors.}
$$\Omega\left(k^3\left(\Vol(PWB^n_1)\right)^{2/k}/\epsilon^2\right)$$
In the formula above, $B^n_1$ is the $\mathcal{L}_1$-unit ball. $\Vol(PWB^n_1)$ is the volume of the unit ball after the linear
transformation $PW$, in which $P$ is any orthogonal linear
transformation matrix from $\mathbb{R}^n\mapsto\mathbb{R}^s$. To prove the lemma, we construct an orthogonal transformation $P=U^T$, where $U$ is obtained form the SVD decomposition of $W$ ($W=U\Sigma V$). According to properties of SVD decomposition, $U^TU$ and $VV^T$ are identity matrices. Thus, we have
$\Vol(PWB^n_1)=\Vol(PUVV^T\Sigma VB^n_1)=\Vol(V(V^T\Sigma
V)B^n_1)=\Vol(VB^n_1)\prod^s_{k=1}\lambda_k$. The last equality holds
due to Lemma 7.5 in \cite{HT10}. Consider the the convex body
$VB^n_1$. By Lemma \ref{lemma:orthvol:bound}, it has a lower bound $\Vol(B^s_2)\cdot  \left(n^{-\frac{s}{2}}\right) $. Note that $\Vol(B^s_2)$ can be computed
using the Gamma function \cite{ball1997elementary}: $\frac{\pi^{s/2}}{\Gamma(1+s/2)}$. Using the Stirling's formula, we know that $\Gamma(1+s/2)$ is roughly $\sqrt{2\pi} e^{-s/2} (s/2)^{s/2 + 1/2}$, so that $\Vol(B^s_2)$ is roughly $\left(\frac{2\pi e}{s}\right)^\frac{s}{2}$. Therefore, the
lower bound can be computed as: $\Omega\left( \frac{s^4}{n}(\frac{2^s}{s!}\prod^s_{k=1}\lambda_k)^{2/s}/\epsilon^2\right)$. We thus reach the conclusion of the lemma.
\end{proof}

Next we compare the error upper and lower bounds. The analysis involves a matrix-theory concept called the \emph{generalized condition number}.

\begin{definition}
\textbf{Generalized condition number.} Given a workload matrix $W$, the generalized condition number $\kappa(W)$ of $W$ defined as the product of the spectral norm of $W$ and that of its pseudo-inverse or equivalently, the ratio between the largest singular value of $W$ to the \emph{nonzero} smallest \cite{chen2005condition,beltran2011estimates}.
$$\kappa(W) \triangleq |||W|||_2 \cdot |||W^\dag|||_2 =\frac{\lambda_1}{\lambda_s}$$
Note that we always have $\kappa(W)\geq1$.
\end{definition}

\begin{theorem}\label{the:opt} When $s>5$, the gap between the upper and lower bounds of the error incurred by mechanism $M_{LRM, \epsilon}(Q,D)$ with the optimal decomposition $W=B^*L^*$ is $\mathcal{O}\left((\kappa(W))^2\frac{n}{s}\right)$.
\end{theorem}
\begin{proof}
The theorem is established by comparing the upper and lower bounds in Lemmas \ref{lem:upper} and \ref{lem:lower}, as follows.

\begin{eqnarray}
 \frac{2\sum^s_{k=1}\lambda^2_k/\epsilon^2}{\frac{s^4}{n} \left(\frac{2^s}{s!}\prod^s_{k=1}\lambda_k\right)^{2/s}/\epsilon^2 }
&\leq& \frac{2\sum^s_{k=1}\lambda_1^2}{ \frac{s^4}{n}\left(\frac{2^s}{s!}\prod^s_{k=1}\lambda_s\right)^{2/s} }\nonumber\\
&\leq& \frac{2ns\lambda^2_1}{\left(\frac{2^s}{s!}\right)^{2/s}
\lambda_s^{2}s^4}\nonumber\\
&=&  \frac{2n\kappa(W)^2}{\left(\frac{2^s}{s!}\right)^{2/s}s^3}\nonumber\\
&\leq& \frac{1}{8} \kappa(W)^2\frac{n}{s} \nonumber
\end{eqnarray}

The last inequality holds due to the fact that $s!<\left(\frac{s}{2}\right)^s$ when $s>5$. Note that all the
inequalities above are tight, and the equalities hold when $\kappa(W)=1$,
i.e. $\lambda_1=\lambda_2=\ldots=\lambda_s$.
\end{proof}

From the theorem above, we draw the following interesting observations. (i) When the rank of the matrix is low (i.e., $s$ is small) and the batch queries are highly correlated ($\kappa(W) \gg 1$), then the ratio of the upper bound to the lower bound is large, meaning that LRM can potentially achieve lower error than NOD. (ii) Conversely, when the rank of the matrix is full rank ($s\rightarrow n~\text{and}~n \leq m$) and the batch queries are almost random or independent ($\kappa(W)\rightarrow 1$), then the achievable error rate of LRM converges to the upper error bound obtained by NOD. Therefore, in this situation, NOD might be good enough and no sophisticated algorithm is needed, which is validated by the experimental results in Section \ref{sec:vary_m}). These results are consistent with the work of \cite{ghosh2012universally}, who show that Laplace mechanism is optimal in a strong sense when answering a single linear query.

\subsection{LRM Error Bounds under ($\epsilon$, $\delta$)-Differential Privacy}\label{sec:appedinex:epsdel}

We first derive an upper bound for the error of LRM. Unlike the case of $\epsilon$-differential privacy, we have a tighter error upper bound than that obtained by naive methods. We introduce the concept of $\rho$-coherence of a matrix, which is similar to $\mu$-coherence \cite{candes2009exact} and $C$-coherence \cite{hardt2012beating} of a matrix in the low-rank optimization literature.

\begin{definition} \textbf{$\rho$-coherence of a matrix.} Given a matrix $W$ with its SVD decomposition that $W=U\Sigma V$, where $U\in \mathbb{R}^{m\times s}, \Sigma\in \mathbb{R}^{s\times s}, V \in \mathbb{R}^{s\times n}$. We say the matrix $W$ is $\rho$-coherent if
$$\rho(W) = \max_{i} \|V_i\|_2,~i=1,...,n$$
where $V_i$ is the $i$-th column of $V$. Note that we have $0<\rho(W)\leq1$.
\end{definition}

\begin{lemma}\label{lem:upper2} \textbf{Error Upper Bound under ($\epsilon$, $\delta$)-differential privacy.} Given a workload matrix $W$ of rank $s$ with singular values $\{\lambda_1,\ldots,\lambda_s\}$, an upper bound of the expected squared error of $M_{LRM, (\epsilon, \delta)}(Q,D)$ w.r.t. the optimal decomposition $W=B^*L^*$  is $(\rho(W))^2\sum^s_{k=1}\lambda_k^2/ h(\epsilon,\delta)^2$.
\end{lemma}

\begin{proof}
To prove the lemma, we perform SVD decomposition of $W$, obtaining $W=U\Sigma V$. Then, we build a decomposition $B=\rho(W) U\Sigma$ and $L=\frac{1}{\rho(W)}V$. This is a valid decomposition of $W$, because $BL=\rho(W) U\Sigma \frac{1}{\rho(W)}V=U\Sigma V=W$.

Next we prove that $\Delta(L)=1$. According to properties of the SVD transformation, column vectors in $V$ are orthogonal vectors; hence, for every column $V_j$ in $V$, we have $\|V_j\|_2 \leq \rho(W)$. Therefore, $\Theta(L) = \max_j \left(\sum_i L_{ij}^2\right)^{1/2} = \max_j \frac{1}{\rho(W)} \|V_j\|_2 = 1$.

The expected squared error of this decomposition is then bounded by:

\begin{eqnarray}
\Phi(B) &=& \mbox{tr}(B^TB)/ h(\epsilon,\delta)^2\nonumber\\
 &=& \mbox{tr}((\rho(W) U\Sigma)^T(\rho(W) U\Sigma)) / h(\epsilon,\delta)^2 \nonumber\\
 &=& \rho(W)^2\mbox{tr}(\Sigma^TU^TU\Sigma))/ h(\epsilon,\delta)^2\nonumber\\
 &=& \rho(W)^2\sum^s_{k=1}\lambda^2_k / h(\epsilon,\delta)^2\nonumber
\end{eqnarray}

We thus reach the conclusion of the lemma.

\end{proof}

Note that since $\rho(W)\leq1$, the above error bound is no worse than the error obtained by NOD. Meanwhile, the proof essentially describes another simple solution whose accuracy is no worse than NOD.

We now focus on the error lower bound of LRM under ($\epsilon$, $\delta$)-differential privacy. This has already been studied in \cite{LiM13}, and we summarize their results with our notations in the following lemma.

\begin{lemma}\label{lem:lower2} \textbf{Error Lower Bound under ($\epsilon$, $\delta$)-differential privacy \cite{LiM13}.} Given a workload matrix $W$ of rank $s$ with \emph{singular values} $\{\lambda_1,\ldots,\lambda_s\}$, the expected squared error of $M_{LRM, (\epsilon, \delta)}(Q,D)$ w.r.t. the optimal decomposition $W=B^*L^*$ is at least
$$ \frac{1}{n h(\epsilon,\delta)^2 }\left( \sum_{i=1}^s \lambda_i \right)^2  $$
\end{lemma}

The proof of the above result in \cite{LiM13} is rather complicated. In the following we provide a simple proof.

\begin{proof}
\begin{eqnarray}
\min_{W=BL,\atop \forall j \sum_i^r L_{ij}^2 \leq 1} \frac{1}{h(\epsilon,\delta)^2}\|B\|_F^2 &\geq& \frac{1}{n h(\epsilon,\delta)^2}\min_{W=BL} \|L\|_F^2\cdot\|B\|_F^2\nonumber\\
&=& \frac{1}{n h(\epsilon,\delta)^2} \left(\|W\|_*\right)^2\nonumber\\
&=& \frac{1}{n h(\epsilon,\delta)^2}\left( \sum_{i=1}^s \lambda_i \right)^2\nonumber
\end{eqnarray}
The first inequality is due to $\sum_{j}^n\left(\sum_i^r L_{ij}^2\right) \leq n$. Note that this inequality above is tight, and the equality holds when every column of $L$ lies on the surface of the unit ball. The first equality is due to the variational formulation of nuclear norm (see, e.g., \cite{srebro2004maximum}) that
$$\|W\|_* = \min_{B,L}~\|L\|_F\cdot||B||_F,~~s.t.~~W=BL.$$
We thus reach the conclusion of the lemma.
\end{proof}

We next compare the error upper bound and the error lower bound for LRM under ($\epsilon$, $\delta$)-differential privacy.

\begin{theorem}\label{the:opt} The ratio between the error upper and lower bounds of mechanism $M_{LRM, (\epsilon, \delta)}(Q,D)$ with the optimal decomposition $W=B^*L^*$ is bounded by $\mathcal{O}\left((\kappa(W))^2\frac{n}{s}\right)$.
\end{theorem}
\begin{proof}

We compare the upper and lower bounds in \ref{lem:upper2} and \ref{lem:lower2}, as follows.

\begin{eqnarray}
\frac{ \rho(W)^2\sum^s_{k=1}\lambda_k^2/ h(\epsilon,\delta)^2 }{ \frac{1}{n  }\left( \sum_{i=1}^s \lambda_i \right)^2 / h(\epsilon,\delta)^2}
&= & \frac{ \rho(W)^2\sum^s_{k=1}\lambda_k^2}{ \frac{1}{n  }\left( \sum_{i=1}^s \lambda_i \right)^2}  \nonumber\\
&\leq& \frac{ s\lambda^2_1 \rho(W)^2}{\frac{1}{n}
\lambda_s^{2}s^2}\nonumber\\
&=&  \left(\kappa(W) \rho(W)\right)^2 \frac{n} {s}\nonumber
\end{eqnarray}

We thus reach the conclusion of the theorem.
\end{proof}

The above theorem leads to similar conclusions as in the case of $\epsilon$-differential privacy, except that here we compare LRM with an improved version of NOD described in the proof of Lemma \ref{lem:upper2}. Meanwhile, the above ratio also involves an additional parameter $\rho$, i.e., the coherence number of the workload matrix.